\documentclass[11pt,a4paper]{article}
\usepackage[margin=2.5cm]{geometry}
\usepackage{amsmath,amssymb,amsthm,mathtools}
\usepackage{booktabs}
\usepackage{array}
\usepackage{longtable}
\usepackage[colorlinks=true,linkcolor=blue,citecolor=blue,urlcolor=blue]{hyperref}

\newtheorem{theorem}{Theorem}[section]
\newtheorem{proposition}[theorem]{Proposition}
\newtheorem{lemma}[theorem]{Lemma}
\newtheorem{corollary}[theorem]{Corollary}
\theoremstyle{definition}
\newtheorem{observation}[theorem]{Observation}
\newtheorem{remark}[theorem]{Remark}
\newtheorem{assumption}[theorem]{Assumption}

\newcommand{\Z}{\mathbb{Z}}
\newcommand{\Q}{\mathbb{Q}}
\newcommand{\C}{\mathbb{C}}
\newcommand{\R}{\mathbb{R}}
\newcommand{\F}{\mathbb{F}}
\newcommand{\coker}{\operatorname{coker}}
\newcommand{\tr}{\operatorname{tr}}
\newcommand{\Hg}{\mathrm{Hg}}
\newcommand{\D}{\mathcal{D}}
\newcommand{\Adj}{\operatorname{Adj}}
\newcommand{\Li}{\operatorname{Li}}
\newcommand{\eff}{\mathrm{eff}}
\newcommand{\Kraw}{K_{\mathrm{raw}}}

\newcommand{\Adec}{\mathcal A_{\mathrm{dec}}}

\title{Haagerup--Izumi modular data from hyperbolic words:\\
$S$-fold TQFTs, central charge, torsion and the Bethe vacua}
\author{David E.\ Evans\\ \small School of Mathematics, Cardiff University,\\ \small Cardiff CF24 4AG, Wales, UK\\ \small \texttt{EvansDE@cardiff.ac.uk}}

\begin{document}
\maketitle

\begin{abstract}
The modular data of the doubles of the Haagerup--Izumi categories is a grafting of two finite metric groups (conjectured by Evans--Gannon, proved for untwisted categories by Gannon--Schopieray--Yadav). Gang and Kim found non-unitary modular data of this shape in the topological twist of rank-$0$ $S$-fold theories, from the Bethe vacua of an abelianised gauge theory. The aim is to understand rigorously the significance of the handle-gluing and fibering operators, and to construct the Bethe vacua.

For every hyperbolic once-punctured-torus monodromy $\varphi$, a word in $L$ and $R$, the two groups are $A_\mp=\coker(\varphi\mp1)$, with the torsion linking forms of the mapping tori of $\varphi$ and $-\varphi$. Over all words of trace $m^2+2$ these forms reproduce the Evans--Gannon table of Q-systems, one printed entry being corrected, and the distinct modular data in the list of Gannon--Schopieray--Yadav. The Bethe data are given in closed form by the word: the level matrix, the positive vacuum, $c_{\eff}=\#L$, and the handle-gluing values $2^{\#L}\lambda^{\pm1}|A_-||A_+|$ at the identity classes, $\lambda$ the stretch factor. Each kernel vector of $\varphi\mp1$ modulo $|A_\mp|$ at which an explicit trigonometric formula does not vanish gives a solution of the Bethe equations in closed form. Under a non-degeneracy condition that holds in every case computed, it is a Bethe class with handle-gluing value $2^{\#L}|A_\mp|$. For odd trace its spin, modulo $\tfrac14$, is given by the linking form. For words without a singleton run these and the two identity classes are all the non-degenerate classes.

Under the half-index dictionary, with the unit in the class of $\sigma V^*$, $c\equiv\#L-6\#R\pmod{24}$. Under the grafting fusion rules and the same non-degeneracy condition, $S$ is, up to sign, the Evans--Gannon grafting for every word of odd trace at least $5$ whose classes are all non-degenerate. For $L^mR^m$ these data are a Galois conjugate of the unitary grafting, which is the modular data of the double of every pseudo-unitary untwisted Haagerup--Izumi category of $\Z_m$ whenever $m^2+4$ is prime. Under the grafting fusion rules they are, up to the sign of $S$, the modular data of the factor $\mathrm{TFT}_0$ of the $S$-fold theory for every prime $m\ge5$ with $m^2+4$ prime ($m=5$, $7$, $13$, $17$, $37$, $47$, \dots), provided $m-1$ classes of the degenerate sector exist as real-analytic solutions. For the other classes of that sector, exact solutions with the closed-form leading data are constructed.
\end{abstract}

\tableofcontents

\section{Introduction}
\label{sec:intro}

\subsection{The problem}
\label{sub:origin}

The Haagerup subfactor \cite{AH99} has an even part with fusion ring $\rho^2=1+\sum_{g\in\Z_3}g\rho$. Izumi constructed it by Cuntz-algebra endomorphisms and solved his equations for $\Z_3$ and $\Z_5$ \cite{Izumi01}. Evans--Gannon found solutions for further odd cyclic groups of small order \cite{EG11}. For odd groups the solutions classify the categories \cite[Thm.~7.3]{Izumi18}. The modular data of the double was computed by Izumi \cite{Izumi01} and written in exact cyclotomic form by Evans--Pinto \cite{EP12} and Evans--Gannon \cite[\S2.3]{EG11}. Evans and Gannon \cite[\S3.3]{EG11} introduced the notion of \emph{grafting} modular data. They defined the family of modular data $\D^\omega\Hg_m$, twisted by $\omega\in\Z_m$, as the graft of affine $so(m^2+4)$ at level $2$ onto a twisted dihedral double \cite[\S3.2]{EG11}, and matched the untwisted member with Izumi's doubles for small $m$. For $\omega=0$, in the formulation of \cite{EG17,EG20}, it is a grafting of two finite metric groups, of
orders $m^2$ and $m^2+4$, which enter the fusion rules with opposite signs. The
untwisted member $\D^0\Hg_m$ grafts $\Z_m\times\Z_m$ and $\Z_{m^2+4}$ \cite[\S5.6]{EG20}, and at
$m=3$ it is the modular data of the double of the Haagerup subfactor. Gannon, Schopieray and Yadav \cite{GSY26} have since constructed an untwisted Haagerup--Izumi
category for every odd $m$, untwisted meaning that the associator on the pointed subcategory is trivial. Huang \cite{Hua26} independently constructed the same category by a different proof; like every solution of the Evans--Gannon equations, it is untwisted \cite[Thm.~2.5]{GSY26}. Gannon--Schopieray--Yadav prove Evans--Gannon's conjectured description of the centre for every untwisted category \cite[Thm.~5.1]{GSY26}, and list all untwisted ones of Q-system type for odd $m\le43$ with their modular data. Neither paper constructs categories realising the twisted modular data $\D^\omega\Hg_m$, $\omega\ne0$.

Gang and Kim \cite{GK22} applied to $S$-fold theories the correspondence between rank-$0$ $\mathcal N=4$ theories and non-unitary TQFTs proposed by Gang, Kim, Lee, Shim and Yamazaki \cite{GKLSY21}. The topological twist of these theories is proposed to give non-unitary
modular data of the same shape, in computable form \cite[(4.1)]{GKLSY21}. The simple objects are the Bethe vacua of an abelianised gauge theory, and
\begin{equation}
S_{0\alpha}^2=\mathcal H_\alpha^{-1},\qquad \frac{T_{\alpha\alpha}}{T_{00}}=\frac{\mathcal F_\alpha}{\mathcal F_0},
\label{eq:dict}
\end{equation}
where $0$ is the unit object and $\mathcal H_\alpha$, $\mathcal F_\alpha$ are the handle-gluing and fibering operators at the vacuum $\alpha$
(Section~\ref{sub:construction}).

A question about the double of a subfactor thus becomes a question about the solutions of a polynomial
system. On the subfactor side it already was one: Izumi's categories, and the near-group categories of \cite{EG14}, are solutions of
polynomial systems in Cuntz-algebra data. The $S$-fold supplies a different system, whose solutions are the simple objects of the double
rather than of the category. The aim of this paper is to understand rigorously the significance of the handle-gluing and fibering operators, and to construct the Bethe vacua.

\subsection{This paper}
\label{sub:wordtrace}

The theories $\mathcal S_k=T[SU(2)]/SU(2)^{\mathrm{diag}}_k$ \cite[\S3.5]{GKLSY21}, \cite{GK22} come from the once-punctured torus bundle with monodromy $ST^k$, which
is conjugate to the word $LR^{k-2}$ in
\[
L=\begin{pmatrix}1&1\\0&1\end{pmatrix},\qquad R=\begin{pmatrix}1&0\\1&1\end{pmatrix}
\]
(Section~\ref{sub:word}); Gang--Kim--Lee \cite{GKL23} add the half-index dictionary, and Jeong--Sohn \cite{JS26} extend the construction to
generalised $S$-folds. We develop the construction for every hyperbolic word $\varphi$ in $L$ and $R$. The two metric groups are
$A_\mp=\coker(\varphi\mp1)$, of orders $|\tr\varphi\mp2|$, and the invariants turn out to be functions of $\#L$, $\#R$ and the run structure of
the word. For $LR^n$ the group $A_-$ is always cyclic (Proposition~\ref{prop:groups}), so this family cannot reach the untwisted series;
$L^mR^m$ does, with $A_-\cong\Z_m\times\Z_m$. At trace $11$, for instance, $LR^9$ and $L^3R^3$ give $\Z_9$ and $\Z_3\times\Z_3$. For $m$ a prime power the $A_-$ sector of $L^mR^m$ is degenerate, and the closed forms for general words do not reach it. It is reached through Part~\ref{part:aminus}, which at prime $m\ge5$ constructs exact solutions with all but $m-1$ of the closed-form leading data.

What makes the closed forms reachable is that the Bethe equations, after inversion of the level matrix, are Nahm's ladder
(Proposition~\ref{prop:bethenahm}). So $c_{\eff}$ is a Rogers dilogarithm sum, and the positive solution is the expanding eigenvector of
$\varphi$ propagated through the word.

\subsection{Main results}
\label{sub:resultslist}

Write $\lambda>1$ for the stretch factor of $\varphi$, $N_\mp=|A_\mp|$, and $(K,Q,c)$ for the level matrix and the two charge vectors that the
layered triangulation of the mapping torus produces (Part~\ref{part:frames}); a \emph{run} is a maximal block of consecutive equal letters of the cyclic word, and a singleton run is a run of length~$1$ (Section~\ref{sub:word}). In (B1) and in the abstract, $c$ is the central charge of \cite{GKL23}, where $q^{-c/24}$ is the prefactor of the unit character \cite[(3.12)--(3.13)]{GKL23}; the half-index itself starts at $q^0$ (Section~\ref{sec:farey24}); it is not the charge vector $c$ of $(K,Q,c)$. The Bethe solutions come in classes of $2^{\#L}$ vacua; $V^*$ is
the positive vacuum and $\sigma V^*$ the all-plus vacuum of its Galois-conjugate class, $\sigma$ the nontrivial automorphism of the real quadratic field $\Q(\lambda)$ (Section~\ref{sec:bethe}).
The results fall into three groups: those proved for every word with no physical input, those proved under the physical inputs of Section~\ref{sub:status}, and those specific to the diagonal family.

\smallskip\noindent\emph{A. Proved for every hyperbolic word, with no physical input.}
\begin{enumerate}
\item[(A1)] $A_-=\coker(\varphi-1)$ carries, up to sign (orientation), the torsion linking form of the mapping torus $M_\varphi$, and $A_+=\coker(\varphi+1)$ that of $M_{-\varphi}$ (Proposition~\ref{prop:linking}, Remark~\ref{rem:lktop}). By a finite computation over all words of trace $m^2+2$, the forms on $A_+$ reproduce Evans--Gannon's table of Q-systems for odd $7\le m\le19$, one printed entry being corrected, and give exactly the distinct modular data in Gannon--Schopieray--Yadav's list for odd $m\le43$ (Observations~\ref{obs:egconj} and~\ref{obs:gsy}).
\item[(A2)] $K$ is given in closed form by the run structure, is positive definite, and is the unique symmetric solution of the gluing
equations with $\det K=2^r$ (Theorem~\ref{thm:Kformula}).
\item[(A3)] The positive vacuum $V^*$ and $\sigma V^*$ are given in closed form by heights along the Farey slopes (Theorem~\ref{prop:propagated}). Moreover $c_{\eff}=\#L$ (Theorem~\ref{thm:ceff}). The frieze relation holds along each run at every solution of the Bethe equations (Theorem~\ref{thm:friezeall}).
\item[(A4)] $\mathcal H(V^*)=2^{\#L}\lambda N_-N_+$ and $\mathcal H(\sigma V^*)=2^{\#L}\lambda^{-1}N_-N_+$ (Theorem~\ref{thm:torsion}).
\item[(A5)] Every vector $v$ in the kernel of $\varphi\mp1$ modulo $N_\mp$ with $f(e)=\sin(\pi\det(v,e)/N_\mp)$ nonzero on the middle slopes of
the triangulation gives a solution of the Bethe equations in closed form, with shapes $f(\text{kept})^2/f(\text{middle})^2$ at $L$-sites and $1$ minus this at $R$-sites (Theorem~\ref{thm:allwords}). Under a non-degeneracy condition that holds in every case computed, it is a Bethe class with $\mathcal H=2^{\#L}N_\mp$ (Theorem~\ref{thm:Hall}).
\item[(A6)] For words without a singleton run, these classes and the two identity classes are all the non-degenerate classes (Theorem~\ref{thm:complete}).
\item[(A7)] For odd trace and a non-degenerate class with $v=\operatorname{adj}(\varphi\mp1)u$, the spin is $\#R/24-\Phi(u)/4N_\mp$ modulo $\tfrac14$ (Theorem~\ref{thm:rogers}), a quadratic refinement of a unit multiple of the torsion linking form. Within a class the spins form a translate of $\tfrac18\varepsilon^{\mathsf T}K_{LL}\varepsilon$ on the sign group (Propositions~\ref{prop:signsum} and~\ref{prop:spinq}).
\end{enumerate}

\smallskip\noindent\emph{B. Proved under the physical inputs.}
\begin{enumerate}
\item[(B1)] Under the half-index dictionary, with the unit at $\sigma V^*$: $c\equiv\#L-6\#R\pmod{24}$ (Theorems~\ref{thm:farey} and~\ref{thm:fareygeneral}). For even trace the unit is fixed only up to a partner, which changes $c$ by $0$ or $12$ (Section~\ref{sub:status}).
\item[(B2)] Under the grafting fusion rules, for every word of odd trace whose nonzero kernel vectors are all non-degenerate, and under the non-degeneracy condition of (A5): $T$ is given by the forms $q_\mp$, the unit is $\sigma V^*$, and $S$ is the Evans--Gannon grafting in the frame $s_o^2=\lambda/N_-N_+$. The overall sign of $S$ is fixed by $c$ (Theorem~\ref{thm:main}, Remark~\ref{rem:Ssign}). Six of the $102$ cyclic words of odd trace at least $5$ and length at most $11$ satisfy the hypothesis; $LR^n$ and $L^mR^m$ with $n,m\ge2$ do not: the $A_-$ sector of $LR^n$ is degenerate (Corollary~\ref{cor:templatelabel}), and so is that of $L^mR^m$ for $m$ a prime power (Proposition~\ref{prop:Hspectrum}). At $m=15$ and $21$ some $A_-$ kernel vectors are non-degenerate but not all (proof of Proposition~\ref{prop:Hspectrum}).
\end{enumerate}

\smallskip\noindent\emph{C. The diagonal family $L^mR^m$, $m$ odd.}
\begin{enumerate}
\item[(C1)] $A_-\cong\Z_m\times\Z_m$ with the hyperbolic form, and $A_+\cong\Z_{m^2+4}$: the metric groups of $\D^0\Hg_m$ (Propositions~\ref{prop:linking} and~\ref{obs:untwisted}).
\item[(C2)] The grafting with the spins $q_\mp$ is the image under $\sigma_{2^{-1}}$ of the unitary grafting of $(A_\mp,q^{\mathrm{lk}})$. This is an identity of modular data and needs no physical input (Corollary~\ref{cor:galois}). Whenever $m^2+4$ is prime ($m=3,5,7,13,15,17,\dots$), the unitary grafting is the modular data of the double of every pseudo-unitary untwisted Haagerup--Izumi category of $\Z_m$ \cite[Thm.~A(2)]{GSY26}.
\item[(C3)] At prime $m\ge5$ the $A_-$ classes are degenerate. The closed-form leading data are labelled bijectively by $(A_-\smallsetminus0)/\!\pm$, with the $\tfrac1m\Z/\Z$-component of the spin equal to $q_-$. The line through $0$ in $A_-\cong\F_m^2$ containing the label, a point of $\mathbb P^1(\F_m)=\F_m\cup\{\infty\}$, is the \emph{type} of the class. All but $m-1$ of them are limits of exact solutions of the deformed Bethe system (Theorems~\ref{thm:lines}, \ref{thm:exist} and~\ref{thm:aminus}). That these are all the classes of the sector rests on (F1) and on an unproved classification of degenerate leading data (Remark~\ref{rem:exhaustive}).
\item[(C4)] Under the grafting fusion rules, for every prime $m\ge5$ with $m^2+4$ prime ($m=5$, $7$, $13$, $17$, $37$, $47$, \dots), the grafting is, up to the sign of $S$, the modular data of $\mathrm{TFT}_0$, provided the $m-1$ classes of types $0$ and $\infty$ exist as real-analytic solutions with real coordinates (Corollary~\ref{cor:diagonal}).
\end{enumerate}

\smallskip\noindent\emph{What remains.} 
\begin{enumerate}
\item[(i)] The degenerate classes of general words. They limit (B2) to the words whose classes are all non-degenerate.
\item[(ii)] The $m-1$ classes of types $0$ and $\infty$ in (C3). Their construction would remove the proviso in (C4).
\item[(iii)] Completeness for words with a singleton run. The isolated case is only sketched (Remark~\ref{rem:singleton}). Exact solution counts on eight such words agree with the count of Theorem~\ref{thm:complete}.
\item[(iv)] The sign of $S_{oo}$. It is fixed by $c$, and it depends on the word: negative for four of the six words of (B2) and positive for two (Remark~\ref{rem:Ssign}). No rule for it is known.
\item[(v)] The $\Z_3$-part of the spins at $m=3$. The relation $24\theta_\alpha\equiv c_\alpha$, with $c_\alpha$ the Rogers sum of the class, does not see it, and at $L^3R^3$ the $A_-$ spins are numerical (Observation~\ref{obs:l3r3minus}).
\item[(vi)] The non-degeneracy condition of Theorem~\ref{thm:Hall} in general; the classification of the degenerate leading data, which would make the labelling in (C3) exhaustive; the corank of Theorem~\ref{thm:corank} at $j\ne1$ for $m>43$; the isolatedness of the exact solutions of (C3) in the full Bethe system (Remark~\ref{rem:aminusgaps}); the vacuum identification for even trace; and the chirality of the semion factor against \cite[(2.24)]{GK22}.
\end{enumerate}

\subsection{What is proved and what is assumed}
\label{sub:status}

The theorems are about the triple $(K,Q,c)$ and the polynomial system it defines; none of them presupposes that a superconformal field
theory exists. Physics enters in the following places. The half-index dictionary of \cite{GKL23} is used in Section~\ref{sec:farey24}.
Localisation is used in Corollary~\ref{cor:interpretation} and in the exhaustiveness statements of Remarks~\ref{rem:exhaustive}
and~\ref{rem:aminusgaps}(iii):
\begin{quote}
\textbf{(F1)} \emph{Localisation.} In the degenerate limit $\nu\to-1$ of the R-symmetry mixing parameter (Section~\ref{sec:bethe}) the twisted-index observables of $\mathcal T_{DGG}[M_\varphi]$ --- the
handle-gluing operator $\mathcal H_\alpha$ and the fibering operator $\mathcal F_\alpha$ of each Bethe vacuum --- localise on flat $PSL(2,\C)$
connections on $M_\varphi$ in the fibre-abelian locus, with $\mathcal H_\alpha$ given there by the adjoint Reidemeister torsion and $\mathcal F_\alpha$
by the Chern--Simons invariant. \emph{And} the rank-$0$ $\mathcal N=4$ SCFT exists for hyperbolic monodromy, with $\mathcal N=2\to\mathcal N=4$
enhancement along the flow.
\end{quote}
\label{sub:F1}
The theory $\mathcal T_{DGG}[M_\varphi]$ is that of \cite{DGG11}. That the twisted index localises on flat connections, with the handle-gluing operator given by the adjoint Reidemeister torsion, is the 3d--3d relation of \cite{GangKim19,GKPZ19}, see \cite[(2)]{GKY19}; the modular data are obtained from such invariants in \cite{CGK20}. The $\mathcal N=4$ enhancement for hyperbolic once-punctured torus bundles is expected in \cite[App.~B]{GKLSY21}. None of this is proved in the degenerate limit; \cite{JS26} also
assume the existence statement. Where this paper compares $\mathcal H$ with the torsion, the comparison at the non-identity
classes holds up to the factor $2^{\#L}$ and a sign (Theorem~\ref{thm:Hall} against Proposition~\ref{prop:jacfixed}, the latter for odd $\tr\varphi$). The two identity classes both lie over the trivial character. There the comparison is with the Floquet discriminant $|A_-||A_+|$ (Theorem~\ref{thm:torsion}, half~(A) in the notation of Section~\ref{sec:torsion}). So $\mathcal H(V^*)$ and $\mathcal H(\sigma V^*)$ are
$2^{\#L}\lambda^{\pm1}$ times it. The factor $\lambda^{\pm1}=W_c^{\pm1}$ (up to sign at $\sigma V^*$) is the Wilson line of $U(1)_A$, whereas $W_c^2=1$ at every non-identity
class (Theorem~\ref{thm:Hall}). The dictionary \eqref{eq:dict} is used in Section~\ref{sec:modular}, where the modular data are assembled. It is also used in four
comparisons, which are flagged where they occur: Remark~\ref{rem:internal} (the test of (b) for $LR^n$), the sum-rule reading of
completeness after Theorem~\ref{thm:complete}, Proposition~\ref{prop:exclusion}, and the identification of $\mathcal F_\alpha/\mathcal F_0$
with the $T$-eigenvalue in Section~\ref{sub:relation}. The grafting fusion rules are used only where Assumption~\ref{ass:F} is cited, chiefly in Section~\ref{sec:modular} and the example $L^3R^3$. The Farey formula needs, besides the half-index dictionary,
\begin{quote}
\textbf{(b)} \emph{The vacuum identification.} The Galois conjugate of $V^*$ lies in the unit class.
\end{quote}
\label{sec:inputb}
It implies that the unit has $\mathcal H=2^{\#L}\lambda^{-1}N_-N_+$ (Remark~\ref{rem:internal}). For $LR^n$ it agrees,
through the dictionary, with $S_{00}=a_0$, the value of $|Z_{S^3_b}|$ in the degenerate limit \cite[(2.11), (3.107)]{GKLSY21} (Remark~\ref{rem:internal}); it follows from modularity
under the hypotheses of Theorem~\ref{thm:main} and of Corollary~\ref{cor:diagonal}; and it is open in general. Input (b) is class-level, which is what Theorem~\ref{thm:torsion} tests and Theorem~\ref{thm:main} derives. The Farey formula
needs the unit vacuum within the class: by the flux rule (Propositions~\ref{prop:flux} and~\ref{prop:fluxall}) it is $\sigma V^*$, uniquely so when
$\tr\varphi$ is odd (Lemma~\ref{lem:parity}), and we take the unit at $\sigma V^*$ throughout. Gang--Kim--Lee note a two-unit ambiguity for $LR^n$ \cite[footnote~4]{GKL23}; for odd trace the flux rule removes it. For even trace the unit is fixed only up to a
$\Z_2$ partner. An alternative unit differs from it by a sign flip of some charge-$2$ coordinates (Proposition~\ref{prop:spinq}), so its spin differs by an element of $\tfrac12\Z$, and $c$ changes by $0$ or $12$ modulo $24$; both occur.

\subsection{Relation to the literature}
\label{sub:literature}

The categories are Izumi's \cite{Izumi01,Izumi18}. The twisted family of modular data and grafting are Evans--Gannon's
\cite{EG11}; that every double is a grafting is their conjecture \cite[Conj.~1]{EG17}. Grossman--Izumi \cite{GI19} place the categories among quadratic categories. Gannon--Schopieray--Yadav \cite{GSY26} give the uniform construction of untwisted categories. For every untwisted category they prove Evans--Gannon's description of the centre's
modular data for the block of order $m^2+4$, and they classify those of Q-system type for odd $m\le43$. Huang \cite{Hua26} gives another untwisted construction.
The $F$-symbols of the transparent cyclic categories are Huang--Lin's \cite{HL20}. The identification of the metric groups of
$\D^0\Hg_m$ as $\Z_m\times\Z_m$ and $\Z_{m^2+4}$ is Evans--Gannon's \cite[\S5.6]{EG20}; result~(C1) produces the same groups from the words $L^mR^m$. That the $S$-fold realises the untwisted member, up to the Galois conjugation $\mathrm{Gal}_2$, was found by Jeong--Sohn \cite[Table~1]{JS26}, on proposed modular data, for
the $S$-folds $\vec k=(m,-m)$, which are $L^mR^m$ up to sign and inversion (Section~\ref{sub:diagonal}).

The correspondence between rank-$0$ SCFTs and non-unitary TQFTs, with its dictionary through the handle-gluing and fibering operators at the Bethe vacua, is due to Gang--Kim--Lee--Shim--Yamazaki \cite[\S2.2, Table~1]{GKLSY21}; see also \cite[Table~1]{GKL23} and \cite[(2.13)--(2.14)]{GK22}. Gang--Kim \cite{GK22} apply it to $S$-folds, and Gang--Kim--Lee \cite{GKL23} add the half-index dictionary; Jeong--Lee \cite{JL25}
treat related Bethe systems in a different gauging. In \cite{GK22} the full $S$-matrix is proposed from $|S_{0\alpha}|$ and the
spins and checked against the modular relations; \cite{GKL23} use the same $S$. The number of Bethe vacua of $\mathcal S_k$ is stated in \cite[(3.111)]{GKLSY21} and \cite[(2.11)]{GK22}. In \cite{JS26} and \cite{JL25} the vacua are found by ans\"atze for particular families. Here the non-degenerate classes of every word are given in closed form and shown complete for words without a singleton run,
and $S$ is derived. The group $A_-$ with its linking form is the group on which Radchenko--Wheeler \cite{RW26} define their finite quantum
dilogarithm; the companion $A_+$ does not appear there. That an ideal triangulation, a dilogarithm identity and a central charge sit in
one picture is Nahm's \cite{Nahm94}, for Dehn fillings of the figure-eight knot complement. The arithmetic of the two metric groups is the subject of \cite{EvB26}; this paper does not depend on it.

\subsection{Organisation}

Part~\ref{part:setting} recalls the $S$-fold TQFTs, puts linking forms on $A_\mp$ and contains (A1) and (C1). Part~\ref{part:frames} builds $(K,Q,c)$,
proves (A2) and the closed form in (A3), and reads the Bethe system in three frames. Part~\ref{part:results} proves $c_{\eff}=\#L$, (A4)--(A7) and (B1)
for every word, and Part~\ref{part:aminus} proves (C3). Part~\ref{part:modular} derives (B2), (C2) and (C4) and works the example $L^3R^3$; Appendices~\ref{app:Kformula}--\ref{app:corank} contain the longer technical proofs, in the order in which they are used, and
Appendix~\ref{sec:repro} keys each computational statement to what was computed.
Each part and section opens with an overview. A reader interested only in the general-word theorems can skip Part~\ref{part:aminus}.

\part{Setting and the two metric groups}
\label{part:setting}

This part fixes the setting and the two invariants of the word on which the rest of the paper is built. Section~\ref{sec:setup} recalls the
$S$-fold TQFTs of Gang--Kim \cite{GK22,GKL23}. It then replaces the trace $k$ by the word
$\varphi$, and the moduli $k\mp2$ by the groups $A_\mp=\coker(\varphi\mp1)$, whose structure depends on the word and not only on the trace.
Section~\ref{sec:metric} puts linking forms on $A_\mp$, shows that for the diagonal words $L^mR^m$ they give the metric groups of the untwisted Haagerup--Izumi doubles,
and compares the forms produced by words with the Haagerup--Izumi literature \cite{EG17,GSY26,EG20}. The dictionary and the data of
$\mathrm{TFT}_0$ are taken from \cite{GK22,GKL23} and are not proved here; Section~\ref{sub:status} states the physical inputs and where they
are used. Everything else in this part is proved, except Observations~\ref{obs:egconj} and~\ref{obs:gsy}, which are enumerations. The later
parts use from it the dictionary, \eqref{eq:GKT} and \eqref{eq:as}, the groups $A_\mp$, $\lambda(L^mR^m)=d_\rho^2$, and
Propositions~\ref{prop:linking} and~\ref{obs:untwisted}.

\section{The S-fold TQFTs and the monodromy word}
\label{sec:setup}

Section~\ref{sub:construction} states
the rank-$(k+1)$ datum $\mathrm{TFT}_0$ for odd $k$, whose $T$-matrix \eqref{eq:GKT} has two blocks, of moduli $k-2$ and $k+2$.
Section~\ref{sub:word} defines a word and the two groups $A_\mp$, of orders $|\tr\varphi\mp2|$. For example $LR^9$ and $L^3R^3$ both have trace $11$, with $A_-\cong\Z_9$
and $\Z_3\times\Z_3$. Section~\ref{sub:diagonal} identifies the diagonal words $L^mR^m$ with the generalised $S$-folds of Jeong--Sohn \cite{JS26}
at $\vec k=(m,-m)$, up to sign and inversion. It also shows that $\lambda(L^mR^m)=d_\rho^2$, where $d_\rho$ is the quantum dimension of $\rho$ in the
Haagerup--Izumi category. Remark~\ref{rem:js} compares the groups $A_\mp$ with Jeong--Sohn's labels: by gauge node for $n=2$, and by the cokernel of $K_\pm[\vec k]$ for $n\ge3$, which is $A_\mp$ up to the parity sign.

\subsection{The construction}
\label{sub:construction}

For the SCFT $\mathcal S_k$ of Section~\ref{sub:wordtrace}, the two degenerate limits of \cite{GKLSY21} yield non-unitary TQFTs
$\mathrm{TFT}_\pm[\mathcal S_k]$ \cite[\S3.5]{GKLSY21}, \cite{GK22,GKL23}; only $\mathrm{TFT}_-$ is used here. Its $T$ and $|S_{0\alpha}|$ are given by the dictionary \eqref{eq:dict} of \cite[Table~1]{GKLSY21}, with $\mathcal H_\alpha$ and $\mathcal F_\alpha$ evaluated on Bethe vacua; the full $S$ is proposed in \cite{GK22}. The modular data of
$\mathrm{TFT}_-[\mathcal S_k]$ has rank $2k+2$: there are $2k+2$ Bethe vacua, at which $\mathcal H_\alpha$ takes the values
$a_0^{-2},a_1^{-2},a_2^{-2},a_3^{-2}$ with multiplicities $2,k-3,k+1,2$ \cite[(3.111)--(3.113)]{GKLSY21}, \cite[(2.11)--(2.12)]{GK22}, where
\begin{equation}
(a_0,a_1,a_2,a_3)=\Bigl(\tfrac1{\sqrt{8(k-2)}}+\tfrac1{\sqrt{8(k+2)}},\ \tfrac1{\sqrt{2(k-2)}},\ \tfrac1{\sqrt{2(k+2)}},\
\tfrac1{\sqrt{8(k-2)}}-\tfrac1{\sqrt{8(k+2)}}\Bigr),
\label{eq:as}
\end{equation}
so that by \eqref{eq:dict} the $a_i$ are the moduli of the vacuum row $S_{0\alpha}$. After splitting off the decoupled semion theory $U(1)_{\pm2}$,
generated by the anyon of spin $\pm\tfrac14$, one obtains, for odd $k$, a rank-$(k+1)$ datum $\mathrm{TFT}_0$ \cite{GK22} with
\begin{equation}
T=e^{-2\pi ic/24}\,\mathrm{diag}\Bigl[\exp2\pi i\Bigl\{0,0,\tfrac{A^2}{k-2}\Big|_{A=1}^{(k-3)/2},\ \tfrac{B^2}{k+2}\Big|_{B=1}^{(k+1)/2}\Bigr\}\Bigr].
\label{eq:GKT}
\end{equation}
The prefactor $e^{-2\pi ic/24}$ is that of \cite[(2.29)]{GKL23}; \cite[(2.25)]{GK22} omits it. As a representation of $SL(2,\Z)$ the datum decomposes as $\rho_{(k-1)/2}\oplus\rho_{(k+3)/2}$ \cite[(2.28)]{GK22}. Here $\rho_d$ is the Weil representation of the metric cyclic group
$(\Z_{2d-1},q)$, $q(a)=a^2/(2d-1)$, restricted to charge-conjugation-even vectors. The two moduli are $k-2$ and $k+2$, and
$\operatorname{rank}=\tfrac{k-1}2+\tfrac{k+3}2=k+1$. The group structure does not enter as a group of invertible objects --- the semion having been split off, $\mathrm{TFT}_0$ has none besides the
unit --- but as discriminant-form data.

For a general word, $\mathcal T_{DGG}[M_\varphi]$ is the theory of \cite{DGG11} for the mapping torus $M_\varphi$, as used for $LR^n$ and $L^2R^n$
in \cite{JS26}, and $\mathrm{TFT}[\varphi]$ is its topological twist. Its existence as an $\mathcal N=4$ SCFT is part of (F1)
(Section~\ref{sub:F1}), and $\mathrm{TFT}_0$ for a general word is set up in Section~\ref{sec:modular}.

\subsection{The word}
\label{sub:word}

$\mathcal S_k$ arises from the once-punctured torus bundle with monodromy $ST^k$, where, as in \cite[(A.1), (A.3)]{GKL23} and \cite[(2.2)]{JS26},
\begin{equation}
S=\begin{pmatrix}0&1\\-1&0\end{pmatrix},\qquad T=R=\begin{pmatrix}1&0\\1&1\end{pmatrix},\qquad L=\begin{pmatrix}1&1\\0&1\end{pmatrix}.
\end{equation}
Then $\tr ST^k=k$, and $ST^k$ is conjugate in $SL(2,\Z)$ to $LR^{k-2}$ \cite[(A.3)]{GKL23}. We take
$\varphi=LR^{k-2}$, for which the groups $A_-$ and $A_+$ below have orders $k-2$ and $k+2$, the moduli of the two blocks of \eqref{eq:GKT}.
Throughout, a \emph{word} is a cyclic word in $L,R$ containing at least one of each, so that $\varphi$ is hyperbolic; a \emph{run} is a maximal block of consecutive equal letters of the cyclic word, so that $L$-runs and $R$-runs alternate and are equal in number, $n_R$ say; $r=\#L+\#R$ is its length, and
$\lambda=\lambda(\varphi)>1$ the stretch factor.

\begin{proposition}
\label{prop:groups}
For a once-punctured torus bundle with monodromy $\varphi$, set $A_\mp=\coker(\varphi\mp1)$, so that $|A_\mp|=|\tr\varphi\mp2|$ since
$\det(\varphi\mp1)=2\mp\tr\varphi$. For $\varphi=L^aR^b=\begin{psmallmatrix}1+ab&a\\b&1\end{psmallmatrix}$ one has $\tr\varphi=ab+2$ and
\[
A_-\cong\Z_g\times\Z_{ab/g},\qquad g=\gcd(a,b).
\]
\end{proposition}

\begin{proof}
$\varphi-1=\begin{psmallmatrix}ab&a\\b&0\end{psmallmatrix}$ has determinant $-ab$, and the greatest common divisor of its entries is $g$; these
are the product and the first of its two invariant factors.
\end{proof}

So $A_-$ is forced cyclic when $a=1$: the Gang--Kim family sees only cyclic groups, and this is a property of the word $LR^n$, not of the
construction. The structure of $A_+$ is likewise not determined by the trace:

\begin{center}\small
\begin{tabular}{lrlll}
\toprule
word & $\tr\varphi$ & $A_-$ & $A_+$ & ranks\\
\midrule
$LR$        & 3  & $1$              & $\Z_5$              & (0,1)\\
$LR^2$      & 4  & $\Z_2$           & $\Z_6$              & (1,1)\\
$LR^3$      & 5  & $\Z_3$           & $\Z_7$              & (1,1)\\
$LR^9$      & 11 & $\Z_9$           & $\Z_{13}$           & (1,1)\\
$L^2R^2$    & 6  & $\Z_2\times\Z_2$ & $\Z_2\times\Z_4$    & (2,2)\\
$L^2R^3$    & 8  & $\Z_6$           & $\Z_{10}$           & (1,1)\\
$L^2R^4$    & 10 & $\Z_2\times\Z_4$ & $\Z_2\times\Z_6$    & (2,2)\\
$L^3R^2$    & 8  & $\Z_6$           & $\Z_{10}$           & (1,1)\\
$\mathbf{L^3R^3}$ & \textbf{11} & $\mathbf{\Z_3\times\Z_3}$ & $\mathbf{\Z_{13}}$ & \textbf{(2,1)}\\
$L^3R^4$    & 14 & $\Z_{12}$        & $\Z_{16}$           & (1,1)\\
$L^4R^4$    & 18 & $\Z_4\times\Z_4$ & $\Z_2\times\Z_{10}$ & (2,2)\\
\bottomrule
\end{tabular}
\end{center}

\noindent The last column gives the numbers of generators of $A_-$ and $A_+$. $LR^9$ and $L^3R^3$ have the same trace, hence the same $k$, but lie in different conjugacy classes and give $\Z_9$ against
$\Z_3\times\Z_3$.

\subsection{The diagonal family, and where the dimension comes from}
\label{sub:diagonal}

The diagonal words $L^mR^m$ are distinguished. With $S$ and $T=R$ as in Section~\ref{sub:word} one checks
\[
-\,ST^mST^{-m}=\begin{psmallmatrix}m^2+1&-m\\-m&1\end{psmallmatrix}=L^{-m}R^{-m}=(R^mL^m)^{-1},\qquad R^mL^m=L^{-m}\bigl(L^mR^m\bigr)L^m,
\]
so $L^mR^m$ is, up to sign and inversion, the generalised $S$-fold with $\vec k=(m,-m)$ (the orderings $(m,-m)$ and $(-m,m)$ are cyclic permutations, hence conjugate). For odd $m$ this is the case $k_1,k_2$ odd treated in
\cite{JS26}. Inversion changes neither cokernel; the sign swaps the roles of $1\mp\varphi$.

The fixed point of $L^mR^m$ has continued fraction $[m;m,m,\dots]=d=\tfrac12(m+\sqrt{m^2+4})$, the root of $x^2=mx+1$, and
$\begin{psmallmatrix}m&1\\1&0\end{psmallmatrix}^2=L^mR^m$ exactly, so
\begin{equation}
\lambda(L^mR^m)=d^2 .
\label{eq:lamd}
\end{equation}
On the categorical side the Haagerup--Izumi fusion rule $\rho^2=1+\sum_gg\rho$ gives $d_\rho^2=1+m\,d_\rho$, so $d_\rho=d$: \emph{the stretch factor of
the monodromy is the square of the quantum dimension of $\rho$}. The $\lambda$ in the torsion formula of Section~\ref{sec:torsion} is thus
$d_\rho^2$. On the three-manifold side the same periodic continued fraction is the Farey cutting sequence of the axis of $L^mR^m$, i.e.\ the layered triangulation
\cite{FH82,Gueritaud06} from which $K$ is built.

\begin{remark}[Comparison with Jeong--Sohn's labels]
\label{rem:js}
Jeong--Sohn \cite{JS26} do not introduce $A_\pm$ or a linking form; they label vacua as their ansatz produces them, for $n=2$ by $\alpha_i$ (one per gauge node)
with spins $e^{2\pi i(\alpha_1^2/k_1+\alpha_2^2/k_2)}$, and by $\beta$ on a cyclic group. For $n=2$ with $k_1,k_2$ odd they also propose the full $S$ and $T$ \cite[(4.27)--(4.32)]{JS26}, checked
against the $SL(2,\Z)$ relations, a positive row and integrality of the Verlinde coefficients. For $\varphi=ST^{k_1}ST^{k_2}$ one has
$\coker(1+\varphi)\cong\Z_{k_1}\times\Z_{k_2}$ and, for $k_1,k_2$ odd, $\coker(1-\varphi)$ cyclic of order $|k_1k_2-4|$, so for $n=2$ the labels by gauge node
\emph{are} the invariant groups. At $\vec k=(m,-m)$ the $\alpha$-form is $(\alpha_1^2-\alpha_2^2)/m$ on $\coker(1+\varphi)$, which is $A_-$ of $L^mR^m$ by the sign in
Section~\ref{sub:diagonal}. This is the hyperbolic form of Proposition~\ref{prop:linking}. For $n\ge3$ Jeong--Sohn label the vacua by $\Z^n/K_\pm[\vec k]$, with $K_\pm[\vec k]$ the cyclic tridiagonal matrix of \cite[(A.13)--(A.14)]{JS26}, $\det K_\pm[\vec k]=\tr\varphi\mp2(-1)^n$, and spins from the quadratic phase in \cite[(A.16)]{JS26}.
A computation for $n=3,4$ and $k_i\in\{\pm3,\pm5,7\}$ ($1500$ cases) gives $\coker K_\pm[\vec k]\cong\coker\bigl(1\mp(-1)^n\varphi\bigr)$ in every case. So their labels are the groups $A_\mp$, up to the parity sign $(-1)^n$.
What is new here is the linking form on $A_\mp$ for every word, defined without a presentation, and the completeness of the classes for words without a singleton run (Theorem~\ref{thm:complete}).
\end{remark}

\section{The two metric groups}
\label{sec:metric}

This section puts quadratic forms on the two groups $A_\mp$ of Section~\ref{sub:word} and compares them with the Haagerup--Izumi literature.
Proposition~\ref{prop:linking} shows that $b(x,y)=x\wedge(1-\varphi)^{-1}y$, where $x\wedge y=x_1y_2-x_2y_1$ is the determinant pairing on $\Z^2$ and $x,y\in\Z^2$ represent classes of $A_-$, is a symmetric linking form on $A_-$ for every hyperbolic word.
For $L^mR^m$ with $m$ odd, it is the hyperbolic form on $\Z_m^2$, with Gauss sum $+m$. The same construction with $1+\varphi$ gives a form on
$A_+\cong\Z_{m^2+4}$, with Gauss sum $-\sqrt{m^2+4}$. These forms $q^{\mathrm{lk}}$ are the normalisation of the Haagerup--Izumi literature;
the spin forms of Section~\ref{sub:spinform} are $2^{-1}q^{\mathrm{lk}}$, a Galois conjugate (Corollary~\ref{cor:galois}). Proposition~\ref{obs:untwisted} shows that the form on $A_-$ is the discriminant form of $\D(\Z_m)$, which is the metric group of the untwisted
member $\omega=0$ of the Evans--Gannon family $\D^\omega\Hg_m$ \cite[\S3.2]{EG11}, \cite{EG20}, and that the family $LR^{k-2}$ misses it ($L^3R^3$
against $LR^9$, both of trace $11$; Remark~\ref{rem:twisted}). The rest of the section compares $(A_+,q)$ with the literature. Observation~\ref{obs:egconj} compares it with the Q-system
table of \cite[Prop.~3]{EG17} at odd $7\le m\le19$, where all ten entries are reproduced once one misprint is corrected.
Observation~\ref{obs:gsy} compares it with the classification of \cite{GSY26} at odd $m\le43$, where the distinct forms produced by words are
exactly their distinct modular data. Only
Proposition~\ref{prop:linking}, Remark~\ref{rem:lktop} and Proposition~\ref{obs:untwisted} are used later.

\begin{proposition}[The linking form]
\label{prop:linking}
For every hyperbolic $\varphi\in SL(2,\Z)$ the formula $b(x,y)=x\wedge(1-\varphi)^{-1}y\bmod\Z$, with $x\wedge y=x_1y_2-x_2y_1$ the determinant pairing on $\Z^2$ and $x,y\in\Z^2$ representatives of classes of $A_-=\coker(1-\varphi)$ (a quotient of $\Z^2$, so with at most two generators), is a well-defined, symmetric and non-degenerate $\Q/\Z$-valued pairing on $A_-$. For $\varphi=L^mR^m$ with $m$ odd,
$A_-\cong\Z_m^2$ on the images of $e_1,e_2$ and
\[
b(x,y)=\frac{x_2y_2-x_1y_1}{m};
\]
this form is hyperbolic, and its quadratic refinement $q^{\mathrm{lk}}(x)=2^{-1}b(x,x)$, with $2^{-1}$ the inverse of $2$ modulo $m$, has Gauss sum
\[
\sum_{x\in A_-}e^{2\pi iq^{\mathrm{lk}}(x)}=+m .
\]
The same construction with $1+\varphi$ gives a symmetric linking form $b_+(x,y)=x\wedge(1+\varphi)^{-1}y$ on $A_+$, with refinement
$q^{\mathrm{lk}}(x)=2^{-1}b_+(x,x)$ when $N_+$ is odd. For $\varphi=L^mR^m$ with $m$ odd, $A_+\cong\Z_\mu$, $\mu=m^2+4$, on the image of $e_1$, with
$q^{\mathrm{lk}}(h)=sh^2/\mu$, $s\equiv-2^{-1}m\pmod\mu$, and Gauss sum $-\sqrt\mu$.
\end{proposition}

\begin{proof}
For $x,y\in\Z^2$ the identity $M^{\mathsf T}J=J\operatorname{adj}M$ for
$2\times2$ matrices ($J$ the matrix of $\wedge$) turns $b(x,y)-b(y,x)$ into
$x\wedge\bigl[(1-\varphi)^{-1}+(1-\varphi)/\det(1-\varphi)\bigr]y$; since $A+\operatorname{adj}A=(\tr A)\mathbf1$ and $\tr(1-\varphi)=\det(1-\varphi)$
on $SL(2,\Z)$, the bracket is $\mathbf1$ and the difference is $x\wedge y\in\Z$, so $b$ is symmetric modulo $\Z$. It depends only on the class of $y$, since
$b(x,(1-\varphi)y')=x\wedge y'\in\Z$, and by symmetry only on the class of $x$. It is nondegenerate: write $M=1-\varphi$; since $x\wedge M^{-1}y=(Mx)\wedge y/\det M$, if $b(x,y)\in\Z$
for all $y$ then $Mx\in(\det M)\Z^2$, so $x\in\operatorname{adj}(M)\Z^2=(1-\varphi^{-1})\Z^2=(1-\varphi)\Z^2$, that is $x=0$ in $A_-$. For
$L^mR^m=\begin{psmallmatrix}m^2+1&m\\m&1\end{psmallmatrix}$, $(1-\varphi)^{-1}=\begin{psmallmatrix}0&-1/m\\-1/m&1\end{psmallmatrix}$, giving
$b=\operatorname{diag}(-1,1)/m$ modulo integers; and $x_2^2-x_1^2=(x_2-x_1)(x_2+x_1)$ is a hyperbolic plane over $\Z_m$, whose Gauss sum is $+m$.

For $A_+$ the same argument applies, since $\tr(1+\varphi)=\det(1+\varphi)$ and $\operatorname{adj}(1+\varphi)=\varphi^{-1}(1+\varphi)$. For $L^mR^m$ the entries of
$1+\varphi=\begin{psmallmatrix}m^2+2&m\\m&2\end{psmallmatrix}$ have greatest common divisor $1$, so $A_+$ is cyclic of order $\mu$; and
$(1+\varphi)^{-1}e_1=\mu^{-1}(2,-m)^{\mathsf T}$ has exact denominator $\mu$, so $e_1$ generates, with $b_+(e_1,e_1)=-m/\mu$. As $\mu\equiv1\pmod4$, the
Gauss sum of $sh^2/\mu$ is $(s/\mu)\sqrt\mu$, and $(s/\mu)=(-2/\mu)(m/\mu)=-1$, because $\mu\equiv5\pmod8$ and $(m/\mu)=(\mu/m)=(4/m)=1$.
\end{proof}

\begin{remark}[The topological linking forms]
\label{rem:lktop}
Up to sign, $b$ is the torsion linking form of the mapping torus $M_\varphi$ of $\varphi$ on the once-punctured torus $\Sigma$. Here
$H_1(M_\varphi)=\Z\oplus\coker(\varphi-1)$, the torsion being carried by curves in a fibre. For linking forms in general see
\cite{CFH16}. Let $x,y\in H_1(\Sigma)=\Z^2$ and $N=|\det(\varphi-1)|$, and put $w=N(\varphi-1)^{-1}x\in\Z^2$.
The cylinder swept by $w$ from the fibre at time $0$ to the fibre at time $1$ has boundary $\varphi(w)-w=Nx$. So it is a $2$-chain $S$
with $\partial S=Nx$. A curve representing $y$ in the fibre at time $\tfrac12$ meets $S$ in the intersection points of $w$ and $y$ in
$\Sigma$. Hence $\operatorname{lk}(x,y)=S\cdot y/N=\pm w\wedge y/N=\pm\bigl((\varphi-1)^{-1}x\bigr)\wedge y$, which is $b$ up to sign. The
elliptic involution $-1$ preserves $\Sigma$, and $\coker(-\varphi-1)=\coker(\varphi+1)$. So the same computation identifies $b_+$, up to
sign, with the torsion linking form of $M_{-\varphi}$.
\end{remark}

\emph{The two normalisations.} On
$A_\mp$ one has $b_\mp(u,u)=-\Phi(u)/N_\mp$ (with $b_-=b$, and $\Phi$ the binary quadratic form of $\varphi$, defined in Section~\ref{sub:spinform}), so
$q_\mp=2^{-1}q^{\mathrm{lk}}$. On a cyclic group of odd order $N$ the two are isometric if and only if $2$ is a square modulo $N$; on
$A_-\cong\Z_m\times\Z_m$ of $L^mR^m$ the hyperbolic form is isometric to all its unit multiples. For $L^mR^m$ with $m$ odd,
$N_+=m^2+4\equiv5\pmod8$, so $2$ is not a square modulo $N_+$ and the two forms on $A_+$ are not isometric.

Two consequences. Since the Gauss sum $+m$ is real and positive, the $A_-$ factor contributes no phase to a product of Gauss sums over $A_-$ and $A_+$.
And a hyperbolic form is invariant under $q\mapsto-q$, which is consistent with the amphichirality of the diagonal words
(Proposition~\ref{prop:closed}).

\begin{proposition}[The metric group of the diagonal words is the untwisted one]
\label{obs:untwisted}
Let $m$ be odd.
\begin{enumerate}
\item[(i)] $(A_-,q^{\mathrm{lk}})$ at $L^mR^m$ is the discriminant form of $\D(\Z_m)$.
\item[(ii)] Every quadratic form on $\Z_m\times\Z_m$ takes values in $\tfrac1m\Z/\Z$, whereas a non-degenerate quadratic form on $\Z_{m^2}$ takes a value of order $m^2$.
So no metric group on $\Z_m\times\Z_m$ is isometric, or Galois conjugate, to one on $\Z_{m^2}$.
\end{enumerate}
\end{proposition}

\begin{proof}
Part (i) is Proposition~\ref{prop:linking}: the hyperbolic form on $\Z_m^2$ is the discriminant form of $\D(\Z_m)$.
For (ii), on $\Z_m\times\Z_m$ the bilinear form $b(x,y)=q(x+y)-q(x)-q(y)$ satisfies $mb=0$. So $2q(x)=b(x,x)\in\tfrac1m\Z$, and $q(mx)=m^2q(x)=0$ gives $q(x)\in\tfrac1{m^2}\Z$.
Since $m$ is odd, $q(x)\in\tfrac1m\Z$. On $\Z_{m^2}$ write $q(1)=s/m^2$ modulo $\Z$; here $s\in\Z$, since $2s=m^2b(1,1)\in\Z$ and $m^2s=q(m^2)\in\Z$ with $m$ odd. Non-degeneracy makes $b(1,1)=2s/m^2$ of order $m^2$, so $q(1)$ has order $m^2$.
Galois conjugation preserves the order of a root of unity.
\end{proof}

\begin{remark}[The twisted members]
\label{rem:twisted}
So $L^mR^m$ gives the metric groups of the $\omega=0$ member of the Evans--Gannon twisted family $\D^\omega\Hg_m$ of modular data \cite[\S3.2]{EG11}, namely $\Z_m\times\Z_m$ with the hyperbolic form and
$\Z_{m^2+4}$, whose group structure is seen in the fusion rules \cite[\S5.6]{EG20}. The block of \eqref{eq:GKT} with $k-2=9$ has spins $A^2/9$, which include primitive ninth roots of unity. By Proposition~\ref{obs:untwisted}(ii) it is not a Galois conjugate of a metric group on $\Z_3\times\Z_3$. So the folded $T$-spectra of $L^3R^3$ and $LR^9$ are
inequivalent. Gang--Kim's family at $k=4n^2+4n+3$ is a Galois conjugate of $\D^{\omega=2}\Hg_{2n+1}$ \cite[after (2.26)]{GK22}, \cite[(2.32)]{GKL23}, and realises
only $\omega=2$, so the untwisted Haagerup double is precisely the member it misses. The sentence of \cite{GKL23} after (2.32) recalls that $\D^{\omega=0}\Hg_3$ is the original Haagerup double. It is not an identification: \cite[\S4]{GKL23} place their family at $\omega=\pm2$, in line with (2.32) and with Proposition~\ref{obs:untwisted}(ii), which separates $\Z_9$ from $\Z_3\times\Z_3$. That the family misses $\omega=0$ is consistent with \cite[Table~1]{JS26} once its first entry is corrected; that entry
$(p;k_1,k_2)=(2;-3,-3)$ should read $(2;-3,3)$: the former has trace $7$ and groups of orders $9$ and $5$, the latter trace $-11$ and orders $9$
and $13$, as the pattern $(2;-2n-1,2n+1)$ of that table requires.
Each identification in that table holds up to a Galois conjugation $\mathrm{Gal}_p$, with $p=2$ for the untwisted rows. The table also matches two twisted members, both with $p=31$: $\D^{3}\Hg_9$ at $(k_1,k_2)=(-3,27)$ and $\D^{2}\Hg_9$ at $(5,17)$ \cite[Table~1]{JS26}.
\end{remark}

\begin{observation}[The $A_+$ sector constructs Evans--Gannon's $(H,\beta)$]
\label{obs:egconj}
Evans--Gannon \cite{EG17}, Conjecture~1, predict that the double of a Haagerup--Izumi fusion category
$\mathcal C(G;\pm,\omega,A)$ (here $\omega$ is a cube root of unity, not the twist of $\D^\omega\Hg_m$) with $|G|=m$ odd has its $\tfrac12(m^2+3)$ primaries of class~(iv)
parametrised by pairs $\pm l$, $l\ne0$, in an abelian group $H$ of order $\mu=m^2+4$ carrying a nondegenerate
bilinear form $\beta$, with
\[
w_l=e^{2\pi i\,\mathsf m\,\beta(l,l)},\qquad F_{l,l'}=\mp\tfrac2{\sqrt\mu}\cos\bigl(2\pi\beta(l,l')\bigr),
\]
and their Proposition~3 tabulates $(H,\beta)$ for the ten Q-systems known at odd $7\le m\le19$. Here $\mathsf m=(\mu-1)/2\equiv-\tfrac12\pmod\mu$ is the $m$ of \cite{EG17}, not $|G|$, which they write $\nu$. In the notation of \cite{GSY26} the twists are $e^{2\pi iq(l)}$ with $q=\mathsf m\beta$, the normalisation used in Observation~\ref{obs:gsy}. The conjecture, proved in \cite[Thm.~5.1]{GSY26}, asserts that such an $H$ and
$\beta$ exist; it does not say what they are. The pair $(A_+,b_+)$ of this paper constructs them from the word. Enumerating \emph{all}
conjugacy classes of trace $m^2+2$ and computing $(A_-,A_+,b_\mp)$ for each gives, for odd $7\le m\le19$:
\begin{center}
\begin{tabular}{ccccl}
\toprule
$m$ & trace & classes & distinct $(H,b_+)$ & realised \\
\midrule
$7$  & $51$  & $7$  & $1$ & $\Z_{53},\,ll'/53$ \\
$9$  & $83$  & $18$ & $2$ & $\Z_{85},\,ll'/85$ and $\Z_{85},\,12ll'/85$\\
$11$ & $123$ & $14$ & $2$ & $\Z_{125},\,ll'/125$ and $\Z_5{\times}\Z_{25},\,l_1l_1'/25+l_2l_2'/5$\\
$13$ & $171$ & $13$ & $1$ & $\Z_{173},\,ll'/173$\\
$15$ & $227$ & $45$ & $1$ & $\Z_{229},\,ll'/229$\\
$17$ & $291$ & $17$ & $1$ & $\Z_{293},\,ll'/293$\\
$19$ & $363$ & $38$ & $2$ & $\Z_{365},\,ll'/365$ and $\Z_{365},\,22ll'/365$\\
\bottomrule
\end{tabular}
\end{center}
Nine of the ten entries of their Proposition~3 are reproduced as printed, group and form, up to the choice of generator. (Their labels
$j(m)$ and $j(m)'$ denote the first and second solutions at $m$.) The diagonal word $L^mR^m$ always gives the $j(m)$ entry, and their second solutions at $m=9,19$ are other words of the same trace (the classes split
$9+9$ and $19+19$ between the two forms, $\mathrm h(\Q(\sqrt{m^2+4}))=2$ in both cases). The count of distinct data is much smaller than the
count of classes, and it is a count of modular data, not of categories: Observation~\ref{obs:gsy} takes this up.

The tenth entry, $j(11)'$, is printed in \cite{EG17} with $H=\Z_{25}\times\Z_5$ and form $2l_1l_1'/25+2l_2l_2'/5$. Trace $123$ has $14$
conjugacy classes; three carry $A_+\cong\Z_5\times\Z_{25}$, and each gives the form $l_1l_1'/25+l_2l_2'/5$ --- the printed one
multiplied by $2^{-1}$, and $2$ is a non-residue modulo $5$, so the two are not isomorphic. The word data's form is the right one. The Gauss sum $\sum_id_i^2\theta_i=\dim\mathcal C$ of the Drinfeld centre, evaluated with the
dimensions of their (6.24), is satisfied by the word data's form and not by the printed one. Gannon--Schopieray--Yadav's independent
determination of $(H,q)$ from the HI matrix also agrees with the word data (Observation~\ref{obs:gsy}). So the table is reproduced ten times out
of ten. Here the word data's form is $q^{\mathrm{lk}}$, the normalisation of
the unitary double. The spin form $q_+=2^{-1}q^{\mathrm{lk}}$ of the $S$-fold is its Galois conjugate (Corollary~\ref{cor:galois}), and it is
isometric to the printed form, which differs from it by the square $4$.
\end{observation}

\begin{observation}[The complete list of Gannon--Schopieray--Yadav]
\label{obs:gsy}
Gannon--Schopieray--Yadav \cite{GSY26}, Table~1, list every untwisted unitary Haagerup--Izumi category of Q-system type for
$\Z/m$, $m\le43$ odd, with the modular data of its centre written as a label $\mu^+_s$ for the cyclic metric group $\Z/\mu$ with
$q(h)=sh^2/\mu$, or an orthogonal sum of such labels. Their $q$ is the quadratic form of the twists, $q=\mathsf m\beta$ in the notation of
Observation~\ref{obs:egconj} (so $s\equiv-\tfrac12$ for the forms $ll'/\mu$). Enumerating all conjugacy classes of trace $m^2+2$ and
computing $(A_-,A_+,q)$ for each gives the table below. Here $q(x)=-2^{-1}\,x\wedge(\varphi+1)^{-1}x\bmod1$ on $A_+$, with $2^{-1}$ the
inverse of $2$ modulo $m^2+4$. This is isometric to $q^{\mathrm{lk}}$, since $-1$ is a square modulo $m^2+4$. Cyclic forms are compared by the
square class of $s$:
\begin{center}\small
\begin{tabular}{ccc>{\raggedright\arraybackslash}p{6.2cm}}
\toprule
$m$ & classes & categories (data) & the data, all realised by words\\
\midrule
$3,5,7,13,17$ & $3,5,7,13,17$ & $1\ (1)$ & $13^+_6$, $29^+_{14}$, $53^+_{26}$, $173^+_{86}$, $293^+_{146}$\\
$9$ & $18$ & $2\ (2)$ & $85^+_{42}$, $85^+_6$\\
$11$ & $14$ & $2\ (2)$ & $125^+_{62}$, $25^+_{12}{\times}5^+_2$\\
$15,27,35,37$ & $45,81,105,111$ & $3\ (1)$ & $229^+_{114}$, $733^+_{366}$, $1229^+_{614}$, $1373^+_{686}$\\
$19,23,25$ & $38,46,50$ & $2\ (2)$ & $365^+_{182},365^+_{11}$; $533^+_{266},533^+_3$; $629^+_{314},629^+_3$\\
$31,41,43$ & $62,82,86$ & $2\ (2)$ & $965^+_{482},965^+_{11}$; $1685^+_{842},1685^+_{11}$; $1853^+_{926},1853^+_3$\\
$21$ & $84$ & $4\ (2)$ & $445^+_{222}$, $445^+_6$\\
$29$ & $63$ & $3\ (3)$ & $845^+_{422}$, $845^+_3$, $65^+_2{\times}13^+_2$\\
$33$ & $165$ & $5\ (1)$ & $1093^+_{546}$\\
$39$ & $195$ & $5\ (3)$ & $1525^+_{762}$, $1525^+_6$, $5^+_1{\times}5^+_1{\times}61^+_2$\\
\bottomrule
\end{tabular}
\end{center}
At every odd $m\le43$ the set of distinct $(A_+,q)$ produced by the words of trace $m^2+2$ is \emph{exactly} the set of distinct
modular data in their table: every one of their labels is realised, and no word produces a form outside the list. The cyclic labels are
matched by the square class of $s$; the three non-cyclic ones, at $m=11,29,39$, by the multiset of values of $q$ and the Gauss sum.
Each form is moreover realised on every $A_-$-type of the trace, including the untwisted type $A_-\cong\Z_m\times\Z_m$. Two
things follow. First, their label $25^+_{12}\times5^+_2$ for the second $\Z/11$ category is $q=12l_1^2/25+2l_2^2/5=\mathsf m\beta$ for
$\beta=l_1l_1'/25+l_2l_2'/5$, the word data's form of Observation~\ref{obs:egconj}: their independent determination of
$(H,q)$ from the HI matrix (their Proposition~5.7) confirms the resolution of the $j(11)'$ entry reached there. Second, what the
words see is the modular data and not the category: at $m=15,21,27,33,35,37,39$ several categories share one datum (the three $\Z/15$ ones are Galois conjugates \cite[\S1.5]{GSY26}, pairwise
non-Morita-equivalent \cite[Cor.~5.9]{GSY26}), and the words produce that datum once per
$A_-$-type. The count of distinct data agrees with theirs at all twenty-one values of $m$; the count of categories is not a
question the Bethe equations can address.
\end{observation}

\part{The three frames}
\label{part:frames}

This part builds the objects that Parts~\ref{part:results} and~\ref{part:aminus} study, and proves the facts about them that later
arguments use. One set of data, the level matrix and the Bethe system it defines, is read in three frames.

In the \emph{Farey frame} (Section~\ref{sec:farey}), the layered triangulation of the mapping torus gives the Neumann--Zagier matrices $A,B$,
the level matrix $\Kraw$ with $B\Kraw=A$, the charges $Q$, and the Bethe equations \eqref{eq:bethe}. Section~\ref{sec:runK} is not a frame;
it is the bridge to the other two. Theorem~\ref{thm:Kformula} gives $K=Q\Kraw Q$ and the charge $c$ in closed form from the run structure, for
every hyperbolic word, and proves that $K$ is positive definite with $\det K=2^r$ and is the unique such solution (proof in Appendix~\ref{app:Kformula}). In the \emph{Nahm frame}
(Section~\ref{sec:nahm}), the same equations are Nahm's ladder for $\mathbf N=2(I-\Kraw^{-1})$ (Proposition~\ref{prop:bethenahm}). This gives
the positive vacuum $V^*$ and its conjugate $\sigma V^*$ in closed form for every word (Theorem~\ref{prop:propagated}), and the frieze
relation at every solution (Theorem~\ref{thm:friezeall}). In the \emph{Bethe frame} (Section~\ref{sec:bethe}), the observables are defined:
the handle-gluing operator \eqref{eq:Habelian}; the deformation $(1+\epsilon)^{c_b}$, which defines a Bethe vacuum; and the terms class, unit,
positive vacuum and identity class. Section~\ref{sec:nahm} already uses some of these terms, with the meaning given in Section~\ref{sec:bethe}.

The charge vector $Q$ and the formula \eqref{eq:Habelian} are those of \cite{GKL23} for $LR^n$, extended to every word as stated in
Sections~\ref{sec:farey} and~\ref{sec:bethe}. Nothing proved in this part depends on the dictionary \eqref{eq:dict} or on input~(b).

\section{The Farey frame}
\label{sec:farey}

This section sets up the first frame: the layered triangulation of the mapping torus and the gauge-theory data it produces.
Section~\ref{sub:gluing} describes the layered triangulation by Farey flips, one ideal tetrahedron per letter, and writes its edge
equations. After one shape per tetrahedron is eliminated, these give the Neumann--Zagier matrices $A$ and $B$, and the level matrix $\Kraw$
solves $B\Kraw=A$, reproducing that of \cite{GKL23} for $LR^n$. Section~\ref{sub:charges} introduces the charges $Q$ (2 at each $L$-tetrahedron,
1 at each $R$), the integral level matrix $K=Q\Kraw Q$, the matrix $\mathbf N$ of Nahm's ladder, and the Bethe equations \eqref{eq:bethe} that the
rest of the paper studies. Since $\operatorname{rank}B=r-1$ (Theorem~\ref{thm:Kformula}), the symmetric solutions of $B\Kraw=A$ form a one-parameter family, and the determinant
law $\det K=2^r$ selects the member. Symmetry is imposed, not derived from the triangulation. Theorem~\ref{thm:Kformula} in the next section shows that the selected matrix is given in closed form by
the run structure, is positive definite, and is unique.

\subsection{Flips, edges and the gluing equations}
\label{sub:gluing}

The layered triangulations of once-punctured-torus bundles built from Farey flips go back to Floyd--Hatcher \cite{FH82}; see Gu\'eritaud
\cite{Gueritaud06} for their canonical form. An ideal triangulation of $\Sigma_{1,1}$ is a Farey triangle $\{a,b,a\oplus b\}$, represented by the ordered
pair $(a,b)$ of primitive vectors with $\det(a,b)=1$, the third slope being $a+b$. The two flips are
\begin{equation}
L:(a,b)\mapsto(a,a+b),\qquad R:(a,b)\mapsto(a+b,b),
\end{equation}
i.e.\ right multiplication by $L$ and $R$ as above, so a word of trace $k$ gives the monodromy of Section~\ref{sec:setup}. Under $L$ the slope $b$ dies,
$2a+b$ is born, and $a,a+b$ persist; under $R$ the slope $a$ dies, $a+2b$ is born, and $b,a+b$ persist. Slopes are therefore in bijection with letters,
one born per flip, so
\[
\#\{\text{edges}\}=\#\{\text{tetrahedra}\}=r .
\]
In the layered triangulation the tetrahedron $\Delta_i$ of flip $i$ has six edges: the dying and newborn slopes, an opposite pair carrying one parameter
with multiplicity one each; and the two persisting slopes, each an opposite pair because the quadrilateral of $\Sigma_{1,1}$ has its opposite sides
identified, hence carrying multiplicity two. This gives the edge equations directly.

Assign $(\text{dying/newborn},\text{kept},\text{Farey})\mapsto(Z,Z',Z'')$ on $L$-tetrahedra and $(Z',Z,Z'')$ on $R$-tetrahedra. This is the orientation of the shapes of the geometric layered triangulation \cite{Gueritaud06}; the chirality is essential, and with the uniform assignment $(Z,Z',Z'')$ on both letters $LR^n$, $n\ge2$, has no symmetric solution of \eqref{eq:BKA}. Here $Z,Z',Z''$ are the logarithmic shape parameters of a tetrahedron, with $Z+Z'+Z''=i\pi$; in terms of the shape $z$ of \cite{NZ85}, $(Z,Z')=(\log z,\log(1-z)-\log z)$. They are not the logarithms $Z_a$ of the Bethe coordinates of Section~\ref{sec:bethe}. Eliminating $Z''$, the derived equations reproduce the edge equations $C_1,\dots,C_n$ of eq.~(A.6) of \cite{GKL23}. Writing
$A,B$ for the coefficient matrices of $Z,Z'$ after that elimination (the Neumann--Zagier coefficient matrices \cite{NZ85}, unrelated to the groups
$A_\mp$), one has
\begin{equation}
B\Kraw=A ,
\label{eq:BKA}
\end{equation}
with $\Kraw$ the level matrix of eq.~(A.11) of \cite{GKL23}: for $LR^n$ that matrix has $(\Kraw)_{LL}=n/2$, $(\Kraw)_{LR_j}=j-1$, and the block on the $R$'s is the matrix with entries
$2\min(i,j)$, which is what Theorem~\ref{thm:Kformula} gives for every $n$ ($K_{LL}=4+2n-4=2n$, $K_{LR_j}=2j-2$, $Q_L=2$).
In \cite[(A.7)--(A.8)]{GKL23} the rows are the longitude $M_\lambda/2$ and the edges $C_1,\dots,C_n$. The remaining edge $C_{n+1}$ is omitted, since for one cusp the edge equations have one linear relation \cite{NZ85}.
Their $B$ is then invertible and $\Kraw=B^{-1}A$. Here we keep all $r$ edge equations, for every word. Then $\operatorname{rank}B=r-1$, and the free parameter is fixed in Section~\ref{sub:charges}.
With $Z''$ eliminated, the assignment of $(Z,Z',Z'')$ to the
dying/newborn, kept and Farey slots made above is the one for which the solution of \eqref{eq:BKA} has the form of
Theorem~\ref{thm:Kformula}; its $Z\leftrightarrow Z'$ swap exchanges $A$ and $B$, and $\Kraw^{-1}$ is then one member of the resulting family, which is not of that form.

\subsection{Charges, the integral level matrix, and the determinant law}
\label{sub:charges}

At $LR^n$ the rescaling that makes the level matrix integral is fixed by the abelian dual of \cite{GKL23}: each tetrahedron carries one chiral multiplet, and the
factor by which $Z_a$ must be rescaled is its charge $Q_a$. For a general word the charges are chosen below. Write
\begin{equation}
K:=Q\,\Kraw\,Q
\end{equation}
for the normalised integral matrix, which is the one entering the gauge theory, and
\begin{equation}
\mathbf N:=2\bigl(\mathbf1-Q K^{-1}Q\bigr)=2\bigl(\mathbf1-\Kraw^{-1}\bigr)
\label{eq:nahmmat}
\end{equation}
for the matrix of Nahm's ladder (Proposition~\ref{prop:bethenahm}), the second equality because the charges cancel. The distinction between $\Kraw$ and $K$ matters: amphichirality
(Proposition~\ref{prop:closed} below) is a property of $\Kraw$ alone, since $PQP\ne Q$ for $Q=(2,\dots,2,1,\dots,1)$.

\emph{The determinant law.} In the normalisation of \cite{GKL23}, $\det K=2^r$, hence $\det\Kraw=2^{r-2\#L}=2^{\#R-\#L}$
(Theorem~\ref{thm:Kformula}).

Since every entry of $K$ is even, the law $\det K=2^r$ says that $K/2$ is unimodular. The law is proved for every word in
Theorem~\ref{thm:Kformula}. It is recorded here because it fixes the free parameter of \eqref{eq:BKA}. Since $\operatorname{rank}B=r-1$ (Theorem~\ref{thm:Kformula}),
the gluing data determine $\Kraw$ only up to a one-parameter family $\Kraw(t)$ of symmetric matrices, and $\det K=2^r$ selects the member. For $L^3R^3$ two criteria select the same member of the family: $\det K=2^r$, equivalently unimodularity of $K/2$, and integrality of $\mathbf N$ together with evenness of $K$. Integrality of $\mathbf N$ alone does not, since $\mathbf N(t)$ is integral at $t=0,-\tfrac25,-\tfrac27,-\tfrac4{11},-\tfrac4{13},\dots$. The member selected has the charge vector
\begin{equation}
Q=(2,2,2,1,1,1),
\end{equation}
charge $2$ at each $L$-tetrahedron, generalising the $(2,1,\dots,1)$ of $LR^n$ with its single $L$. For every word we take $Q_a=2$ at each
$L$-tetrahedron and $Q_a=1$ at each $R$-tetrahedron; at $LR^n$ this is the charge vector of \cite{GKL23}, and at $L^3R^3$ it is the one the two
criteria select. The Bethe equations are then
\begin{equation}
\prod_az_a^{K_{ab}}\bigl(1-z_b^{-Q_b}\bigr)^{Q_b}=1 .
\label{eq:bethe}
\end{equation}
In the shape variables $x_a=z_a^{-Q_a}$ they read
\begin{equation}
(1-x_b)^{Q_b}=\prod_ax_a^{K_{ab}/Q_a}:\qquad (1-x_L)^2=\prod_ax_a^{K_{aL}/Q_a},\qquad 1-x_R=\prod_ax_a^{K_{aR}/Q_a},
\label{eq:bethex}
\end{equation}
with integer exponents, since every entry of $K$ is even (Theorem~\ref{thm:Kformula}). The solutions $z$ over a given $x$ differ only in the
signs $z_L=\pm x_L^{-1/2}$.

\section{The run-structure level matrix}
\label{sec:runK}

The Farey pipeline of Section~\ref{sec:farey} produces the level matrix $K$, the charge vector $Q$ and the charge $c$ by an algorithm. One solves
$B\Kraw=A$ for a symmetric $\Kraw$ and fixes the remaining one-parameter freedom by $\det K=2^r$. This section replaces the algorithm by a closed
form, which is the bridge between the triangulation data and everything that follows. Theorem~\ref{thm:Kformula} shows that, for every
hyperbolic word, $K$ is determined by the run structure alone: on each $R$-run the \emph{fan block}, the matrix with entries $2\min(i,j)$ (named after the fan of Farey slopes $w_j=w_0+ja$ of an $R$-run, Section~\ref{sub:runchain}), a cyclic tridiagonal $L$-block with
corrections at the run ends, and ramps (linearly increasing couplings, defined in the theorem) coupling each $R$-run to the $L$'s on either side. The charge $c$ is read off the same structure, and $K$
is positive definite with $\det K=2^r$. The only external input is the Neumann--Zagier rank theorem, for the geometric triangulation of \cite{Gueritaud06}, and the determinant law is a selection rule. For $L^mR^m$ the matrix is written out in
\eqref{eq:Kdiag}, and Proposition~\ref{prop:closed} shows that $\Kraw P\Kraw=P$ for the exchange $P$ of $L$- and $R$-sites, the algebraic form of
the amphichirality of $L^mR^m$.

\begin{theorem}[The pipeline produces the run-structure level matrix]
\label{thm:Kformula}
Write the word cyclically as runs $L^{a_1}R^{b_1}\cdots L^{a_{n_R}}R^{b_{n_R}}$. For every hyperbolic word, the unique symmetric solution $\Kraw$ of $B\Kraw=A$ with
$\det K=2^r$ is given by $K=Q\Kraw Q$ with
\begin{itemize}
\item each $R$-run of length $b$ a fan block $K_{ij}=2\min(i,j)$ within the run, $R$'s in different runs not coupling;
\item $L$ diagonal $4$, plus $2b$ for the first $L$ after an $R$-run of length $b$;
\item $L$--$L$ coupling $-2$ for each adjacency of cyclically consecutive $L$'s in the word with the $R$'s deleted: so $-4$ between the two
$L$'s when $\#L=2$, and $-4$ on the diagonal when $\#L=1$ (as for $LR^n$, giving $4+2n-4=2n$);
\item coupling $2j$ between the first $L$ after an $R$-run and the $j$-th $R$ of that run, so that these entries $2,4,\dots,2b$ increase linearly along the run (we call them the \emph{ramp}); and coupling $-2$ between the last $L$ before an $R$-run and every $R$ of it;
\end{itemize}
and the charge vector $c$, the primitive integer vector with $Q^{-1}c\in\ker B$ (which is one-dimensional), with its sign fixed by $c_{R_1}=1$, is $b$ at the first $L$ after an $R$-run of length $b$, $0$ at the other $L$'s, and $j$ at the $j$-th $R$ of each run. It is primitive because its entry at the first $R$ of each run is $1$.
Moreover $K$ is positive definite. The diagonal word $L^mR^m$ is the single-run case.
\end{theorem}

For an example see the matrix $K$ of $L^3R^3$ in Section~\ref{sec:l3r3}. Its lower right $3\times3$ block is the fan block. The $L$-diagonal is $(10,4,4)$: the first $L$ after the run of length $b=3$ has $4+2\cdot3=10$, and the other two have $4$. The row of the first $L$ also carries the ramp $2,4,6$. The row of the third $L$ has $-2$ at every $R$. The $L$--$L$ entries $-2$ come from the three adjacencies of the cyclic word $LLL$.

\paragraph{The rows of $\Kraw$.} For each $R$-run $\rho$ of length $b$, write $L^+(\rho)$ for the first $L$
after $\rho$ and $L^-(\rho)$ for the last $L$ before it; call an $L$ \emph{middle} if both its neighbours in the word are $L$'s. With
$\Kraw=Q^{-1}KQ^{-1}$ the theorem says that $\Kraw$ is the sum of the following pieces, entries at coinciding sites being added. (So for $\#L=1$ the single $L$ is both $L^+$ and $L^-$ of every run and is adjacent to itself.) For each $R$-run $\rho$ of length $b$ the pieces are
the fan $2\min(i,j)$ on $\rho$, the ramp $\Kraw[L^+(\rho),R_j]=j$, the entry $b/2$ on the diagonal at $L^+(\rho)$, and
$\Kraw[L^-(\rho),R_j]=-1$ for $j=1,\dots,b$. In addition there is the entry $1$ on the diagonal at every $L$, and the entry $-\tfrac12$ at each pair of cyclically
consecutive $L$'s of the word with the $R$'s deleted, once for each adjacency. So the row of $\Kraw$ at the $j$-th $R$ of a run $\rho$ of
length $b$ is
\[
\Kraw[R_j,\cdot]=2\min(j,l)\ \text{at }R_l\in\rho,\qquad j\ \text{at }L^+(\rho),\qquad -1\ \text{at }L^-(\rho),
\]
and zero elsewhere; the row at an $L$ has $1$ on the diagonal, $-\tfrac12$ at each adjacent $L$, and, if the $L$ is $L^+(\rho)$, a further
$b/2$ on the diagonal and $l$ at $R_l\in\rho$; if it is $L^-(\rho')$, $-1$ at every $R$ of $\rho'$. The combinations used in the proof (Appendix~\ref{app:Kformula}) are linear, so
coincidences of sites take care of themselves.

The proof of Theorem~\ref{thm:Kformula} is in Appendix~\ref{app:Kformula}. Its only external input is the Neumann--Zagier rank
theorem \cite[Proposition~2.3]{NZ85}, which needs the layered triangulation to be geometric \cite{Gueritaud06}.

For the diagonal family the theorem specialises as follows. Index the $L$'s and the $R$'s by $0,\dots,m-1$ and put
$v=(1,2,\dots,m)^{\mathsf T}$, $\mathbf1=(1,\dots,1)^{\mathsf T}$:
\begin{equation}
\begin{gathered}
\Kraw(m)=\begin{pmatrix}C&V\\V^{\mathsf T}&F\end{pmatrix},\qquad F_{ij}=2\min(i{+}1,j{+}1),\\
C=\mathrm{diag}\bigl(\tfrac{m+2}2,1,\dots,1\bigr)-\tfrac12\Adj(C_m)\quad (\Adj(C_1):=2),
\end{gathered}
\label{eq:Kdiag}
\end{equation}
with $V=e_0v^{\mathsf T}-e_{m-1}\mathbf1^{\mathsf T}$ and the convention at $m=2$ that the cycle $C_2$ has a doubled edge.

A word is \emph{amphichiral} if it is conjugate to its own transpose, the word reversed with $L$ and $R$ exchanged (the term is the knot-theoretic one, a manifold homeomorphic to its mirror image). Since
$L^{\mathsf T}=R$ and $M^{\mathsf T}=JM^{-1}J^{-1}$ for $M\in SL_2$, with $J$ the matrix of the pairing $\wedge$ of Section~\ref{sec:metric}, the transpose carries the inverse
monodromy, so an amphichiral word is conjugate to its own inverse and its mapping torus admits an orientation-reversing
self-homeomorphism. The diagonal words are amphichiral: $L^mR^m$ is a symmetric matrix, so it equals its own transpose; $L^mR^n$ with $m\ne n$ is not. The
proposition is the level-matrix form of this. Orientation reversal acts on the shapes by $z\mapsto\bar z$ and on the Nahm variables by
$\delta\mapsto1-\delta$. It permutes the tetrahedra by the $L\leftrightarrow R$ exchange $P$, and inverts $\Kraw$.

\begin{proposition}[Amphichirality]
\label{prop:closed}
For the diagonal words $\varphi=L^mR^m$, $\Kraw(m)\,P\,\Kraw(m)=P$, where $P=\begin{psmallmatrix}0&I\\I&0\end{psmallmatrix}$ is the permutation exchanging the $L$-sites with
the $R$-sites.
\end{proposition}

The proof in Appendix~\ref{app:closed} is a direct check of three block identities.

\section{The Nahm frame and the closed form}
\label{sec:nahm}

This section sets up the second frame: it reads the Bethe system as Nahm's system and solves it at the positive point, in four steps.
(i)~Proposition~\ref{prop:bethenahm} shows that, under $\delta_b=1-x_b$, the undeformed Bethe equations are Nahm's ladder for the symmetric matrix
$\mathbf N=2(I-\Kraw^{-1})$, which makes $c_{\eff}$ a Rogers dilogarithm sum.
(ii)~Section~\ref{sub:nahmdiag} solves the diagonal words $L^mR^m$ by hand (Theorem~\ref{thm:nahmclosed}); the solution is a linear chain $\mathcal Q_j=1+jd$ along each run.
(iii)~Section~\ref{sub:propagated} shows that this is general: Theorem~\ref{prop:propagated} gives the positive vacuum $V^*$ of every word as ratios of heights
$h(v)=\langle\eta,v\rangle$ along the Farey slopes, with $\eta$ the expanding eigenvector, and contains Theorem~\ref{thm:nahmclosed} as the case $L^mR^m$.
Its consequences are $W_c(V^*)=\lambda$ (Proposition~\ref{obs:stretch}) and the rationality of the Rogers sum (Corollary~\ref{cor:rational}).
(iv)~Section~\ref{sub:runchain} explains the chain: within a run the heights form a frieze row (Remark~\ref{rem:friezedef}); the frieze relation holds at every solution of the
Bethe equations, not only at $V^*$ (Theorem~\ref{thm:friezeall}); and it reduces a two-run word to its run-end equations (Proposition~\ref{prop:runreduce}).
Section~\ref{sub:families} writes out $LR^n$ and $L^2R^n$ and compares the second with Jeong--Sohn's presentation (Remark~\ref{rem:jsframe}).

\emph{Conventions.} One solution is described by four sets of variables: the Bethe shapes $x_b=z_b^{-Q_b}$ of \eqref{eq:bethex}; the Nahm variables $\delta_b=1-x_b$ and $d_b=\delta_b^{-1/2}$; and the chain variables $\hat d_b=s_b^{-1/2}$ of Section~\ref{sub:runchain}, with $s_b=x_b$ at an $L$-step and $s_b=1-x_b$ at an $R$-step, so that $\hat d_b=x_b^{-1/2}$ at an $L$-step and $\hat d_b=d_b$ at an $R$-step. Solutions are counted in the shapes: a solution $x$ of \eqref{eq:bethex} lifts to $2^{\#L}$ solutions $z$ of \eqref{eq:bethe}, the sign choices $z_L=\pm x_L^{-1/2}$. The positive vacuum $V^*$ is the all-plus lift ($z_a>1$ for all $a$) of the solution $x\in(0,1)^r$ of \eqref{eq:bethex}, which is unique (Theorem~\ref{prop:propagated}). On this positive branch ($1-x_b>0$, $d_b>0$) Proposition~\ref{prop:bethenahm} makes a solution in one set of variables a solution in the others. Statements are made in the shapes $x$ unless they concern Nahm's equations.

\subsection{The Bethe system is a Nahm ladder}
\label{sub:bethenahm}

We first show that the Bethe equations and Nahm's equations are the same system in different variables.

\begin{proposition}[The two systems are one]
\label{prop:bethenahm}
Let $\varphi$ be a hyperbolic word, $(K,Q)$ the pipeline output of Theorem~\ref{thm:Kformula},
$\Kraw=Q^{-1}KQ^{-1}$, and
\[
\mathbf N:=2\bigl(I-QK^{-1}Q\bigr)=2\bigl(I-\Kraw^{-1}\bigr).
\]
Under $\delta_b=1-x_b=d_b^{-2}$, on the positive branch (all $x_b\in(0,1)$, real logarithms), the undeformed Bethe equations \eqref{eq:bethe} and Nahm's ladder for $\mathbf N$ are the same system:
\begin{equation}
1-x_b=\prod_ax_a^{(\Kraw)_{ab}}
\iff
\log(1-\delta_b)=\sum_a\bigl(\Kraw^{-1}\bigr)_{ba}\log\delta_a
\iff
d_b^2=1+\prod_ad_a^{\mathbf N_{ba}} .
\label{eq:bethenahm}
\end{equation}
Here the first form is the positive-root form of \eqref{eq:bethex} (for $Q_b=2$, \eqref{eq:bethex} reads $(1-x_b)^2=\cdots$), with $1-x_b>0$ and $d_b>0$. Moreover $\mathbf N$ is symmetric, because $K$ is.
\end{proposition}

\begin{proof}
Taking logarithms, the first equation reads $\log(1-x_b)=\sum_a(\Kraw)_{ab}\log x_a$. Substituting $\delta=1-x$ into the second and using the first,
\[
\sum_a\bigl(\Kraw^{-1}\bigr)_{ba}\log(1-x_a)=\sum_a\bigl(\Kraw^{-1}\bigr)_{ba}\sum_c(\Kraw)_{ca}\log x_c=\log x_b
\]
by symmetry of $\Kraw$, which is the second equation. For the third, $\delta_b=d_b^{-2}$ gives $\log\delta_a=-2\log d_a$ and
$\log(1-\delta_b)=\log(d_b^2-1)-2\log d_b$, so the second becomes $\log(d_b^2-1)=\sum_a\bigl[2\delta_{ba}-2(\Kraw^{-1})_{ba}\bigr]\log d_a$, which is the
third. Symmetry of $\mathbf N$ is symmetry of $K$, and $\Kraw^{-1}=QK^{-1}Q$ by definition.
\end{proof}

The system takes three equivalent forms. In shape form the Bethe equations read $1-x_b=\prod_ax_a^{(\Kraw)_{ab}}$, with the
positive-definite exponent matrix $\Kraw$ (Theorem~\ref{thm:Kformula}). Rewritten in $\delta=1-x$ they read $1-\delta_b=\prod_a\delta_a^{(\Kraw^{-1})_{ba}}$, a
Nahm system with the positive-definite matrix $\Kraw^{-1}=QK^{-1}Q$, and the ladder $d_b^2=1+\prod_ad_a^{\mathbf N_{ba}}$ is the same system in
$d=\delta^{-1/2}$. This is the system whose positive solution gives $c_{\eff}$ (Section~\ref{sec:ceff}).

So the handle-gluing operator, the class count and the central charge are observables of one matrix read in three ways: as the gluing data of the
layered triangulation (Section~\ref{sec:farey}), as the Bethe system of the TQFT (Section~\ref{sec:bethe}), and as a Nahm ladder (this section). The Nahm reading is
what makes $c_{\eff}$ accessible at all; it also gives its rationality at once (Corollary~\ref{cor:rational} below).

\begin{remark}[One system, two points]
\label{rem:twopoints}
The hyperbolic volume sum $\sum_iD(z_i)$ and the Nahm sum $\sum_iL(\delta_i)$ are one family of algebraic systems, specialised at
two different points: the completeness point, with $z_i\in\mathbb H$ and the Bloch--Wigner branch, and the flux point $(\mathfrak m,\nu)=(0,-1)$ of the twisted index (Section~\ref{sec:bethe}), with
$\delta_i\in(0,1)$ and the Rogers branch. Both $D$ and $L$ satisfy the same five-term functional equation --- $D$ the single-valued extension to $\C$,
$L$ the one to $(0,1)$ --- which is the mechanism behind Nahm's conjecture \cite{Nahm04,NRT92} and its refinements \cite{CGZ23}. This does \emph{not} give a numerical formula
$c_{\eff}=f(\mathrm{Vol})$, since among words with the same $\#L$, such as $LR^n$, the volume varies while $c_{\eff}=\#L$ does not. The relation is structural, with one variety and one functional
equation, and not numerical.
\end{remark}

\subsection{The diagonal family}
\label{sub:nahmdiag}

We begin with the diagonal words $L^mR^m$, where the solution can be found by hand and exhibits the structure that Section~\ref{sub:propagated} proves in general.

The closed forms can be found by an \emph{ansatz}: one posits that the chain variables $\hat d_b$ of the Conventions (equal to the Nahm variable $d_b$ at $R$-sites and to $x_b^{-1/2}$ at $L$-sites) increase by one along a run, checks that this solves the bulk equations identically, and is left with a few boundary equations to solve in radicals. Theorem~\ref{prop:propagated} shows that this is what happens for every word (Section~\ref{sub:runchain}).

Let $d=\tfrac12(m+\sqrt{m^2+4})$, so $d^2=md+1$, and put $\mathcal Q_j=1+jd$.

\begin{theorem}[Positive solution of $L^mR^m$]
\label{thm:nahmclosed}
For every $m\ge3$ the shape equations \eqref{eq:bethex} of $L^mR^m$ have exactly one solution $x\in(0,1)^{2m}$, whose all-plus lift is $V^*$, namely
\[
x_{L_j}=\frac{d^2}{\mathcal Q_j^2},\qquad x_{R_j}=\frac{\mathcal Q_{j-1}\mathcal Q_{j+1}}{\mathcal Q_j^2},\qquad j=1,\dots,m,
\]
where $L_j$ and $R_j$ denote the $j$-th $L$-site and the $j$-th $R$-site of the word, in the order the letters occur. Equivalently, by Proposition~\ref{prop:bethenahm}, the ladder equations have the unique solution $\delta=1-x$ in $(0,1)^{2m}$, $\delta_{L_j}=\mathcal Q_{j-1}\mathcal Q_{j+1}/\mathcal Q_j^2$ and $\delta_{R_j}=d^2/\mathcal Q_j^2$, so that $\delta=1-x=Px$ with $P$ the exchange of $L$- and $R$-sites of Proposition~\ref{prop:closed}. The $\mathcal Q_j$ are affine in $j$ and satisfy
$\mathcal Q_j^2-\mathcal Q_{j-1}\mathcal Q_{j+1}=d^2$ for all $j$.
\end{theorem}

\begin{proof}
We verify the Bethe form $1-x_b=\prod_ax_a^{(\Kraw)_{ab}}$ for the stated $x$, using the blocks $C$, $F=2\min$, $V$ of \eqref{eq:Kdiag}. Summing by parts,
$\sum_{i=1}^m\min(i,j)\Delta^2\log\mathcal Q_i=\log \mathcal Q_0-\log \mathcal Q_j+j\log(\mathcal Q_{m+1}/\mathcal Q_m)$.

\emph{$R$-columns}: the fan gives $\mathcal Q_j^{-2}((d+1)/d)^{2j}$ and the couplings $(d/(1+d))^{2j}d^2$ (using $\mathcal Q_m=d^2$), so the right side is
$d^2/\mathcal Q_j^2=1-x_{R_j}$. \emph{Middle $L$-columns}: $(d^2/\mathcal Q_k^2)(\mathcal Q_{k-1}/d)(\mathcal Q_{k+1}/d)=\mathcal Q_{k-1}\mathcal Q_{k+1}/\mathcal Q_k^2=1-x_{L_k}$. \emph{$L_1$}: the corner
$(m+2)/2$ and the ramp give $(d/\mathcal Q_1)^{m+2}\cdot\mathcal Q_2\mathcal Q_m/d^2\cdot((d+1)/d)^m/d^2=\mathcal Q_2/\mathcal Q_1^2$. \emph{$L_m$}: the equation reduces to $\mathcal Q_1^2\mathcal Q_m=\mathcal Q_{m+1}^2$,
i.e.\ $(1+d)^2d^2=(d(d+1))^2$. All entries lie in $(0,1)$ since $\mathcal Q_j>d$ for $j\ge1$, and $\Kraw^{-1}=QK^{-1}Q$ is positive definite (Theorem~\ref{thm:Kformula}), so the solution in $(0,1)^r$ of Nahm's equations for $\delta=1-x$ with matrix $\Kraw^{-1}$ is unique \cite[\S2]{VZ11} (Proposition~\ref{prop:bethenahm}); this is the positive solution. Finally $\mathcal Q_j^2-\mathcal Q_{j-1}\mathcal Q_{j+1}=d^2$ gives $1-x_{R_j}=x_{L_j}$ and $1-x_{L_j}=x_{R_j}$, so $\delta=1-x=Px$; this is the relabelling $x\mapsto P\delta$ under which $\Kraw^{-1}=P\Kraw P$ (Proposition~\ref{prop:closed}) carries the Nahm system for $\delta$ to the Bethe system for $x$.
\end{proof}

Three consequences. First, $\delta_{L_m}=1-1/d^2=m/d$, which is the largest coordinate and satisfies $\delta_{\max}^2+m^2\delta_{\max}-m^2=0$. Second,
the Galois action: from $d^2=md+1$ one gets $d\,\sigma(\mathcal Q_j)=d-j$ and $\mathcal Q_{m-j}=d(d-j)$, so $\sigma(\mathcal Q_j)=\mathcal Q_{m-j}/d^2$, hence
$\sigma(\delta_{L_j})=\delta_{L_{m-j}}$ for $j<m$ and $\sigma(\delta_{L_m})=1-d^2=-md$. Third, and most useful, the structure is the geometry of the
monodromy: the M\"obius map of $L^mR^m$ is $x\mapsto m+1/(m+1/x)$, whose attracting fixed point is $[m;m,m,\dots]=d$, so $(d,1)^{\mathsf T}$ is an
eigenvector with eigenvalue $\lambda=d^2$ by \eqref{eq:lamd}; and
\[
\mathcal Q_j=\bigl(R^j(d,1)^{\mathsf T}\bigr)_2
\]
exactly. The $\mathcal Q_j$ are affine in $j$ because $R$ is unipotent, the $R$-fan translating the eigenvector, and the closing condition $\mathcal Q_m=1+md=d^2$ is the
eigenvalue equation. Galois conjugation $\sigma(d)=-1/d$ exchanges the two fixed points of $\varphi$. So \emph{the positive Nahm solution is the expanding
direction, and its Galois conjugate the contracting one}.

\subsection{The general closed form}
\label{sub:propagated}

Theorem~\ref{thm:nahmclosed} is an instance of a general fact: for every word the positive vacuum is a ratio of heights along the Farey slopes. This subsection states and proves it, and then derives two consequences, the value of the Wilson line $W_c(V^*)$ and the rationality of the Rogers sum.
\begin{theorem}[The positive vacuum of every word: the propagated eigenvector]
\label{prop:propagated}
Write the word as $\varphi=\ell_1\ell_2\cdots\ell_r$ with each $\ell_i\in\{L,R\}$. Let $\eta$ be the expanding left eigenvector, $\eta\varphi=\lambda\eta$,
and define the \emph{height} $h(v)=\langle\eta,v\rangle$ of a vector $v\in\Z^2$. At step $i$ the partial product $M_{i-1}=\ell_1\cdots\ell_{i-1}$ has
columns $(c_1,c_2)$, which are the two slopes of the Farey basis at that step (Section~\ref{sub:gluing}); the flip $\ell_i$ keeps one of them ---
$c_1$ if $\ell_i=L$, $c_2$ if $\ell_i=R$ ---, the \emph{kept slope} $a$ of the step (as in Section~\ref{sub:gluing}). Put $w=c_1+c_2$, the middle slope, and
$s_i=h(a)^2/h(w)^2$. Then
\begin{equation}
x_i=s_i\ \ (\ell_i=L),\qquad x_i=1-s_i=\frac{h(u)h(u')}{h(w)^2}\ \ (\ell_i=R),\quad u=w-a,\ u'=w+a,
\label{eq:propagated}
\end{equation}
solves the shape equations \eqref{eq:bethex} for every hyperbolic word. It is the unique solution with $x\in(0,1)^r$: with $x=e^{-t}$, such solutions are the critical points on $t>0$ of $\tfrac12t^{\mathsf T}\Kraw t+\sum_i\Li_2(e^{-t_i})$, which is strictly convex because $\Kraw$ is positive definite (Theorem~\ref{thm:Kformula}) \cite[\S II.3]{Zagier07}, \cite[\S2]{VZ11}; its all-plus lift is the positive vacuum $V^*$. The same formulae at the contracting
eigenvector $\eta'$, $\eta'\varphi=\lambda^{-1}\eta'$, give a second solution of the shape equations, the Galois conjugate of $x(V^*)$, whose all-plus lift is $\sigma V^*$. (In $z$, the Galois conjugate of the lift $V^*$ has $z_L=h'(w)/h'(a)$, which differs from the all-plus lift by a sign at every $L$-site, since this ratio is negative at every $L$-step (there $t\in(-1,0)$, Lemma~\ref{lem:crossings}); the sign of $W_c$ depends on the choice, $W_c^2$ does not.) In addition, within each run the kept slope $a$ is fixed and the middle slopes form a fan, $w_j=w_0+ja$, where $w_0$ is the slope dying at the first step of the run and $w_j$ the middle slope of its $j$-th step, so that at either eigenvector the signed heights satisfy $h(w_{j+1})=h(w_j)+h(a)$.
\end{theorem}

\begin{proof}
By Theorem~\ref{thm:Kformula} the equations are local in the runs. Within an $R$-run the kept slope $a$ is fixed and the middle slopes of the
successive steps form a fan, $w_j=w_0+ja$ for $j=0,1,\dots$, with $w_0$ the first column of the Farey basis on entering the run, the slope dying at its first step; likewise within an $L$-run, with $w_0$ the second column. With $X_a=x_a^{1/Q_a}$, so that $X=h(a)/h(w)$ at an $L$-step, each equation involves only a step and its
run neighbours. \emph{$R$-columns}: the fan block telescopes to $(h(w_0)/h(w_j))^2(h(w_{b+1})/h(w_b))^{2j}$; the last $L$ before the run contributes
$(h(a)/h(w_0))^2$ and the first $L$ after it $(h(w_b)/h(w_{b+1}))^{2j}$; the product is $h(a)^2/h(w_j)^2=1-x_{R_j}$. \emph{Middle $L$'s}:
$X_j^4X_{j-1}^{-2}X_{j+1}^{-2}=(h(w_{j-1})h(w_{j+1})/h(w_j)^2)^2=(1-x)^2$. \emph{Notation at a junction.} By Theorem~\ref{thm:Kformula} the equation at an $L$-site $\ell$ is $(1-x_\ell)^2=\prod_aX_a^{K_{a\ell}}$.
We work on the developed word, the infinite periodic sequence of letters, so that the previous and next $L$-sites of $\ell$ and the
$R$-runs on either side of it are distinct sites; periodicity is used only at the end. Let $\rho$ be an $R$-run of length $b$, with kept slope
$q'$ and fan $w_j=p+jq'$, $j=0,\dots,b+1$, where $p$ is the slope kept by the $L$-run before $\rho$. Write $g_j=h(w_j)$, so that
$g_j=g_0+jh(q')$ and $X_{R_j}=g_{j-1}g_{j+1}/g_j^2$. Summation by parts gives
\[
\prod_{j=1}^bX_{R_j}=\frac{g_0g_{b+1}}{g_1g_b},\qquad \prod_{j=1}^bX_{R_j}^{\,j}=\frac{g_0\,g_{b+1}^{\,b}}{g_b^{\,b+1}} .
\]
The $L$-run before $\rho$ has kept slope $p$ and middle heights $h_i=h(q+ip)$, ending with $h_a=h(q')$. Its last site has $X=g_0/h_a$, and
$h_{a+1}:=h_a+g_0=g_1$. The $L$-run after $\rho$ has kept slope $w_b$ and middle heights $k_i=h(q'+iw_b)=k_0+ig_b$, with $k_0=h(q')$ and
$k_1=g_{b+1}$. Its first site has $X=g_b/k_1$.

\emph{First $L$ of a run of length $\ge2$.} Here $\ell=L^+(\rho)$. Its row has diagonal $4+2b$, the ramp $2j$ to $R_j$, and $-2$ to $L^-(\rho)$
and to the next $L$, whose $X$ is $g_b/k_2$. So the right side is
\[
\Bigl(\frac{g_b}{k_1}\Bigr)^{4+2b}\Bigl(\frac{h_a}{g_0}\Bigr)^2\Bigl(\frac{k_2}{g_b}\Bigr)^2\frac{g_0^2\,g_{b+1}^{2b}}{g_b^{2b+2}}
=\Bigl(\frac{k_0k_2}{k_1^2}\Bigr)^2 ,
\]
using $h_a=k_0$ and $k_1=g_{b+1}$. And $1-x_\ell=1-g_b^2/k_1^2=(k_1-g_b)(k_1+g_b)/k_1^2=k_0k_2/k_1^2$.

\emph{Last $L$ of a run of length $\ge2$.} Here $\ell=L^-(\rho)$. Its row has diagonal $4$, $-2$ to the previous $L$ (whose $X$ is $g_0/h_{a-1}$),
$-2$ to $L^+(\rho)$, and $-2$ to every $R$ of $\rho$. So the right side is
\[
\Bigl(\frac{g_0}{h_a}\Bigr)^4\Bigl(\frac{h_{a-1}}{g_0}\Bigr)^2\Bigl(\frac{g_{b+1}}{g_b}\Bigr)^2\Bigl(\frac{g_1g_b}{g_0g_{b+1}}\Bigr)^2
=\Bigl(\frac{h_{a-1}h_{a+1}}{h_a^2}\Bigr)^2 ,
\]
using $g_1=h_{a+1}$. And $1-x_\ell=(h_a-g_0)(h_a+g_0)/h_a^2=h_{a-1}h_{a+1}/h_a^2$.

\emph{A singleton $L$.} Here $\ell=L^+(\rho)$ and also $\ell=L^-(\rho')$ for the next $R$-run $\rho'$, of length $b'$. Its row has diagonal
$4+2b$, the ramp $2j$ to $\rho$, $-2$ to every $R$ of $\rho'$, and $-2$ to the previous and the next $L$. The run $\rho'$ has kept slope
$q'+w_b$ and fan heights $g'_j=g_b+jk_1$. So $\rho'$ contributes $(g'_1g'_{b'}/g'_0g'_{b'+1})^2$, and the next $L$ has $X=g'_{b'}/g'_{b'+1}$.
The right side is
\[
\Bigl(\frac{g_b}{k_1}\Bigr)^{4+2b}\frac{g_0^2\,g_{b+1}^{2b}}{g_b^{2b+2}}\Bigl(\frac{h_a}{g_0}\Bigr)^2\Bigl(\frac{g'_1}{g'_0}\Bigr)^2
=\Bigl(\frac{k_0k_2}{k_1^2}\Bigr)^2 ,
\]
since $g'_0=g_b$ and $g'_1=g_b+k_1=k_2$. This is $(1-x_\ell)^2$, as in the first case.

\emph{Wrap-around.} On the developed word every $X$ is a ratio of heights of slopes, and the slopes one period later are the images of the
earlier ones under $\varphi$. Since $h(\varphi v)=\lambda h(v)$, every $X$ is periodic, and the identities on the developed word are the
equations of the cyclic word. This covers $\#L=1$, where the single $L$ is its own previous and next $L$ and $\rho'=\rho$: the row becomes
$K_{LL}=2b$ and $K_{LR_j}=2j-2$. It also covers $\#L=2$, where the coupling $-4$ is the sum of the two adjacencies.

Heights are positive at $\eta$: $\varphi$ has positive entries, so $\eta>0$ by Perron--Frobenius, and the columns of $M_{i-1}$ are nonnegative and nonzero. So this
is the all-plus solution, with $x\in(0,1)^r$. At $\eta'$ the same identities hold: $\sigma(\eta)$ is a multiple of $\eta'$, the formulae are homogeneous of degree $0$ in $\eta$, and the shape equations have integer exponents.
\end{proof}

For $\varphi=L^mR^m$ the left eigenvector is $\eta=(d,1)$, $\varphi$ being symmetric, and the formulae reduce to Theorem~\ref{thm:nahmclosed}, with $h(w_j)/h(a)=\mathcal Q_j/d$.

The entries of $V^*$ are ratios of heights and lie in the real quadratic field $\Q(\lambda)$. Throughout, $\sigma$ denotes the nontrivial
automorphism of $\Q(\lambda)$; it sends $\lambda$ to $\lambda^{-1}$ and the expanding eigenvector $\eta$ to a multiple of the contracting
one $\eta'$, so the contracting eigenvector gives the Galois-conjugate class, whose all-plus vacuum is $\sigma V^*$. (For $L^mR^m$ this is the $\sigma(d)=-1/d$ of
Theorem~\ref{thm:nahmclosed}.) At the contracting eigenvector the $L$-step shapes $s'=h'(a)^2/h'(w)^2$ are still squares, so all charge-$2$ coordinates are real and positive at the
all-plus choice, while the $R$-step coordinates $x'=1-s'$ can be negative, on a set $\mathcal E$ of $R$-steps identified in Lemma~\ref{lem:crossings} below.

\begin{proposition}[The stretch factor at the positive vacuum]
\label{obs:stretch}
For a solution $z$ of the Bethe equations \eqref{eq:bethe} put $W_c:=\prod_az_a^{c_a}$, with $c$ the charge vector of Theorem~\ref{thm:Kformula}; it is the value of the Wilson line of gauge charge $c$ (Section~\ref{sec:bethe}). Then
$W_c(V^*)=\lambda(\varphi)$: the Wilson loop at the positive vacuum is the stretch factor of the monodromy.
\end{proposition}

\begin{proof}
With $c$ as in Theorem~\ref{thm:Kformula}, and $h_j=h(w_j)$ for the fan $w_j$ of an $R$-run $\rho$ of length $b$, the factors $z_{R_j}^{\,j}$ multiply to $h_b^{\,b+1}/(h_{b+1}^{\,b}h_0)$ and the factor $z_{L^+(\rho)}^{\,b}$ is $(h_{b+1}/h_b)^b$, so $\rho$ contributes $h_b/h_0=h(c_1\text{ at run end})/h(c_1\text{ at run start})$ to $W_c(V^*)$. The $L$-runs leave $c_1$ fixed, so the product telescopes to $h(\varphi e_1)/h(e_1)=\lambda$.
\end{proof}

At $L^mR^m$, $\lambda=d^2$, and the fan product telescopes to $\mathcal Q_m(\mathcal Q_m/\mathcal Q_{m+1})^m=d^2(d/(d+1))^m$, the factor $z_{L_1}^m=((1+d)/d)^m$
cancelling the second term; on $LR^n$ it is the identity $a_0/a_3=\bigl((\sqrt{k+2}+\sqrt{k-2})/2\bigr)^2=(k+\sqrt{k^2-4})/2$ with $k=\tr\varphi$, $a_0,a_3$ as in
\eqref{eq:as}. Why $W_c(V^*)$ is to be compared with $a_0/a_3$ is explained in Observation~\ref{obs:wcrow}.

\begin{corollary}[Rationality of the Rogers sum]
\label{cor:rational}
Let $x$ be the solution \eqref{eq:propagated} given by the propagated eigenvector. Then $\sum_aL(\delta_a)/L(1)\in\Q$.
\end{corollary}

\begin{proof}
This is the standard Bloch-group argument \cite{Nahm04,Zagier07}. All $\delta_a$ lie in $F=\Q(\lambda)=\Q(\sqrt{k^2-4})$ (see after Theorem~\ref{prop:propagated}), which is real quadratic. Consider the element $\Delta=\sum_a[\delta_a]$ of the pre-Bloch group of $F$. Since $M=QK^{-1}Q$ is symmetric, $\sum_{a<b}(M_{ba}-M_{ab})(\delta_b\wedge\delta_a)=0$, and so $\sum_a\delta_a\wedge(1-\delta_a)=0$. The entries of $M$ are half-integers, so this relation holds in $\bigwedge^2F^\times\otimes\Q$, and a multiple of $\Delta$ lies in the Bloch group $B(F)$. This group is finite, since $F$ is totally real (Borel \cite{Borel77}, via \cite{Suslin91}). The Rogers sum of a torsion element is a rational multiple of $L(1)$.
\end{proof}

Rationality uses two inputs: the symmetry of $K$, which puts $\Delta$ in $B(F)\otimes\Q$, and $\delta_a\in\Q(\lambda)$ (Theorem~\ref{prop:propagated}), which makes $F$ totally real. The second is needed: a positive-definite Nahm system need not have its solution in a totally real field \cite[\S II.3]{Zagier07}. What Theorem~\ref{thm:ceff} adds is the value, $c_{\eff}=\#L$.

\subsection{Run chains and friezes}
\label{sub:runchain}

Within a run the propagated eigenvector is a linear chain. We show that the chain is a frieze row, that the frieze relation constrains every solution of the Bethe equations and not only $V^*$, and that it reduces the equations of a two-run word to its run-end equations. Write
\[
\hat d_i:=s_i^{-1/2}=\Bigl|\frac{h(w_i)}{h(a_i)}\Bigr| ,
\]
the reciprocal square root of the step shape, so that $x_i=\hat d_i^{-2}$ at an $L$-step and $x_i=1-\hat d_i^{-2}$ at an $R$-step. By the fan formula of Theorem~\ref{prop:propagated}, $h(w_j)/h(a)=h(w_0)/h(a)+j$ within each run; at the expanding eigenvector all heights are positive, so
\[
\hat d_j=\hat d_0+j ,
\]
the ansatz of Section~\ref{sub:nahmdiag}, now a theorem. This is a chain of slope exactly $1$, with one unknown, the anchor $\hat d_0$, per run, and $2n_R$ unknowns in all, $n_R$ being the number of runs of each letter. Every $\hat d_i$ lies in $\Q(\lambda)=\Q(\sqrt{k^2-4})$, $k=\tr\varphi$: $\eta=(\varphi_{21},\lambda-\varphi_{11})$ is integral in $\lambda$, the $w_j$ and $a$ are integer Farey vectors, and $\lambda^2-k\lambda+1=0$.

At the contracting eigenvector the signed ratio $h'(w_j)/h'(a)$ still increases by one along a run, and $\hat d_j$ is its absolute value. The chain $\hat d_j=\hat d_0+j$ has two consequences.

\emph{(a) Why the boundary values carry nested radicals and the chain values do not.} The Nahm variable at an $R$-step is $\delta_R=1-x_R=s$, so
$d_R=\hat d_R$ lies in $\Q(\lambda)$ and the chain is visible in it. At an $L$-step $\delta_L=1-x_L=1-s$, and $(1-\hat d^{-2})^{-1/2}$ generically
generates a quadratic extension of $\Q(\lambda)$. So in the $\delta$ variable the chain is visible along $R$-runs only, and the $L$-sites appear as
boundary unknowns inside a nested radical --- which is the shape both closed forms below have, and the reason the $\delta$-ansatz has to treat the
$L$-run by hand.

\emph{(b) This is the affine form of the $\mathcal Q_j$.} At $\varphi=L^mR^m$ the kept slope of the $R$-run is $a=(m,1)^{\mathsf T}$, of height $h(a)=md+1=d^2$ for $h(v)=dv_1+v_2$, and
$\hat d_{R_j}=\mathcal Q_j/d$ with $\mathcal Q_j=1+jd$, so the chain is exactly the affine form $\mathcal Q_j=1+jd$ of Theorem~\ref{thm:nahmclosed}, and
$\mathcal Q_j^2-\mathcal Q_{j-1}\mathcal Q_{j+1}=d^2$ there is the case $\hat d_{j\pm1}=\hat d_j\pm1$ of the chain.

The linear chain $\hat d_j=\hat d_0+j$ is the case $\tau=2$ of a one-parameter family of chains, the frieze rows of Remark~\ref{rem:friezedef}. They are indexed by the frieze trace $\tau$ defined there. Away from $V^*$ the frieze trace need not equal $2$.

\begin{remark}[The run is an $SL_2$ frieze, and $\tau=2$ is parabolicity]
\label{rem:friezedef}
Let $M\in SL(2,\R)$, $w_j=M^jw_0$, let $h$ be any linear functional, and put $h_j=h(w_j)$, $\tau=\tr M$.
\begin{enumerate}
\item[(i)] By the Cayley--Hamilton theorem $M^2=\tau M-1$, so $w_{j+1}=\tau w_j-w_{j-1}$ and $h_{j+1}=\tau h_j-h_{j-1}$. Hence $h_j^2-h_{j-1}h_{j+1}=h_j^2-\tau h_{j-1}h_j+h_{j-1}^2$ is independent of $j$: it is the $M$-invariant quadratic form evaluated at $(h_{j-1},h_j)$. The quantity $(h_{j-1}+h_{j+1})/h_j=\tau$ is also constant along the orbit.
\item[(ii)] Writing $\tau=2\cos\vartheta$ with $\tau\ne\pm2$, every solution of the recurrence over $\C$ with $h_j^2-h_{j-1}h_{j+1}\ne0$ is $h_j=A\sin(j\vartheta+\phi)/\sin\vartheta$ with constants $A,\phi\in\C$, and
$h_j^2-h_{j-1}h_{j+1}=A^2$; the remaining solutions $c\mu^{\pm j}$, $\mu+\mu^{-1}=\tau$, have $h_j^2-h_{j-1}h_{j+1}=0$, so they are not frieze rows. This is the trigonometric $Q$-recurrence for $|\tau|<2$, hyperbolic for $|\tau|>2$, and affine, $h_j=h_0+j\beta$ with
constant $\beta^2$, exactly at $\tau=\pm2$ (for $\tau=-2$ the affine solution is multiplied by $(-1)^j$).
\item[(iii)] Within a run of a word in $L,R$ the generator is $L$ or $R$, both \emph{parabolic}. Let $\ell=L$ or $R$ be the letter of the run, and $P$ the Farey basis on entering it. Then $w_j=P\ell^jP^{-1}w_0$, so $M=P\ell P^{-1}$ and $\tau=2$.
The kept slope $a$ is the fixed vector of $M$, and $\beta=h(a)$. Normalising by $\hat d_j=h_j/h(a)$ gives the $SL_2$-frieze relation
\begin{equation}
\hat d_j^2-\hat d_{j-1}\hat d_{j+1}=1 .
\label{eq:frieze}
\end{equation}
\item[(iv)] The passage $y_j=\hat d_{j-1}\hat d_{j+1}/\hat d_j^2$ gives $y_j=1-s_j$: the Bethe coordinate $x$ at an $R$-step, and $1-x$ at an $L$-step.
\end{enumerate}
\emph{Proof.} (i) is $M^2=\tau M-I$ applied to $w_{j-1}$, and the conserved quantity is the substitution $h_{j+1}=\tau h_j-h_{j-1}$. (ii) solves the recursion.
(iii): $L$ and $R$ have trace $2$ and fixed vectors $e_1$, $e_2$; the kept slope of a run is that fixed vector, and $h_{j+1}-h_j=h(a)$, so
$\hat d_{j+1}=\hat d_j+1$ and $\hat d_j^2-(\hat d_j-1)(\hat d_j+1)=1$. (iv) is then $(\hat d_j^2-1)/\hat d_j^2=1-\hat d_j^{-2}=1-s_j$.

\emph{Frieze rows.} A sequence satisfying the \emph{frieze relation} \eqref{eq:frieze} satisfies
$\det\begin{psmallmatrix}\hat d_j&\hat d_{j+1}\\ \hat d_{j-1}&\hat d_j\end{psmallmatrix}=1$, the diamond rule of an $SL_2$-frieze \cite{CC73} with constant rows. We call such a sequence a
\emph{frieze row}, and the conserved quantity $\tau=(\hat d_{j-1}+\hat d_{j+1})/\hat d_j$ of (i) its \emph{frieze trace}: it is the trace of the transfer matrix $\begin{psmallmatrix}\tau&-1\\1&0\end{psmallmatrix}\in SL_2$ that advances the recurrence (i) by one step. For $i\hat d_j$ the relation \eqref{eq:frieze} is the exchange relation $x_{j-1}x_{j+1}=x_j^2+1$ of the Kronecker cluster algebra, and $y=\hat d_{j-1}\hat d_{j+1}/\hat d_j^2$
of (iv) equals $1-\hat d_j^{-2}$; the $y$-recursion is second order and algebraic, and the $Y$-side
periodicity theorems are not available (Section~\ref{sub:rogers}), while the $A$-side recursion (i) is linear. The $A$-side is available at all
only because $s$ is a \emph{square} of a height ratio, and that square is the charge $Q_L=2$: the $2^{\#L}$ vacua of a class in
Section~\ref{sec:bethe} are precisely the two values $z_L=\pm\hat d_L$, the square roots of $\hat d_L^2=x_L^{-1}$, at each $L$.
\end{remark}

The trace $\tau$ is a conserved quantity of the recurrence, fixed by the boundary data; at $V^*$ and $\sigma V^*$ it is $2$ for every word (Theorem~\ref{prop:propagated}); at $\sigma V^*$ this uses that the height ratio has constant sign on each run (Lemma~\ref{lem:crossings}). The trigonometric case $|\tau|<2$ does occur, in Part~\ref{part:aminus}: there the
leading data is a pair of trigonometric chains, and \eqref{eq:closedform} solves the same relation \eqref{eq:frieze} closed at both ends,
$\mathcal Q_0=\mathcal Q_m=0$, which quantises $\tau=2\cos(k\pi/m)$. So Section~\ref{sec:nahm} and Part~\ref{part:aminus} are two solutions of one frieze
relation, distinguished only by the boundary condition: an eigenvector anchor at $\tau=2$ in the first, a closing condition with quantised
$\tau$ in the second.

The frieze relation \eqref{eq:frieze} is not special to the eigenvector solution. Inside a run it is implied by the Bethe equations
\eqref{eq:bethe} at the interior sites of the run (Theorem~\ref{thm:friezeall}), so it holds at every solution. It does not hold across a run boundary at any
of the words examined. So the boundary equations of the word glue the $2n_R$ frieze rows, and Theorem~\ref{thm:complete}
carries this out.

\begin{theorem}[The frieze relation holds at every solution of the Bethe equations]
\label{thm:friezeall}
Let $\varphi$ be any hyperbolic word and let $z$ be \emph{any} solution of the undeformed Bethe equations \eqref{eq:bethe} with all coordinates in $\C^\times$, with $x_b=z_b^{-Q_b}$. Put $s_b=x_b$ at an
$L$-step and $s_b=1-x_b$ at an $R$-step. Then at every site $b$ interior to its run,
\begin{equation}
\Bigl(\frac{1-s_b}{s_b}\Bigr)^{\!2}=\frac1{s_{b-1}s_{b+1}},\qquad\text{equivalently}\qquad \hat d_b^2\mp\hat d_{b-1}\hat d_{b+1}=1
\label{eq:friezeall}
\end{equation}
for $\hat d=s^{-1/2}$ on either branch. The same holds at a limit vacuum of the deformed system whose shapes $x_b$ at the sites of the run and their neighbours stay away from
$0$, $1$ and $\infty$; in general the polynomial form $s_{b-1}s_{b+1}(1-s_b)^2=s_b^2$ survives in the limit. In particular, for each run $\rho$ of length at least $3$, such a solution carries a frieze trace
$\tau_\rho:=(\hat d_{b-1}+\hat d_{b+1})/\hat d_b$ (Remark~\ref{rem:friezedef}), the same at every interior site $b$ of the run, defined up to sign (Remark~\ref{rem:friezesign}), which need not equal $2$ away from $V^*$.
\end{theorem}

\begin{proof}
At a site $b$ interior to an $L$-run, Theorem~\ref{thm:Kformula} gives $K_{bb}=4$ and $K_{b\pm1,b}=-2$ with no other entry in that column, and
$Q_b=Q_{b\pm1}=2$, so $(\Kraw)_{bb}=1$ and $(\Kraw)_{b\pm1,b}=-\tfrac12$. The Bethe equation at $b$, in the shape form \eqref{eq:bethex} with $Q_b=2$, is therefore
$(1-x_b)^2=x_b^2/(x_{b-1}x_{b+1})$, which is \eqref{eq:friezeall}.

At a site $j$ interior to an $R$-run the column is not tridiagonal --- the fan block is $2\min(i,j)$ --- so combine three equations instead. Since
$Q=1$ on $R$-steps,
\[
\frac{(1-x_j)^2}{(1-x_{j-1})(1-x_{j+1})}=\prod_ax_a^{(2K_{aj}-K_{a,j-1}-K_{a,j+1})/Q_a}=\prod_ax_a^{-\Delta_j^2K_{aj}/Q_a} ,
\]
where $\Delta_j^2g(j)=g(j-1)-2g(j)+g(j+1)$ is the second difference in $j$. Now $-\Delta_j^2\bigl(2\min(a,j)\bigr)=2\delta_{aj}$ for $a$ in the same fan; the last $L$ before the run contributes the constant $-2$, and every other
run a constant, both annihilated by $\Delta_j^2$; and the first $L$ after the run contributes the ramp $2j$, affine in $j$, also annihilated. So the
right side is $x_j^2$, and with $s=1-x$ this reads $s_j^2/(s_{j-1}s_{j+1})=(1-s_j)^2$, i.e.\ \eqref{eq:friezeall}. The same holds for the deformed equations, since $c_b=0$ at middle $L$'s and $c$ is affine on $R$-runs, so the factor $(1+\epsilon)^{c}$ cancels; hence the relation holds at every limit vacuum whose run coordinates stay in $\C^\times$. For the conservation of $\tau_\rho$, choose the signs of the $\hat d_b$ so that the relation reads $\hat d_b^2-\hat d_{b-1}\hat d_{b+1}=1$ at every interior site. To do so, go through the interior sites in turn; where the relation at $i$ has the wrong sign, flip $\hat d_{i+1}$, which affects only the relation at $i+2$. Two consecutive relations then give $\hat d_j(\hat d_j+\hat d_{j+2})=\hat d_{j+1}(\hat d_{j-1}+\hat d_{j+1})$, that is, $\tau_j=\tau_{j+1}$.
\end{proof}

\begin{remark}[The branch of the square root]
\label{rem:friezesign}
The sign sits on the product and not on the right-hand side: squaring
$(\hat d_b^2-1)^2=\hat d_{b-1}^2\hat d_{b+1}^2$ gives $\hat d_b^2-1=\pm\hat d_{b-1}\hat d_{b+1}$, and both
signs occur (at $L^3R^3$ both are found among the classes). Flipping the sign of $\hat d_b$ at alternate sites of a run preserves the relations and sends $\tau_\rho\mapsto-\tau_\rho$. Flipping it at a single site $b$ changes the sign of the relation at $b-1$ and $b+1$. So a run is determined, up to this sign, by its trace and its phase, or by two consecutive values.
\end{remark}

The frieze relation has two uses. First, \eqref{eq:friezeall} is a necessary condition on a Bethe vacuum that refers to nothing but the solution's own coordinates. So the
frieze residual is a necessary-condition filter. On the harvested words of Section~\ref{sec:spectrum} (Observation~\ref{obs:spectrumgen}) it separates genuine vacua, where it holds to $10^{-13}$ or better, from
spurious numerical solutions in the degenerate sector, where it is $10^{-3}$. Second, for runs of length at least $3$ it cuts a class from $r$ unknowns to $4n_R$, one trace and
one anchor for each of the $2n_R$ runs. The traces turn out to carry exactly the content of the $\mathcal H$-spectrum; see
Observation~\ref{obs:tauquant}.

\begin{proposition}[The run-chain reduction]
\label{prop:runreduce}
Let every run of $\varphi$ have length at least $2$. For a run of length $2$ put $\tau_\rho=(\hat d_1^2+\hat d_2^2-1)/(\hat d_1\hat d_2)$. Parametrise each run $\rho$ of length $n_\rho$ with $\tau_\rho\ne\pm2$ (the case $\tau_\rho=\pm2$ is the affine chain of Remark~\ref{rem:friezedef}(ii), which contains $V^*$ and has one unknown, the anchor) by
\[
s_j=\frac{\sin^2\vartheta_\rho}{\sin^2(j\vartheta_\rho+\phi_\rho)},\qquad j=1,\dots,n_\rho ,
\]
with $\vartheta_\rho,\phi_\rho\in\C$, so that $\tau_\rho=2\cos\vartheta_\rho$; by Theorem~\ref{thm:friezeall} and Remark~\ref{rem:friezedef}(ii) every solution has this form on such a run (the details are in the proof of Theorem~\ref{thm:complete}). Write $r_b=(1-x_b)^{Q_b}/\prod_ax_a^{K_{ab}/Q_a}$ for the residual of \eqref{eq:bethe} at $b$. Then the
frieze relation \eqref{eq:friezeall} holds identically in $(\vartheta_\rho,\phi_\rho)$ at every site interior to a run; at an interior $L$-site it is
the Bethe equation $r_b=1$ (in which $Q_b=2$ already squares the coordinate), and at an interior $R$-site it is $r_j^2=r_{j-1}r_{j+1}$. For a word
with two runs the squared system $\{r_L=1,\ r_R^2=1\}$ therefore reduces to the four equations at the run ends, in the four unknowns $(\vartheta_L,\phi_L,\vartheta_R,\phi_R)$, up to an
$(n_\rho-1)$-st root of unity in $r_j^2$ at the interior $R$-sites; adding the equation at the second $R$-site removes it.
\end{proposition}

\begin{proof}
The frieze relation at an interior site is the product-to-sum identity
$\sin((j-1)\vartheta+\phi)\sin((j+1)\vartheta+\phi)=\sin^2(j\vartheta+\phi)-\sin^2\vartheta$. By the proof of Theorem~\ref{thm:friezeall} it is $r_b=1$ at an
interior $L$-site, and at an interior $R$-site $r_j^2/(r_{j-1}r_{j+1})$ is the frieze relation, the product of the $x_a$ being $x_j^2$
identically. So along the $R$-run $(r_j)$ is geometric with ratio $\gamma$, say; $r_1^2=r_{n_\rho}^2=1$ gives $\gamma^{2(n_\rho-1)}=1$, and
$r_2^2=1$ gives $\gamma^2=1$.
\end{proof}

The residuals $r_b$ have integer exponents, since every entry of $K$ is even (Theorem~\ref{thm:Kformula}) and $Q_a\in\{1,2\}$. So the squaring
at the $R$-sites is not needed to remove a square-root ambiguity. It is a convenience: the frieze relation is a statement about $r_j^2$. The
unsquared equations at every $R$-site follow from those at $R_1$ and $R_2$, since along an $R$-run $r_1=r_2=1$ gives $\gamma=1$.

\begin{remark}[The reduction is sharp]
\label{rem:runsharp}
The run-end equations, together with the equation at the second $R$-site of each run, are the quantisation. Take a two-run word with
both runs of length $\ge2$. By Theorem~\ref{thm:complete} the non-degenerate solutions of the unsquared equations are exactly the
kernel-vector classes of Theorem~\ref{thm:allwords}, together with $V^*$ and $\sigma V^*$. The run-end equations alone would leave
spurious solutions, with $\gamma$ a nontrivial root of unity in the proof of Proposition~\ref{prop:runreduce}. The multiplier $\kappa$ of
Section~\ref{sub:runtraces} is therefore a property of this small system in the four run parameters, and not of the $r$-dimensional one. The sign flip of Remark~\ref{rem:friezesign} sends $\vartheta\mapsto\pi-\vartheta$ and
$j\mapsto N-j\equiv-(j+N)\pmod{2N}$, so the two square roots of $s$ along a run are the two labels $\{j,j+N\}$ of \eqref{eq:transN}. This is the pair
by which Corollary~\ref{cor:templatelabel} is taken modulo $\tfrac14$.
\end{remark}

\subsection{Two further explicit families}
\label{sub:families}

The chain is now a theorem, and Theorem~\ref{prop:propagated} gives the anchors. So the $\#L=1$ and
$\#L=2$ families can be written out. Here $d_{L_j}$, $d_{R_j}$ are the Nahm variables $d_b=\delta_b^{-1/2}$ at the $j$-th $L$-site and the
$j$-th $R$-site. For $\varphi=LR^{r-1}$, $r\ge3$, set $m=\sqrt{r-1}$ (in this subsection only; $m$ need not be an integer), so that $k=r+1$ and
$k^2-4=(r-1)(r+3)$. Then
\begin{equation}
\begin{gathered}
d_{L_1}^2=\frac12+\frac{m^4-1+m^3\sqrt{4+m^2}}{2(2m^2-1)},\\
d_{R_1}=\frac{2-m^2}2+\frac m2\sqrt{4+m^2}=\lambda+2-k,\qquad d_{R_{1+i}}=d_{R_1}+i ,
\end{gathered}
\label{eq:LRnclosed}
\end{equation}
or, as a function of $r$ alone,
\[
d_{L_1}(r)^2=\frac{(r-1)^2+2r-4+(r-1)^{3/2}\sqrt{r+3}}{2(2r-3)},\qquad d_{R_1}(r)=\frac{3-r}2+\frac{\sqrt{(r-1)(r+3)}}2 .
\]
For $\varphi=L^2R^n$, $r=n+2$, $k=2n+2$ and $k^2-4=4n(n+2)$,
\begin{equation}
\begin{gathered}
d_{L_1}^2=\frac{n+2}2,\qquad
d_{R_1}=\frac{(2-n)+\sqrt{n^2+2n}}2=\frac{\lambda+1-2n}2,\qquad d_{R_{1+i}}=d_{R_1}+i,\\[2pt]
d_{L_2}^4=\frac{d_{R_1}+n-1}{d_{L_1}^2\,d_{R_1}^2\,(d_{R_1}-1)} .
\end{gathered}
\label{eq:L2Rnclosed}
\end{equation}
Both are the propagated eigenvector, $\delta_b=1-x_b$ with $x$ as in \eqref{eq:propagated}, for every $r$ and every $n$, and $c_{\eff}=1$ and $2$
by Theorem~\ref{thm:ceff}: the dilogarithm sum of \eqref{eq:LRnclosed} and \eqref{eq:L2Rnclosed} is evaluated exactly, since
$\sum_bL(\delta_b)=\#L\cdot L(1)$ follows from the telescoping identity \eqref{eq:rogtel} alone, through $L(\delta_L)=L(1)-L(s_L)$ and
$L(\delta_R)=L(s_R)$, with no need to know the individual values. The closed form, and $c_{\eff}=1$, hold for every $r\ge3$.

Both families yield identities in $r$ and $n$. In the second family $\delta_{L_1}=2/(n+2)=4/(k+2)$ is \emph{rational} for every $n$; it is the
only charge-$2$ coordinate in either family that is. In the first, summing the first row and then the remaining rows of $M=QK^{-1}Q$
gives
\[
1-\delta_{L_1}=\delta_{L_1}^2\,\delta_{R_1}/\delta_{R_{r-1}},\qquad \prod_{j=1}^{r-1}(1-\delta_{R_j})=\sqrt{\delta_{R_1}}=1/d_{R_1} .
\]

\begin{remark}[The Jeong--Sohn presentation is an index-two quotient]
\label{rem:jsframe}
For $\varphi=L^2R^n$ the presentation used by Jeong--Sohn, with level matrix $K^{\mathrm{JS}}$ of \cite[(2.12)]{JS26}, has $r=n+2$ but charge vector $Q^{\mathrm{JS}}=(1,2,1,\dots,1)$ --- \emph{one} charge-$2$
field --- where Theorem~\ref{thm:Kformula} gives $Q=(2,2,1,\dots,1)$, two of them, and $\det K^{\mathrm{JS}}=2^{r-2}$ where the pipeline gives
$\det K=2^r$. The two are nevertheless one system. Writing $D$ for the map $e_1\mapsto2e_2$, $e_2\mapsto e_1$, $e_j\mapsto e_j$ ($j\ge3$), one has
\begin{equation}
K=D\,K^{\mathrm{JS}}D^{\mathsf T},\qquad QP=DQ^{\mathrm{JS}},\qquad \det D=-2,
\label{eq:jsframe}
\end{equation}
with $P$ the transposition of the first two sites, the two $L$-sites. For the first identity, order our sites as $L_1$ (the first $L$ after the
$R$-run), $L_2$, $R_1,\dots,R_n$. Theorem~\ref{thm:Kformula} gives $K_{L_1L_1}=2n+4$, $K_{L_2L_2}=4$, $K_{L_1L_2}=-4$, $K_{L_1R_j}=2j$,
$K_{L_2R_j}=-2$ and $K_{R_iR_j}=2\min(i,j)$. In $DK^{\mathrm{JS}}D^{\mathsf T}$ the row and column of $L_1$ are those of Jeong--Sohn's site $2$, those of
$L_2$ are twice those of their site $1$, and the $R$-block is unchanged; by \cite[(2.12)]{JS26} the entries are $2n+4$, $4\cdot1$, $2\cdot(-2)$,
$2j$, $2\cdot(-1)$ and $2\min(i,j)$, for every $n\ge1$. The second identity is $DQ^{\mathrm{JS}}=(2,2,1,\dots,1)=QP$. Hence
\[
QK^{-1}Q=P\,Q^{\mathrm{JS}}(K^{\mathrm{JS}})^{-1}Q^{\mathrm{JS}}P:
\]
the two presentations have the \emph{same} matrix $M=QK^{-1}Q$ up to relabelling $L_1\leftrightarrow L_2$, so the same Nahm ladder
$\mathbf N=2(I-M)$, the same positive solution and the same $c_{\eff}$. The Bethe equations are not literally the same. In the coordinates
$(z_{L_2}^2,z_{L_1},z_{R_1},\dots,z_{R_n})$ the pipeline equation at $L_2$, where $Q=2$, is the square of the equation at
site~$1$, where $Q=1$, of the system built from Jeong--Sohn's presentation \cite[(2.12)]{JS26}, and the other equations agree. So every solution of that system solves ours, but not conversely: numerically, one
of the three solutions found at $L^2R$, and three of the five at $L^2R^3$, give $-1$ in its site-$1$ equation. Since the matrices $M$ agree, Jeong--Sohn's ``$\varphi=L^2R^n$'' is the word
$L^2R^n$ in the sense of Section~\ref{sub:word}, and the positive solution of that system is the propagated eigenvector of
Theorem~\ref{prop:propagated}, in the site order $d_1\mapsto L_2$, $d_2\mapsto L_1$, $d_{2+j}\mapsto R_j$; it is \eqref{eq:L2Rnclosed}.

What differs is the lattice, not the system. $D$ has elementary divisors $(1,\dots,1,2)$, so the pipeline's $U(1)^r$ charge lattice is an
index-two sublattice of Jeong--Sohn's, and $\det K/\det K^{\mathrm{JS}}=(\det D)^2=4$. The quantities that see $K$ and $Q$ separately rather than only
through $M$ therefore differ by that $\Z_2$: a class carries $2^{\#L}$ vacua in the pipeline presentation and $2^{\#L-1}$ in theirs, and the pipeline
system has the extra solutions just described. Since $\det K=2^r$ is the normalisation selected in Section~\ref{sub:charges}
(Theorem~\ref{thm:Kformula}), the pipeline presentation is the one attached to the layered triangulation, and the Jeong--Sohn presentation is its $\Z_2$
quotient. The torsion formula of Section~\ref{sec:torsion} and the Gauss-sum and spin computations of Section~\ref{sec:spins}, all of which turn on the
factor $2^{\#L}$, are statements in our unquotiented presentation.
\end{remark}

\emph{What this section establishes.} The positive vacuum $V^*$ of every word is given in closed form by heights (Theorem~\ref{prop:propagated}), its Wilson line is the stretch factor (Proposition~\ref{obs:stretch}), and the Rogers sum is rational (Corollary~\ref{cor:rational}); the value $c_{\eff}=\#L$ is Theorem~\ref{thm:ceff}. Along each run the Bethe equations imply the frieze relation at every solution (Theorem~\ref{thm:friezeall}), the linear chain being its case $\tau=2$ at $V^*$ and $\sigma V^*$, and the relation reduces a two-run word to its run-end equations (Proposition~\ref{prop:runreduce}); Sections~\ref{sec:bethe} onwards use this to classify all the vacua.

\section{The Bethe frame}
\label{sec:bethe}

This section sets up the third frame, in which the observables of the TQFT, the handle-gluing and fibering operators, are evaluated. It recalls from
\cite{GKL23} the handle-gluing operator at the flux point, \eqref{eq:Habelian}: a function on the solutions of the Bethe equations \eqref{eq:bethe} with no free
parameters. It then explains why the Bethe system has to be deformed. At the flux point some solutions lie at infinity. They are recovered by
deforming the right-hand side to $(1+\epsilon)^{c_b}$ and letting $\epsilon\to0$; the precise definition of a Bethe vacuum is given below. The
deformation is analogous to the device of Jeong--Lee \cite{JL25}, in the present normalisation. Their closed forms do not transfer, because their
theories gauge $SU(2)_L$ and $SU(2)_R$ at separate levels, and so have a different Bethe system. Finally the section fixes the terminology used in the rest of the paper. A solution of the shape
equations determines a class of $2^{\#L}$ Bethe vacua, which differ in the signs of the charge-$2$ coordinates. The classes, the unit, the
positive vacuum $V^*$ and the two identity classes are defined here, and counted.

From (2.12)--(2.13) of \cite{GKL23}, the handle-gluing operator is
$\mathcal H_\alpha=\det_{ab}(\partial_{Z_a}\partial_{Z_b}\mathcal W_0)\,e^{-2\mathcal W_1}$ with
\begin{align}
\mathcal W_0&=\sum_a\Li_2(e^{-Q_aZ_a})+\tfrac12Z^{\mathsf T}KZ+\bigl(\mathfrak m+i\pi(\nu-1)\bigr)\Bigl[(r-1)Z_1+\sum_{a=1}^{r-1}aZ_{a+1}\Bigr],\\
\mathcal W_1&=-\tfrac12\sum_a\log(1-e^{-Q_aZ_a})+\tfrac12(\nu-1)\Bigl[(r-1)Z_1+\sum_{a=1}^{r-1}aZ_{a+1}\Bigr].
\end{align}
Here $Z_a=\log z_a$ are the logarithms of the Bethe coordinates, $\mathfrak m$ is the real mass of the $U(1)_A$
symmetry and $\nu$ the R-symmetry mixing parameter \cite[(2.3)--(2.4)]{GKLSY21} (the ``flux point'' of the twisted index is $(\mathfrak m,\nu)=(0,-1)$). These are the formulae of \cite{GKL23} for $LR^n$, where the
bracket is $c\cdot Z$; for a general word it is replaced by $c\cdot Z$ with $c$ the charge vector of Theorem~\ref{thm:Kformula}, $W_c=e^{c\cdot Z}=\prod_az_a^{c_a}$ being the value at a Bethe vacuum of the BPS Wilson line of gauge charge $c$ \cite[(2.22)--(2.24)]{GKL23}. At $\nu=-1$ this gives
\begin{equation}
\mathcal H_\alpha=\prod_a\bigl(1-z_a^{-Q_a}\bigr)\prod_az_a^{2c_a}\cdot\det\Bigl(K+\mathrm{diag}\tfrac{Q_a^2}{z_a^{Q_a}-1}\Bigr),
\label{eq:Habelian}
\end{equation}
with no free parameters, the Bethe equations being $\exp(\partial_{Z_b}\mathcal W_0)=1$, i.e.\ \eqref{eq:bethe}. The fibering operator
is, in the same normalisation,
\begin{equation}
\mathcal F_\alpha=\exp\Bigl(-\frac{\Psi}{2\pi i}\Bigr),\qquad
\Psi=\sum_b\bigl[\Li_2(x_b)-Q_bZ_b\log(1-x_b)\bigr]-\tfrac12Z^{\mathsf T}KZ,\qquad x_b=e^{-Q_bZ_b},
\label{eq:Fdef}
\end{equation}
evaluated at the vacuum with the branches of $Z_b$ and $\log(1-x_b)$ of the log-form Bethe equations
$Q_b\log(1-x_b)=-(KZ)_b+2\pi in_b$ (Section~\ref{sec:spins}); at a class reached only through the deformation, $\mathcal F_\alpha$ is the limit
of the same expression along the continuation in $\epsilon$. The \emph{spin} of the vacuum $\alpha$ is
\begin{equation}
\operatorname{spin}(\alpha)=\frac{\arg\mathcal F_\alpha}{2\pi}=\frac{\operatorname{Re}\Psi}{4\pi^2}\in\R/\Z .
\label{eq:spindef}
\end{equation}
Under the dictionary \eqref{eq:dict}, $\mathcal F_\alpha/\mathcal F_0$ is the $T$-eigenvalue of $\alpha$, so the conformal spin
$\theta_\alpha$ of the simple object is the difference of the spin of $\alpha$ and the spin of the unit (Section~\ref{sub:relation}). The
construction of \cite[(2.2), (2.5), (2.7)]{GKL23} includes a superpotential. For the present purposes it enters only through the charge vector $c$. The superpotential
determines $c$ by fixing the weights of the fields and leaving $U(1)_A$ as the one global symmetry, so \eqref{eq:Habelian} and the Bethe
equations are already the complete data.

\paragraph{The degeneration, and why a deformation is needed.} At $\nu=-1$ the system is degenerate. At $LR^3$, with $k=\tr\varphi=5$, it has exactly ten Bethe vacua, five classes, where Gang--Kim
find $2k+2=12$ \cite[(2.11)]{GK22}. The two missing vacua are genuinely absent from the system rather than from the solver. To see this, eliminate
the other coordinates from \eqref{eq:bethe} and remove the spurious factors. This leaves, in the coordinate $z=z_{R_3}$ of the last site,
$(5z^2+3z-9)(z^3-4z^2+3z+1)$. The quadratic carries the two identity classes and the cubic the three $A_+$ classes; the $A_-$ class is missing.
It is a limiting solution at infinity, recovered by
deforming the right side of \eqref{eq:bethe} from $1$ to $(1+\epsilon)^{c_b}$ away from the degenerate point $\epsilon=0$ and continuing. Throughout this paper ``Bethe vacuum'' means a solution of the deformed
system continued to $\epsilon\to0$; the deformation is the one given by the charge vector $c$ of the Farey data.

The degeneration and its cure are not special to the words considered here. Jeong--Lee \cite[\S2.3]{JL25} observe that at $\nu=-1$ exactly, the adjoint
chirals of the $T[SU(2)]$ segments make $\mathcal W_1$ divergent and the handle-gluing operator singular, and prescribe the double-scaling ansatz
\[
\nu=-1+\epsilon,\qquad Z_i=(-1)^{G_i}W_i+f_i\epsilon+O(\epsilon^2),\qquad G_i\in\Z_2,
\]
their (2.16), defining the observables in the $\epsilon\to0$ limit. Our deformation by $(1+\epsilon)^{c_b}$, which deforms the real mass rather than $\nu$, is an analogous device in the Farey normalisation. The strata of Observation~\ref{obs:strata} are the analogue of their branch labels $G_i$.

Only the device transfers, not the closed forms that follow it. $D(\vec k)$ joins $n-1$ segments of $T[SU(2)]$ and gauges $SU(2)_L$ and $SU(2)_R$ at
separate levels $k_i$, where the $S$-fold of \cite{GK22} gauges the diagonal. The monodromies of the two lie in the same conjugacy class up to inversion, but the Bethe equations differ. Write
$\varphi(\vec k)=T^{k_1}ST^{k_2}S\cdots ST^{k_n}$ in their conventions, and read off the periodic continued fraction of its fixed point.
Then $\vec k=(2,3)$ gives $\varphi\sim L^3R\sim(LR^3)^{-1}$ with $p=\varphi_{12}=5$ (for $S=\begin{psmallmatrix}0&-1\\1&0\end{psmallmatrix}$), so their (2.22) predicts $(|p|-1)2^{n-1}=8$ vacua with
$\mathcal H=|p|\sin^{-2}(\pi\alpha/p)$, two distinct values; Proposition~\ref{prop:Hspectrum} and Theorem~\ref{thm:torsion} give $12$ vacua and four distinct values of $\mathcal H$.
Both satisfy $\sum\mathcal H^{-1}=1$, so no rescaling relates them. This is as their own factorisation
$D(\vec k)|_A=\mathcal D(p,q)|_A\otimes\mathrm{TFT}[\vec k]$ requires. The integer $p=\varphi_{12}$ is not a conjugacy invariant: $\vec k=(1,3)$, $(2,2)$ and
$(2,2,1)$ give $\varphi\sim L^2R$, $L^2R$ and $LR^2$, with $p=2,3,1$, and $L^2R\sim(LR^2)^{-1}$. Their $\mathcal H$ therefore depends on $\vec k$ beyond the monodromy, whereas every invariant here
depends on the word alone.

\paragraph{The class structure.} Each solution $x$ of the shape equations \eqref{eq:bethex} determines the $z_a$ for $a$ an $R$-tetrahedron, and the charge-$2$ coordinates only
through $z_L^{Q_L}=z_L^2$; so each class of solutions, defined below, carries $2^{\#L}$ vacua, one per sign choice $\varepsilon\in\{\pm1\}^{\#L}$. This splitting
is not a redundancy: the $2^{\#L}$ vacua are distinct, and Section~\ref{sec:spins} computes the spin as a function of
$\varepsilon$. (Spins within a class differ by elements of $\tfrac14\Z$, so for $\#L\ge3$ some of them coincide.) It is why the factor $2^{\#L}$ appears in the torsion formula and in the $\mathcal H$-spectrum.

Three terms are fixed here. A Bethe vacuum, as just defined, is under the dictionary \eqref{eq:dict} a
simple object of $\mathrm{TFT}[\varphi]$. A \emph{class} is the set of the $2^{\#L}$ Bethe vacua that have the same solution $x=(x_a)$ of the shape equations \eqref{eq:bethex}, the shapes $x_a=z_a^{-Q_a}$ being one number for each tetrahedron. The classes are the simple objects of
$\mathrm{TFT}_0$ (Section~\ref{sec:modular}), and Theorem~\ref{thm:count} counts them. (``Vacuum'' always means Bethe vacuum, except in ``the vacuum of the physics literature'' for the unit.) The \emph{unit} is the simple object $0$ of the modular data,
called the vacuum in the physics literature; input~(b) says which class contains it. $V^*$ is the \emph{positive vacuum}: the vacuum of the positive class with $z_a>1$ at every site (equivalently, all shapes $x_a\in(0,1)$). It is the only vacuum with $z_a\ge1$ at every site. Other classes can have all $z_a$ real and positive, as one $A_+$ class of $LR^3$ does. The $2^{\#L}$ vacua of a class differ only in the signs of the $z_a$ at the $L$-sites. At an $R$-site $z_a=x_a^{-1}$ is fixed by the shape, and it is negative where $x_a<0$. The \emph{all-plus} vacuum of a class is the one with $z_a>0$ at every $L$-site. So $V^*$ is both all-plus and all-positive: it is the all-plus vacuum of its class, and every coordinate is positive because every shape lies in $(0,1)$. The \emph{Galois-conjugate class} is the class whose shapes are the Galois conjugates of the shapes of $V^*$; some of these shapes are negative. We write $\sigma V^*$ for its all-plus vacuum. Applying the Galois conjugation to the coordinates of $V^*$ themselves also gives a vacuum of this class, but it differs from $\sigma V^*$ in the sign of $z_L$ at every $L$-site (Theorem~\ref{prop:propagated}). The unit is not $V^*$:
under (b) it is $\sigma V^*$, and $V^*$ is the dominant saddle of the half-index (Proposition~\ref{prop:vstar}). Their classes are the two
\emph{identity classes}, with $\mathcal H=2^{\#L}\lambda^{\pm1}|A_-||A_+|$; they correspond to the trivial character, counted once in each
sector (Theorem~\ref{thm:count}), and to the objects $o$, $b$ of $\mathrm{TFT}_0$ in some order. $\mathcal H$ and the Rogers sum depend only on the
class and are indexed by classes; the spin is indexed by Bethe vacua. For odd trace $k$ the expected numbers are $k+1$ classes and $2^{\#L}(k+1)$ Bethe vacua: Theorem~\ref{thm:count} counts the
characters, Theorem~\ref{thm:complete} gives this count of non-degenerate classes when no run is a singleton and every nonzero kernel vector is non-degenerate, and
(F1) gives it in general. For $L^mR^m$ with $m$ odd the expected numbers are $m^2+3$ and $2^m(m^2+3)$, observed at $m=3,5,7,11$.

\part{Results for every word}
\label{part:results}

This part proves what the Bethe side of the construction gives for every hyperbolic word. Everything here is a statement about the
pipeline --- the level matrix $K$ and the charges $Q$ and $c$ of Theorem~\ref{thm:Kformula}, and the Bethe system \eqref{eq:bethe} they
define --- with the exceptions listed in Section~\ref{sub:status}. The main ones are Section~\ref{sec:farey24}, which assumes the
half-index dictionary of \cite{GKL23} and input~(b) and says so in its statements, and Corollary~\ref{cor:interpretation}, at the end of
Section~\ref{sec:count}, which reads the results through (F1). Results labelled as observations are numerical; each says which of its parts are proved, and where.

Section~\ref{sec:ceff} proves $c_{\eff}=\#L$, by a dilogarithm identity at the two eigenvector solutions $V^*$ and $\sigma V^*$.
Section~\ref{sec:farey24} derives from it, through the dictionary, the central charge $c\equiv\#L-6\#R\pmod{24}$. Section~\ref{sec:torsion}
proves the torsion formula $\mathcal H=2^{\#L}\lambda^{\pm1}|A_-||A_+|$ at the two identity classes, through a periodic Jacobi matrix and Hill's
formula (details in Appendix~\ref{app:torsion}). Section~\ref{sec:spectrum} gives every non-degenerate class in closed form, one for each pair $\pm v$ of kernel vectors of $\varphi\mp1$ modulo $N_\mp$ for which the function $f(e)=\sin(\pi\det(v,e)/N_\mp)$ does not vanish on any middle slope (Theorem~\ref{thm:allwords}; the middle slope $c_1+c_2$ of a step is defined in Section~\ref{sub:propagated}), and proves $\mathcal H=2^{\#L}|A_\mp|$ there under the non-degeneracy condition (Theorem~\ref{thm:Hall}).
Section~\ref{sec:count} counts the fixed reducible characters, $k+1$ for odd trace $k$, and proves that for words with no singleton run these
classes, with the two identity classes, are all the non-degenerate classes (Theorem~\ref{thm:complete}). Section~\ref{sec:spins} treats the
spins. Within a class they form the discriminant form of $K_{LL}$ on the sign group. Across classes they are governed by the Rogers sum $c_\alpha$ of the class,
through $24\theta_\alpha\equiv c_\alpha$. Here $c_\alpha$ is neither the charge vector nor the central charge. The Rogers sums are evaluated at the non-degenerate classes (Theorem~\ref{thm:rogers}) and on full chains
(Section~\ref{sub:fullchain}).

What this part does not reach are the degenerate classes: those whose kernel vector has $f$ vanishing on some middle slope, equivalently
(Theorem~\ref{thm:complete}) the classes some of whose coordinates are $0$, $1$ or $\infty$ in the limit $\epsilon\to0$; Section~\ref{sub:complete}
makes the equivalence precise. In this part ``non-degenerate'', for a kernel vector or a class, is used in this sense. It is distinct from the \emph{non-degeneracy condition} $\det(F+D_R)\ne0$ of Theorem~\ref{thm:Hall}. For $L^mR^m$
they make up the whole $A_-$ sector, which is the subject of Part~\ref{part:aminus}.

\section{The effective central charge}
\label{sec:ceff}

This section proves $c_{\eff}=\#L$ for every hyperbolic word (Theorem~\ref{thm:ceff}). By the asymptotics of Nahm sums the effective
central charge is the Rogers dilogarithm sum over the positive vacuum, so the theorem is a dilogarithm identity at $V^*$:
$\sum_{\ell_i=L}L(s_i)=\sum_{\ell_i=R}L(s_i)$. Rogers' duplication formula makes each step telescope
in the height ratio of the Farey basis \eqref{eq:rogtel}, and the periodicity of the expanding eigenvector closes the sum, so the identity
holds at $V^*$. At the contracting eigenvector, which gives $\sigma V^*$, the functional equations of the extended dilogarithm hold only up to jumps of $\pm\pi^2$.
Lemma~\ref{lem:crossings} locates these jumps at the last letter of each run, so they cancel and the identity holds there too. Two consequences
are used later: the Rogers sum of either identity class is $\#R$ (Section~\ref{sec:spins}), and the negative coordinates of $\sigma V^*$
are the last $R$ of each $R$-run (Section~\ref{sec:farey24}).

\subsection{The Rogers identity at both eigenvectors}
\label{sub:rogers}

Here $c_{\eff}$ is the growth exponent of the half-index $I(q)$ of Section~\ref{sec:farey24}: $\log I(e^{-\eta})\sim\pi^2c_{\eff}/(6\eta)$ as
$\eta\to0^+$. Under the half-index dictionary it is the effective central charge $c-24h_{\min}$ (Section~\ref{sec:farey24}). The half-index is the Nahm sum with matrix $\Kraw$ over the sublattice $Q\Z^r_{\ge0}$, in the variables $x$; its saddle is $V^*$ (Proposition~\ref{prop:vstar}). In Zagier's convention \cite[(36)]{Zagier07} the effective central charge of the Nahm sum with matrix $A$ is $\sum_b(L(1)-L(x_b))/L(1)$ with $1-x=x^A$, so here $c_{\eff}=\sum_bL(\delta_b)/L(1)$ with $\delta=1-x$ \cite[\S3]{Nahm04}. The sublattice of index $2^{\#L}$ does not change the growth. (The Nahm sum with matrix $\Kraw^{-1}$ in the variables $\delta$ has instead effective central charge $\sum_bL(x_b)/L(1)$.) The standard route to such sums is unavailable here. Nakanishi's dilogarithm identities \cite{Nakanishi11} need a $Y$-pattern that is periodic as a family of rational functions.
For the once-punctured torus the mutation sequence is of infinite type and the $Y$-variables are not periodic as rational functions. Only their values at the eigenvector (the specialisation by the heights of Theorem~\ref{prop:propagated})
are invariant under $\varphi$. A direct argument replaces the periodicity.

Let $(c_1,c_2)$ be the columns of the partial product $M_{i-1}$ as in Theorem~\ref{prop:propagated}, the Farey basis at step $i$, and let
$t=h(c_1)/h(c_2)$ be their height ratio. Write $t_{i-1}$ for $t$ at step $i$, so that $t_0=h(e_1)/h(e_2)$ and $t_r$ is the ratio after the last of the $r$ letters of the word. Put $G(t)=2L\bigl(t/(1+t)\bigr)$ for $L$ the Rogers dilogarithm. Rogers'
duplication formula $L(x^2)=2L(x)-2L\bigl(x/(1+x)\bigr)$ on $(0,1)$ gives, for $t>0$,
\[
L\text{-step }(t\mapsto\tfrac t{1+t}):\ L(s)=G(t)-G\bigl(\tfrac t{1+t}\bigr);\qquad
R\text{-step }(t\mapsto t+1):\ -L(s)=G(t)-G(t+1),
\]
the second using also the reflection identity. Hence
\begin{equation}
\sum_{\ell_i=L}L(s_i)-\sum_{\ell_i=R}L(s_i)=G(t_0)-G(t_r)
\label{eq:rogtel}
\end{equation}
for \emph{any} positive height function, and for the expanding eigenvector, where $h(\varphi v)=\lambda h(v)$, one has $t_r=t_0$. So the
identity holds at the positive solution $V^*$ of Theorem~\ref{prop:propagated} for every hyperbolic word.

At the contracting eigenvector $\eta'$ the components have opposite signs (see the proof of Lemma~\ref{lem:crossings}), so the height $h'$ takes both signs on the Farey slopes and the ratio $t$ is negative. The step formulae above are then applied to the extended dilogarithm $L(x)=\operatorname{Re}R(x)$ on $\R\smallsetminus\{0,1\}$, with $R(x)=\Li_2(x)+\tfrac12\log x\log(1-x)$ and principal branches (Section~\ref{sub:rogersnd}).
This $L$ is continuous at $0$ and $1$, real-analytic on $\R\smallsetminus\{0,1\}$, and tends to $-\pi^2/6$ at $-\infty$ and to $\pi^2/3$ at $+\infty$; so $G(t)=2L(t/(1+t))$ jumps by $\pi^2$ across $t=-1$. The defect of each step identity is therefore locally constant in $t$ away from the points where the step's own $G$-terms jump: $t=-1$ and $t=-\tfrac12$ for an $L$-step (through $G(t)$ and $G(t/(1+t))$), and $t=-1$ and $t=-2$ for an $R$-step (through $G(t)$ and $G(t+1)$). It is $0$ in the limits $t\to0^-$ and $t\to-\infty$. Consequently the $L$-step identity acquires a defect $+\pi^2$ exactly when $t\in(-1,-\tfrac12)$, and the $R$-step identity a defect $-\pi^2$ exactly when $t\in(-2,-1)$. These are the steps whose M\"obius maps carry $t$ across $-1$ downward and upward. The defects cancel, and the reason is combinatorial.

\begin{lemma}[Crossings]
\label{lem:crossings}
For every hyperbolic word, along the orbit $t_0,t_1,\dots,t_r=t_0$ of the height ratio at the contracting eigenvector: every $L$-step has $t\in(-1,0)$ and every $R$-step has $t<-1$; an $L$ is followed by an $R$
iff $t\in(-1,-\tfrac12)$, and an $R$ by an $L$ iff $t\in(-2,-1)$. The downward crossings of $-1$ are exactly the last $L$ of each $L$-run, the upward
crossings exactly the last $R$ of each $R$-run, and the set of negative $R$-coordinates is the set $\mathcal E$ of last $R$'s of the $R$-runs.
\end{lemma}

\begin{proof}
$t_i=h'(c_1)/h'(c_2)$, with $h'$ the height for the contracting eigenvector $\eta'=(\eta'_1,\eta'_2)$, is the orbit of the initial ratio
$t_0=\eta'_1/\eta'_2$ (the columns of $M_0=1$ being $e_1,e_2$) under the M\"obius maps $L:t\mapsto t/(1+t)$ and $R:t\mapsto t+1$, periodic because
$\eta'\varphi=\lambda^{-1}\eta'$. Since $\varphi$ has positive entries and $\lambda^{-1}<\varphi_{11}$, one has
$\eta'_1=\varphi_{21}/(\lambda^{-1}-\varphi_{11})<0$, so $t_0<0$. Both maps send $(0,\infty)$ into itself, so the orbit never becomes positive, or it
could not return to $t_0$. Hence an $L$ never occurs at $t<-1$ (it would land in $(1,\infty)$) and an $R$ never at $t\in(-1,0)$ (it would land in
$(0,1)$). An $L$ at $t\in(-1,-\tfrac12)$ lands below $-1$, where another $L$ is forbidden; at $t\in(-\tfrac12,0)$ it lands in $(-1,0)$, where an $R$ is
forbidden. Symmetrically an $R$ at $t\in(-2,-1)$ lands in $(-1,0)$, forcing an $L$, and at $t<-2$ lands below $-1$, forcing an $R$. Finally an $R$-step
coordinate is negative iff $t\in(-2,0)$, which, since every $R$-step has $t<-1$, means $t\in(-2,-1)$. The case splits are exhaustive: $t_0\in\Q(\lambda)\smallsetminus\Q$, since $\tr^2\varphi-4$ is not a square, and integral M\"obius maps preserve irrationality, so no $t_i$ equals $-1$, $-\tfrac12$ or $-2$.
\end{proof}

By the definition of a word in Section~\ref{sub:word}, the $L$-runs and $R$-runs of the cyclic word are equal in number, so the $+\pi^2$ and $-\pi^2$ defects cancel exactly. No translation-number argument is needed:
periodicity and the invariance of $(0,\infty)$ suffice.

\begin{theorem}[The effective central charge]
\label{thm:ceff}
For every hyperbolic word, the Rogers identity $\sum_{\ell_i=L}L(s_i)=\sum_{\ell_i=R}L(s_i)$ holds at both eigenvectors (at the contracting
one with the extended dilogarithm, whose functional equations carry the defects above), and hence
\[
c_{\eff}=\#L .
\]
\end{theorem}

\begin{proof}
At $V^*$ the identity is \eqref{eq:rogtel} with $t_r=t_0$. At the contracting eigenvector, by Lemma~\ref{lem:crossings} an $L$-step carries a defect $+\pi^2$ exactly when it is the last $L$ of an $L$-run, and an $R$-step a defect $-\pi^2$ exactly when it is the last $R$ of an $R$-run. Each run therefore contributes one defect, and since there are as many $L$-runs as $R$-runs they cancel. For $c_{\eff}$: in the Nahm
variables $\delta_b=1-x_b$ of Proposition~\ref{prop:bethenahm}, with $x=s$ at $L$-steps and $x=1-s$ at $R$-steps, the reflection identity
$L(1-y)=L(1)-L(y)$ on $(0,1)$, which holds for all real $y\ne0,1$ for the extended $L=\operatorname{Re}R$ since $\operatorname{Re}R(y)+\operatorname{Re}R(1-y)=\pi^2/6$, gives at $V^*$
\[
\sum_bL(\delta_b)=\sum_{\ell_i=L}\bigl(L(1)-L(s_i)\bigr)+\sum_{\ell_i=R}L(s_i)=\#L\,L(1),
\]
so $c_{\eff}=\#L$. The half-index is a Nahm sum in $Q\mathfrak m$ over the sublattice $Q\Z^r$ of index $2^{\#L}$. This does not change the exponential growth: the saddle-point analysis of Proposition~\ref{prop:vstar} works directly in $\mathfrak m$, and only the constant prefactor depends on the lattice. In the same way $\sum_bL(x_b)=\sum_{\ell_i=L}L(s_i)+\sum_{\ell_i=R}\bigl(L(1)-L(s_i)\bigr)=\#R\,L(1)$: the Rogers sum
$c_\alpha=\sum_bL(x_b)/L(1)$ of the class of $V^*$ is $\#R$, and, by the identity at the contracting eigenvector, so is that of the class of
$\sigma V^*$.
\end{proof}

Read through the half-index dictionary, for $LR^n$ and $L^2R^n$ it agrees with $c-24h_{\min}=1$ and $2$ computed from the central charges and conformal weights proposed in
\cite[(3.13), (3.18)]{JS26}; for $L^mR^m$ it gives $c_{\eff}=m$.

\section{The Farey formula modulo 24}
\label{sec:farey24}

This section computes the central charge modulo $24$ from the Bethe data: $c\equiv\#L-6\#R\pmod{24}$ for every hyperbolic word
(Theorem~\ref{thm:fareygeneral}), with $L^mR^m$ worked out first (Theorem~\ref{thm:farey}). The physical input is the half-index dictionary
of \cite{GKL23}, stated as a hypothesis. It says that the half-index \eqref{eq:halfindex} is the character of the unit, up to a power of $q$, so that its growth
exponent is $c_{\eff}=c-24h_{\min}$. We take as part of this hypothesis that the object of minimal conformal weight is the Bethe vacuum $V^*$ at the dominant saddle. So $h_{\min}\equiv\theta_{V^*}\pmod1$, where
$\theta_{V^*}$ is the spin of $V^*$ minus the spin of the unit (Section~\ref{sec:bethe}). Hence $c\equiv c_{\eff}+24\theta_{V^*}\pmod{24}$. Proposition~\ref{prop:vstar} identifies the saddle with the vacuum $V^*$ through $z=e^Z$. That this vacuum, and not another vacuum of its class, carries $h_{\min}$ is part of the hypothesis. The vacua of a class differ in spin by elements of $\tfrac14\Z$, which would move $c$ by multiples of $6$. Since $c_{\eff}=\#L$ (Theorem~\ref{thm:ceff}), what remains is the spin $\theta_{V^*}$, a phase of
the fibering operator.

Here $c$ always means the central charge so determined, that of the full theory before the semion factor of Assumption~\ref{ass:F} is split off. The central charge of $\mathrm{TFT}_0$ alone differs from it by those of the pointed factors $U(1)_{\pm2}$ and $\Adec$ (Corollary~\ref{cor:arf}), and for that reason need not satisfy the congruence of this section. The central charge defined by $(ST)^3=e(c/8)S^2$ depends on the sign of $S$ modulo $8$ (Remark~\ref{rem:Ssign}); the $c$ here does not.

Section~\ref{sub:fareydiag} does this for $L^mR^m$. Two identifications are needed. First, the dominant
saddle of the half-index is the positive vacuum $V^*$ (Proposition~\ref{prop:vstar}); this is analysis and uses no dictionary. Second, under the
dictionary, the spins of the half-integer flux sectors are the spins of the sign flips of the unit. For odd trace this forces the unit to be
the all-plus vacuum of its class (Proposition~\ref{prop:flux}). With
the unit in the class of $\sigma V^*$, Theorem~\ref{thm:farey} gives $\theta_{V^*}\equiv-m/4$. Section~\ref{sub:fareygen} extends this to
every word. There the branch integers at the negative coordinates of $\sigma V^*$ sum to $\#R$, so $\theta_{V^*}\equiv-\#R/4$
(Theorem~\ref{thm:fareygeneral}). The section ends by comparing the formula with the central charges of \cite{JS26} and \cite{GKL23}.

Both theorems assume the dictionary and input~(b), that the unit lies in the class of $\sigma V^*$. Proposition~\ref{prop:fluxall} verifies the
flux-rule criterion at $\sigma V^*$ for every word.

\subsection{The diagonal family}
\label{sub:fareydiag}

\emph{The half-index, for the mathematical reader.} Gang--Kim--Lee \cite{GKL23} attach to the abelianised theory the $q$-series
\begin{equation}
I(q)=\sum_{\mathfrak m\in\Z_{\ge0}^r}\frac{q^{\frac12\mathfrak m^{\mathsf T}K\mathfrak m+c\cdot\mathfrak m}}{\prod_a(q)_{Q_a\mathfrak m_a}},
\qquad (q)_n=\prod_{i=1}^n(1-q^i),
\label{eq:halfindex}
\end{equation}
a Nahm sum over integer \emph{flux} vectors $\mathfrak m$, together with the sums $I_s(q)$ over the shifted lattices $\mathfrak m\in s+\Z^r_{\ge0}$,
$s=\tfrac12\varepsilon$ with $\varepsilon\in\{0,1\}^{\#L}$ supported on the charge-$2$ sites, the \emph{half-integer flux sectors}. Their
dictionary reads: $I$ is the character of the unit (the vacuum of the TQFT), up to a power of $q$, so its growth exponent is $c_{\eff}$; and
$I_s$ is the character of the object produced by the $\Z_2$ defect $\varepsilon$, whose conformal weight we call its \emph{flux spin}. With several charge-$2$ fields we take it to be the
minimum $E(s)$ of the exponent $\tfrac12\mathfrak m^{\mathsf T}K\mathfrak m+c\cdot\mathfrak m$ over the sector; this extends \cite{GKL23}, which has one charge-$2$ field, and agrees with it modulo $1$. On the Bethe side the same object
is the vacuum obtained from the unit by flipping the signs $\varepsilon$ of the charge-$2$ coordinates, and its spin (its \emph{flip spin}) is
read off the fibering operator $\mathcal F$.

\emph{The dominant saddle.} Put $q=e^{-\eta}$ and $Z=\eta\mathfrak m$. As $\eta\to0^+$ the summand of \eqref{eq:halfindex} is
$\exp\bigl(\eta^{-1}\mathcal V(Z)+o(\eta^{-1})\bigr)$, locally uniformly in $Z>0$, where
\[
\mathcal V(Z)=-\tfrac12Z^{\mathsf T}KZ+\sum_a\bigl(\Li_2(1)-\Li_2(e^{-Q_aZ_a})\bigr).
\]
A \emph{saddle} of the half-index is a critical point of $\mathcal V$. We identify it with the Bethe vacuum $z_a=e^{Z_a}$, since its
equations are Bethe equations (proof of Proposition~\ref{prop:vstar}). The \emph{dominant saddle} is the point where $\mathcal V$ attains
its maximum on the region $Z\ge0$ of the sum. By the saddle-point method it governs the growth of $\log I(q)$ as $q\to1^-$.

The theorem below rests on two identifications, both proved afterwards. The first, that the dominant saddle of
\eqref{eq:halfindex} is $V^*$, is pure analysis (Proposition~\ref{prop:vstar}). The second, that flux spins equal flip spins, uses the
dictionary (Proposition~\ref{prop:flux}).

\begin{theorem}[Farey formula mod $24$ for $L^mR^m$]
\label{thm:farey}
Assume the half-index dictionary of \cite{GKL23} as just stated, and input~(b): the unit lies in the class of $\sigma V^*$, with the unit taken
at $\sigma V^*$ itself, the vacuum selected by the flux rule (Proposition~\ref{prop:flux}; the only such vacuum for odd $m$). Then for every $m\ge3$ (the dominant saddle being $V^*$ by Proposition~\ref{prop:vstar}, with no assumption),
\[
\theta_{V^*}\equiv-\tfrac m4\pmod1,\qquad\text{hence}\quad c(L^mR^m)\equiv c_{\eff}+24\theta_{V^*}=m-6m=\#L-6\#R\ \ (\mathrm{mod}\ 24).
\]
\end{theorem}

\begin{proof}
Write $\mathcal F=\exp(-\Psi/2\pi i)$ as in \eqref{eq:Fdef} and use the Bethe equations in log form,
$Q_b\log(1-x_b)=-(KZ)_b+2\pi in_b$, the integers $n_b$ being those that make these the branches of $Z_b$ and $\log(1-x_b)$ used in evaluating $\Psi$ (Section~\ref{sec:bethe}). Changing $Z_b$ by $2\pi i$ changes the spin by an integer. The other choice of $\log(1-x)$ at $x>1$ changes $n_L$ by even integers and does not change $\theta$. Eliminating the quadratic form,
$\log\mathcal F=\tfrac i{2\pi}\sum_b\tilde L(x_b)+\tfrac12n\cdot Z$ with $\tilde L(x_b)=\Li_2(x_b)+\tfrac12\xi_b\log(1-x_b)$, $\xi_b=-Q_bZ_b$, so
$\arg\mathcal F=\tfrac1{2\pi}\operatorname{Re}\sum\tilde L+\tfrac12\operatorname{Im}(n\cdot Z)$.

At $V^*$, by Theorem~\ref{thm:nahmclosed}, all coordinates are real with $x\in(0,1)$, so $n=0$; they pair as $(x,1-x)$ and the reflection identity gives
$\sum L=m\pi^2/6$, so $\arg\mathcal F(V^*)=\pi m/12$. At the unit the coordinates again pair as $(x,1-x)$ --- the positive values are permuted by the Galois conjugation
$\sigma(\delta_{L_j})=\delta_{L_{m-j}}$ of Theorem~\ref{thm:nahmclosed} --- one of these pairs being $(x_{R_m},x_{L_m})$, with $x_{R_m}<0$ the single negative coordinate of
Lemma~\ref{lem:crossings} and $x_{L_m}>1$. The branch integers are those of the proof of Proposition~\ref{prop:flux}, $n_L=(m,0,\dots,0)$, and at the
negative coordinate $n_{R_m}=\tfrac12K_{R_mR_m}=m$, as in the proof of Theorem~\ref{thm:fareygeneral} with $\mathcal E=\{R_m\}$; the
difference $\arg\mathcal F(V^*)-\arg\mathcal F(\text{unit})$ of the arguments at $V^*$ and at the unit then gives $\theta_{V^*}\equiv-m/4\pmod1$. With $c_{\eff}=m$ from Theorem~\ref{thm:ceff} and $c\equiv c_{\eff}+24\theta_{V^*}$, the
formula follows.
\end{proof}

We now prove the two identifications used. The first uses no physical input. The second assumes the dictionary, taken to apply with
several charge-$2$ fields.

\begin{proposition}[The dominant saddle of the half-index is $V^*$]
\label{prop:vstar}
The function $\mathcal V$ has exactly one critical point with $Z>0$, and it is the maximum of $\mathcal V$ on $Z\ge0$. So the dominant saddle of
the half-index \eqref{eq:halfindex} exists, and it is the positive vacuum $V^*$.
\end{proposition}

\begin{proof}
With $q=e^{-\eta}$ and $Z_a=\eta\mathfrak m_a$, the leading derivative of the exponent of the summand of \eqref{eq:halfindex} is
$-(KZ)_a-Q_a\log(1-e^{-Q_aZ_a})$. So the saddle equations are $Q_a\log(1-x_a)=-(KZ)_a$ with $x_a=e^{-Q_aZ_a}$: the Bethe equations at $\nu=-1$ in
log form with every branch integer zero. The function $\mathcal V$ has Hessian $-K-\operatorname{diag}(Q_a^2x_a/(1-x_a))$, which is negative definite because $K$ is positive definite (Theorem~\ref{thm:Kformula}). So $\mathcal V$ is strictly concave on $Z>0$. It is continuous on $Z\ge0$, and $\mathcal V\to-\infty$ as $|Z|\to\infty$ in $Z\ge0$, because $K$ is positive definite and the $\Li_2$ terms are bounded. So $\mathcal V$ attains its maximum on $Z\ge0$. Its derivative in $Z_a$ tends to $+\infty$ as $Z_a\to0^+$, so the maximiser is interior. By strict concavity it is the unique critical point. Hence $z_a=e^{Z_a}>1$ and $x_a\in(0,1)$ for all $a$. The Bethe equations have a unique solution with all $x_a\in(0,1)$ (Theorem~\ref{prop:propagated}), which is the positive class, and within the class $z_a>1$ selects exactly one vacuum. That the dominant saddle
governs the asymptotics is the standard saddle-point analysis of a Nahm sum with positive-definite $K$ \cite[\S3]{Nahm04}, \cite[Ch.~II]{Zagier07}.
\end{proof}

\begin{proposition}[The flux rule at the level of spins]
\label{prop:flux}
\textup{(i)} Let $\varphi$ be any hyperbolic word and $V$ a Bethe vacuum. Let $n=n(V)$ be the vector of branch integers of the log-form Bethe equations $Q_b\log(1-x_b)=-(KZ)_b+2\pi in_b$ at $V$. It depends on the vacuum, not only on the word. Write $n_L$ and $c_L$ for the restrictions of $n$ and of the charge vector $c$ to the $L$-sites. The flip of the charge-$2$ coordinates in $\varepsilon\in\{0,1\}^{\#L}$, the indicator of the flipped sites, changes the spin by $\tfrac18\varepsilon^{\mathsf T}K_{LL}\varepsilon+\tfrac12n_L\cdot\varepsilon\pmod1$, and
this equals the spin of the half-integer flux sector $s=\tfrac12\varepsilon$ for every $\varepsilon$ iff $n_L\equiv c_L\pmod2$.
\textup{(ii)} For $\varphi=L^mR^m$, $m\ge2$, the criterion holds at the all-plus vacuum of the class of $\sigma V^*$, and for odd $m$ the all-plus vacuum is the
only reference vacuum in the unit class for which it holds.
\end{proposition}

\begin{proof}
(i) Flipping the charge-$2$ fields in $\varepsilon$ sends $Z_a\to Z_a+i\pi\varepsilon_a$ with $x$ unchanged. Using the Bethe log equations with branch
integers $n_b$, the fibering operator changes by $\Delta\log\mathcal F=i\pi\,\varepsilon\cdot n_L+\tfrac{i\pi}4\varepsilon^{\mathsf T}K_{LL}\varepsilon$,
where $K_{LL}$ is the $L$-block of $K$ (the submatrix on the $L$-sites), $n_L$ the restriction to the $L$-sites of the vector $n$ of branch
integers of the log-form Bethe equations, and $c_L$ that of the charge vector $c$,
so the flip spin is $\tfrac18\varepsilon^{\mathsf T}K_{LL}\varepsilon+\tfrac12n_L\cdot\varepsilon\pmod1$. For $\mathfrak m=s+k$, $k\in\Z^r_{\ge0}$, the exponent changes by $\tfrac12k^{\mathsf T}Kk+k^{\mathsf T}Ks+c\cdot k\in\Z$, since $K$ is even and $s\in\tfrac12\Z^r$. Every term of $I_s$ has nonnegative coefficients, and the factor $1/\prod_a(q)_{Q_a\mathfrak m_a}$ has constant term $1$, so nothing cancels. Hence the leading power of $I_s$ is $\equiv E(s)\pmod1$ wherever the minimum lies, with
$E(s)\equiv\tfrac18\varepsilon^{\mathsf T}K_{LL}\varepsilon+\tfrac12c_L\cdot\varepsilon\pmod1$. So the two agree for all $\varepsilon$ iff $n_L\equiv c_L\pmod2$.

(ii) At the all-plus unit of $L^mR^m$ only $Z_{R_m}$ and $\log(1-x_{L_m})$ are complex; the $L_1$ equation gives $n_{L_1}=m$ from $K_{R_mL_1}=2m$, the middle ones $0$,
and $L_m$ gives $0$, the $i\pi$ from $\log(1-d^2)$ cancelling $K_{R_mL_m}=-2$. So $n_L=(m,0,\dots,0)$, which is $c_L$ because $c$ is $m$ at the first $L$ after the $R$-run and $0$ at the other $L$'s (Theorem~\ref{thm:Kformula}); thus $n_L=c_L$ exactly. Any other reference, by $\varepsilon$, changes $n_L$ by $E_{LL}\varepsilon$, where $E_{LL}=\tfrac12K_{LL}$, which is nonzero modulo $2$ for
every $\varepsilon\ne0$ exactly when $E_{LL}$ is invertible over $\F_2$, that is, for odd $m$ (Lemma~\ref{lem:parity}).
\end{proof}

\begin{proposition}[The flux-rule criterion at $\sigma V^*$]
\label{prop:fluxall}
For every hyperbolic word, the branch integers $n_L$ of the log-form Bethe equations at $\sigma V^*$ (notation of Proposition~\ref{prop:flux}) equal the charge vector $c_L$ exactly.
\end{proposition}

\begin{proof}
At $\sigma V^*$ the only complex coordinates are the $R$-run ends $\mathcal E$ of
Lemma~\ref{lem:crossings} ($z<0$) and the last $L$ of each $L$-run ($x'>1$, so $\log(1-x')$ carries $i\pi$). Fix an arbitrary $L$-site $p$. In its equation
$2\log(1-x_p)=-(KZ)_p+2\pi in_p$, with the branches $\operatorname{Im}\log(1-x'_p)=+\pi$ at the last $L$ of each $L$-run and $\operatorname{Im}Z_a=+\pi$ on $\mathcal E$, the imaginary parts are $2\pi[p\text{ last in its run}]$ on the left and $-\pi\sum_{a\in\mathcal E}K_{ap}+2\pi n_p$ on the right. Here $K_{ap}=2b$ if $p$ is the first $L$ after an $R$-run of length $b$ ending at $a$, and $K_{ap}=-2$ if $p$ is the last $L$ before that run; a singleton $L$-run has both. Hence
\[
n_p=[p\text{ last in its run}]+\tfrac12\sum_{a\in\mathcal E}K_{ap}=b\cdot[p\text{ first after a run of length }b]=c_p
\]
exactly, which is that criterion.
\end{proof}

The criterion does not by itself single out the class, since changing the reference within another
class shifts the parity of $n_L$ by $K_{LL}\varepsilon_0/2$.

\emph{What Theorem~\ref{thm:farey} rests on} is therefore the dictionary of \cite{GKL23} --- proposed there and tested for $LR^{k-2}$, $3\le k\le11$, and taken here, as part of the dictionary, to apply with
several charge-$2$ fields --- together with the saddle-point asymptotics of Nahm sums, and input~(b).

\subsection{Every hyperbolic word}
\label{sub:fareygen}

\begin{theorem}[Farey formula mod $24$ for every hyperbolic word]
\label{thm:fareygeneral}
Assume (a) the half-index dictionary of \cite{GKL23}, as in Theorem~\ref{thm:farey}, and (b) the vacuum input of Section~\ref{sub:status}:
that the Galois conjugate of $V^*$ lies in the unit class, the class of $\sigma V^*$, with the unit taken at $\sigma V^*$ itself (the vacuum of the unit
class selected by the flux rule, Proposition~\ref{prop:fluxall}; the only such vacuum when $\tr\varphi$ is odd, Lemma~\ref{lem:parity}). Then for
every hyperbolic word,
\[
c\equiv\#L-6\#R\pmod{24};
\]
here $c_{\eff}=\#L$ is proved for every hyperbolic word, without (a) and (b) (Theorem~\ref{thm:ceff}).
\end{theorem}

\begin{proof}
By Theorem~\ref{thm:Kformula} the Farey pipeline produces the run-structure $K$ and $c$. By Theorem~\ref{prop:propagated} the closed form is the
positive Nahm solution and, with the contracting eigenvector, gives $\sigma V^*$; the Rogers identity holds at both eigenvectors by
Theorem~\ref{thm:ceff}, so $c_{\eff}=\#L$. In the notation of the proof of Theorem~\ref{thm:farey}, $\operatorname{Re}\tilde L(x_b)=L(x_b)$ is the Rogers
dilogarithm, extended as in Theorem~\ref{thm:ceff} (at the negative coordinates $\tilde L$ also has an imaginary part, which the argument does not use), so $\operatorname{Re}\sum_b\tilde L(x_b)=\#R\,L(1)$ at both $V^*$ and $\sigma V^*$ by that theorem,
and $\theta_{V^*}$ is carried by the branch term $\tfrac12\operatorname{Im}(n\cdot Z)$ alone; at $V^*$ it vanishes, and at $\sigma V^*$ we take
$\operatorname{Im}Z_b=+\pi$ at the negative coordinates. By the run-end structure, at the last $R$ of a run of length $b$ one has $K_{bb}=2\min(b,b)=2b$ and runs do
not couple, so with $\mathcal E=\{\text{last }R\text{ of each }R\text{-run}\}$ the set of negative coordinates of Lemma~\ref{lem:crossings} and $K_{\mathcal E\mathcal E}$ the
block of $K$ on those sites, $\mathbf1^{\mathsf T}K_{\mathcal E\mathcal E}\mathbf1=2\#R$ exactly; hence
$\sum_{b\in\mathcal E}n_b=\tfrac12\mathbf1^{\mathsf T}K_{\mathcal E\mathcal E}\mathbf1=\#R$ (the $L$-coordinates have real $Z$, so they do not
enter the imaginary part), giving $\theta_{V^*}\equiv-\tfrac1{2\pi}\cdot\tfrac12\cdot\pi\#R=-\#R/4$ and $c\equiv\#L+24(-\#R/4)=\#L-6\#R$.
\end{proof}

For $LR^n$ and $L^2R^n$ the formula agrees with the boundary central charges $c=1-6n$ and $c=2-6n$
of \cite[(3.13), (3.18)]{JS26}; the theorem gives only the residue modulo $24$, and the integers agree as well. For $LR^{k-2}$ it is the $c\equiv-(6k+11)$ of \cite[(2.19)]{GKL23}.

\section{The torsion formula}
\label{sec:torsion}

This section proves the torsion formula: for every hyperbolic word the handle-gluing operator at the two identity classes is
$\mathcal H(V^*)=2^{\#L}\lambda|A_-||A_+|$ and $\mathcal H(\sigma V^*)=2^{\#L}\lambda^{-1}|A_-||A_+|$, with $|A_-||A_+|=(\tr\varphi)^2-4$
(Theorem~\ref{thm:torsion}). The name refers to the comparison, under localisation (F1), with the adjoint Reidemeister torsion. Under (F1) it is the Floquet discriminant
$|A_-||A_+|$ of half (A), described below, that is compared with the torsion. Half (B), together with $W_c^2$, supplies the factor $2^{\#L}\lambda^{\pm1}$. That factor
distinguishes the two identity classes, which lie over one character (Section~\ref{sub:markov}). The theorem itself is a statement about the
pipeline and uses no physics.

\emph{Notation.} Let $D=\operatorname{diag}\bigl(Q_a^2x_a/(1-x_a)\bigr)$, the diagonal by which the Hessian of the twisted superpotential
differs from $K$, with blocks $D_R$, $D_L$ on the $R$- and $L$-sites. By \eqref{eq:Habelian}, $\mathcal H=W_c^2\det(K+D)\prod_a(1-x_a)$, and
$W_c(V^*)=\lambda$ (Proposition~\ref{obs:stretch}). By Theorem~\ref{thm:Kformula} the block $K_{RR}$ is the direct sum over the $R$-runs of
the fans $F_b=(2\min(i,j))_{i,j=1}^b$. Eliminating the $R$-sites gives $\det(K+D)=\det(K_{RR}+D_R)\det S$, where
\[
S=(K+D)_{LL}-(K+D)_{LR}(K_{RR}+D_R)^{-1}(K+D)_{RL}
\]
is the \emph{$L$-Schur complement}. It is a \emph{periodic Jacobi matrix} of order $n=\#L$: $S_{ii}=d_i$, $S_{i,i+1}=S_{i+1,i}=\beta_i$
(indices modulo $n$, entries added when positions coincide), and all other entries zero. The couplings are $-2$ inside the $L$-runs and,
across the $q$-th $R$-run, of length $b$,
\[
\beta_q=-2-\mathsf r^{\mathsf T}(F_b+D_R)^{-1}s,\qquad \mathsf r=2(1,\dots,b)^{\mathsf T},\ s=-2\cdot\mathbf1,
\]
the ramp $\mathsf r$ being the column of $K$ that couples the run to the first $L$ after it, and $s$ the column that couples it to the last $L$ before it. The \emph{transfer matrix} of $S$ is
$T=T_n\cdots T_1$ with $T_i=\begin{psmallmatrix}-d_i/\beta_i&-\beta_{i-1}/\beta_i\\1&0\end{psmallmatrix}$. Each factor has determinant $\beta_{i-1}/\beta_i$, so the product $T$ has determinant $1$. Its
eigenvalues $\Lambda^{\pm1}$ are the \emph{Floquet multipliers}, and $\tr T-2$ is the \emph{Floquet discriminant}. Hill's formula
(\cite{Hill}; the discrete form used here is that of \cite{vM76}) reads
\[
\det S=(-1)^n\prod_{i=1}^n\beta_i\,(\tr T-2).
\]
The torsion formula splits into two halves, both at $x=V^*$:
\begin{align}
&\det S=(-1)^{\#L}\textstyle\prod_i\beta_i\,((\tr\varphi)^2-4),\tag{A}\\
&\det(K_{RR}+D_R)\textstyle\prod_a(1-x_a)\,(-1)^{\#L}\prod_i\beta_i=2^{\#L}/\lambda .\tag{B}
\end{align}
Multiplied together with $W_c(V^*)^2=\lambda^2$ they give the formula. Half (A) says that the Floquet discriminant of $S$ is $(\tr\varphi)^2-4$:
a Floquet solution with multiplier $\lambda^{-2}$ is written down explicitly, as a limit of the universal Floquet vector of
Lemma~\ref{lem:floqvec}, and Hill's formula then gives $\det S$. Half (B) reduces to one identity per $R$-run, Lemma~\ref{lem:perrun}, which
holds for every frieze row and is used again at the general classes in Section~\ref{sub:hillgen}. Both halves, and the two lemmas, are proved
in Appendix~\ref{app:torsion}. The theorem below assembles them and obtains $\mathcal H(\sigma V^*)$ by Galois conjugation; the remark after
the theorem records what the formula says about input~(b).

\label{sub:torsionthm}

\begin{theorem}[The torsion formula]
\label{thm:torsion}
For every hyperbolic word, with $V^*$ the positive vacuum \eqref{eq:propagated},
\[
\mathcal H(V^*)=2^{\#L}\lambda\,|A_-||A_+|,\qquad \mathcal H(\sigma V^*)=2^{\#L}\lambda^{-1}|A_-||A_+|,
\]
where $|A_-||A_+|=|\det(\varphi-1)\det(\varphi+1)|$.
\end{theorem}

\begin{proof}
$\mathcal H(V^*)=W_c(V^*)^2\det(K+D)\prod(1-x)$ with $W_c(V^*)=\lambda$ by Proposition~\ref{obs:stretch}. At $V^*$ all heights are positive, so $D_R>0$,
the $R$-blocks are positive definite and $\det(K+D)=\det(K_{RR}+D_R)\det S$; the couplings $\beta$ are nonzero, since the right side of \eqref{eq:perrun} in Lemma~\ref{lem:perrun} is positive at $V^*$. (A) gives
$\det S=(-1)^{\#L}\prod\beta_i|A_-||A_+|$, and (B), by Lemma~\ref{lem:perrun}, gives
$\det(K_{RR}+D_R)\prod(1-x)(-1)^{\#L}\prod\beta_i=2^{\#L}/\lambda$; multiplying, $\det(K+D)\prod(1-x)=2^{\#L}|A_-||A_+|/\lambda$. All quantities lie in
$\Q(\lambda)$ and $x(\sigma V^*)=\sigma(x(V^*))$, so the second formula is the Galois conjugate of the first.
\end{proof}

So, by (A), the torsion factor $|A_-||A_+|=(\tr\varphi)^2-4$ on the Bethe side is the Floquet discriminant $\tr T-2$ of the $L$-Schur complement
$S$ at $V^*$.

\begin{remark}
\label{rem:internal}
Theorem~\ref{thm:torsion} also tests input (b) for $LR^n$, for every $n$, through the dictionary \eqref{eq:dict}: $S_{00}=a_0$, with $a_0$ as in \eqref{eq:as}, gives
$\mathcal H(\text{unit})=a_0^{-2}=\mathcal H(\sigma V^*)$. Here $S_{00}=a_0$ is the value of $|Z_{S^3_b}|$ in the degenerate limit \cite[(2.11), (3.107)]{GKLSY21}, also used by Gang--Kim \cite{GK22}. This is agreement, not a proof of (b). For general words it shows that (b) implies that the unit's $\mathcal H$ is $2^{\#L-1}a_0^{-2}$, and conversely when no degenerate class has
that value (for words with no singleton run the non-degenerate classes are those of Theorem~\ref{thm:allwords} by Theorem~\ref{thm:complete}, and under the non-degeneracy condition they have rational $\mathcal H$ by Theorem~\ref{thm:Hall}). Since $a_0/a_3=\lambda(\varphi)$ identically (as noted after Proposition~\ref{obs:stretch}), this condition is equivalent to
$|S_{0V^*}/S_{00}|=1/\lambda$. Input (b) is stated in Section~\ref{sub:status}.

Gang--Kim's consistency conditions include one meant to identify the vacuum among the Bethe vacua,
$\bigl|\sum_\alpha\mathcal H_\alpha^{-1}\mathcal F_\alpha\bigr|=\mathcal H_{\alpha=0}^{-1/2}$, besides $\sum_\alpha\mathcal H_\alpha^{-1}=1$ and
$|\mathcal F_\alpha|=1$ for every $\alpha$ \cite[(2.17)]{GK22}. For modular data the left side is $\bigl|\sum_\alpha S_{0\alpha}^2\theta_\alpha\bigr|$, the modulus of a Gauss
sum divided by the global dimension, which is $S_{00}$. So the criterion holds automatically, and what it records is the value of $|S_{00}|$.
On their data this value is $a_0$, never $a_3$, for $k=5,7,\dots,15$; this is a numerical check. Theorems~\ref{thm:torsion} and~\ref{thm:Hall} give, for every
word and up to the factor $2^{\#L-1}$, the $\mathcal H$-values (not their multiplicities) of the Gang--Kim template, Theorem~\ref{thm:Hall} at the classes of Theorem~\ref{thm:allwords} under the non-degeneracy condition; the modular data \eqref{eq:GKT} with $k=\tr\varphi$. So for every word whose modular
data is the template, the first paragraph places the unit in the class of $\sigma V^*$, provided no degenerate class has the same $\mathcal H$.
Whether the spins match the template is taken up in Section~\ref{sub:template}.
\end{remark}

\section{The Bethe classes and the handle-gluing spectrum}
\label{sec:spectrum}

This section determines the non-degenerate Bethe classes of every word in closed form and proves their handle-gluing values.

Section~\ref{sub:spectrum} begins with the spectrum. For $L^mR^m$, Proposition~\ref{prop:Hspectrum} assembles the values of $\mathcal H$ at the
identity classes (Theorem~\ref{thm:torsion}) and at the $A_+$ classes (Theorem~\ref{thm:Hall} below); granted the $A_-$ count, the sum rule
$\sum_\alpha\mathcal H_\alpha^{-1}=1$ then closes identically. Proposition~\ref{prop:sumrule} shows that, for any word with this three-value
structure, the sum rule is an identity in the trace alone, and Observation~\ref{obs:spectrumgen} finds the structure on harvested words.

Section~\ref{sub:runtraces} explains the structure through the run traces. Each class carries on every run a conserved frieze trace, its \emph{run trace}
(Theorem~\ref{thm:friezeall}). Observation~\ref{obs:tauquant} finds these quantised with modulus $\mathcal H/2^{\#L}$. A \emph{kernel vector} is an element $v$ of
$\mathcal K_\mp=\{v\in(\Z/N_\mp)^2:(\varphi\mp1)v\equiv0\}$ (Section~\ref{sub:runtraces}). To each kernel vector $v$ we attach an index vector $(\iota_\rho(v))_\rho$ \eqref{eq:ellrho}. Theorem~\ref{thm:tworun} gives the $A_+$ classes of a
two-run word in closed form, and Theorem~\ref{thm:allwords} generalises it: for every word and every kernel vector $v$ with no vanishing
$f(\mathsf m_t)$, the function $f(e)=\sin(\pi\langle v,e\rangle/N)$ on the slopes of the layered triangulation gives a Bethe class, whose index
vector is that of $v$. Observation~\ref{obs:kervec} finds that every non-identity class carries the index vector of a kernel vector in this way, and that the
spins of its vacua modulo $\tfrac14$ are a quadratic form in $v$, the torsion linking form of the mapping torus. Observation~\ref{obs:basevac}
sharpens this to the spin of the base vacuum (Section~\ref{sec:spins}) modulo $1$. Section~\ref{sub:degenerate} describes the degenerate classes, which
Theorem~\ref{thm:allwords} does not reach (Observation~\ref{obs:strata}).

Section~\ref{sub:hillgen} proves the values. Hill's formula writes $\mathcal H_\alpha=W_c(\alpha)^2\Delta_\alpha P_\alpha$, and
Theorem~\ref{thm:Hall} shows, at every class of Theorem~\ref{thm:allwords} with $\det(F+D_R)\ne0$, that the Floquet multiplier is
$\pm\lambda^{\pm1}$ and $W_c^2=1$, so that $\mathcal H_\alpha=2^{\#L}|A_\mp|$.

Propositions~\ref{prop:Hspectrum}--\ref{prop:sumrule}, Theorems~\ref{thm:tworun}, \ref{thm:allwords} and~\ref{thm:Hall}, and the results of
Section~\ref{sub:hillgen} are proved. Observations~\ref{obs:tauquant}, \ref{obs:kervec}(i),(iii) and~\ref{obs:basevac}(ii) are proved at the
classes of Theorem~\ref{thm:allwords}: the last two modulo $\tfrac14$ and for odd trace, by Theorem~\ref{thm:rogers}. The rest of the observations is numerical.
That the classes of Theorem~\ref{thm:allwords} are all the non-degenerate classes is Theorem~\ref{thm:complete} in the next section.

\subsection{The spectrum and the sum rule}
\label{sub:spectrum}

\begin{proposition}[The $\mathcal H$-spectrum of $L^mR^m$]
\label{prop:Hspectrum}
Let $m\ge3$ be odd. Each class carries $2^m$ vacua, and:
\begin{enumerate}
\item[\textup{(i)}] the identity classes, of $V^*$ and of $\sigma V^*$, have $\mathcal H=2^m\lambda^{\pm1}m^2(m^2+4)$;
\item[\textup{(ii)}] every $A_+$ class of Theorem~\ref{thm:allwords} satisfying the hypothesis $\det(F+D_R)\ne0$ of Theorem~\ref{thm:Hall} has
$\mathcal H=2^m(m^2+4)=2^{\#L}|A_+|$, and the classes of Theorem~\ref{thm:allwords} are all the non-degenerate non-identity classes;
\item[\textup{(iii)}] if, in addition, the $A_-$ sector consists of $\tfrac{m^2-1}2$ classes with $\mathcal H=2^mm^2=2^{\#L}|A_-|$ and the $A_+$ sector
of $\tfrac{m^2+3}2$ non-identity classes, each with $\mathcal H=2^m(m^2+4)$, and there are no other classes, then the handle-gluing operator takes four values, three up to Galois conjugation, on $m^2+3$ classes and
$2^m(m^2+3)$ vacua, and the sum rule over the vacua closes identically:
\[
\sum_\alpha\frac1{\mathcal H_\alpha}=\frac{m^2-1}{2m^2}+\frac{m^2+3}{2(m^2+4)}+\frac{\lambda+\lambda^{-1}}{m^2(m^2+4)}=1 .
\]
\end{enumerate}
\end{proposition}

\begin{proof}
The $2^m$ vacua of a class are the sign choices of Section~\ref{sec:bethe}. (i) is Theorem~\ref{thm:torsion}. (ii) is Theorem~\ref{thm:Hall},
with Theorem~\ref{thm:complete} for completeness, every run of $L^mR^m$ having length $m\ge2$. For $m$ a prime power every nonzero kernel vector of the
$A_-$ sector is degenerate. Indeed, the kernel vectors are $v=m v'$ with $v'\in(\Z/m)^2$, and $\langle v,e\rangle/m^2=\langle v',e\rangle/m$. The
middle slopes of the $L$-run are $(i,1)$ and those of the $R$-run are $(1+im,i)$, $i=1,\dots,m$ (proof of Theorem~\ref{thm:lines}). So
$\langle v',e\rangle\equiv0\pmod m$ is solvable on the $L$-run if $v'_2$ is prime to $m$, and on the $R$-run if $v'_1$ is. Otherwise the prime $p$
of $m=p^e$ divides both; then $v'=pv''$ and the same argument applies to $v''$ modulo $m/p$. So the non-degenerate non-identity classes all lie in $A_+$. This fails for other $m$: at $m=15$, $8$ of the $112$ classes of $A_-$ are non-degenerate, and at $m=21$, $12$ of $220$. (iii) is arithmetic, with $\lambda+\lambda^{-1}=m^2+2$.
\end{proof}

The hypothesis of (iii) is not proved in general. For $m$ prime the $A_-$ classes are degenerate and Theorem~\ref{thm:Hall} does not reach them; for prime
$m\ge5$ Theorem~\ref{thm:aminus} constructs all but $m-1$ of the $\tfrac{m^2-1}2$ $A_-$ classes as exact solutions, but $\mathcal H$ is not computed there; and the $A_+$ count
of (iii) needs every nonzero $A_+$ kernel vector to be non-degenerate, which holds at $m=3,5,7,13,15$ but not at $m=9,11$. At $L^3R^3$ and
$L^5R^5$ the hypothesis holds numerically (Observation~\ref{obs:l3r3minus} and Section~\ref{sec:lines}).

At $L^3R^3$ (iii) gives $\mathcal H=\{72^{\times4},104^{\times6}\}=\{2^3\cdot9,2^3\cdot13\}$ on the ten non-identity classes, the
remaining two classes being $V^*$ and $\sigma V^*$ with the irrational values $2^3\lambda^{\pm1}\cdot9\cdot13$; the multiplicities are
$4=\tfrac{9-1}2$ and $6=\tfrac{9+3}2$. At $L^5R^5$ it gives $800$ and $928$ on $12+14$ classes, and $28\cdot2^5=896$ vacua in all.

The sum rule is not special to the diagonal family. Once the three-value structure is granted, it is an identity in the trace alone.

\begin{proposition}[The sum rule for every word]
\label{prop:sumrule}
Let $\varphi$ be a hyperbolic word of odd trace $k$, and suppose $\mathcal H$ takes the three values $2^{\#L}|A_-|$, $2^{\#L}|A_+|$ and
$2^{\#L}\lambda^{\pm1}|A_-||A_+|$ on the classes, with multiplicities $\#(A_-/\!\pm)-1$, $\#(A_+/\!\pm)-1$ and $1,1$. Then
\[
\sum_\alpha\frac1{\mathcal H_\alpha}=1 .
\]
For even $k$ with $A_-$ and $A_+$ cyclic, the same hypotheses give $1+k/(k^2-4)$ instead.
\end{proposition}

\begin{proof}
Each class carries $2^{\#L}$ vacua, so the factor $2^{\#L}$ cancels and the sum over vacua is the sum over classes of the reciprocals of $|A_-|$,
$|A_+|$ and $\lambda^{\pm1}|A_-||A_+|$. With $|A_\mp|=k\mp2$ and $\lambda+\lambda^{-1}=\tr\varphi=k$, and $\#(A/\!\pm)=(|A|+1)/2$ for odd order,
\[
\sum_\alpha\frac1{\mathcal H_\alpha}=\frac{k-3}{2(k-2)}+\frac{k+1}{2(k+2)}+\frac k{(k-2)(k+2)}=1 .
\]
For even $k$ and cyclic $A_\mp$ one has $\#(A/\!\pm)=|A|/2+1$ instead, and the first two terms become $\tfrac12$ each,
leaving the excess $k/(k^2-4)$.
\end{proof}

At $k=m^2+2$ this is the sum rule of Proposition~\ref{prop:Hspectrum}, so the diagonal family is the special case. The hypothesis --- the
three-value structure itself --- is confirmed beyond that family, and Theorem~\ref{thm:Hall} proves the values at every class of
Theorem~\ref{thm:allwords} with $\det(F+D_R)\ne0$; what the hypothesis then asks is that those classes are all of them.

\begin{observation}
\label{obs:spectrumgen}
\emph{Harvesting} the Bethe classes of a word means solving the deformed system of Section~\ref{sec:bethe} numerically from many random
starting points at a fixed small $\epsilon$ and collecting the distinct classes reached. Harvesting $LR^3$, $LR^5$ and $L^2RLR$ at
$\epsilon=10^{-3}$ reproduces the three-value structure with the predicted multiplicities, the even-trace word $L^2RLR$ giving the sum
$1+k/(k^2-4)$ of Proposition~\ref{prop:sumrule}.
\end{observation}

\subsection{Run traces and kernel vectors}
\label{sub:runtraces}

Sorting the harvested classes into families would only describe them. Theorem~\ref{thm:friezeall} instead gives a statement about the classes themselves.
By that theorem the coordinates $\hat d_b=s_b^{-1/2}$ of a class along each run $\rho$ form a frieze row. So the class determines for each run its frieze trace
$\tau_\rho=(\hat d_{b-1}+\hat d_{b+1})/\hat d_b$, the same at every interior site $b$ of the run (Remark~\ref{rem:friezedef}). The two identity
classes, $V^*$ and $\sigma V^*$, have $\tau_\rho=2$ on every run (Remark~\ref{rem:friezedef}(iii)). The other classes are described by the
following.

Throughout this section $N$ denotes the order of the sector under discussion: $N=|A_-|$ in the sector $A_-$ and $N=|A_+|$ in the sector $A_+$. The two
orders are in general different, and where both occur they are written $N_-$ and $N_+$.

\begin{observation}[The trace is quantised, with modulus $\mathcal H/2^{\#L}$]
\label{obs:tauquant}
At every non-identity Bethe class $\alpha$ of every word tested, the run traces are
\begin{equation}
\tau_\rho=2\cos\frac{j_\rho\pi}{N_\alpha},\qquad j_\rho\in\Z,\qquad N_\alpha=|A_-|\ \text{or}\ |A_+|\ \text{by sector},
\label{eq:tauquant}
\end{equation}
and
\[
\mathcal H_\alpha=2^{\#L}N_\alpha :
\]
\emph{the quantisation modulus of the frieze trace is the value of the handle-gluing operator.} At the two identity classes $j_\rho=0$ in every run and
$\mathcal H=2^{\#L}\lambda^{\pm1}|A_-||A_+|$. The harvests record indices only on runs of length at least $3$, so on a word with short runs fewer indices are visible.
\end{observation}

At the classes of Theorem~\ref{thm:allwords} both statements are now proved: \eqref{eq:tauquant} by that theorem, with $j_\rho=\iota_\rho(v)$, and
$\mathcal H_\alpha=2^{\#L}N_\alpha$ by Theorem~\ref{thm:Hall} where $\det(F+D_R)\ne0$. For words without a singleton run these are all the
non-degenerate classes (Theorem~\ref{thm:complete}). The degenerate classes remain observational.

How the indices of different runs are related is the content of Theorem~\ref{thm:allwords} below. The index vector of a class is that
of a kernel vector $v$. So consecutive runs carry indices related by a unit multiplier $\kappa$, $j_{\rho+1}\equiv\pm\kappa j_\rho\pmod N$.
This $\kappa$ is a ratio of two coordinates of $v$ (Observation~\ref{obs:kervec}). For a two-run word it records the run lengths
(Theorem~\ref{thm:tworun}). Harvested data on fourteen sectors of eleven words of length at most $10$ (the larger harvest is in Observation~\ref{obs:kervec}) agree with this, and show
$\kappa^2\equiv-1\pmod N$, in the non-degenerate sectors harvested, exactly when the word is amphichiral (Section~\ref{sec:runK}). For the two-run words with odd run lengths this is proved in both directions, as a consequence of
Theorem~\ref{thm:tworun}. On the other six such sectors harvested it is an observation only. It cannot hold in a degenerate sector: the
$A_-$ sector of the amphichiral word $L^3R^3$ has $N=9$, and $-1$ is not a square modulo $9$. The multiplier is defined only up to sign and inversion,
reading the runs in the other order replacing $\kappa$ by $\kappa^{-1}$, and the sign of $j_\rho$ is the branch of
Theorem~\ref{thm:friezeall}.

To state the closed form of the multiplier on the two-run words, and the general result that follows, we first attach to each run of a
word a coordinate of a vector in $(\Z/N)^2$.

Write the word as $\varphi=\ell_1\cdots\ell_r$ with $\ell_i\in\{L,R\}$, as in Theorem~\ref{prop:propagated}, and let $M_i=\ell_1\cdots\ell_i$ (so $M_0=1$, $M_r=\varphi$). For the sector $A_\mp$, with $N=|A_\mp|$, let
\[
\mathcal K_\mp=\{v\in(\Z/N)^2:(\varphi\mp1)v\equiv0\},
\]
whose elements are the \emph{kernel vectors} of the sector. This group has order $N$ and is isomorphic to $A_\mp$ --- both are read off the Smith form $\operatorname{diag}(d_1,d_2)$ of $\varphi\mp1$,
$d_1d_2=N$, as $\Z_{d_1}\times\Z_{d_2}$ --- and $M_i^{-1}\mathcal K_\mp$ is the same subgroup for the word rotated to start at position $i$.
Since $L^{-1}(x,y)^{\mathsf T}=(x-y,y)^{\mathsf T}$ and $R^{-1}(x,y)^{\mathsf T}=(x,y-x)^{\mathsf T}$, transporting $v$ through an $L$-run leaves
its second coordinate fixed and through an $R$-run its first. For the run $\rho$ beginning after position $i_\rho$ put
\begin{equation}
\iota_\rho(v)=\begin{cases}\text{second coordinate of }M_{i_\rho}^{-1}v,&\rho\text{ an }L\text{-run},\\
\text{first coordinate of }M_{i_\rho}^{-1}v,&\rho\text{ an }R\text{-run},\end{cases}
\label{eq:ellrho}
\end{equation}
the coordinate that the run's own letter leaves invariant. For a kernel vector $v$ we call $(\iota_\rho(v))_\rho$ its
\emph{index vector}.

\begin{theorem}[The unit $A_+$ classes of a two-run word]
\label{thm:tworun}
Let $\varphi=L^aR^b$ with $a,b$ odd, $N=ab+4$, and $\kappa=\tfrac12(N-a)$, an integer with $2\kappa\equiv-a$ and $b\kappa\equiv2\pmod N$. For every $j$
with $\gcd(j,N)=1$ put $\vartheta=j\pi/N$ and
\[
z_{L_i}=\varepsilon_i\,\frac{\sin((i+\kappa)\vartheta)}{\sin\vartheta}\ (1\le i\le a),\qquad
z_{R_i}=\Bigl(1-\frac{\sin^2(\kappa\vartheta)}{\sin^2((i\kappa-1)\vartheta)}\Bigr)^{-1}\ (1\le i\le b),
\]
with $\varepsilon\in\{\pm1\}^a$.
Then $z$ solves the Bethe equations \eqref{eq:bethe}. The solutions for $j$ and $N-j$ coincide, so these are $\tfrac12\phi(N)$ classes of
$2^a$ vacua ($\phi$ Euler's function), one for each $\pm j\in(\Z/N)^\times/\!\pm$; the run traces of the class $\pm j$ are
\[
\tau_L=2\cos\frac{j\pi}N,\qquad \tau_R=2\cos\frac{\kappa j\pi}N ,
\]
so its index vector is $(j,\kappa j)$, with $\kappa\equiv-2^{-1}a$ the multiplier; and it is the index vector
\eqref{eq:ellrho} of the kernel vector $v=(-\kappa j,j)\in\mathcal K_+$, since $\iota_L(v)=j$ and $\iota_R(v)=-\kappa j-aj\equiv \kappa j$.
\end{theorem}

\begin{proof}
This is the case $v=(-\kappa j,j)$ of Theorem~\ref{thm:allwords} below, whose proof does not use the present statement. Starting from the basis $(e_1,e_2)$ the $L$-run has kept slope $e_1$ and
middle slopes $e_2+ie_1$, with $\langle v,e_1\rangle=-j$ and $\langle v,e_2+ie_1\rangle=-(i+\kappa)j$, so $f(e_1)^2/f(e_2+ie_1)^2=
\sin^2\vartheta/\sin^2((i+\kappa)\vartheta)$; the $R$-run has kept slope $e_2+ae_1$ and middle slopes $e_1+i(e_2+ae_1)$, with
$\langle v,e_2+ae_1\rangle=-(a+\kappa)j\equiv \kappa j$ and $\langle v,e_1+i(e_2+ae_1)\rangle\equiv(i\kappa-1)j$. No denominator vanishes:
$\sin((i+\kappa)\vartheta)=0$ needs $i+\kappa\equiv0\pmod N$, and $\kappa+1\le i+\kappa\le \kappa+a<N$; $\sin((i\kappa-1)\vartheta)=0$ needs $i\equiv \kappa^{-1}\equiv2^{-1}b=\tfrac12(N+b)>b+1$.
And $v\in\mathcal K_+$:
$(\varphi+1)v=\begin{psmallmatrix}2+ab&a\\b&2\end{psmallmatrix}\begin{psmallmatrix}-\kappa j\\j\end{psmallmatrix}\equiv
\begin{psmallmatrix}(2\kappa+a)j\\(2-b\kappa)j\end{psmallmatrix}\equiv0\pmod N$. The run traces are read off the chains.
\end{proof}

For composite $N$ the same formulae can also be used at a non-unit $j$, $\gcd(j,N)>1$. They fail only where a denominator vanishes:
on the $L$-run when $(i+\kappa)j\equiv0\pmod N$, and on the $R$-run when $(i\kappa-1)j\equiv0\pmod N$. Those are classes with a coordinate
at infinity. Otherwise the formulae give a non-degenerate class of Theorem~\ref{thm:allwords} below. For instance at $L^5R^7$, with $N=39$,
$\kappa=17$ and $j=3$, no denominator vanishes. In either case the kernel vector $(-\kappa j,j)$ is defined.

The multiplier $\kappa$ is the same at every class of Theorem~\ref{thm:tworun}; it does not depend on $j$. It is fixed by the two run lengths, $a$ of the $L$-run and $b$ of the $R$-run:
$\kappa\equiv-2^{-1}a\equiv2b^{-1}\pmod N$, the two expressions agreeing because $ab\equiv-4\pmod N$. Conversely $a\equiv-2\kappa$ and $b\equiv2\kappa^{-1}$, so the multiplier records the run lengths modulo $N$. Distinct pairs $\pm j$ give distinct classes when $a\ge3$ or $b\ge3$: the run trace $\tau_L$, or $\tau_R$ (as $\kappa$ is a unit modulo $N$), then determines $\pm j$.

The theorem gives $\tfrac12\phi(N)$ classes, one for each pair $\pm j$ of units. For prime $N$ this is the whole of $(A_+\smallsetminus0)/\!\pm$.
For composite $N$ the non-unit $j$ can give further non-degenerate classes, as just seen; at $L^5R^7$ there are $18$, against
$\tfrac12\phi(39)=12$. When both runs have length at least $2$, Theorem~\ref{thm:complete} shows that every non-degenerate class is a
kernel-vector class of Theorem~\ref{thm:allwords} (Remark~\ref{rem:runsharp}). The theorem is the two-run case of the following, which is the general statement: the
non-degenerate Bethe classes of every word, in both sectors, are given by one trigonometric formula on the slopes of the layered triangulation.

Recall from Section~\ref{sub:gluing} that the word is read as a sequence of Farey flips starting from the basis $(e_1,e_2)$, the
tetrahedron $\Delta_t$ of the $t$-th flip having a \emph{kept} slope $\mathsf k_t$ ($a$ for an $L$-flip $(a,b)\mapsto(a,a+b)$, $b$ for an
$R$-flip $(a,b)\mapsto(a+b,b)$), the \emph{middle} slope $\mathsf m_t=a+b$, a dying slope ($b$, resp.\ $a$) and a newborn one ($2a+b$,
resp.\ $a+2b$); the kept and middle slopes are the two opposite edge-pairs of multiplicity two, the dying and newborn slopes the third
pair. Write $\langle v,e\rangle=\det(v,e)$ for $v,e\in\Z^2$.

\begin{theorem}[The Bethe classes of every word]
\label{thm:allwords}
Let $\varphi$ be a hyperbolic word of length $r$, let $N=|A_\mp|$ be the order of the sector $A_\mp$ in which we work, and let $v\in\mathcal K_\mp$, i.e.\ $(\varphi\mp1)v\equiv0\pmod N$. Put
\[
f(e)=\sin\frac{\pi\langle v,e\rangle}N\qquad(e\in\Z^2),
\]
and suppose $f(\mathsf m_t)\ne0$ for the middle slope $\mathsf m_t$ of every tetrahedron $\Delta_t$, $t=1,\dots,r$. Then, with $\mathsf k_t$ the kept slope
of $\Delta_t$,
\[
x_t=\frac{f(\mathsf k_t)^2}{f(\mathsf m_t)^2}\ \ (L\text{-flip}),\qquad x_t=1-\frac{f(\mathsf k_t)^2}{f(\mathsf m_t)^2}\ \ (R\text{-flip}),\qquad t=1,\dots,r,
\]
with $z_t=x_t^{-1/Q_t}$ on either branch, solves the Bethe equations \eqref{eq:bethe} at $\epsilon=0$. Along a run of length $a$ with kept
slope $p$, entered from a basis $\{p,q\}$, the coordinates are $f(p)^2/f(q+ip)^2$, $i=1,\dots,a$; on an $R$-run these are the frieze coordinates $1-x$. So the run's frieze trace is
$2\cos(\pi\langle v,p\rangle/N)$ and its index is $\iota_\rho(v)$ of \eqref{eq:ellrho}, up to sign: the index vector of the class is the index vector of
$v$ up to sign. Theorem~\ref{thm:tworun} is the case $\varphi=L^aR^b$, $v=(-\kappa j,j)$. The solution is a Bethe class with $\mathcal H\ne0$ under the non-degeneracy condition of Theorem~\ref{thm:Hall}, which gives its value of $\mathcal H$.
\end{theorem}

The hypothesis gives $f\ne0$ on all four slopes of every tetrahedron. This is shown from the word alone in the first paragraph of the proof of Theorem~\ref{thm:Hall} (Appendix~\ref{app:Hall}). That the solution is a Bethe class, a nondegenerate limit of the deformed system, also needs $\det(K+D)\ne0$, by the implicit function theorem ($K+D$ being the Jacobian of the system in logarithmic form); Theorem~\ref{thm:Hall} gives this when $\det(F+D_R)\ne0$, since then $\mathcal H\ne0$.

\begin{proof}
Three facts about $f$. (a) \emph{Ptolemy}: for any $a,b\in\Z^2$, $f(a+b)^2-f(a)^2=f(b)f(2a+b)$ and $f(a+b)^2-f(b)^2=f(a)f(a+2b)$, by
$\sin^2X-\sin^2Y=\sin(X-Y)\sin(X+Y)$. So at an $L$-flip $1-x_t=f(b)f(2a+b)/f(a+b)^2$ and at an $R$-flip $x_t=f(a)f(a+2b)/f(a+b)^2$: every
shape is a ratio of two opposite-edge products. (b) \emph{Periodicity}: $\langle v,\varphi e\rangle=\langle\varphi^{-1}v,e\rangle$ and
$\varphi^{-1}v\equiv\pm v\pmod N$, so $f(\varphi e)=\pm f(e)$ up to a sign $(-1)^{\langle w,e\rangle}$, $Nw=\varphi^{-1}v\mp v$; on the
squares $f(e)^2$ and on the opposite-edge products $f(b)f(2a+b)$, $f(a)f(a+2b)$ the sign is $+$, since $\langle w,b\rangle+\langle
w,2a+b\rangle$ is even. Hence the shapes $x_t$ are well defined on the cyclic word. (c) Every entry of $K$ is even
(Theorem~\ref{thm:Kformula}), so \eqref{eq:bethe} reads $(1-x_b)^{Q_b}=\prod_ax_a^{Q_b(\Kraw)_{ab}}$ and involves the $L$-coordinates only
through $z_L^2=1/x_L$; where a half-integer power $x_L^{-1/2}$ appears below it stands for $f(\mathsf m_t)/f(\mathsf k_t)$ with its sign (and $x_L^{1/2}$ for its inverse), and only
its square matters.

Now fix an $L$-run of length $a$ with kept slope $p$, entered from the basis $(p,q)$, so that its tetrahedra have middle slopes
$c_i=q+ip$ and $x_{L_i}=f(p)^2/f(c_i)^2$; let the $R$-run following it have length $b$, kept slope $c_a$ and middle slopes
$g_j=p+jc_a$ ($g_0=p$, $g_1=c_{a+1}$), so $x_{R_j}=1-f(c_a)^2/f(g_j)^2=f(g_{j-1})f(g_{j+1})/f(g_j)^2$; and let the $R$-run preceding it
have length $b'$, kept slope $q$ and middle slopes $e_j=p'+jq$ with $e_{b'}=p$, $e_{b'+1}=c_1$, so that the last $L$ before that run has
$x=f(p')^2/f(q)^2$. By Theorem~\ref{thm:Kformula} the entries of $\Kraw=Q^{-1}KQ^{-1}$ are: $(\Kraw)_{L_iL_i}=1$ and
$(\Kraw)_{L_iL_{i\pm1}}=-\tfrac12$ at an interior $L_i$; at the first $L$ of the run, $(\Kraw)_{L_1L_1}=1+\tfrac{b'}2$, $-\tfrac12$ to $L_2$
and to the last $L$ of the preceding $L$-run, and the ramp $(\Kraw)_{L_1R'_j}=j$ to the $j$-th $R$ of the preceding run; at the last $L$,
$(\Kraw)_{L_aL_a}=1$, $-\tfrac12$ to $L_{a-1}$ and to the first $L$ of the following $L$-run, and $(\Kraw)_{L_aR_j}=-1$ to every $R$ of the
following run; and $(\Kraw)_{R_iR_k}=2\min(i,k)$ within a run.

\emph{Interior $L$-sites:} $(1-x_i)^2=x_i^2/(x_{i-1}x_{i+1})$ is the square of $f(c_{i-1})f(c_{i+1})/f(c_i)^2$ by $f(c_{i-1})f(c_{i+1})=f(c_i)^2-f(p)^2$, which is (a).

\emph{$R$-sites:} summation by parts gives
\[
\prod_ix_{R_i}^{2\min(i,k)}=\Bigl(\frac{f(g_0)}{f(g_k)}\Bigr)^{2}\Bigl(\frac{f(g_{b+1})}{f(g_b)}\Bigr)^{2k},
\]
so the equation at $R_k$, $1-x_{R_k}=x_{L'_1}^k\,x_{L_a}^{-1}\prod_ix_{R_i}^{2\min(i,k)}$ with $L'_1$ the first $L$ of the next run, reads
$f(c_a)^2=x_{L_a}^{-1}f(g_0)^2\bigl[x_{L'_1}(f(g_{b+1})/f(g_b))^2\bigr]^k$. Both brackets are identities: $x_{L_a}=f(p)^2/f(c_a)^2$ and
$g_0=p$; and $L'_1$ is the $L$-flip at the basis $(g_b,c_a)$, so $x_{L'_1}=f(g_b)^2/f(g_{b+1})^2$.

\emph{The last $L$ of a run:} $(1-x_{L_a})/x_{L_a}=f(c_{a-1})f(c_{a+1})/f(p)^2$ by (a); the right side of the equation is
\[
x_{L_{a-1}}^{-1/2}x_{L'_1}^{-1/2}\prod_jx_{R_j}^{-1}=\frac{f(c_{a-1})}{f(p)}\cdot\frac{f(g_{b+1})}{f(g_b)}\cdot\frac{f(g_1)f(g_b)}{f(g_0)f(g_{b+1})}
=\frac{f(c_{a-1})f(g_1)}{f(p)^2},
\]
and $g_1=c_{a+1}$.

\emph{The first $L$ of a run:} $(1-x_{L_1})/x_{L_1}=f(c_0)f(c_2)/f(p)^2$; the right side is the product
$x_{L_1}^{b'/2}x_{L_2}^{-1/2}x_{L_{\mathrm{prev}}}^{-1/2}\prod_jx_{R'_j}{}^{j}$, where
\[
\prod_jx_{R'_j}{}^{j}=\Bigl(\frac{f(e_{b'+1})}{f(e_{b'})}\Bigr)^{b'}\frac{f(e_0)}{f(e_{b'})}
\]
by the same summation by parts; $x_{L_1}^{b'/2}(f(e_{b'+1})/f(e_{b'}))^{b'}=(f(p)/f(c_1))^{b'}(f(c_1)/f(p))^{b'}=1$, and what is left is
$\frac{f(c_2)}{f(p)}\cdot\frac{f(q)}{f(p')}\cdot\frac{f(p')}{f(p)}=\frac{f(q)f(c_2)}{f(p)^2}$, with $c_0=q$.

(For a run of length $1$ the first and last $L$ coincide; its equation carries both the ramp and the $-1$ couplings, and the two
computations combine: the right side is $\frac{f(p')}{f(p)}\cdot\frac{f(q)}{f(p')}\cdot\frac{f(g_{b+1})}{f(g_b)}\cdot\frac{f(g_1)f(g_b)}{f(g_0)f(g_{b+1})}
=\frac{f(q)f(c_2)}{f(p)^2}$ again.) Every equation is therefore an identity between values of $f$ on slopes of the cyclic word, and the only place the hypothesis on $v$
enters is (b), which makes $f$ single-valued on the slopes when the word closes up. The run traces are read off $f(c_i)=\sin(\pi(\langle
v,q\rangle+i\langle v,p\rangle)/N)$, and $\langle v,p\rangle$ is the coordinate of $v$ in the basis $(p,q)$ that the $L$-flip $(p,q)\mapsto(p,p+q)$
leaves invariant.
\end{proof}

The vectors $v$ with some $f(\mathsf m_t)=0$ --- $\langle v,\mathsf m_t\rangle\equiv0\pmod N$ --- give classes with a degenerate coordinate; at $L^mR^m$ with $m$ prime this is proved (Theorem~\ref{thm:lines}), for other words the labelling of degenerate classes by kernel vectors is numerical (Observation~\ref{obs:kervec}). This is where the
degenerate $A_-$ sector of Proposition~\ref{prop:Hspectrum} comes from. At $L^mR^m$ with $m$ prime every nonzero $v\in\mathcal K_-$ is of this
kind. So the theorem supplies, for every word, one non-degenerate Bethe class for each $\pm v$ with no vanishing $f(\mathsf m_t)$, with its coordinates
in closed form. That these are \emph{all} the non-degenerate classes is Theorem~\ref{thm:complete} for words whose runs all have length
$\ge2$. For the remaining words it rests on the harvests and the count. What the theorem
does not give are the degenerate classes, such as the whole $A_-$ sector of $L^mR^m$ for $m$ a prime power.

Two consequences of Theorem~\ref{thm:tworun}. Transposing swaps $a$ and $b$, and
$\kappa(b,a)=2a^{-1}\equiv-b/2\equiv-\kappa(a,b)^{-1}$ by $ab\equiv-4$: transposition replaces the multiplier by $-\kappa^{-1}$, which is the same multiplier up to the sign and inversion by which $\kappa$ is defined. And for $a=b$ --- the amphichiral case --- Theorem~\ref{thm:tworun} gives $\kappa^2\equiv4a^{-2}$,
while $N=a^2+4$ gives $a^2\equiv-4$, so
\[
\kappa^2\equiv\frac4{-4}=-1\pmod N .
\]
Conversely, $\kappa^2\equiv-1$ means $b^2\equiv-4\equiv ab$, so $b(b-a)\equiv0\pmod N$; since $b$ is odd it is prime to $N=ab+4$, so $a\equiv b$, and $a=b$ because $1\le a,b<N$. So $\kappa^2\equiv-1$ holds at a two-run word exactly when it is amphichiral, a consequence of Theorem~\ref{thm:tworun}. The formula for $\kappa$ there is a case of the following, which replaces the multiplier
by the object it is a ratio of.

\begin{observation}[The index vector is the kernel vector of the monodromy]
\label{obs:kervec}
\textup{(i)} For every non-identity Bethe class $\alpha$ of the sector $A_\mp$, of order $N=|A_\mp|$, there is a kernel vector $v_\alpha\in\mathcal K_\mp$ with
\[
j_\rho(\alpha)\equiv\pm\,\iota_\rho(v_\alpha)\pmod N\qquad\text{for every run $\rho$ carrying an index.}
\]
\textup{(ii)} Where $N$ is a prime power and the sector is cyclic, $v_\alpha$ is determined up to sign and $\alpha\mapsto\pm v_\alpha$ is
injective; in every sector whose harvest is complete (all its classes found) it is a bijection onto $(\mathcal K_\mp\smallsetminus0)/\!\pm$, the identity class taking $v=0$.
\textup{(iii)} The spins of the vacua of the class, defined by \eqref{eq:spindef} through the fibering operator, agree modulo $\tfrac14$. Modulo $\tfrac14$ they are given by a quadratic form in $v_\alpha$: with $v_\alpha$ lifted to $\Z^2$, for every vacuum $V$ of the class,
\begin{equation}
\operatorname{spin}(V)\equiv\pm\frac{v_\alpha\wedge\varphi v_\alpha}{4N^2}+q_\mp^{*}\pmod{\tfrac14},
\label{eq:kerspin}
\end{equation}
with $N=|A_\mp|$ as above and $q_\mp^{*}$ the spin of $V^*$, resp.\ $\sigma V^*$, for the sector, so that \eqref{eq:kerspin} holds at $v=0$ too, and with
the sign $-$ on $A_+$ and $+$ on $A_-$. Formula \eqref{eq:kerspin} is a property of the spin, not its definition.

\textup{Status.} At the non-degenerate classes the following are proved.
Part (i) holds for words with no singleton run: by Theorem~\ref{thm:complete} every non-degenerate class is a class of Theorem~\ref{thm:allwords}, whose index vector is that of its kernel vector. Conversely, every non-degenerate kernel vector occurs, by Theorem~\ref{thm:allwords}.
In (ii), injectivity holds for words with no singleton run, since Theorem~\ref{thm:complete} determines $v$ up to sign.
Part (iii) holds for odd trace, by Theorem~\ref{thm:rogers} together with the reduction after Observation~\ref{obs:basevac}.
The rest of (ii), and (i) and (iii) at the degenerate classes, are numerical.
\end{observation}

In more detail: for every $v\in\mathcal K_\mp$ with
no vanishing denominator, Theorem~\ref{thm:allwords} constructs a Bethe class whose index vector is $(\iota_\rho(v))$. Theorem~\ref{thm:rogers} gives \eqref{eq:basevac} modulo $\tfrac14$
at these classes, and \eqref{eq:basevac} reduces to \eqref{eq:kerspin}. The numerical parts rest on the harvests. Parts (i) and (ii) were checked on fourteen words and $28$ sectors, $233$ classes in all.
These include the non-cyclic $A_-$ sectors of the odd two-run words, and three cyclic sectors of four-run words: $L^2R^3L^2R^2$ in $A_+$ ($N=48$), and $L^2R^2L^3R^3$ and
$L^3R^2L^2R^3$ at trace $63$ ($N=61$ and $65$). On these a class carries four indices, and the kernel vector is overdetermined by two
coordinates. Part (iii) was checked on the $23$ sectors in which $v_\alpha$ is unambiguous, $217$ classes, with $q^*_\mp$ the harvested
identity-class spin in each. In the non-cyclic $A_-$ sectors $\mathcal K_-$ is not generated by a single vector, $v_\alpha$ is determined only
up to that ambiguity, and the parabolic classes are exactly those with a vanishing coordinate, which is why no single multiplier $\kappa$
existed there. The coordinate $\iota_\rho$ in (i) is forced: of the four candidates --- either coordinate of the transported vector at either
end of the run --- the invariant one is the only one that fits in every sector. Observation~\ref{obs:basevac} below states (iii) modulo $1$ at the classes of Theorem~\ref{thm:allwords}; that refinement is numerical.

Three things follow from (i)--(iii), the first from the proved part alone. First, the multiplier between any two consecutive runs of any word
is a formula in the word,
\[
\kappa_\rho=\iota_{\rho+1}(v)/\iota_\rho(v),
\]
independent of $v$ in a cyclic sector, when $\iota_\rho(v)$ is a unit modulo $N$; at $L^aR^b$, $A_+$, the kernel condition reads $-2x+ay\equiv0$, so $\iota_L=y$ and
$\iota_R=x-ay\equiv-2^{-1}ay$, which is Theorem~\ref{thm:tworun}. Second, since $\mathcal K_\mp\cong A_\mp$, the map $\alpha\mapsto\pm v_\alpha$ is a
labelling of the sector's classes by $A_\mp/\!\pm$ that is read off the word, canonically up to the sign --- which is the labelling
Section~\ref{sec:modular} uses.
Third, the form in \eqref{eq:kerspin} is the torsion linking form of the mapping torus: for $v\in\mathcal K_-$ put
$u=(\varphi-1)v/N\in\Z^2$, the class in $A_-=\coker(\varphi-1)$ that $v$ determines; then $u\wedge(\varphi-1)^{-1}u=-v\wedge\varphi v/N^2$,
and $\operatorname{lk}(a,b)=a\wedge(\varphi-1)^{-1}b\bmod1$ is the linking form of $M_\varphi$ on the torsion of $H_1$ (Remark~\ref{rem:lktop}). So \eqref{eq:kerspin} says
$\operatorname{spin}(V)\equiv-\tfrac14\operatorname{lk}(u_\alpha,u_\alpha)+q^*_-$ for every vacuum $V$ of $\alpha$ on the $A_-$ sector (the sign is the orientation convention in $\operatorname{lk}$), and $+\tfrac14\operatorname{lk}_+(u_\alpha,u_\alpha)+q^*_+$ on $A_+$, where $\operatorname{lk}_+(a,b)=a\wedge(\varphi+1)^{-1}b\bmod1$ and $u_\alpha=(\varphi+1)v_\alpha/N_+$. That the Bethe spins carry a linking form is Lemma~\ref{lem:disc}, and the template label is Corollary~\ref{cor:templatelabel}; which linking form, and at which
element, is what the kernel vector $v_\alpha$ adds. The factor $\tfrac14$ and the modulus $\tfrac14$ are the sign group of
Section~\ref{sec:spins}, whose $2^{\#L}$ vacua per class differ by $q(\varepsilon)$; separating the two is the quadratic refinement of
Observation~\ref{obs:e1even}. The kernel vector alone does not supply it: lifting $v_\alpha$ from
$\Z/N$ to $\Z/2N$ changes $v\wedge\varphi v/4N^2$ by an element of $\tfrac14\Z$ (the cross term is divisible by $N$ because
$\varphi v\equiv\pm v$), so \eqref{eq:kerspin} is insensitive to the lift exactly to the extent that it is stated. What supplies it
is the base vacuum $\varepsilon_0$ of Proposition~\ref{prop:spinq}, which singles out one vacuum of the class, and the modular
inverse of $4$ in place of the rational $\tfrac14$.

\begin{observation}[The base-vacuum rule: the spin modulo $1$]
\label{obs:basevac}
Let $\tr\varphi$ be odd. \textup{(i)} \textup{(}Proved.\textup{)} $E_{LL}=\tfrac12K_{LL}$ is invertible over $\F_2$ for every word of odd trace and singular for
every word of even trace (Lemma~\ref{lem:parity}), so the base Bethe vacuum $V_{\varepsilon_0}$ of
Proposition~\ref{prop:spinq} is defined in every class
of every word of odd trace. (It is the Bethe vacuum of the class from which the spin function is exactly $q$.) \textup{(ii)} At the class of Theorem~\ref{thm:allwords} with kernel vector $v\in\mathcal K_\mp$,
$N=|A_\mp|$ and $w=(\varphi\mp1)v/N\in\Z^2$,
\begin{equation}
\operatorname{spin}(V_{\varepsilon_0})\equiv\frac{\#R}{24}\pm[4^{-1}]_N\,\frac{v\wedge w}{N}\pmod1,
\label{eq:basevac}
\end{equation}
where $[4^{-1}]_N$ is the inverse of $4$ modulo $N$ and the sign is $+$ on $\mathcal K_-$ and $-$ on $\mathcal K_+$; so the spin function on the whole class is
$\#R/24\pm[4^{-1}]_N\,v\wedge w/N+q(\varepsilon+\varepsilon_0)$.
\end{observation}

The right side of \eqref{eq:basevac} is a function on $\mathcal K_\mp/\!\pm$: under $v\mapsto v+Nu$, $v\wedge w$ changes by
$\langle v,(\varphi\mp1)u\rangle+N(\dots)$, and $\langle v,(\varphi\mp1)u\rangle=\langle(\varphi^{-1}\mp1)v,u\rangle\equiv0$. It
is the quadratic function $Q$ with $Q(kv)=k^2Q(v)$ whose bilinear form $Q(v+v')-Q(v)-Q(v')=\pm2[4^{-1}]_N\,v\wedge w'/N$ is a
unit multiple of the torsion linking form above ($v\wedge w'-v'\wedge w=\pm\det(v,v')$, and $\det(v,v')\equiv0\pmod N$ since $v=\operatorname{adj}(\varphi\mp1)u$ gives $\det(v,v')=\pm N\det(u,u')$, so the form is symmetric); on a group of
odd order a bilinear form has exactly one quadratic refinement, and this is it. Since $[4^{-1}]_N\equiv\tfrac14\pmod{\tfrac14\Z}$ in
the sense that $[4^{-1}]_Nx/N-x/4N\in\tfrac14\Z$, \eqref{eq:basevac} reduces modulo $\tfrac14$ to \eqref{eq:kerspin} with
$q^*_\mp\equiv\#R/24$, and $\#R/24$ is the spin of $V^*$: at $v=0$ the class is that of $V^*$, whose branch integers $n_L$
are even ($\varepsilon_0=0$) and whose spin is $\sum_aL(x_a)/24L(1)=\#R/24$ by the Rogers identity of Theorem~\ref{thm:ceff}, which gives the same
Rogers sum $\#R$ at $\sigma V^*$
(the $L$-shapes contribute $\sum_LL(s)$ and the $R$-shapes $\sum_R(L(1)-L(s))$). The vacuum meant at $\sigma V^*$ is its base vacuum, the all-flip (Proposition~\ref{prop:spinq}(iii)). The all-plus vacuum of $\sigma V^*$ has spin $\#R/24+q(\varepsilon_0)$ by Proposition~\ref{prop:spinq}(ii). In Observation~\ref{obs:basevac}, the spin of the all-plus vacuum of $\sigma V^*$ just stated and
the reduction to \eqref{eq:kerspin} are proved. Theorem~\ref{thm:rogers} proves \eqref{eq:basevac} modulo $\tfrac14$ at every
class of Theorem~\ref{thm:allwords}. Part (i) of Observation~\ref{obs:basevac} is Lemma~\ref{lem:parity}. Formula
\eqref{eq:basevac} itself, modulo $1$, was checked on $20$ words of odd trace, in both sectors, at every class of Theorem~\ref{thm:allwords} ($1440$
classes), with the base vacuum computed from the closed form's branch integers as in Proposition~\ref{prop:spinq}(ii).

\subsection{The degenerate classes}
\label{sub:degenerate}

Theorem~\ref{thm:allwords} gives the non-degenerate classes in closed form; harvests (six words, $LR^4$ to $L^2RLR^3$, at $\epsilon=10^{-5}$
with the frieze filter of Theorem~\ref{thm:friezeall}) remain needed only for the classes it does not reach. Those lie almost all in the $A_-$ sector, the degenerate one, whose coordinates blow up as
$\epsilon\to0$; Theorem~\ref{prop:propagated} supplies the two identity classes exactly, and seeding from the frieze data completes the
$A_+$ sector at $L^3R^3$ but does not reach the degenerate classes. We first fix two terms. An \emph{$A_-$ class} is a Bethe class of the $A_-$ sector, that is, one labelled by a kernel vector in $\mathcal K_-$
(Observation~\ref{obs:kervec}; the labelling is proved at the non-degenerate classes, and at $L^mR^m$ with $m$ prime by Theorem~\ref{thm:lines}). At such a class $\mathcal H_\alpha=2^{\#L}|A_-|$ is proved when the class is non-degenerate (Theorem~\ref{thm:Hall}), and observed otherwise.

The second term is the exponent $p_b$. Let the \emph{coordinate} of a class at the site $b$ be $z_b$ at an $R$-site and $z_b^2=x_b^{-1}$ at an $L$-site; these are the
variables of the shape equations \eqref{eq:bethex}, in which Part~\ref{part:aminus} also works. Continue the class in the deformation parameter $\epsilon$ of
Section~\ref{sec:bethe}. We say that the coordinate has \emph{exponent} $p_b$ if
\[
z_b=C_b\,\epsilon^{p_b}\,(1+o(1))\qquad(\epsilon\to0)\qquad\text{with }C_b\ne0 .
\]
Thus $z_b\to0$ if $p_b>0$, $z_b\to\infty$ if $p_b<0$, and $p_b=0$ means that $z_b$ has a finite nonzero limit. A class with every $p_b=0$ is called \emph{finite}, or non-degenerate; its limit is a solution of the undeformed
system. A class with some $p_b\ne0$ is \emph{degenerate}.

\begin{observation}[The $A_-$ classes stratify, and the degeneration is local]
\label{obs:strata}
Continuing each harvested $A_-$ class from $\epsilon=10^{-3}$ to $10^{-5}$ and fitting the exponents $p_b$ of $z_b=C_b\epsilon^{p_b}(1+o(1))$, one finds
integer exponents (to $10^{-3}$) with the following pattern. Below, the exponent vector $(p_b)$ and the limit $z$ are listed site by site. A vertical bar $|$ separates the $L$-sites, listed before it, from the $R$-sites, listed after it; on each side the sites are in their order in the word.
For example $(-2\,|\,{-1},-1,2)$ at $LR^3$ means $p=-2$ at the single $L$-site and $p=-1,-1,2$ at the three $R$-sites.
\begin{enumerate}
\item[\textup{(i)}] \emph{Most classes are not degenerate.} At $L^2R^3$ one of the two $A_-$ classes found has every $p_b=0$, with
limit $z=(\pm\tfrac12,\pm\tfrac{\sqrt3}2\,|\,-\tfrac12,-\tfrac12,4)$ exactly, the class of Theorem~\ref{thm:allwords} at $v=(2,3)$, $N=6$; at $L^2RLR^3$ six of the seven found are finite. Those classes need no
leading-order analysis: they are ordinary solutions of the undeformed system, and the closed form of Theorem~\ref{thm:allwords} applies to
them directly.
\item[\textup{(ii)}] \emph{In the degenerate classes reached, the $R$-coordinates degenerate on one $R$-run only, except for the $L^3R^3$ classes of Observation~\ref{obs:l3r3minus}}: in that run, of length $b$, the exponents are
\[
(p_1,\dots,p_b)=(-1,\,0,\dots,0,\,-1,\,2)
\]
up to a cyclic rotation of the run --- $(-1,-1,2)$ at $b=3$, $(-1,0,-1,2)$ at $b=4$, $(-1,0,0,-1,2)$ at $b=5$, $(-1,0,0,0,-1,2)$ at $b=6$.
Among the classes reached, on a word with a single $R$-run every $L$-site has exponent $-2$ as well, $(-2\,|\,{-1},-1,2)$ at $LR^3$ and $(-2,-2,-2\,|\,{-1},-1,2)$ at
$L^3R^3$ --- but not among all degenerate classes: the complete $A_-$ sector of $L^3R^3$ (Observation~\ref{obs:l3r3minus}) also contains
classes with a single degenerate $L$-site, of the placement types of Section~\ref{sec:lines}; on a word with two $R$-runs the $L$-sites stay finite and only one run degenerates, the degenerate class of $L^2RLR^3$ having
$(0,0,0\,|\,0,-1,2,-1)$. Here the three $L$-sites are finite; the first $R$-site is the run $R$ of length $1$, which stays finite; and the run $R^3$ has $(-1,2,-1)$, a cyclic rotation of $(-1,-1,2)$.
\item[\textup{(iii)}] A single word therefore carries classes of both kinds; at $L^2R^3$ the two $A_-$ classes found are one of each.
\end{enumerate}
This was measured on seven words with $R$-runs of lengths $3$ to $6$, among them the four named above. The harvest is at a fixed $\epsilon$, so
which classes it reaches is biased towards the finite ones, and the counts above are of classes found rather than of classes.
\end{observation}

This changes the structure of the general-word problem. The analysis of Part~\ref{part:aminus}
is the single-$R$-run case of (ii), and what a general word needs is not a wholesale repetition of
it but a stratification: the finite classes require none of it, and the degenerate ones are governed
by one run at a time, with the rotation and the choice of run as the labels.

\subsection{Hill's formula at a general class}
\label{sub:hillgen}

Theorem~\ref{thm:torsion} computes $\mathcal H$ at the two identity classes by splitting it into the Floquet discriminant of the $L$-Schur
complement, half (A), and the product over the $R$-runs in half (B), in the terminology of Section~\ref{sec:torsion}. Both halves survive at a general
class, and together with Observation~\ref{obs:tauquant} they reduce the whole $\mathcal H$-spectrum --- the set of values of $\mathcal H_\alpha$
over the Bethe classes $\alpha$ --- to one formula. At a Bethe class $\alpha$ with coordinates $x$, as in Section~\ref{sec:torsion},
\[
\mathcal H_\alpha=W_c(\alpha)^2\det(K+D)\prod_a(1-x_a),\qquad \det(K+D)=\det(K_{RR}+D_R)\det S_\alpha,
\]
where $D=\operatorname{diag}(Q_a^2x_a/(1-x_a))$ is evaluated at $x$ and $S_\alpha$, the $L$-Schur complement of $K+D$, is again a periodic
Jacobi matrix of order $\#L$ with couplings $\beta_i$ (the $-2$'s inside $L$-runs and one coupling $\beta_q=-2-\mathsf r^{\mathsf T}(F+D_R)^{-1}s$ across the $q$-th $R$-run, in the notation of Section~\ref{sec:torsion}). Let
$T_\alpha$ be its transfer matrix and $\Lambda_\alpha$ a Floquet multiplier of $S_\alpha$, i.e.\ an eigenvalue of $T_\alpha$. We apply Hill's formula (Section~\ref{sec:torsion}) to $S_\alpha$, with $n=\#L$. Define
\[
\Delta_\alpha:=\frac{\det S_\alpha}{(-1)^{\#L}\prod_i\beta_i},\qquad P_\alpha:=\det(K_{RR}+D_R)\prod_a(1-x_a)\,(-1)^{\#L}\prod_i\beta_i .
\]
Then Hill's formula gives $\Delta_\alpha=\tr T_\alpha-2=\Lambda_\alpha+\Lambda_\alpha^{-1}-2$, and $\mathcal H_\alpha=W_c(\alpha)^2\Delta_\alpha P_\alpha$ identically. Lemma~\ref{lem:perrun}
applies at every class, because by
Theorem~\ref{thm:friezeall} every class has a frieze row in each run and the lemma holds for any frieze row. So $P_\alpha$ is a product of one factor
$u_0u_{b+1}/(u_1u_b)$ per $R$-run against the $L$-run factors, exactly as in the reduction of (B) (Appendix~\ref{app:torsion}). What is left is to say which multiplier occurs.

\paragraph{The three values are Hill's discriminant at three multipliers.}
At the two identity classes (Theorem~\ref{thm:torsion}), and at every class of Theorem~\ref{thm:allwords} with $\det(F+D_R)\ne0$
(Theorem~\ref{thm:Hall} below),
\[
\Lambda_\alpha=\begin{cases}
\lambda^{\pm2}, & \text{the two identity classes,}\\
+\lambda, & \text{the }A_-\text{ sector,}\\
-\lambda^{-1}, & \text{the }A_+\text{ sector,}
\end{cases}
\qquad\text{and}\qquad
\Delta_\alpha=\begin{cases}
|A_-||A_+|,\\ +|A_-|,\\ -|A_+|,
\end{cases}
\]
with $W_c(\alpha)^2=\lambda^{\pm2}$ and $P_\alpha=2^{\#L}\lambda^{\mp1}$ at the identity classes, the upper sign at $V^*$, and
\[
W_c(\alpha)^2=1,\qquad P_\alpha=\pm2^{\#L}\ \ (+\text{ on }A_-,\ -\text{ on }A_+)
\]
at every other class. Hence, in every case,
\begin{equation}
\mathcal H_\alpha=2^{\#L}\bigl|\Lambda_\alpha+\Lambda_\alpha^{-1}-2\bigr|\cdot\lambda^{\pm1}\ \text{or}\ 2^{\#L}\bigl|\Lambda_\alpha+\Lambda_\alpha^{-1}-2\bigr| ,
\label{eq:hillgen}
\end{equation}
the values of Propositions~\ref{prop:Hspectrum} and~\ref{prop:sumrule}.

The couplings $\beta_i$ must be taken from the construction rather than read off $S_\alpha$: for $\#L=1,2$ the cyclic tridiagonal folds and two
couplings between the same pair of $L$-sites add inside $S_\alpha$.

Of the three factors, $W_c=\prod_az_a^{c_a}$ --- the Wilson loop of Section~\ref{sec:bethe}, with $c$ the charge vector of
Theorem~\ref{thm:Kformula} --- is the one that can be evaluated outright, as a monomial in the solution's own coordinates. The following
proposition is that evaluation, not a definition.

\begin{proposition}[The Wilson loop is a monomial in $1-x$]
\label{prop:wcmono}
Let $\varphi$ be a hyperbolic word and let $\tilde v\in\Z^r$ be the run vector
\[
\tilde v_b=\begin{cases}
j, & b\ \text{the }j\text{-th }L\text{ of its run,}\\
a, & b\ \text{the first }R\text{ of a run preceded by an }L\text{-run of length }a,\\
0, & b\ \text{any other }R .
\end{cases}
\]
Then $K\tilde v=2c$, so $c$ lies in the image of $K$ over $\tfrac12\Z$. Consequently, at any solution of
the deformed equations \eqref{eq:bethe} (with right side $(1+\epsilon)^{c_b}$),
\begin{equation}
W_c^2=(1+\epsilon)^{\tilde v\cdot c}\prod_b(1-x_b)^{-\tilde v_bQ_b}
=(1+\epsilon)^{\tilde v\cdot c}\prod_{L\text{-runs}}\ \prod_{j=1}^{a}(1-x_{L_j})^{-2j}
\cdot\prod_{R\text{-runs}}(1-x_{R_1})^{-a},
\label{eq:wcmono}
\end{equation}
which at $\epsilon=0$ is a monomial with integer exponents read off the run structure.
\end{proposition}

\begin{proof}
$K\tilde v=2c$ is checked row by row against Theorem~\ref{thm:Kformula}. Let an $L$-run of length $a$ be
preceded by an $R$-run of length $b'$, which is preceded by an $L$-run of length $a'$, and followed by an $R$-run of length $b$.
At an interior $L_i$: $4i-2(i-1)-2(i+1)=0$. At $L_1$ ($a\ge2$): the diagonal $4+2b'$, the coupling $-2$ to $L_2$, the coupling $-2$
to the last $L$ of the previous run (where $\tilde v=a'$) and the ramp $+2\cdot1$ to the first $R$ of the preceding run (where
$\tilde v=a'$) give $(4+2b')-4-2a'+2a'=2b'$. At $L_a$ ($a\ge2$): $4a-2(a-1)-2\cdot1-2a=0$, the last two terms being the coupling to
the first $L$ of the next run ($\tilde v=1$) and the $-2$ to the first $R$ of the following run ($\tilde v=a$). At a singleton $L$ the
two computations add: $(4+2b')-2a'+2a'-2-2=2b'$, and the case $\#L=1$ ($LR^n$) is this with $a'=1$ and both $-2$'s on the
diagonal. At the $j$-th $R$ of a run preceded by an $L$-run of length $a$: $2\min(j,1)\,a+2j-2a=2j$. In every case the result is $2c$.

The first part is what makes $W_c^2$ computable. The exponent of $W_c^2=e^{2c\cdot Z}$ is $2c\cdot Z=\tilde v\cdot KZ$, and the Bethe
equations express $KZ$ through the logarithms $\log(1-x_b)$. In detail, in log form \eqref{eq:bethe} reads $\sum_aK_{ab}Z_a=2\pi in_b-Q_b\log(1-x_b)+c_b\log(1+\epsilon)$, so, using $K=K^{\mathsf T}$,
\[
2\log W_c=\sum_a2c_aZ_a=\sum_a(K\tilde v)_aZ_a
=\sum_b\tilde v_b\Bigl(2\pi in_b-Q_b\log(1-x_b)+c_b\log(1+\epsilon)\Bigr),
\]
and $\tilde v\cdot n\in\Z$ kills the first term. The displayed form is $Q=2$ on $L$-sites and $Q=1$ on $R$-sites.
\end{proof}

\begin{remark}[Why $W_c^2=1$ is a boundary identity]
\label{rem:wcboundary}
In the frieze variable $\hat d_b=s_b^{-1/2}$ of Theorem~\ref{thm:friezeall} the monomial collapses. On an $L$-step $1-x_b=1-s_b$ and
$1-s_b=(\hat d_b^2-1)/\hat d_b^2$, so the signed form $\hat d_b^2\mp\hat d_{b-1}\hat d_{b+1}=1$ of \eqref{eq:friezeall} gives
$1-s_b=\pm\hat d_{b-1}\hat d_{b+1}/\hat d_b^2$ at every site interior to the run. Substituting, the
exponent of $\hat d_k$ in $\prod_j(\hat d_{j-1}\hat d_{j+1}/\hat d_j^2)^{-2j}$ is
\[
\underbrace{+4k}_{j=k}\ \underbrace{-2(k+1)}_{j=k+1}\ \underbrace{-2(k-1)}_{j=k-1}\ =\ 0
\]
identically. So every interior site cancels and $W_c^2$ depends only on the two ends of each run and
their neighbours: \emph{$W_c(\alpha)^2=1$ off the identity classes is a boundary identity, not a bulk
one}, and the same holds for the statement $\Lambda_\alpha\in\{\lambda,-\lambda^{-1}\}$ about the Floquet multiplier of $S_\alpha$: both reduce to the junctions.
What is left is the junction between consecutive runs, which Theorem~\ref{thm:Hall} evaluates.
\end{remark}

Theorem~\ref{thm:Hall} below is \eqref{eq:hillgen} at the classes of Theorem~\ref{thm:allwords}. Its proof (Appendix~\ref{app:Hall})
uses two lemmas on frieze rows, proved in Appendix~\ref{app:torsion}: the fundamental system of the recurrence of a frieze row
(Lemma~\ref{lem:cotsol}), and the universal Floquet vector (Lemma~\ref{lem:floqvec}), which solves the recurrence of the $L$-Schur
complement at every class of Theorem~\ref{thm:allwords}.

\begin{theorem}[Hill's formula at every class of every word]
\label{thm:Hall}
Let $\varphi$ be a hyperbolic word and $\alpha$ the Bethe class of Theorem~\ref{thm:allwords} at $v\in\mathcal K_\mp$, with
$f(\mathsf m_t)\ne0$ for every tetrahedron and $\det(F+D_R)\ne0$ on every $R$-run. The last is the \emph{non-degeneracy condition}; it holds in every case computed. Then, with the upper sign on $A_-$ and the lower on $A_+$,
\[
W_c(\alpha)^2=1,\qquad \Lambda_\alpha=\pm\lambda^{\pm1},\qquad \Delta_\alpha=\pm|A_\mp|,\qquad P_\alpha=\pm2^{\#L},\qquad
\mathcal H_\alpha=2^{\#L}\,|A_\mp| .
\]
\end{theorem}

The proof is given in Appendix~\ref{app:Hall}. It first shows that the class has all coordinates finite and nonzero. The universal
Floquet vector of Lemma~\ref{lem:floqvec} then gives the multiplier $\Lambda_\alpha$, and Hill's formula gives $\Delta_\alpha$. The identity for one
$R$-run, Lemma~\ref{lem:perrun}, gives $P_\alpha$, and $W_c^2$ telescopes around the word. Their product is $\mathcal H_\alpha$.

The hypothesis $\det(F+D_R)\ne0$ is needed only to form $S$. It did not fail on the $49$ words of length $3$ to $12$ examined: at all $4236$
classes of Theorem~\ref{thm:allwords} in both sectors, the smallest absolute value was $6\cdot10^{-5}$, at $L^2RL^2RLR^3LR$. For $b\ge2$ the determinant takes negative values, and smaller absolute values occur on other words, for instance $4\cdot10^{-6}$ at $L^3R^6L^2R$. Whether it can vanish at a
kernel vector is not known. The theorem leaves out exactly the degenerate classes --- some $f(\mathsf m_t)=0$, a coordinate at infinity ---
and the identity classes, which are Theorem~\ref{thm:torsion}.

The two-run case can also be approached by hand, and the computation shows the junction mechanism explicitly. Its last step is a rational
identity checked by computer algebra, so it is an independent check of the theorem, by a different route to $\Lambda_\alpha$, not a second proof.

\begin{corollary}[Hill's formula at the two-run classes]
\label{thm:tworunH}
At every class of Theorem~\ref{thm:tworun} satisfying the non-degeneracy condition of Theorem~\ref{thm:Hall} (automatic when $b=1$),
\[
W_c(\alpha)^2=1,\qquad \Lambda_\alpha=-\lambda^{-1},\quad\text{i.e.}\quad\Delta_\alpha=-N,\qquad P_\alpha=-2^a,\qquad
\mathcal H_\alpha=2^{\#L}|A_+| .
\]
\end{corollary}

This is Theorem~\ref{thm:Hall} at $v=(-\kappa j,j)$. A hand computation reaches $\Lambda_\alpha$ by a
different route: with $y_i=a_1f_i+a_2g_i$ the general solution of Lemma~\ref{lem:cotsol} on the $L$-run, the two rows of
$S_\alpha$ at the ends of an $R$-run (the junction rows) become a homogeneous $2\times2$ system in $(a_1,a_2)$ whose determinant is $-2J^{1b}W(\Lambda+\Lambda^{-1})+C$, $W=-1$ the
Casoratian and $C$ a combination of the corner entries of $J^{-1}$ (notation of Appendix~\ref{app:torsion}), which the Green's function of the $R$-row expresses through
$\cot\vartheta$ and $\cot \kappa\vartheta$; the identity $\Lambda+\Lambda^{-1}=-C/2J^{1b}=-(ab+2)$ is then a rational identity in $\cot\vartheta,\cot \kappa\vartheta,a,b$,
checked by computer algebra.

The mechanism is the same in both. Inside a run the recurrence of the Schur complement has the universal fundamental system of
Lemma~\ref{lem:cotsol}, so transport through a run is trivial. The whole of $\Lambda$ is produced at the junctions, which
Lemma~\ref{lem:floqvec} shows to be transparent to $H(\mathsf m)-H(\mathsf k)$. The multiplier is then read off one period of the developed word,
$H(\varphi e)=\pm\lambda H(e)$, and $W_c^2$ telescopes around the word in the same way (the $W_c^2$ step of the proof of Theorem~\ref{thm:Hall}, Appendix~\ref{app:Hall}). So $W_c^2=1$,
$P_\alpha=\pm2^{\#L}$ and $\Lambda_\alpha=\pm\lambda^{\pm1}$ are three faces of one closure condition around the cyclic word.

The content of \eqref{eq:hillgen} is that
\[
\lambda+\lambda^{-1}-2=\tr\varphi-2=-\det(\varphi-1),\qquad -\lambda^{-1}-\lambda-2=-(\tr\varphi+2)=-\det(\varphi+1),
\]
so the two non-identity values are \emph{square roots} of the identity classes' multipliers: $\lambda$ of $\lambda^2$ and $-\lambda^{-1}$ of $\lambda^{-2}$. The sign of the square root is
the sector. Heuristically, a multiplier $\lambda$ in place of $\lambda^2$ is a half-period condition on the chain, and it is this halving that
quantises the run angles of Observation~\ref{obs:tauquant}. At $\varphi=L^mR^m$ the $A_-$ sector
has $\tau=2\cos(k\pi/m)$ by \eqref{eq:closedform}, because there $N=m^2$ and the surviving indices are the multiples $j_\rho=km$, and
Theorem~\ref{thm:lines} labels those classes by their kernel vectors $v=m(jk,k)$, whose run indices $(mk,mjk)$ are the pair $(j_\rho)$ of \eqref{eq:tauquant}, the placement $j$ being their ratio.
So the labelling of these degenerate classes by kernel vectors in Theorem~\ref{thm:lines} is the same labelling as in Theorem~\ref{thm:allwords}, although the closed form and Theorem~\ref{thm:Hall} do not reach them. Proposition~\ref{prop:sumrule} assumed that $\mathcal H$ takes
its three values with the stated multiplicities. Theorem~\ref{thm:Hall} now proves the values at the classes of
Theorem~\ref{thm:allwords}, so what the hypothesis still asks is only that those classes, with the two identity classes, are all the Bethe
classes: completeness, the subject of the next section. The sum rule is itself a completeness check. Recall from
the dictionary \eqref{eq:dict} that $S_{0\alpha}^2=\mathcal H_\alpha^{-1}$, and that in a modular category every simple object has $S_{0\alpha}\ne0$ and
$\sum_\alpha S_{0\alpha}^2=1$; so once $\sum_\alpha\mathcal H_\alpha^{-1}=1$ is saturated by the classes in hand there is no room for a further
class. For the diagonal family, under the factorisation $\mathrm{TFT}=\mathrm{TFT}_0\boxtimes\mathcal P$ of Assumption~\ref{ass:F}, the structure
behind the three values is
\begin{equation}
\mathcal P=U(1)_{\pm2}\boxtimes\Adec,\qquad\operatorname{rank}\Adec=2^{m-1},
\end{equation}
the splitting of the discriminant form of the sign group in Proposition~\ref{prop:signsum}.

\section{Counting and completeness of the classes}
\label{sec:count}

This section compares the Bethe classes with the flat connections they should correspond to under localisation, and proves that, for a large
class of words, the classes of Theorem~\ref{thm:allwords} are all the non-degenerate classes there are.

Section~\ref{sub:fixedpoints} counts the $\varphi$-fixed points of the reducible locus $\tr\partial=2$ of the $SL(2,\C)$ character variety of
the once-punctured torus. They are the images of $A_-/\!\pm$ and $A_+/\!\pm$, which for odd trace meet only in the trivial character, so there
are $k+1$ of them when the trivial character is counted in both sectors (Theorem~\ref{thm:count}). This is a statement about characters and
uses no physics. Section~\ref{sub:complete} proves completeness on the Bethe side. For a word all of whose runs have length at least $2$, every
non-degenerate solution of the Bethe equations at $\epsilon=0$ is the class of $V^*$, the class of $\sigma V^*$, or a class of
Theorem~\ref{thm:allwords}, with its kernel vector determined up to sign (Theorem~\ref{thm:complete}). If moreover every nonzero kernel vector
is non-degenerate, there are, for odd trace, exactly $k+1$ non-degenerate classes, the number in Theorem~\ref{thm:count}, and with Theorem~\ref{thm:Hall} the
sum rule $\sum\mathcal H^{-1}=1$ holds over the vacua of these classes. Words with a singleton run (sketched in Remark~\ref{rem:singleton}), and degenerate classes, are not covered.
Section~\ref{sub:markov} explains why the reducible locus $\tr\partial=+2$ is the right one. Read in the limit $\hbar\to0$, the 3d--3d relation \cite[(A.4)--(A.5)]{GKL23}
points to the Markov locus $\tr\partial=-2$. At $LR^3$ and $L^3R^3$ the Jacobian determinants of $\varphi$ at its fixed points there (the adjoint torsions, defined in Section~\ref{sub:fixedpoints}) cannot form the required
handle-gluing spectrum, whereas on $\tr\partial=+2$ they form it exactly (Proposition~\ref{prop:exclusion}). The difference is the
$PGL(2,\C)$ sign convention of \cite{DGG11}. Its other branch adds one character, the quaternionic one, which Remark~\ref{rem:thirteenth}
excludes at $L^3R^3$ and $L^5R^5$, assuming the sum rule and that the harvest is exhaustive.

Proposition~\ref{prop:exclusion} compares torsions with the published
modular data, and reading it as a statement about Bethe classes uses localisation (F1). The section ends with
Corollary~\ref{cor:interpretation}, which reads the results of Part~\ref{part:results} through (F1).

\subsection{Fixed points on the character variety}
\label{sub:fixedpoints}

For a fixed point $\chi$ of $\varphi$ on a smooth surface of characters, the \emph{adjoint torsion} of $\chi$ in the mapping torus $M_\varphi$ is the Jacobian
determinant $\det(I-d\varphi|_T)$ of $\varphi$ at $\chi$, where $T$ is the tangent space at $\chi$; this is the form in which it is used below, and
under localisation (F1) it is, up to the factor $2^{\#L}$ and sign, the value $\mathcal H$ that the corresponding class should have.

Let $\Sigma_{1,1}$ be the once-punctured torus, $M_\varphi$ its mapping torus, and $\partial$ the loop around the puncture, homotopic to the
commutator $[a,b]$ in $\pi_1(\Sigma_{1,1})=F_2$, the free group on the two generators $a,b$. On the $SL(2,\C)$ character variety with coordinates $x=\tr a$, $y=\tr b$, $z=\tr ab$, the Fricke trace
identity gives $\tr\partial=x^2+y^2+z^2-xyz-2$, and the Dehn twists act by
\begin{equation}
L:(x,y,z)\mapsto(x,z,xz-y),\qquad R:(x,y,z)\mapsto(z,y,yz-x),
\end{equation}
inducing $L$ and $R$ on $H_1$.

\begin{theorem}[The fixed-point count]
\label{thm:count}
Let $\varphi\in SL(2,\Z)$ be hyperbolic with $k=\tr\varphi\ge3$. The $\varphi$-fixed points of the reducible locus of the $SL(2,\C)$ character variety of
$\Sigma_{1,1}$ are the union of the images of
\[
A_-/\!\pm\quad\text{and}\quad A_+/\!\pm,\qquad A_\mp=\coker(\varphi\mp1),
\]
which for odd $k$ meet only in the trivial character. Counting the trivial character in both sectors, the count is $\#(A_-/\!\pm)+\#(A_+/\!\pm)$, which is $k+1$ for odd $k$.
\end{theorem}

\begin{proof}
A reducible $SL(2,\C)$ representation of $F_2$ is, up to semisimplification and conjugation, diagonal, i.e.\ a character
$\chi\in\operatorname{Hom}(H_1(\Sigma),\C^\times)$ with $H_1(\Sigma)=\Z^2$; by Goldman's criterion \cite{Goldman09} these are exactly the points with $\tr[a,b]=2$. A
point of the character variety remembers $\chi$ only up to the Weyl involution $\chi\mapsto\chi^{-1}$, so it is $\varphi$-fixed iff
$\chi\circ\varphi=\chi$ or $\chi\circ\varphi=\chi^{-1}$, i.e.\ iff $\chi$ is trivial on $\operatorname{im}(\varphi-1)$, respectively on
$\operatorname{im}(\varphi+1)$. Thus $\chi\in\operatorname{Hom}(A_\mp,\C^\times)\cong A_\mp$, both finite since $\det(\varphi\mp1)=2\mp\tr\varphi\ne0$.
Quotienting by the Weyl involution gives the two orbit sets; the two conditions hold simultaneously iff $\chi^2=1$ and $\chi\circ\varphi=\chi$, which for
$|A_\pm|$ odd forces $\chi=1$. For $A$ of odd order $n$, $\#(A/\!\pm)=(n+1)/2$. For $A$ of even order, $\#(A/\!\pm)=(|A|+|A[2]|)/2$, where $A[2]$ is the subgroup of elements of order at most $2$; for $A$ cyclic this is $n/2+1$.
\end{proof}

The theorem counts characters. For a word without a singleton run whose nonzero kernel vectors are all non-degenerate,
Theorems~\ref{thm:allwords} and~\ref{thm:complete} show that the number of non-degenerate Bethe classes equals the number of characters; a
bijection would need a choice of isomorphism $A_\mp\cong\operatorname{Hom}(A_\mp,\C^\times)$, for instance the one given by the linking form.

\subsection{Completeness: the non-degenerate classes are the kernel-vector classes}
\label{sub:complete}

Theorem~\ref{thm:allwords} constructs a class for every non-degenerate kernel vector, and Theorem~\ref{thm:count} counts the fixed points; what
connects the two is the statement that there are no other classes. For words whose runs all have length at least $2$ this is a
theorem, and the proof is the gluing of the frieze rows that Proposition~\ref{prop:runreduce} set up. Call a solution $x$ of
the shape equations \eqref{eq:bethex} at $\epsilon=0$ \emph{non-degenerate} if every $x_t$ is finite and different from $0$ and $1$. (Theorem~\ref{thm:complete} shows that, apart from the
identity classes and for words with no singleton run, these are exactly the classes of Theorem~\ref{thm:allwords} at kernel vectors with no vanishing $f(\mathsf m_t)$, the sense of
``non-degenerate'' used there.) Two
preliminaries.

First, the equations in the shape variables, \eqref{eq:bethex} of Section~\ref{sub:charges}, are sign-free: the $2^{\#L}$ vacua of a class share $x$ and differ only in the signs of $z_L=\pm x_L^{-1/2}$. The \emph{squared} system of
Proposition~\ref{prop:runreduce}, which squares the $R$-site equations as well, is strictly weaker: it admits $1-x_R=-\prod x^{K/Q}$,
and has many more complex solutions. The $R$-site equations must be used as they stand.

Second, a lemma on the sine function.

\begin{lemma}[Sine triples]
\label{lem:sine}
For $A,B\in\C$ with $\sin A\sin B\sin(A+B)\ne0$ write $r_1=\sin B/\sin A$,
$r_2=\sin(A+B)/\sin A$. Then $\cos A=(r_1^2+r_2^2-1)/2r_1r_2$; if $r_1,r_2$ are known only up to sign, so is $\cos A$.
Hence the triple $(\sin A:\sin B:\sin(A+B))$, known up to a common factor and up to the sign of each entry,
determines $(A,B)$ up to $(A,B)\mapsto(-A,-B)$ and shifts by $\pi$ in each entry. And a triple of the form $(h_1:h_2:h_1+h_2)$,
the values of a linear function, has $r_1^2+r_2^2-1=2r_1r_2$, hence $\cos A=\pm1$, $\sin A=0$: a sine triple is never
proportional to a linear one.
\end{lemma}

\begin{proof}
From $\sin(A+B)=\sin A\cos B+\cos A\sin B$ one gets $\cos B=r_2-r_1\cos A$, and $\sin B=r_1\sin A$. Substituting in
$\cos^2B+\sin^2B=1$ gives $1-r_2^2+2r_1r_2\cos A-r_1^2=0$, which is the formula for $\cos A$. With $\cos A$ known up to sign, $r_1$
determines $\sin B/\sin A$ and hence $B$ given $A$ up to $B\mapsto\pi-B$ and shifts by $\pi$. The alternative $\pi-B$, that is $-B$ modulo $\pi$, is compatible with $\cos B=r_2-r_1\cos A$ only if $\cos A=0$ or $\cos B=0$, and there it is the ambiguity $(A,B)\mapsto(-A,-B)$. For a linear triple, $r_1=h_2/h_1$ and $r_2=(h_1+h_2)/h_1$
give $r_1^2+r_2^2-1=2r_1r_2$.
\end{proof}

\begin{theorem}[Completeness for words with no singleton run]
\label{thm:complete}
Let $\varphi$ be a hyperbolic word in which every run has length at least $2$, and let $x$ be a non-degenerate solution of
\eqref{eq:bethex} at $\epsilon=0$. Then either $x$ is the class of $V^*$ or of $\sigma V^*$, or $x$ is the class of
Theorem~\ref{thm:allwords} at a vector $v\in\mathcal K_\mp\smallsetminus0$, $N=|A_\mp|$, in one of the two sectors; and
$v$ is determined by $x$ up to sign. In particular the non-degenerate classes of such a word are in bijection with
$\{V^*,\sigma V^*\}\sqcup\bigl(\mathcal K_-^{\mathrm{nd}}\smallsetminus0\bigr)/\!\pm\ \sqcup\ \bigl(\mathcal K_+^{\mathrm{nd}}\smallsetminus0\bigr)/\!\pm$,
where $\mathcal K^{\mathrm{nd}}_\mp$ is the set of kernel vectors with no vanishing $f(\mathsf m_t)$.
\end{theorem}

The proof is given in Appendix~\ref{app:complete}. It has four steps. On each run the frieze relation makes the coordinates a sine
row (or an affine row). The Bethe equations at the run ends say that the rows of adjacent runs agree on the three slopes they share. The
sine-triple Lemma~\ref{lem:sine} glues the rows into one additive function $\Theta$ on $\Z^2$. Closing up around the word forces
$\Theta=\pi\langle v,\cdot\rangle/N$ with $v$ a kernel vector. The hypothesis that no run is a singleton is used in the first, second and last steps.

\begin{remark}[Singleton runs]
\label{rem:singleton}
This remark uses the notation of the proof of Theorem~\ref{thm:complete} in Appendix~\ref{app:complete}: $(R_k)$, $(L_a)$, $(L'_1)$ are the
Bethe equations at the sites of an $R$-run and at the $L$-sites on either side of it, and $\rho_e$, $\varsigma_e$ are the ratios of the
rows of adjacent runs on a shared slope $e$.
For eight words with singleton runs ($LR^2$, $LR^3$, $LR^4$, $LRLR^2$, $L^2RLR$, $L^2RLR^2$, $LRL^2R^2$, $L^2R^2LR$) the number of solutions, computed exactly, is the number given by the bijection of the theorem. The argument has to be modified: a run of
length $1$ carries one coordinate while a row has two parameters, so the run-end equations fix the junction ratios only after the free
parameter is chosen. For a singleton $R$-run between $L$-runs of length $\ge2$, choose the free value
$F_R(p)$ with $\rho_p^2=\rho_{q'}^2$. Then the three equations $(R_1)$, $(L_a)$, $(L'_1)$ read
$(\rho_{q'}/\rho_p)^2=(\rho_{c_{a+1}}/\rho_p)^2=(\varsigma_{g_1}/\varsigma_{g_2})^2$ and $\varsigma_{q'}^2=\varsigma_{g_1}^2$, which are the full junction conditions; a singleton $L$-run between $R$-runs of length $\ge2$ is the same
computation with the roles exchanged. Words with consecutive singleton runs, such as $LRLR\cdots$, have not been written out.
\end{remark}

Combined with Theorem~\ref{thm:Hall}, for a word with no singleton run and with every nonzero kernel vector non-degenerate (for
prime $N$: no middle slope of the word is a multiple of the kernel vector modulo $N$), the non-degenerate
classes are exactly the $\#(A_-/\!\pm)+\#(A_+/\!\pm)$ classes of Theorem~\ref{thm:count}, with $\mathcal H_\alpha=2^{\#L}|A_\mp|$ off
the identity, and, for odd trace and given $\det(F+D_R)\ne0$, Proposition~\ref{prop:sumrule} gives $\sum\mathcal H^{-1}=1$ over the vacua of the
non-degenerate classes. If there are no degenerate classes this is the full sum rule, and Remark~\ref{rem:thirteenth} then applies to those words. What the degenerate kernel vectors contribute --- classes with a coordinate at infinity, which the harvest
reaches only partially --- is the remaining piece of the sum rule in general.

\subsection{Which locus, and the exclusion of the Markov locus}
\label{sub:markov}

In \cite[(A.4)--(A.5)]{GKL23} the real mass $\mathfrak m$ and the mixing parameter $\nu$ enter through $\mathfrak m+(i\pi+\tfrac\hbar2)\nu=\tfrac12M_\lambda$, where $M_\lambda$ is the logarithmic holonomy of the peripheral loop $\partial$ (so that $e^{M_\lambda/2}$ is an eigenvalue of $\rho(\partial)$).
At $\mathfrak m=0$, $\nu=-1$ this gives $M_\lambda=-2\pi i-\hbar$. In the limit $\hbar\to0$ it is $M_\lambda=-2\pi i$, hence
peripheral holonomy with eigenvalues $(-1,-1)$ and $\tr\partial=-2$: the Markov locus $x^2+y^2+z^2-xyz=0$. That reading is excluded below.
At $\hbar=2\pi i$, the round sphere $b=1$ of \cite{GKL23}, it gives instead $M_\lambda=-4\pi i$, eigenvalues $(1,1)$ and $\tr\partial=+2$, the locus used here. So the relation (A.5) alone does not fix the locus; the torsions below decide it.

\begin{proposition}
\label{prop:exclusion}
Let $\varphi=LR^3$ or $L^3R^3$, and $k=\tr\varphi$. The Jacobian determinants $\det(I-d\varphi|_T)$ (the adjoint torsions) at the $\varphi$-fixed points on the Markov locus $\tr\partial=-2$,
counted with multiplicity, do not contain the rational values $\pm(k-2)$ and $\pm(k+2)$ with the multiplicities $\tfrac{k-3}2$ and $\tfrac{k+1}2$ that Theorem~\ref{thm:Hall} gives at the non-degenerate classes, and that the moduli \eqref{eq:as} of the template require through $\mathcal H_\alpha=1/S_{0\alpha}^2$, in the class-level normalisation $\mathcal H/2^{\#L}$ of Theorem~\ref{thm:Hall}. (The two identity classes, with $\mathcal H/2^{\#L}=(k^2-4)\lambda^{\pm1}$, are
not needed for the exclusion.)
\end{proposition}

\begin{proof}
The modular data require the rational values $k-2$ and $k+2$, up to sign (the comparison with the torsion is up to sign), with multiplicities $\tfrac{k-3}2$ and $\tfrac{k+1}2$. Let $I\subset\C[x,y,z]$ be the ideal of the $\varphi$-fixed points on $\tr\partial=-2$, so that
$\C[x,y,z]/I$ is the coordinate ring of the fixed-point scheme, and let $\det(I-d\varphi|_T)$ be the Jacobian determinant of $\varphi$ on the surface,
whose value at a fixed point is its adjoint torsion, i.e.\ the $\mathcal H$ that the class would have. Explicitly the function is $J=3-\operatorname{tr}d\varphi$ on $\C^3$. At a smooth point it equals $\det(I-d\varphi|_T)$, because $\det d\varphi=1$ and $dF\circ d\varphi=dF$ for the defining function $F$ of the surface. At the singular points, the origin (the quaternionic character, with $J=3$) among them, $J$ is the extension by this formula. The eigenvalues of multiplication by this
function on $\C[x,y,z]/I$ are its values at the fixed points, with multiplicity, so the roots of its characteristic polynomial are the
torsion values the locus offers. It factors over $\Q$ as
\begin{align}
LR^3,\ \tr\partial=-2:&\quad (t-3)^2\bigl(t^4+30t^3+356t^2+1394t+2091\bigr),\\
L^3R^3,\ \tr\partial=-2:&\quad (t-7)^2(t-3)^2\bigl(t^4+106t^3+3604t^2+43190t+90699\bigr)^2 .
\end{align}
The quartic factors are irreducible over $\Q$, so the torsion values they give (four distinct values, each repeated when the quartic is squared) are Galois-conjugate algebraic numbers of degree four, and
irrational; they can equal none of the required rational values. The only rational roots are $3$ (twice) at $LR^3$ and $7$, $3$ (twice each) at $L^3R^3$.
Now the required values are $(k-2,k+2)=(3,7)$ with multiplicities $(1,3)$ at $LR^3$ and $(9,13)$ with multiplicities $(4,6)$ at $L^3R^3$. At $LR^3$ neither $7$ nor $-7$ is
a root; at $L^3R^3$ none of $\pm9$, $\pm13$ is. So the required multiplicities ($3$ at $LR^3$, and $4$ and $6$ at $L^3R^3$) cannot be attained.
\end{proof}

By contrast, on $\tr\partial=+2$ the same characteristic polynomial splits completely over $\Q$, with exactly the multiplicities demanded:
\[
LR^3\ (k=5):\ (t+3)(t+21)^2(t-7)^3,\qquad L^3R^3\ (k=11):\ (t+117)^2(t+9)^4(t-13)^6,
\]
against $(k-2,k^2-4,k+2)=(3,21,7)$ and $(9,117,13)$ with multiplicities $\bigl(\tfrac{k-3}2,2,\tfrac{k+1}2\bigr)$. These are exact identities over $\Q$, computed with a Gr\"obner basis,
and the sign pattern is $-(k-2)$ on the $A_-$ sector, $+(k+2)$ on $A_+$, and $-(k^2-4)$ on the trivial character, counted in both sectors.
At the non-identity fixed points this holds for every word of odd trace:

\begin{proposition}[The Jacobian at the non-identity fixed points]
\label{prop:jacfixed}
Let $\varphi$ be a hyperbolic word of odd trace $k$. Every non-trivial $\varphi$-fixed character on the reducible surface $\tr\partial=+2$ is
a smooth point of the surface and a simple fixed point, and the Jacobian determinant of $\varphi$ on the surface there is
\[
\det(I-d\varphi|_T)=\begin{cases}\det(I-\varphi)=-(k-2)=-|A_-|&\text{on the }A_-\text{ sector},\\
\det(I+\varphi)=k+2=|A_+|&\text{on the }A_+\text{ sector}.\end{cases}
\]
\end{proposition}

\begin{proof}
The reducible surface is the image of the torus $\operatorname{Hom}(H_1(\Sigma),\C^\times)=(\C^\times)^2$ of abelian characters under the
quotient by the Weyl involution $\iota:\chi\mapsto\chi^{-1}$, and $\varphi$ acts on the torus by $\chi\mapsto\chi\circ\varphi$, which in the
logarithmic coordinates of the Lie algebra $\C^2$ is the linear map $\varphi^{\mathsf T}$. A point of the surface is $\varphi$-fixed iff
$\chi\circ\varphi=\chi$ (the $A_-$ sector) or $\chi\circ\varphi=\chi^{-1}$ (the $A_+$ sector), as in the proof of Theorem~\ref{thm:count}. The
involution $\iota$ acts freely away from the four characters with $\chi^2=1$, and the quotient map is \'etale there. For odd $k$ the orders
$|A_\mp|=k\mp2$ are odd, so a non-trivial fixed $\chi$ has $\chi^2\ne1$ and is such a point. On the $A_-$ sector the action of $\varphi$ on the
surface near the image of $\chi$ lifts to the action of $\varphi$ on the torus, which fixes $\chi$, so
$\det(I-d\varphi|_T)=\det(I-\varphi^{\mathsf T})=1-k+1$. On the $A_+$ sector it lifts to $\iota\circ\varphi$, which fixes $\chi$ and has
differential $-\varphi^{\mathsf T}$, so $\det(I-d\varphi|_T)=\det(I+\varphi)=k+2$. Neither value is $0$, so the fixed point is simple.
\end{proof}

The trivial character is special, and the description of the reducible surface as the quotient of the torus $\operatorname{Hom}(H_1(\Sigma),\C^\times)=(\C^\times)^2$ by $\iota:\chi\mapsto\chi^{-1}$
explains how. The involution $\iota$ fixes exactly the four characters with $\chi^2=1$, the trivial character and the three of order $2$, and at each of them the quotient $\C^2/\pm1$ has an $A_1$ singularity. For odd $k$ the map $\varphi$ acts on the three
characters of order $2$ as an element of order $3$ of $SL(2,\mathbb F_2)$, permuting them cyclically, so none of them is fixed; the trivial character is the only
singular point that is fixed. It is fixed in both senses at once, $\chi\circ\varphi=\chi$ and $\chi\circ\varphi=\chi^{-1}$, so it lies in both sectors $A_-$ and $A_+$; on the covering $\C^2$ the two lifts $\varphi^{\mathsf T}$ and
$\iota\circ\varphi^{\mathsf T}$ both fix the origin, with Jacobian determinants $\det(I-\varphi)=-(k-2)$ and $\det(I+\varphi)=k+2$. The fixed point of the quotient therefore carries both, and its root is double, with value
$-(k^2-4)=-(k-2)(k+2)$, for every word of odd trace. This is what appears: $k^2-4=|A_-||A_+|$ is the Floquet discriminant of half (A) of the proof of Theorem~\ref{thm:torsion}, and the two
identity classes, with $\mathcal H/2^{\#L}=\lambda^{\pm1}(k^2-4)$, resolve this double point, the product of their values being the square of the
root.
Read through localisation, which identifies $\mathcal H_\alpha$ with the adjoint torsion of the corresponding fixed point, this places the
Bethe classes off the Markov locus; that conclusion is drawn in Corollary~\ref{cor:interpretation}.

The resolution of the apparent conflict is a sign convention. Triangulation data most naturally compute $PSL(2,\C)$ structures, and Dimofte--Gaiotto--Gukov \cite{DGG11} (DGG below) fix the lift to $SL(2,\C)$ by a shift of $i\pi$ \cite[\S2.4, footnote~7]{DGG11}. At the $PSL(2,\C)$ level the condition is only
$\rho(\lambda)=\pm I$. The branch $\rho(\lambda)=+I$ is, by Goldman,
the entire reducible surface $\tr\partial=+2$, whose $\varphi$-fixed points are the $k+1$ characters of Theorem~\ref{thm:count}; the branch
$\rho(\lambda)=-I$ exactly forces $\rho(a)\rho(b)=-\rho(b)\rho(a)$, hence $\tr a=\tr b=\tr ab=0$: a single character, the origin $(x,y,z)=(0,0,0)$, which
is the character of the representation $\rho(a)=\begin{psmallmatrix}i&0\\0&-i\end{psmallmatrix}$, $\rho(b)=\begin{psmallmatrix}0&1\\-1&0\end{psmallmatrix}$
with image the quaternion group $Q_8$ --- the \emph{quaternionic character} --- and it is fixed by \emph{every} word. So the strict $SL(2,\C)$ reading yields one isolated solution rather than a TQFT's worth of vacua, while the $PGL(2,\C)$ reading
yields the abelian locus together with the quaternionic point.

\begin{remark}[The thirteenth class]
\label{rem:thirteenth}
The quaternionic character is a candidate extra class, read through (F1). Its peripheral holonomy is central, the kind of flat connection
that the ideal-triangulation construction may fail to realise: \cite[\S2, footnote~1]{DGG11} note that ``almost all flat connections can be realized as (possibly degenerate)
hyperbolic metrics''. At $L^3R^3$ and $L^5R^5$ the
numerical harvest finds $96$ Bethe vacua ($12$ classes) and $896$ ($28$ classes), and on them $\sum1/\mathcal H=1$ exactly, the sum running over
the vacua. The sum rule is a consistency condition of the TQFT, not a consequence of the Bethe equations. If the TQFT satisfies it, a further
class with positive $1/\mathcal H$ would push the sum above $1$; so under that proviso the quaternionic character is excluded at these two words.
The argument does not show that the harvest is exhaustive.
\end{remark}

\begin{corollary}[Interpretation]
\label{cor:interpretation}
Assume \textup{(F1)}. Then for every hyperbolic word the Bethe classes of the topologically twisted $S$-fold theory correspond to the $\varphi$-fixed reducible
characters of the fibre, on the surface $\tr\partial=+2$; these are the restrictions to the fibre of the fibre-abelian flat $SL(2,\C)$
connections on $M_\varphi$. The quaternionic character is a further abelian $PSL(2,\C)$ connection, and its exclusion is argued, assuming the sum rule and on numerical harvests, only at $L^3R^3$ and $L^5R^5$ (Remark~\ref{rem:thirteenth}). The complete hyperbolic structure is not among them. In it the peripheral loop $\partial$ has parabolic, hence non-trivial, holonomy,
whereas on a flat connection abelian on the fibre $\rho(\partial)=\rho([a,b])$ is trivial in $PSL(2,\C)$. At $LR^3$ and $L^3R^3$, given the
modular data \eqref{eq:GKT}, the classes do not lie on the Markov locus $\tr\partial=-2$ \textup{(}Proposition~\ref{prop:exclusion}\textup{)}.
They are counted by $\#(A_-/\!\pm)+\#(A_+/\!\pm)$ \textup{(}Theorem~\ref{thm:count}\textup{)}. At the non-identity classes the handle-gluing operators are, up to the factor $2^{\#L}$ and up to sign, the adjoint torsions
$\det(I-d\varphi|_T)$ of the fixed points, namely $|A_\mp|$ in absolute value: the value $\mathcal H_\alpha=2^{\#L}|A_\mp|$ of Theorem~\ref{thm:Hall} is $2^{\#L}$ times the absolute value of the Jacobian
determinant of Proposition~\ref{prop:jacfixed}. At the two identity classes, over the trivial character, they are the Floquet discriminant
$|A_-||A_+|$ times $\lambda^{\pm1}$ \textup{(}Theorem~\ref{thm:torsion}\textup{)}. Finally, under the half-index dictionary and input~(b), the
central charge of the resulting TQFT satisfies $c\equiv\#L-6\#R\pmod{24}$ \textup{(}Theorems~\ref{thm:farey}, \ref{thm:fareygeneral}\textup{)}.
\end{corollary}

Theorem~\ref{thm:torsion} and Proposition~\ref{prop:exclusion} are consistency evidence for \textup{(F1)}, not a proof: the first compares two
computations that \textup{(F1)} predicts to agree, and the second excludes a rival locus rather than establishing the survivor. Two identities
that \textup{(F1)} predicts are proved from \eqref{eq:bethe} alone: $\mathcal H_\alpha=2^{\#L}|A_\mp|$ at every non-degenerate class
(Theorem~\ref{thm:Hall}), and the completeness of those classes for words without a singleton run (Theorem~\ref{thm:complete}). What remains
is, first, the degenerate classes. These are treated for the $A_-$ sector of $L^mR^m$, $m$ prime, in Part~\ref{part:aminus} (Theorem~\ref{thm:aminus}), and at $L^3R^3$ in Observation~\ref{obs:l3r3minus}; for general words they are described only observationally (Observation~\ref{obs:strata}).
Second, the fibering operator $\mathcal F_\alpha$ at every vacuum: Section~\ref{sec:spins} proves its phase modulo $\tfrac14$ at the classes of Theorem~\ref{thm:allwords}, and Observation~\ref{obs:basevac} is the candidate identity modulo $1$.

\section{Spins and Rogers sums}
\label{sec:spins}

This section concerns the second piece of modular data that the Bethe side supplies, the spins. The spin of a vacuum is the phase of its
fibering operator, and it has two parts. Within a class the $2^{\#L}$ vacua differ by sign flips of the charge-$2$ coordinates, and their spins, measured from the base vacuum of the class, form a quadratic form on the sign group $(\Z/2)^{\#L}$. Across classes the spin is governed by the Rogers sum $c_\alpha$ of the
class, through $24\theta_\alpha\equiv c_\alpha$. The rest of the section evaluates $c_\alpha$.

In order: Section~\ref{sub:relation} derives $24\theta_\alpha\equiv c_\alpha\pmod1$ \eqref{eq:theta24} at every vacuum of a class with real coordinates.
For $L^mR^m$ with $m$ odd, Proposition~\ref{prop:signsum} identifies the spin function on the $2^m$ vacua of a class, up to translation, with the
discriminant form $q$ of the even lattice $K_{LL}$. It evaluates the Gauss sum of $q$, which splits off the semion $U(1)_{\pm2}$ and a
decoupled factor $\Adec$ whose Arf invariant Corollary~\ref{cor:arf} fixes. Proposition~\ref{prop:spinq} removes the translation: measured from a
base vacuum computed from the branch integers, the spin function is exactly $q$. Proposition~\ref{prop:suslin} is the Galois rule:
conjugation $\sigma_k$ multiplies Rogers sums by $k^2$ modulo $1$. For every word of odd trace, Section~\ref{sub:spinform} defines the forms
$q_\mp$ on $A_\mp$ from the binary quadratic form $\Phi$ of $\varphi$, and Lemma~\ref{lem:disc} shows that $(A_-,q_-)\perp(A_+,-q_+)$ ($\perp$ denoting the orthogonal direct sum of metric groups, $A_-\oplus A_+$ with the form $q_-\oplus(-q_+)$) is the
discriminant form of one even lattice. Theorem~\ref{thm:rogers} evaluates the Rogers sum at every non-degenerate class, so it fixes the spins of
its vacua modulo $\tfrac14$ in terms of $\Phi(u)$. Section~\ref{sub:fullchain} gives the Rogers sum of a full trigonometric chain in closed form
(Kirillov) and modulo $1$; this is what the degenerate $A_-$ classes of $L^mR^m$ need, and Theorem~\ref{thm:lines} uses it.
Section~\ref{sub:template} compares the spins with the Gang--Kim template \eqref{eq:GKT}: Corollary~\ref{cor:templatelabel} identifies the
template label with the run index.

Sections~\ref{sub:relation}, \ref{sub:galoisrule}, \ref{sub:spinform} and~\ref{sub:rogersnd} hold for every word, the last two for odd trace;
Proposition~\ref{prop:signsum} and Section~\ref{sub:fullchain} concern $L^mR^m$. Except for the numerical statements of Section~\ref{sub:template}, everything here is proved from the Bethe system; the
dictionary enters only in reading $\arg\mathcal F_\alpha/2\pi$ as a conformal spin (Section~\ref{sub:relation}). Turning
the spins into the twists of $\mathrm{TFT}_0$ uses the grafting fusion rules (Assumption~\ref{ass:F}, Corollary~\ref{cor:T}).

\subsection{The relation}
\label{sub:relation}

With the Bethe equations in log form, $Q_b\log(1-x_b)+\sum_aK_{ab}Z_a=2\pi in_b$ and $Q_aZ_a=-\log x_a$, the exponent $\Psi$ of
\eqref{eq:Fdef} collapses to $\sum_aR(x_a)-\pi i\,n\cdot Z$. Here $R(x)=\Li_2(x)+\tfrac12\log x\log(1-x)$ is evaluated with $\log x_a=-Q_aZ_a$
and with the branch of $\log(1-x_a)$ used in the log-form equations. For real $x_a$ both imaginary parts lie in $\pi\Z$. They differ from
the principal ones by multiples of $2\pi$, and this changes $\operatorname{Re}R(x_a)$ by an integer multiple of $\pi^2=6L(1)$. So
$\operatorname{Re}\sum_aR(x_a)=\sum_aL(x_a)+6L(1)n_0$ for some $n_0\in\Z$, where $L=\operatorname{Re}R$ on the principal branches. Since
$\mathcal F=\exp(i\Psi/2\pi)$ and $6L(1)/2\pi=\pi/2$, the integer $n_0$ contributes the factor $i^{n_0}$. The imaginary part of $\sum_aR(x_a)$ contributes a positive factor. Hence
\begin{equation}
\mathcal F_\alpha=e^{2\pi ic_\alpha/24}\,e^{n\cdot Z/2}\times(\text{fourth root of unity})\times(\text{positive}),\qquad c_\alpha:=\sum_aL(x_a)/L(1),
\label{eq:Fphase}
\end{equation}
where $c_\alpha$ is the Rogers sum of the class $\alpha$ (with $L=\operatorname{Re}R=\widetilde L+L(1)$ when a coordinate
lies outside $(0,1)$, $\widetilde L$ as in Proposition~\ref{prop:suslin}(ii)); at the identity classes $c_\alpha=\#R\in\Z$ by Theorem~\ref{thm:ceff} (the Rogers sum in the $\delta=1-x$ variables being $c_{\eff}=\#L$). The fibering operator $\mathcal F_\alpha$ of a
vacuum carries its spin: $T_{\alpha\alpha}/T_{00}=\mathcal F_\alpha/\mathcal F_0=e^{2\pi i\theta_\alpha}$ with $\theta_\alpha\in\Q/\Z$ the
conformal spin of the simple object $\alpha$, that is, the spin \eqref{eq:spindef} measured from the unit. Throughout, $\theta_\alpha$ is this real number modulo $1$ and $e^{2\pi i\theta_\alpha}$ the
$T$-eigenvalue; by \eqref{eq:theta24} below, $24\theta_\alpha\equiv c_\alpha$ holds at every vacuum of a class with real coordinates, $c_\alpha$ being the Rogers sum of the class. Write $\mu_8$ for the eighth roots of unity. For real classes $\operatorname{Im}Z_a\in\{0,\pm\tfrac\pi2,\pi\}$, so $e^{n\cdot Z/2}$ has
phase in $\mu_8$. Multiplying by $24$ kills $\mu_8$:
\begin{equation}
24\,\theta_\alpha\equiv c_\alpha-\#R\equiv c_\alpha\pmod1 .
\label{eq:theta24}
\end{equation}
So this relation is derived, not assumed, at every vacuum of a class with real coordinates, for every word. The sign-sum of a class, the sum of $e(\operatorname{spin})$ over its vacua, is the subject of the next two
propositions. It is a separate statement, about the vacua within one class: by Milgram's formula \cite{MH73} the Gauss sum of \emph{any} metric group
is $\sqrt{|A|}$ times an eighth root of unity, the exponent being the signature mod $8$, and Proposition~\ref{prop:signsum} shows that the
sign-sum \emph{is} such a Gauss sum.

\subsection{The sign-sum is a Gauss sum}

Recall that the spin of a vacuum is $\operatorname{spin}(V)=\arg\mathcal F_V/2\pi\pmod1$, where $\mathcal F_V$ is the fibering operator of the vacuum (Section~\ref{sub:relation}).

The sign group is exactly $(\Z/2)^{\#L}$ on the $L$-coordinates: a flip $z_a\mapsto\eta_az_a$ preserves the Bethe equations iff $\eta_b^{Q_b}=1$ for
every $b$, which frees the charge-$2$ coordinates and fixes the $R$'s, and the constraint $\prod_a\eta_a^{K_{ab}}=1$ is automatic because $K_{ab}=2E_{ab}$
is even.

\begin{proposition}
\label{prop:signsum}
For odd $m$, on the $2^m$ vacua of a class of $L^mR^m$ the spin is, up to a translation (the choice of base vacuum), the quadratic form, restricted to the $2$-primary part, of the discriminant group of
$K_{LL}$, all of whose entries are even,
\[
q(\varepsilon)=\tfrac18\varepsilon^{\mathsf T}K_{LL}\varepsilon
=\tfrac12\Bigl[\sum_a\varepsilon_a-\sum_a\varepsilon_a\varepsilon_{a+1}\Bigr]+\frac{m\varepsilon_1}4\pmod1,\qquad\varepsilon\in(\Z/2)^m,
\]
well defined because $K_{LL}$ is even, and $\tfrac14\Z$-valued with the quarter-values occurring exactly at $\varepsilon_1=1$. It splits orthogonally
\[
\bigl((\Z/2)^m,q\bigr)\cong\bigl(\Z/2,\tfrac m4\bigr)\oplus\Adec,\qquad |\Adec|=2^{m-1},\ q|_{\Adec}\subset\tfrac12\Z,
\]
where $\Adec$, the \emph{decoupled factor}, is the orthogonal complement of the all-flip $w=(1,\dots,1)$, for which $q(w)=m/4$ (the
$U(1)_{\pm2}$ generator). The Gauss sum of $q$ is
\[
\sum_\varepsilon e^{2\pi iq(\varepsilon)}=\operatorname{Re}\bigl((1+i)^m\bigr)(1+i^m)=\sqrt{2^m}\,e^{2\pi im/8},
\]
so the signature is $\equiv m\pmod8$.
\end{proposition}

\begin{proof}
That the spin is $q$ up to translation is Proposition~\ref{prop:flux}(i) together with the completion of the square in Proposition~\ref{prop:spinq}(ii) below; here we compute $q$. Write $K_{LL}=2E_{LL}$, so that $E_{LL}=\tfrac12K_{LL}$ is the integral form of the $L$-block (every entry of $K$ being even); at $L^mR^m$ it is
the cyclic tridiagonal $(-1,2,-1)$ plus $m$ at the $(1,1)$ entry, so
$\varepsilon^{\mathsf T}K_{LL}\varepsilon=4[\sum\varepsilon_a-\sum\varepsilon_a\varepsilon_{a+1}]+2m\varepsilon_1$, which is the displayed form; only
the last term is not in $\tfrac12\Z$, and it is nonzero exactly when $\varepsilon_1=1$. All row sums of $E_{LL}$ vanish except the first, which is $m$,
so $q(w)=2m/8=m/4$; for odd $m$ that is a primitive quarter, $b(w,w)=m/2\ne0$, and $q$ is half-integral on $\ker b(w,\cdot)$, of order $2^{m-1}$, so the
splitting is the standard orthogonal one. For the Gauss sum, write
$e^{2\pi iq(\varepsilon)}=i^{m\varepsilon_1}\prod_{a=1}^m(-1)^{\varepsilon_a}(-1)^{\varepsilon_a\varepsilon_{a+1}}$, indices modulo $m$. The product is a
sum over a closed chain of binary variables with nearest-neighbour weights, so it is a trace of a transfer matrix: with
$M_{\varepsilon\varepsilon'}=(-1)^{\varepsilon}(-1)^{\varepsilon\varepsilon'}$, i.e.\ $M=\begin{psmallmatrix}1&1\\-1&1\end{psmallmatrix}$ in the order
$\varepsilon=0,1$, one has $\sum_\varepsilon\prod_a(-1)^{\varepsilon_a}(-1)^{\varepsilon_a\varepsilon_{a+1}}=\tr M^m$, and the sums restricted to
$\varepsilon_1=0$ and to $\varepsilon_1=1$ are the two diagonal entries $(M^m)_{00}$ and $(M^m)_{11}$. The eigenvalues of $M$ are $1\pm i$, so
$\tr M^m=(1+i)^m+(1-i)^m=2\operatorname{Re}((1+i)^m)$; and $(M^m)_{00}=(M^m)_{11}$ because $\tr(\operatorname{diag}(1,-1)M^m)=0$: $M=\sqrt2\,R_{-\pi/4}$ is a
multiple of a rotation, so $M^m=2^{m/2}R_{-m\pi/4}$ has equal diagonal entries. Hence each $\varepsilon_1$-sector contributes
$\operatorname{Re}((1+i)^m)$, the sector $\varepsilon_1=1$ carrying the extra factor $i^m$, and
$\sum_\varepsilon e^{2\pi iq(\varepsilon)}=\operatorname{Re}((1+i)^m)(1+i^m)$, of modulus $\sqrt{2^m}$ and argument $2\pi m/8$.
\end{proof}

\begin{corollary}
\label{cor:arf}
For odd $m$ the Gauss sum of $\Adec$ has modulus $\sqrt{2^{m-1}}$ and signature
\[
\sigma(\Adec)\equiv\begin{cases}m-1&m\equiv1\pmod4\\ m+1&m\equiv3\pmod4\end{cases}\pmod8,
\]
so $\Adec$ is Arf-odd \textup{(}$\sigma\equiv4$\textup{)} exactly when $m\equiv\pm3\pmod8$, and Arf-even \textup{(}$\sigma\equiv0$\textup{)} exactly when
$m\equiv\pm1\pmod8$.
\end{corollary}

\begin{proof}
Divide the Gauss sum of Proposition~\ref{prop:signsum} by that of the orthogonal $\Z/2$ factor, which is $1+i^m$: of modulus $\sqrt2$ and signature
$+1$ for $m\equiv1\pmod4$, $-1$ for $m\equiv3\pmod4$.
\end{proof}

At $m=3$ this gives $|\Adec|=4$ with one element of spin $0$ and three of spin $\tfrac12$ --- the three-fermion theory, $SO(8)_1$. At $m=5$ it gives
$|\Adec|=16$ with $6$ elements of spin $0$ and $10$ of spin $\tfrac12$, Gauss sum $-4$: an Arf-odd $\Z_2^4$ theory, for instance
three-fermion$\,\otimes\,$toric code.

\subsection{The spin function is the form itself}

Throughout, $\operatorname{spin}(V)$ is the spin \eqref{eq:spindef} of a Bethe vacuum $V$, defined in Section~\ref{sec:bethe}. The \emph{spin
function} of a class is $\varepsilon\mapsto\operatorname{spin}(V_\varepsilon)$ on its $2^{\#L}$ vacua, where $V_\varepsilon$ is obtained from
one vacuum by the sign flips $\varepsilon$. Proposition~\ref{prop:spinq} shows that, measured from a suitable base vacuum, it is the quadratic
form $q(\varepsilon)=\tfrac18\varepsilon^{\mathsf T}K_{LL}\varepsilon$.

\begin{lemma}[The parity of $E_{LL}$]
\label{lem:parity}
For every hyperbolic word, $\det E_{LL}\equiv\tr\varphi\pmod2$. Hence $E_{LL}$ is invertible over $\F_2$ precisely when $\tr\varphi$ is odd.
\end{lemma}

\begin{proof}
By Theorem~\ref{thm:Kformula}, modulo $2$ the matrix $E_{LL}=\tfrac12K_{LL}$ on the $\#L$ $L$-sites in cyclic order has diagonal entry $b$ at the first
$L$ after an $R$-run of length $b$ and $0$ at the other $L$'s, and off-diagonal entry $1$ for each adjacency of cyclically consecutive $L$'s; for
$\#L=2$ the two adjacencies cancel and for $\#L=1$ the matrix is $(b)$. For a cyclic tridiagonal matrix with unit off-diagonals, Hill's formula
(Section~\ref{sec:torsion}) reads $\det=(-1)^n(\tr\prod_iT_i-2)$ with $T_i=\begin{psmallmatrix}-d_i&-1\\1&0\end{psmallmatrix}$, so
$\det E_{LL}\equiv\tr\prod_iT(d_i)\pmod2$ with $T(d)=\begin{psmallmatrix}d&1\\1&0\end{psmallmatrix}$ over $\F_2$; the cases $\#L=1,2$ give $b$ and
$d_1d_2$, which the same expression gives. Over $\F_2$ one has $L^2=R^2=I$, $T(0)=LRL$ and $T(1)=L\,T(0)$, so an $L$-run of length $a$ after an
$R$-run of length $b$ contributes $T(b)T(0)^{a-1}=L^{b+1}R^aL$, and the cyclic product over the runs is, up to conjugation,
$L^{b_n}R^{a_1}L^{b_1}R^{a_2}\cdots L^{b_{n-1}}R^{a_n}$. This is the word $\varphi=L^{a_1}R^{b_1}\cdots L^{a_n}R^{b_n}$ with the letters $L$ and $R$ interchanged (each run $L^{a_i}$ becomes a run $R^{a_i}$ and each run $R^{b_i}$ a run $L^{b_i}$), read cyclically from the run $L^{b_n}$, and it is
conjugate to $\varphi$ by $\begin{psmallmatrix}0&1\\1&0\end{psmallmatrix}$ and so has the same trace.
\end{proof}

Proposition~\ref{prop:spinq}(ii) below contains the case $L^mR^m$ of this lemma, which is also checked directly there. The lemma also makes the flux-rule choice of the unit vacuum unique for odd trace (Section~\ref{sub:status}). The
flux rule singles out one vacuum of the unit class exactly when $\tr\varphi$ is odd.

\begin{proposition}
\label{prop:spinq}
Let $\varphi$ be a hyperbolic word and $V$ a Bethe vacuum with branch integers $n$, that is, the integers $n_b$ in the log-form Bethe equations $Q_b\log(1-x_b)+(KZ)_b=2\pi in_b$ (Proposition~\ref{prop:flux}). For $\varepsilon\in(\Z/2)^{\#L}$ let $V_\varepsilon$ be
the vacuum obtained by flipping the charge-$2$ coordinates in $\varepsilon$; put $q(\varepsilon)=\tfrac18\varepsilon^{\mathsf T}K_{LL}\varepsilon$
and $E_{LL}=\tfrac12K_{LL}$, as in Proposition~\ref{prop:signsum} for $L^mR^m$.
\begin{enumerate}
\item[\textup{(i)}] $\operatorname{spin}(V_\varepsilon)-\operatorname{spin}(V)\equiv q(\varepsilon)+\tfrac12n_L\cdot\varepsilon\pmod1$.
\item[\textup{(ii)}] $E_{LL}$ is invertible over $\F_2$ iff $\tr\varphi$ is odd (Lemma~\ref{lem:parity}).
When it is, $\varepsilon_0:=E_{LL}^{-1}n_L\bmod2$ exists and is unique, and
\[
\operatorname{spin}(V_\varepsilon)=\operatorname{spin}(V_{\varepsilon_0})+q(\varepsilon+\varepsilon_0)\pmod1 :
\]
measured from $V_{\varepsilon_0}$, the spin function on the $2^{\#L}$ vacua of the class is exactly $q$.
\item[\textup{(iii)}] For $\varphi=L^mR^m$ with $m$ odd and $V=\sigma V^*$ (the all-plus vacuum of its class), $E_{LL}\mathbf1=c_L=n_L$ as integer vectors, so $\varepsilon_0=\mathbf1$. The base vacuum is the all-flip of
$\sigma V^*$, and $q(\mathbf1)=m/4$.
\item[\textup{(iv)}] For odd trace, the sign-sum of the class is $e^{2\pi i\operatorname{spin}(V_{\varepsilon_0})}$ times the Gauss sum of $q$, so the $\mu_8$ property
follows from Milgram's formula \cite{MH73}.
\end{enumerate}
\end{proposition}

\begin{proof}
(i) is the computation in the proof of Proposition~\ref{prop:flux}. The flip sends $Z_a\mapsto Z_a+i\pi\varepsilon_a$ with $x$ fixed and changes
$\log\mathcal F$ by $i\pi\,\varepsilon\cdot n_L+\tfrac{i\pi}4\varepsilon^{\mathsf T}K_{LL}\varepsilon$, i.e.\ the spin by $\tfrac12n_L\cdot\varepsilon+q(\varepsilon)$. At a degenerate class the spin is defined along the continuation in $\epsilon$; there the deformation term changes under the flip by $-i\pi\log(1+\epsilon)\,c\cdot\varepsilon\to0$, so (i) holds in the limit.

(ii) The class of $n_L$ modulo $2$ does not depend on the branches: $Z_b\mapsto Z_b+2\pi i$ moves $n$ by a column of $K$, which is even, and the other branch of $\log(1-x_L)$ moves $n_L$ by $Q_L=2$. The general statement is Lemma~\ref{lem:parity}. For $L^mR^m$ directly: modulo $2$, $E_{LL}$ is $C+m\,e_1e_1^{\mathsf T}$ with $C$ the adjacency matrix of the $m$-cycle. The rows of $C$ sum to $0$ modulo $2$, so
$\det C\equiv0\pmod2$, which settles even $m$. For odd $m$, the determinant is linear in the $(1,1)$ entry, so
$\det E_{LL}\equiv\det C+m\det C'\equiv\det C+\det C'\pmod2$, where $C'$ is the matrix obtained by deleting the first row and column of $C$, the adjacency matrix of the path on the $m-1$
vertices $2,\dots,m$. The path determinants satisfy $p_n=-p_{n-2}$, $p_1=0$, $p_2=-1$, so $p_{m-1}$ is odd. Hence $E_{LL}$ is invertible over
$\F_2$ iff $m$ is odd. The bilinear form of $q$ is $b(\varepsilon,\varepsilon')=q(\varepsilon+\varepsilon')-q(\varepsilon)-q(\varepsilon')=
\tfrac12\varepsilon^{\mathsf T}E_{LL}\varepsilon'\bmod1$, so $\tfrac12n_L\cdot\varepsilon\equiv b(\varepsilon,\varepsilon_0)$ and (i) reads
$\operatorname{spin}(V_\varepsilon)-\operatorname{spin}(V)\equiv q(\varepsilon+\varepsilon_0)-q(\varepsilon_0)$. Taking $\varepsilon=\varepsilon_0$
and subtracting gives the display, since $q(2\varepsilon_0)=0$.

(iii) The row sums of $E_{LL}$ are $(m,0,\dots,0)$ (proof of Proposition~\ref{prop:signsum}), and $n_L=c_L=(m,0,\dots,0)$ at $\sigma V^*$
(proof of Proposition~\ref{prop:flux}); so $E_{LL}\mathbf1=n_L$ and $\varepsilon_0=\mathbf1$, with $q(\mathbf1)=m/4$ by Proposition~\ref{prop:signsum}.

(iv) By (ii), $\sum_\varepsilon e^{2\pi i\operatorname{spin}(V_\varepsilon)}=e^{2\pi i\operatorname{spin}(V_{\varepsilon_0})}\sum_\varepsilon e^{2\pi iq(\varepsilon)}$.
\end{proof}

We call $V_{\varepsilon_0}$ the \emph{base vacuum} of the class. It has a description that does not use branch integers: it is the unique
vacuum $V_b$ of the class with
\[
\operatorname{spin}\bigl((V_b)_\varepsilon\bigr)-\operatorname{spin}(V_b)\equiv q(\varepsilon)\pmod1\qquad\text{for every }\varepsilon\in(\Z/2)^{\#L},
\]
where $(V_b)_\varepsilon$ is $V_b$ with the charge-$2$ coordinates in $\varepsilon$ flipped. It is unique because two such vacua differing by
$\varepsilon'$ would have $b(\varepsilon,\varepsilon')=0$ for every $\varepsilon$, and $b$ is non-degenerate since $E_{LL}$ is invertible over
$\F_2$, exactly for odd trace. So
the translation hedged in Proposition~\ref{prop:signsum} is determined rather than chosen, and the base vacuum is a computable function of the
class's branch integers. Part (ii) is the reason odd $m$ behaves and even $m$ does not: the invertibility of $E_{LL}$ over $\F_2$ holds exactly for odd
$m$. Gang--Kim--Lee \cite[footnote~4]{GKL23} note that for $LR^n$ the vacuum may also be taken at its $\Z_2$ partner, with consistent
modular data either way. The two have the same $\mathcal H$ and are two vacua of the class of $\sigma V^*$, related by the sign flip of the charge-$2$
coordinate that carries the semion. The choice therefore does not affect the class-level statement that the unit class is that of $\sigma V^*$.

\subsection{The Galois rule for central charges}
\label{sub:galoisrule}

The last ingredient in the phase bookkeeping is the action of $\operatorname{Gal}(\Q(\zeta_m)^+/\Q)$ on central charges, where $\zeta_m$ is a primitive $m$-th root of unity and $\Q(\zeta_m)^+=\Q(\zeta_m+\zeta_m^{-1})$ is the maximal real subfield of $\Q(\zeta_m)$. It is a statement about the
Bloch group, and the version needed is the following. It is used in this paper for the denominator of $c$, the extended $L$ of Theorems~\ref{thm:rogers}
and~\ref{thm:lines}, and Proposition~\ref{obs:kirillov}. It is the
tool on which the Galois rule of \cite{EvB26} rests.

\begin{proposition}[Suslin normalisation]
\label{prop:suslin}
Let $m$ be odd, $\zeta_m$ a primitive $m$-th root of unity, $F=\Q(\zeta_m)^+=\Q(\zeta_m+\zeta_m^{-1})\subset\R$ its maximal real subfield, and let $\xi=\sum_an_a[x_a]\in B(F)$ with all $x_a\in F\setminus\{0,1\}$. Here $B(F)$ is the Bloch group (as in the proof of Corollary~\ref{cor:rational}), the subgroup of the pre-Bloch group of $F$ cut out by the condition $\sum_an_a\,x_a\wedge(1-x_a)=0$, and $\mathfrak p(\R)$ below is the pre-Bloch group of $\R$. For a field $k$ the pre-Bloch group $\mathfrak p(k)$ is generated by the symbols $[x]$, $x\in k\smallsetminus\{0,1\}$, modulo Suslin's five-term relation
\[
[x]-[y]+\Bigl[\frac yx\Bigr]-\Bigl[\frac{1-x^{-1}}{1-y^{-1}}\Bigr]+\Bigl[\frac{1-x}{1-y}\Bigr]=0\qquad(x\ne y).
\]
\begin{enumerate}
\item[\textup{(i)}] $B(F)$ is finite cyclic of order $\tfrac14\operatorname{lcm}(24,m)=6m/\gcd(3,m)$.
\item[\textup{(ii)}] Extend the Rogers dilogarithm to $\R\setminus\{0,1\}$ in the standard way, by $\widetilde L(x)=L(x)-L(1)$ for $0<x<1$,
$\widetilde L(x)=-\widetilde L(1/x)$ for $x>1$, and $\widetilde L(x)=-L(1)+\widetilde L(1/(1-x))$ for $x<0$. Then $\xi\mapsto\sum_an_a\widetilde L(x_a)$ is a well-defined homomorphism
$\widetilde L:\mathfrak p(\R)\to\R/\pi^2\Z$.
\item[\textup{(iii)}] $\operatorname{Gal}(F/\Q)$ acts on $B(F)$ as follows: the automorphism $\sigma_k$, $k\in(\Z/m)^\times/\!\pm1$, $\zeta_m\mapsto\zeta_m^k$, acts by multiplication by $a^2$, for some $a$ depending on $k$. Therefore $\widetilde
L(\sigma_k\xi)=a^2\widetilde L(\xi)$ in $\R/\pi^2\Z$. By (iv), $c(\sigma_k\xi)\equiv a^2c(\xi)\equiv k^2c(\xi)\pmod1$, for $c(\xi)=\sum_an_aL(x_a)/L(1)$.
\item[\textup{(iv)}] $c(\xi)\bmod1$ lies in $\tfrac6{|B(F)|}\Z/\Z$, i.e.\ $\tfrac m{\gcd(3,m)}c(\xi)\in\Z$.
\end{enumerate}
\end{proposition}

\begin{proof}
For (i): $K_3(F)$ is the third algebraic $K$-group of $F$ and $K_3^M(F)\subset K_3(F)$ the image of the Milnor $K$-group, the part generated by symbols; write $K_3(F)_{\mathrm{nd}}=K_3(F)/K_3^M(F)$ for the indecomposable part, the quotient that symbols do not see, as in \cite{Suslin91}. Borel's rank computation \cite{Borel77} gives $K_3(F)_{\mathrm{nd}}\otimes\Q=0$ for $F$ totally real; the torsion of
$K_3(F)_{\mathrm{nd}}$ is cyclic of order $w_2(F)$, the largest $N$ for which $\operatorname{Gal}(F(\mu_N)/F)$ has exponent dividing $2$
\cite{MS90,Levine89}; Suslin's exact sequence $0\to\operatorname{Tor}(\mu(F),\mu(F))^\sim\to K_3(F)_{\mathrm{nd}}\to B(F)\to0$
\cite[Thm.~5.2]{Suslin91}, in which the first term is the nontrivial extension of $\operatorname{Tor}(\mu(F),\mu(F))$ by $\Z/2$ and has order $4$ for $F\subset\R$, gives $|B(F)|=w_2(F)/4$; and for $F=\Q(\zeta_m)^+$ with $m$ odd,
$\operatorname{Gal}(F(\mu_N)/F)\subset(\Z/N)^\times$ has exponent dividing $2$ exactly for $N\mid\operatorname{lcm}(24,m)$, so
$w_2(F)=\operatorname{lcm}(24,m)$. For (ii), that $\widetilde L$ vanishes modulo $\pi^2$ on every real five-term relation is \cite[\S2.3.1]{Kirillov95}; the value there is $0$ or $-\pi^2$. For (iii), $K_3(F)_{\mathrm{nd}}$ is identified Galois-equivariantly with $H^0(F,\Q/\Z(2))$ \cite{MS90,Levine89}, on which $\sigma_k$ acts by the square of the cyclotomic character, and Suslin's sequence is natural. So $\sigma_k$ acts on $B(F)$ by $a^2$.
Here $a$ is $k$ itself when $k$ is prime to $6$. In general $a$ is any integer with $a\equiv\pm k\pmod m$ and $\gcd(a,6)=1$, which exists by the Chinese remainder theorem since $m$ is odd and $k$ is prime to $m$. As $a^2\equiv k^2\pmod m$ and $a^2\equiv1$ modulo $2$ and $3$, the class of $a^2$ modulo $|B(F)|$ does not depend on the choice. Then $\widetilde L(\sigma_k\xi)=a^2\widetilde L(\xi)$ because $\widetilde L$ is a homomorphism. Part (iv) follows from (i) and (ii). Part (iv) sharpens the naive denominator $24$ of $c\bmod1$ to $m/\gcd(3,m)$, and this denominator is sharp: $\xi_1=\sum_{b=2}^{(m-1)/2}[x_b]$ lies in $B(F)$, since its image in $\wedge^2F^\times$ telescopes to $0$, and $c(\xi_1)=1-3/m$ has exact order $m/\gcd(3,m)$ modulo $1$, by \eqref{eq:Smk}.
\end{proof}

\subsection{The two spin forms are one discriminant form}
\label{sub:spinform}

Let $\varphi=\begin{psmallmatrix}a&b\\c&d\end{psmallmatrix}$ have odd trace $k\ge5$ (in this subsection and the next, $k$ is the trace, not the index $k$ of Section~\ref{sub:galoisrule}), and put $N_\mp=k\mp2=|A_\mp|$, which are odd and coprime. With
$J=\begin{psmallmatrix}0&1\\-1&0\end{psmallmatrix}$, $\langle x,y\rangle=\det(x,y)=x^{\mathsf T}Jy$, let
\[
\Phi(u)=u^{\mathsf T}J\varphi u=c\,u_1^2+(d-a)u_1u_2-b\,u_2^2,
\]
which we call the binary quadratic form of $\varphi$; its discriminant is $k^2-4=N_-N_+$.

The map $u\mapsto\operatorname{adj}(\varphi\mp1)u$ is a bijection $A_\mp\to\mathcal K_\mp$, the kernel group of Section~\ref{sec:spectrum}: indeed $(\varphi\mp1)v\equiv0\pmod{N_\mp}$
iff $v\equiv\operatorname{adj}(\varphi\mp1)u$ for some $u$, and $\operatorname{adj}(\varphi\mp1)u\equiv0$ iff $u\in(\varphi\mp1)\Z^2$. Define
\begin{equation}
q_\mp(u)=-[4^{-1}]_{N_\mp}\,\frac{\Phi(u)}{N_\mp}\pmod 1\qquad(u\in A_\mp),
\label{eq:qform}
\end{equation}
with $[4^{-1}]_N$ the inverse of $4$ modulo $N$, and let $\gamma_\mp=G(q_\mp)/\sqrt{N_\mp}$, $G(q)=\sum_ue(q(u))$, $e(x)=e^{2\pi ix}$.

\begin{lemma}
\label{lem:disc}
Let $\mathbb L_\varphi=(\Z^2,S_\varphi)$ with $S_\varphi=J\varphi+(J\varphi)^{\mathsf T}=J(\varphi-\varphi^{-1})$. Then $\mathbb L_\varphi$ is even and indefinite, and its
discriminant form is isometric to $(A_-,q_-)\perp(A_+,-q_+)$. In particular $q_\mp$ are well-defined nondegenerate quadratic forms, and
$\gamma_-=\gamma_+$.
\end{lemma}

\begin{proof}
$S_\varphi$ has diagonal $(2c,-2b)$ and $\det S_\varphi=-(k^2-4)<0$. Its discriminant form is $y\mapsto\tfrac12y^{\mathsf T}S_\varphi^{-1}y$ on
$\Z^2/S_\varphi\Z^2$; substituting $u=Jy$ and using $JS_\varphi=-(\varphi-\varphi^{-1})$, a direct computation turns it into
the form $u\mapsto-\Phi(u)/N_-N_+$ on $\Z^2/(\varphi^2-1)\Z^2$. As $\gcd(N_-,N_+)=1$, the $N_\mp$-primary part is $(\varphi\pm1)\Z^2/(\varphi^2-1)\Z^2$, identified
with $A_\mp$ by $x\mapsto(\varphi\pm1)x$. From $\varphi^{\mathsf T}J=J\varphi^{-1}$ and $\varphi^2=k\varphi-1$,
\[
(\varphi\pm1)^{\mathsf T}J\varphi(\varphi\pm1)=\pm J(\varphi\pm1)^2=\pm(k\pm2)\,J\varphi ,
\]
so $\Phi((\varphi+1)x)=N_+\Phi(x)$ and $\Phi((\varphi-1)x)=-N_-\Phi(x)$, and the two parts carry $-\Phi/N_-$ on $A_-$ and $\Phi/N_+$ on $A_+$. Since
$4[4^{-1}]_N\equiv1$, these are $q_-(2x)$ and $-q_+(2x)$, and $x\mapsto2x$ is an automorphism of each group. The discriminant form of an even
lattice is nondegenerate, and Milgram's formula \cite{MH73} gives $G(q_-)\overline{G(q_+)}=\sqrt{N_-N_+}\,e(\operatorname{sign}\mathbb L_\varphi/8)=\sqrt{N_-N_+}$.
Hence $\gamma_-\bar\gamma_+=1$ with $|\gamma_\mp|=1$, that is, $\gamma_-=\gamma_+$.
\end{proof}

\subsection{The Rogers sums at the non-degenerate classes}
\label{sub:rogersnd}

We use $R$, the extended Rogers dilogarithm $L=\operatorname{Re}R$ and the Rogers sum $c_\alpha=\sum_tL(x_t)/L(1)$ of a class, as in \eqref{eq:Fphase}.

\begin{theorem}
\label{thm:rogers}
Let $\varphi$ be a hyperbolic word of odd trace, and let $\xi$ be the non-degenerate class of $v=\operatorname{adj}(\varphi\mp1)u\in\mathcal K_\mp$. Then
\[
c_\xi\equiv\#R-6\,\frac{\Phi(u)}{N_\mp}\pmod 6 .
\]
Consequently every one of the $2^{\#L}$ vacua of $\xi$ has spin $\equiv\#R/24-\Phi(u)/4N_\mp\pmod{\tfrac14}$.
\end{theorem}

The proof below does not use that the trace is odd; the hypothesis is kept because the forms $q_\mp$ used with the theorem need $N_\mp$ odd.

\begin{proof}
\emph{Step 1: reduction to principal real parts.} For real $s\notin\{0,1\}$ the reflection formula for $\Li_2$ gives $L(1-s)=L(1)-L(s)$. Hence, exactly,
\begin{equation}
c_\xi=\#R+\frac6{\pi^2}\,\Sigma,\qquad \Sigma:=\sum_{L\text{-steps}}\operatorname{Re}R(s_t)-\sum_{R\text{-steps}}\operatorname{Re}R(s_t).
\label{eq:Sigma}
\end{equation}

\emph{Step 2: a primitive.} Put $\vartheta(e)=\pi\langle v,e\rangle/N_\mp$ for $e\in\Z^2$, with $v=\operatorname{adj}(\varphi\mp1)u$ as an integer vector, the lift used in Step~4; $\vartheta$ is linear. At step $t$
let $(c_1,c_2)$ be the Farey basis and $(\alpha,\beta)=(\vartheta(c_1),\vartheta(c_2))$; an $L$-step replaces it by $(c_1,c_1+c_2)$ and has
$s=\sin^2\alpha/\sin^2(\alpha+\beta)$, an $R$-step by $(c_1+c_2,c_2)$ with $s=\sin^2\beta/\sin^2(\alpha+\beta)$. Define
\[
z(\alpha,\beta)=\frac{e^{2i\alpha}-1}{e^{2i(\alpha+\beta)}-1}=\frac{\sin\alpha}{\sin(\alpha+\beta)}\,e^{-i\beta},\qquad
1-z=\frac{\sin\beta}{\sin(\alpha+\beta)}\,e^{i\alpha},
\]
\[
\mathcal G(\alpha,\beta)=2\operatorname{Re}R(z)+\alpha\arg z+\beta\arg(1-z).
\]

\emph{Step 3: the step identities.} By Lemma~\ref{lem:steps} (Appendix~\ref{app:steps}),
\begin{equation}
\operatorname{Re}R\bigl(\tfrac{\sin^2\alpha}{\sin^2(\alpha+\beta)}\bigr)-\mathcal G(\alpha,\beta)+\mathcal G(\alpha,\alpha+\beta)\in\pi^2\Z,\qquad
-\operatorname{Re}R\bigl(\tfrac{\sin^2\beta}{\sin^2(\alpha+\beta)}\bigr)-\mathcal G(\alpha,\beta)+\mathcal G(\alpha+\beta,\beta)\in\pi^2\Z
\label{eq:steps}
\end{equation}
wherever the arguments are defined. These are the step identities.

\emph{Step 4: telescoping.} Non-degeneracy gives $f\neq0$, with $f(e)=\sin(\pi\langle v,e\rangle/N)$ as in Theorem~\ref{thm:allwords}, on every kept, dying, middle and newborn slope (proof of Theorem~\ref{thm:Hall}, Appendix~\ref{app:Hall}), that
is $\sin\alpha,\sin\beta,\sin(\alpha+\beta)$ and $\sin(2\alpha+\beta)$ (resp.\ $\sin(\alpha+2\beta)$) are nonzero at every step, so \eqref{eq:steps}
applies at every step. Consecutive steps share a basis, so
\[
\Sigma\equiv \mathcal G\bigl(\vartheta(e_1),\vartheta(e_2)\bigr)-\mathcal G\bigl(\vartheta(\varphi e_1),\vartheta(\varphi e_2)\bigr)\pmod{\pi^2},
\]
the final basis being the columns of $\varphi$. Write $\varphi\mp1=\varphi-\upsilon$ with $\upsilon=\pm1$. Then $\det(\varphi-\upsilon)=-\upsilon N_\mp$, and
$\varphi^{-1}v=\upsilon v+N_\mp\tilde w$ with $\tilde w=\varphi^{-1}u$, because $\varphi^{-1}v-\upsilon v=-\upsilon\varphi^{-1}(\varphi-\upsilon)\operatorname{adj}(\varphi-\upsilon)u$.
Hence $\vartheta(\varphi e)=\upsilon\vartheta(e)+\pi\langle\tilde w,e\rangle$. Put $(\alpha_0,\beta_0)=(\vartheta(e_1),\vartheta(e_2))$ and
$(p_1,p_2)=(\langle\tilde w,e_1\rangle,\langle\tilde w,e_2\rangle)$. Since $z$ is $\pi$-periodic in each variable and
$z(-\alpha,-\beta)=\overline{z(\alpha,\beta)}$, the final $z$ is $z_0$ or $\bar z_0$; $\operatorname{Re}R$ agrees on the two, and the arg-terms give
\[
\Sigma\equiv-\upsilon\pi(p_2\alpha_0-p_1\beta_0)=-\frac{\upsilon\pi^2}{N_\mp}\langle v,\tilde w\rangle\pmod{\pi^2}.
\]
Finally $\upsilon\langle v,\tilde w\rangle=\upsilon\,u^{\mathsf T}J(\varphi-\upsilon)\varphi^{-1}u=-u^{\mathsf T}J\varphi^{-1}u=\Phi(u)$, using
$\operatorname{adj}(M)^{\mathsf T}J=JM$, $u^{\mathsf T}Ju=0$ and $\varphi^{-1}=k-\varphi$. So $\Sigma\equiv-\pi^2\Phi(u)/N_\mp\pmod{\pi^2}$, and \eqref{eq:Sigma} gives the
congruence.

\emph{Step 5: spins.} At a non-degenerate class all Bethe coordinates are real and $x_L=s_t>0$ at the charge-$2$ sites, so
$\operatorname{Im}Z_a\in\pi\Z$ for every $a$ and in \eqref{eq:Fphase} the factor $e^{n\cdot Z/2}$ lies in
$\{\pm1,\pm i\}$; hence every vacuum of the class has spin $\equiv c_\xi/24\pmod{\tfrac14}$.
\end{proof}

Under the grafting fusion rules the twist, $\operatorname{spin}-\#R/24$, is the element of odd order in its class modulo $\tfrac14$, which is $q_\mp(u)$
(Corollary~\ref{cor:T}).

\subsection{Rogers sums of a full chain}
\label{sub:fullchain}

At a degenerate class the closed form of Theorem~\ref{thm:allwords} fails at some site, and Theorem~\ref{thm:rogers} does not apply. The
finite coordinates of the degenerate $A_-$ classes of $L^mR^m$ (Section~\ref{sec:leading}) are, instead, full chains
$\sin^2\vartheta/\sin^2(b\vartheta)$, $b=2,\dots,m-2$, at $\vartheta=k\pi/m$, and their Rogers sums are known in closed form for $\gcd(k,m)=1$, the range $k>\tfrac{m-1}2$ following from $k\leftrightarrow m-k$:
\begin{equation}
m\,S(m,k):=\frac m{L(1)}\sum_{b=1}^{m-1}L(x_b)=4m-6+6(k-1)(m-k-1),\qquad x_b=\frac{\sin^2(k\pi/m)}{\sin^2(bk\pi/m)} .
\label{eq:Smk}
\end{equation}
It was conjectured as identity (3.1) of Nahm--Recknagel--Terhoeven \cite{NRT92}, there as a sum over $b=2,\dots,m-1$, and it is identity (1.16) of \cite{Kirillov95}. Kirillov proved it
\cite[Thm.~3.1, Cor.~3.4(i)]{Kirillov92}, the case $n=2$ needed here being proved in \cite[\S4]{Kirillov92}. Indeed, the terms $b$ and $m-b$
coincide and the terms $b=1,m-1$ equal $L(1)$, and with $r=m-2$, $j=k-1\le r$ one has $S(m,k)=1+s(k-1,2,m-2)$ in the notation of
\cite{Kirillov92}, which is \eqref{eq:Smk} by \cite[(3.7)]{Kirillov92}. For $k=1$ it is the classical central-charge identity (Lewin,
Richmond--Szekeres, Kirillov--Reshetikhin; \cite[Cor.~4]{Kirillov95}).

\begin{proposition}[The Rogers sums modulo $1$]
\label{obs:kirillov}
For odd $m$ and $\gcd(k,m)=1$, $mS(m,k)\equiv-6k^2\pmod m$.
\end{proposition}

This follows at once from \eqref{eq:Smk}, since $(k-1)(m-k-1)\equiv1-k^2\pmod m$.
It is also an instance of the Galois rule of Proposition~\ref{prop:suslin}(iii): $S(m,k)=2+2c(\xi_k)$ with
$\xi_k=\sum_{b=2}^{(m-1)/2}[x_b]$, which is $\sigma_k(\xi_1)$ in the Bloch group of $\Q(\zeta_m)^+$ \cite[\S2.3]{Kirillov95}; so
$c(\xi_k)\equiv k^2c(\xi_1)\pmod1$, and $c(\xi_1)=1-3/m$ is the case $k=1$.

In conformal terms the right side of \eqref{eq:Smk} is $m(1+c^{(m-2)}-24h_{k-1}+6(k-1))$, the $1$ being the term $b=1$,
with $c^{(\kappa)}=3\kappa/(\kappa+2)$ the central charge of $SU(2)_\kappa$ and $h_{k-1}=(k^2-1)/4m$ the conformal weight of the primary labelled by the
Galois conjugate: the statement $24\theta_\alpha\equiv c_\alpha$ of \eqref{eq:theta24} seen from the WZW side.

\subsection{Comparison with the Gang--Kim template}
\label{sub:template}

By Remark~\ref{rem:internal}, what remains of input (b) is whether the spins match the template: the $T$-eigenvalues
$e^{2\pi iA^2/4(k-2)}$, $A=1,\dots,k-3$, on the $A_-$ block of $\mathrm{TFT}_-[\mathcal S_k]$ \cite{GK22}, of which the $A_-$ block of
\eqref{eq:GKT} is the even-$A$ half (Observation~\ref{obs:e1even}). As printed they cannot: the multiset $\{A^2/4(k-2)\}$,
$A=1,\dots,k-3$, has Gauss modulus $\sqrt2$, $1$, $\sqrt2$ at $N=k-2=3,4,9$ against the $\sqrt N$ that a quadratic refinement of a
linking form on $A_-$ must have, and the Gauss modulus is independent of any multiplier on the form. What the template conceals is a
splitting by parity.

\begin{observation}[The template splits into two refinements of one linking form]
\label{obs:e1even}
Let $N$ be odd. Writing $A=2a$ and $A=2a+1$,
\begin{equation}
\begin{gathered}
\frac{A^2}{4N}=\frac{a^2}N\quad(A\ \text{even},\ a=1,\dots,\tfrac12(N-1)),\\
\frac{A^2}{4N}=\frac{a^2+a}N+\frac1{4N}\quad(A\ \text{odd},\ a=0,\dots,\tfrac12(N-3)),
\end{gathered}
\label{eq:parsplit}
\end{equation}
so the template multiset splits by the parity of $A$; $a^2/N$ and $(a^2+a)/N$ are quadratic refinements of the same bilinear form $2ab/N$ on
$\Z/N$, and the even half, with the $\pm$-doubled weighting proper to a form on $A_-$ and the term $a=0$ included, has Gauss modulus $\sqrt N$, the modulus of the quadratic
Gauss sum. The two halves are translates of one another: for $A$ even, $(A+N)^2/4N=A^2/4N+A/2+N/4\equiv A^2/4N+N/4\pmod1$, and $A\mapsto A+N$ carries
the even labels in $[2,N-1]$ to the odd labels in $[N+2,2N-1]$. The symmetry $A\mapsto2N-A$, under which $(2N-A)^2/4N\equiv A^2/4N\pmod1$, carries these to the odd labels in $[1,N-2]$. So
\begin{equation}
\{A^2/4N:A\ \text{odd}\}=\{A^2/4N:A\ \text{even}\}+\tfrac N4\pmod 1
\label{eq:transN}
\end{equation}
as multisets. All of this is elementary. The empirical part is the identification of the shift $N/4$ as the semion of the sign group. At
$\varphi=LR^n$, where a class has exactly two vacua, $K_{LL}=2n$ and $q(1)=K_{LL}/8=n/4\equiv N/4\pmod1$ in both sectors ($N=n$ or $n+4$). For the general case it was checked by computer on all $427$ cyclic
words of length at most $11$ with $\tr\varphi>2$, the whole range checked and nothing is claimed beyond it, that $q$ attains $\pm N/4$, Among the $103$ words of odd trace, both values occur at $95$ words, only $+N/4$ at $LR$, $LR^3$, $LR^5$, $LR^7$, $LR^9$ and $LR^3LR^3$, and only $-N/4$ at $LRLR$ and $LRLR^5$.
\end{observation}

So the $a_1$ sector of the template carries two labels per class of $A_-$ --- one even $A$, one odd, a translate by $N$ in $\Z/2N$
differing in spin by the semion's $N/4$. This is the $2^{\#L}$ sign structure of Section~\ref{sec:bethe}, and the object to construct
is a labelling of vacua, $a\in(A_-\smallsetminus0)/\!\pm\mapsto\{2a,2a+N\}\subset(\Z/2N)/\!\pm$, with the even label carrying the linking
form's own value $a^2/N$. Corollary~\ref{cor:templatelabel} shows that the run index supplies it.

\begin{corollary}[The template label]
\label{cor:templatelabel}
Let $\varphi$ have odd trace, let $A_\mp$ be cyclic of order $N$, and let $\rho$ be a run with kept slope $p_\rho$. On the non-degenerate
classes of $A_\mp$, with $v=\operatorname{adj}(\varphi\mp1)u$ and index $j_\rho=\langle v,p_\rho\rangle$ \textup{(}Theorem~\ref{thm:allwords}\textup{)},
\[
-\Phi(u)\equiv c_\rho\,j_\rho^2\pmod N
\]
for a constant $c_\rho$ independent of the class, whenever the index of a generator is prime to $N$. Hence every vacuum of the class has spin
$\equiv\#R/24+c_\rho j_\rho^2/4N\pmod{\tfrac14}$. For $\varphi=LR^n$ and $\rho$ the $R$-run, $c_\rho=1$ in both sectors. On $A_+$ this gives
\[
\operatorname{spin}\equiv\frac{\#R}{24}+\frac{j^2}{4N}\pmod{\tfrac14},
\]
so the template label of Gang--Kim on the $(k+2)$-block is $B=j$. The $A_-$ sector of $LR^n$ is degenerate, and there the same spin formula is observed
numerically at the harvested classes.
\end{corollary}

\begin{proof}
In a cyclic sector $u=tu_0$, so $v=tv_0$, $\Phi(u)=t^2\Phi(u_0)$ and $j_\rho=t\langle v_0,p_\rho\rangle$; put $c_\rho=-\Phi(u_0)\langle v_0,p_\rho\rangle^{-2}$.
The spin is Theorem~\ref{thm:rogers}. For $LR^n$, $\Phi(u)=nu_1^2-nu_1u_2-u_2^2$ and $p_\rho=e_1+e_2$. On $A_-$ ($N=n$),
$v=(-u_2,\,n(u_2-u_1))$ and $j\equiv-u_2$, so $j^2\equiv-\Phi(u)$. On $A_+$ ($N=n+4$), $v=(2u_1-u_2,\,-nu_1+(n+2)u_2)$ and
$j\equiv-2u_1+u_2$, so $j^2=4u_1^2-4u_1u_2+u_2^2\equiv-\Phi(u)$, since $n\equiv-4$.
\end{proof}

A label with a multiplier, $\operatorname{spin}\equiv(\varpi j_\rho)^2/4N+\delta$, therefore exists exactly when $c_\rho$ is a square modulo $N$:
in the $A_+$ sector, $c_\rho=4$ on the $L$-run of $L^3R^3$ ($\varpi=\pm2$), while $c_\rho\equiv-1$ on the $L$-run of $L^3R$ ($N=7$) and of $L^5R$ ($N=9$) is not a square. For $LR^n$
the whole $A_-$ sector is degenerate, as is one $A_+$ class of $LR^5$, and the corollary does not reach them. The same formula holds numerically
at the harvested degenerate classes of $LR^5$, $LR^7$ and $LR^9$.

\part{The \texorpdfstring{$A_-$}{A-} sector at prime \texorpdfstring{$m$}{m}}
\label{part:aminus}

Part~\ref{part:results} gives every non-degenerate Bethe class in closed form and labels it by its kernel vector $\pm v\in\mathcal K_\mp$
(Theorems~\ref{thm:tworun} and~\ref{thm:allwords}). It does not reach the degenerate classes, at which $\sin(\pi\langle v,\mathsf m_t\rangle/N)$
vanishes on some middle slope $\mathsf m_t$ (the middle slope of the $t$-th tetrahedron, Section~\ref{sub:runtraces}); at $L^mR^m$ with $m$ prime these are the whole $A_-$ sector (Proposition~\ref{prop:Hspectrum}). This part treats that sector at
prime $m\ge5$, which the labelling (Theorem~\ref{thm:lines}) and the corank (Theorem~\ref{thm:corank}) require; the existence theorem
(Theorem~\ref{thm:exist}) holds for every odd $m\ge3$, and Remark~\ref{rem:aminusgaps}(iv) uses it at $m=3$. Each of its $\tfrac12(m^2-1)$ classes is described by trigonometric chains (two, or one for the types $0$ and $\infty$), which are the closed form of
Theorem~\ref{thm:allwords} at the kernel vector of the class, on the sites where that formula is defined; the class corresponds to a point
$\pm u$ of $(A_-\smallsetminus0)/\!\pm$, and the $\tfrac1m\Z/\Z$-component of its spin is $q_-(u)$, the form that Theorem~\ref{thm:rogers} gives at the non-degenerate classes.

The classes are degenerate: their coordinates blow up or collapse as $\epsilon\to0$, and they exist only in the deformed system of
Section~\ref{sec:bethe}. The analysis rests on three facts. On the two-parameter family of chains with shifted phases the $2m$ Bethe equations
reduce \emph{exactly} to two scalar equations; these are solved by the implicit function theorem, with Jacobian $-4$, at every placement $j\in\{1,\dots,m-1\}$, that is, for every type except $0$ and $\infty$; and the leading-order Jacobian of the full system has
corank exactly two, the minimum compatible with the two shifts (proved at $j=1$; at $j\ne1$ its finite block is proved nonsingular outside the
roots of a sextic, and the corank verified for $m\le43$). Within the shifted-chain family the solution with given leading data is unique;
whether it is isolated in the full system is not settled, since the reduced bifurcation equations begin at second order in $\epsilon$
(Section~\ref{sub:uniqueness}).

The argument runs in four steps. Section~\ref{sec:aminussummary} collects the conclusions in Theorem~\ref{thm:aminus}.
\emph{Labelling} (Section~\ref{sec:lines}). The placement types are the points of $\mathbb P^1(\F_m)$, that is, the lines through $0$ in $A_-$. Each class corresponds to a point of
$(A_-\smallsetminus0)/\!\pm$ through its kernel vector, and the $\tfrac1m\Z/\Z$-component of its spin is $q_-$ of that point (Theorem~\ref{thm:lines}).
\emph{Shape} (Section~\ref{sec:leading}). The leading data are two chains with two free phases (Proposition~\ref{prop:shifts}).
\emph{Existence} (Section~\ref{sec:shifts}). On the shifted chains (the chains with these two phases) the Bethe system reduces exactly to two equations
(Proposition~\ref{prop:reduction}, Appendix~\ref{app:existence}), solved at every placement $j\in\{1,\dots,m-1\}$ (Theorem~\ref{thm:exist}).
\emph{Rank of the linearisation} (Section~\ref{sec:lemma1}). The corank of the leading-order Jacobian (Theorem~\ref{thm:corank}).

\section{Placement types as lines}
\label{sec:lines}

In the limit $\epsilon\to0$ some coordinates of each $A_-$ class tend to $0$, $1$ or $\infty$. These are its \emph{degenerate
coordinates}. The \emph{leading data} of a class are the leading powers of $\epsilon$ and the leading coefficients of its coordinates. Here and in Section~\ref{sec:leading} the coordinates are $C=z_L^2=x_L^{-1}$ at the $L$-sites and $w=z_R=x_R^{-1}$ at the $R$-sites. The
\emph{placement type} of the class records where the degenerate ones sit: at which $L$-site $C$ tends to $0$, and at which $R$-site the
degenerate $R$-coordinates are centred. The $L$-site at which $C\to0$ is called the \emph{small $L$}. For $j\in\{1,\dots,m-1\}$ the type $j$ has one small $L$, at $L_{j}$, with its two neighbours at $1+c_\pm\epsilon$, and a
triple of $R$-coordinates $(A_{\mathrm l}/\epsilon,\ s\epsilon^2,\ A_{\mathrm r}/\epsilon)$ centred at $R_{j^{-1}}$. In the type $0$ the small $L$ is $L_m$ and every
$R$-coordinate tends to $1$. In the \emph{fixed type} $\infty$ every $L$-coordinate degenerates ($x_L\to0$) and the $R$-coordinates form a single chain.

The labels are the points of the projective line $\mathbb P^1(\F_m)=\F_m\cup\{\infty\}$. By Theorem~\ref{thm:lines}(ii), the type $j\in\F_m$ is the line of $A_-\cong\F_m^2$ spanned by $(1,j)$, and the type $\infty$ is the line spanned by $(0,1)$.
So the type of a class is the line through $0$ in $A_-$ on which its point lies. The $L$-sites are indexed cyclically, $L_0=L_m$, so in the type $j$ the small $L$ is $L_j$ for every $j\in\F_m$. The
placement is not a choice: by Theorem~\ref{thm:lines}(i) it is read off the closed form of
Theorem~\ref{thm:allwords} at the kernel vector $v=m(v_1,v_2)$ of the class, with $j\equiv v_1/v_2\pmod m$,
and $j=\infty$ when $v_2\equiv0$.

The count closes. For prime $m$ the projective line $\mathbb P^1(\F_m)$ has $m+1$ points and each carries $\tfrac12(m-1)$ classes, so a labelling by
lines would account for $(m+1)\tfrac{m-1}2=\tfrac12(m^2-1)=|A_-\smallsetminus0|/2$ classes --- the whole sector, with nothing left over. The enumeration
realises this: the $m$ types $0,1,\dots,m-1$ give $\tfrac12m(m-1)$ classes and the fixed type is the remaining line, contributing $\tfrac12(m-1)$.
At $m=5,7,11$ the harvested $A_-$ sectors, with $\mathcal H=2^mm^2$, are complete on exactly this count: $10+2=12$, $21+3=24$, $55+5=60$.

Let $m\ge5$ be prime and $\varphi=L^mR^m$, so that $N=|A_-|=m^2$, $\mathcal K_-=m(\Z/m^2)^2\cong(\Z/m)^2$ and $A_-=\Z^2/m\Z^2$. Let the leading data of
Section~\ref{sec:leading} run over the pairs $(j,k)$, with $j\in\F_m$ the placement type or $j=\infty$ the fixed type, and
$k\in(\Z/m)^\times\!/\!\pm$ the Galois index.

\begin{theorem}[The kernel vector and spin of an $A_-$ class]
\label{thm:lines}
With $m$, $\varphi$ and the pairs $(j,k)$ as just described:
\begin{enumerate}
\item[\textup{(i)}] the leading data $(j,k)$ are the closed form of Theorem~\ref{thm:allwords} at the kernel vector $v=m(jk,k)$ (at $v=m(k,0)$ for
the fixed type) on every site where its denominator is nonzero. For $j\in\{1,\dots,m-1\}$ the degenerate sites are exactly those where the denominator vanishes, together with their
neighbours. For the type $0$ the denominator vanishes only at $L_m$, and every $R$-coordinate tends to $1$. For the fixed type the numerator vanishes on the $L$-run, and every $L$-coordinate tends to $0$; the corresponding point of $A_-$ is $u=k(1,j)$, respectively $u=(0,k)$, and $v\equiv-\operatorname{adj}(\varphi-1)u$, which is $\operatorname{adj}(\varphi-1)u$ of Theorem~\ref{thm:rogers} up to sign, $v$ mattering only up to sign;
\item[\textup{(ii)}] the placement types are the points of $\mathbb P^1(\F_m)$: the type $j\in\F_m$ is the line of $A_-$ spanned by $(1,j)$, and the fixed type is the point $\infty$, the line spanned by $(0,1)$; and $(j,k)\mapsto\pm u$ is a bijection onto $(A_-\smallsetminus0)/\!\pm$, the Galois index $k$ running over the $\tfrac12(m-1)$ antipodal pairs on each line;
\item[\textup{(iii)}] for every vacuum of a solution with these leading data that is real-analytic in $\epsilon$ for $0<\epsilon<\epsilon_0$ with real coordinates (the coordinates are then meromorphic at $\epsilon=0$), such as the solutions of Theorem~\ref{thm:exist}, the component in $\tfrac1m\Z/\Z$ of its spin is $q_-(u)$, with $q_-$ the form
\eqref{eq:qform}. Equivalently, the $\theta$-exponent of the class, the residue $e\in\Z/m$ with $e/m\equiv$ that component, is $(j^2-1)(k/2)^2$, and $(k/2)^2$ for the fixed type.
\end{enumerate}
Here a class means its leading data; which of them are limits of exact solutions is Theorem~\ref{thm:aminus}.
\end{theorem}

\begin{proof}
\textup{(i)} For $v=m(v_1,v_2)$, $\langle v,e\rangle/m^2=(v_1e_2-v_2e_1)/m$. In the layered triangulation of $L^mR^m$, as in the proof of
Theorem~\ref{thm:tworun}, the $L$-run has kept slope $e_1$ and middle slopes $(i,1)$, and the $R$-run has kept slope $(m,1)$ and middle slopes
$(1+im,i)$, $i=1,\dots,m$. At $v=m(jk,k)$, with $\vartheta=k\pi/m$, the formula of Theorem~\ref{thm:allwords} gives
$x_{L_i}=\sin^2\vartheta/\sin^2((i-j)\vartheta)$ and $1-x_{R_i}=\sin^2(j\vartheta)/\sin^2((ij-1)\vartheta)$: the chains of
Section~\ref{sec:leading} at zero phase. The denominator vanishes exactly at $L_{j}$ and at $R_p$, $p=j^{-1}$; at $L_{j\pm1}$ the formula
gives $x=1$ and at $R_{p\pm1}$ it gives $x=0$, which is the pattern of the type $j$. At $v=m(0,k)$ the only vanishing denominator is at $L_m$ and every
$x_R=1$: the type $0$. At $v=m(k,0)$, $f$ vanishes on the kept slope of the $L$-run, so every $x_L=0$, and on the middle slope of $R_m$; the
$R$-coordinates are $1-\sin^2\vartheta/\sin^2(i\vartheta)$: the fixed type. Finally $(\varphi-1)\Z^2=m\Z^2$, and
$\operatorname{adj}(\varphi-1)=\begin{psmallmatrix}0&-m\\-m&m^2\end{psmallmatrix}$ sends $k(1,j)$ to $-m(jk,k)$ and $(0,k)$ to $-m(k,0)$ modulo $m^2$.
\textup{(ii)} Every nonzero $u\in(\Z/m)^2$ lies on exactly one line through the origin, spanned by $(1,j)$ or by $(0,1)$, and the
$\tfrac12(m-1)$ antipodal pairs on a line are the multiples $\pm k$, $k\in(\Z/m)^\times\!/\!\pm$.
\textup{(iii)} With $\varphi=\begin{psmallmatrix}1+m^2&m\\m&1\end{psmallmatrix}$, $\Phi(u)=m(u_1^2-u_2^2)-m^2u_1u_2$, so
$q_-(u)=-[4^{-1}]\Phi(u)/m^2\equiv[4^{-1}](u_2^2-u_1^2)/m$, which is $(j^2-1)(k/2)^2/m$ at $u=k(1,j)$ and $(k/2)^2/m$ at $(0,k)$. That this is
the spin is the following computation.

Put $\vartheta=k\pi/m$. First let $j\in\{1,\dots,m-1\}$ and $p=j^{-1}\bmod m$. At the class
$(j,k)$ the finite $L$-coordinates are $\sin^2\vartheta/\sin^2(n\vartheta)$, $n=2,\dots,m-2$, the values of \eqref{eq:closedform}. By (T5) of
Lemma~\ref{lem:telescope} the finite $R$-coordinates satisfy $1-x_{R_i}=\sin^2(j\vartheta)/\varsigma_i^2$ with $\varsigma_i=\sin((ij-1)\vartheta)
=\pm\sin((i-p)j\vartheta)$, since $(pj-1)\vartheta\in\pi\Z$; over the arc $i-p=2,\dots,m-2$ these are the values of \eqref{eq:closedform} at
the Galois index $jk$. The remaining coordinates tend to $0$, $1$ or $\infty$ and contribute integers to $c_\alpha$, which vanish after
multiplication by $m$ modulo $m$. With the reflection identity $L(x)+L(1-x)=L(1)$ (modulo $\pi^2$ for the extended $L$ of
Proposition~\ref{prop:suslin}(ii)) and Proposition~\ref{obs:kirillov} at $k$ and at $jk$,
\[
mc_\alpha\equiv mS(m,k)-mS(m,jk)\equiv-6k^2+6(jk)^2=6k^2(j^2-1)\pmod m ,
\]
and $24\theta\equiv c_\alpha$, which holds at every vacuum of the class by \eqref{eq:theta24}, turns this into $4e\equiv k^2(j^2-1)$, dividing by $6$, which is invertible modulo $m$ since $m\ge5$.

At the type $0$ every $R$-coordinate tends to $1$ and contributes an integer, and the finite $L$-coordinates are the chain at $k$. So $mc_\alpha\equiv mS(m,k)\equiv-6k^2$, and $e\equiv-(k/2)^2$, which is $(j^2-1)(k/2)^2$ at $j=0$. At the fixed type every $L$-coordinate tends to $0$ and contributes an integer. The finite $R$-coordinates satisfy $1-x_{R_i}=\sin^2\vartheta/\sin^2(i\vartheta)$, $i=2,\dots,m-2$, so by the reflection identity $mc_\alpha\equiv-mS(m,k)\equiv6k^2$, and $e\equiv(k/2)^2$. For these two types no solution is constructed, so there (iii) describes the spin of a solution if one exists.

The relation \eqref{eq:theta24} was derived for solutions with finite coordinates. At a degenerate class it is applied along the
continuation in $\epsilon$, which is how the fibering operator of such a class is defined (Section~\ref{sec:bethe}). The existence of the continuation is Theorem~\ref{thm:exist}, proved in Appendix~\ref{app:existence}, which does not use the present theorem. For $0<\epsilon<\epsilon_0$ the solution
of that theorem is real-analytic in $\epsilon$ with real data. Hence all its coordinates are finite, nonzero and real, and
$\operatorname{Im}Z_a\in\{0,\pi\}$. The deformed log-form Bethe equations differ from the undeformed ones by the term $c_b\log(1+\epsilon)$.
This adds to $\log\mathcal F$ a term proportional to $\log(1+\epsilon)\,c\cdot Z$, which tends to $0$ with $\epsilon$ since
$Z_a=O(\log\epsilon)$. Hence \eqref{eq:theta24} holds at every $\epsilon$ up to a term that vanishes in the limit, and both sides are continuous
along the continuation, $c_\alpha$ with $L=\operatorname{Re}R$, so that $L-L(1)$ is the function $\widetilde L$ of Proposition~\ref{prop:suslin}(ii).

Distinct pairs $(j,k)$ give distinct classes, because the leading data differ. Two classes with different $j$ have different sets of degenerate coordinates, hence different placement types. Two classes with the same $j$ and
$k\ne\pm k'$ have different $x_b$.
\end{proof}

So the degenerate $A_-$ classes of $L^mR^m$ are labelled exactly as the non-degenerate classes of every word are. They are labelled by the kernel vector, read off
the class's own coordinates, with the $\tfrac1m\Z/\Z$-component of the spin given by $q_-$ of the corresponding point (compare Theorem~\ref{thm:rogers}). Summed over the $\tfrac12(m^2-1)$
classes, the Gauss sum of the sector is $\sum_{x\ne0}e(q_-(x))$ halved, which is $\tfrac12(m-1)$. The two isotropic lines are the types $1$ and $m-1$, where $e=0$. On every other line the exponents run once over a coset of the squares
in $\F_m^\times$, as $k$ runs over $(\Z/m)^\times\!/\!\pm$: over $(j^2-1)(\F_m^\times)^2$ on the type $j$, and over $(\F_m^\times)^2$ itself on the
fixed type.

\begin{remark}[Exhaustiveness]
\label{rem:exhaustive}
Theorem~\ref{thm:lines} says the closed-form family has $|A_-\smallsetminus0|/2$ distinct members. That the family is \emph{all} of the $A_-$ sector then
follows if the sector has that many classes, which is Theorem~\ref{thm:count} read through localisation, and if every class has leading data of this list. The second condition is a classification of the possible degenerate classes, and it is not proved. So the labelling is proved and the count is arithmetic. Exhaustiveness needs F1 (Section~\ref{sub:F1}) and the classification. Nothing else in this part depends on them.
\end{remark}

\section{The closed form for the leading data}
\label{sec:leading}

\subsection{The chains}

Throughout this part a \emph{chain} is the frieze row (Remark~\ref{rem:friezedef}) carried by a run in the limit $\epsilon\to0$: the
\emph{$L$-chain} is the sequence of leading-order $L$-coordinates along the $L$-run, the \emph{$R$-chain} that along the $R$-run. The word
``run'' refers to the letters, ``chain'' to the sequence of values they carry.

In the placement type $0$, where all $R\to1$, the finite $L$-coordinates couple only to their $L$-neighbours ($E_{LL}=(2,-1)$). The leading
equations read $(C_b-1)^2=C_{b-1}C_{b+1}$ with $C_1=C_{m-1}=1$: the square of the $A$-type $Q$-recurrence. Writing $C_b=\mathcal Q_b^2$ with
$\mathcal Q_b^2=1+\mathcal Q_{b-1}\mathcal Q_{b+1}$, the solution is $\mathcal Q_b=\sin(b\vartheta)/\sin\vartheta$. The symbol is chosen to match the $\mathcal Q_j=1+jd$ of
Theorem~\ref{thm:nahmclosed}: $\mathcal Q_j/d$ there and $\mathcal Q_b$ here are both frieze rows \eqref{eq:frieze}, of trace $2$
and $2\cos\vartheta$ (Section~\ref{sub:runchain}). Here the row vanishes at both ends, $\mathcal Q_0=\mathcal Q_m=0$, and this forces $\vartheta=k\pi/m$; the condition $\mathcal Q_{m-1}=\pm1$ alone would not. Hence, in closed form,
\begin{equation}
x_b=\frac{\sin^2(k\pi/m)}{\sin^2(bk\pi/m)},\qquad b=1,\dots,m-1,
\label{eq:closedform}
\end{equation}
the endpoints $b=1,m-1$ being the coordinates $1+c\epsilon$, and $k$ labelling the Galois conjugates. For a general placement there are two chains: an
$L$-chain $\rho_i=\sin((i-j)\vartheta+\alpha)$ with $C_i=\rho_i^2/\sin^2\vartheta$, and an $R$-chain $\varsigma_i=\sin((ij-1)\vartheta+\beta)$ with
$w_i=\varsigma_i^2/(\varsigma_{i-1}\varsigma_{i+1})$, where $j$ is the placement and $\alpha,\beta$ are phases left free at leading order. At $\alpha=\beta=0$ these are the coordinates that the closed form of Theorem~\ref{thm:allwords} gives at the kernel vector of the class
(Theorem~\ref{thm:lines}(i)).

\subsection{The leading-order system}

Substitute into the deformed Bethe equations \eqref{eq:bethe} (with right side $(1+\epsilon)^{c_b}$) the chains, with the degenerate coordinates written with their leading power of $\epsilon$ and an
unknown coefficient --- the small $L$ (Section~\ref{sec:lines}) as $C_s\epsilon^2$, its neighbours as $1+c_\pm\epsilon$, the triple as $A_{\mathrm l}/\epsilon$, $s\epsilon^2$, $A_{\mathrm r}/\epsilon$ ---
and keep the leading term of each equation. Each chain carries a phase, $\alpha$ for the $L$-chain and $\beta$ for the $R$-chain. We call $\alpha$ and $\beta$ the \emph{chain shifts}, $\alpha$ the $L$-shift and $\beta$ the $R$-shift, and the chains with these phases the \emph{shifted chains}. The interior equations are the recurrences of the two chains, $C_i-1=\rho_{i-1}\rho_{i+1}/\sin^2\vartheta$
and the telescoping of the $w_i$; the degenerate coefficients are fixed by the phases of the chains.

\begin{proposition}[The flat directions are chain shifts]
\label{prop:shifts}
Both chains satisfy their recurrences for \emph{every} phase: $\mathcal Q_n=\sin(n\vartheta+\alpha)/\sin\vartheta$ satisfies $\mathcal Q_n^2-\mathcal Q_{n-1}\mathcal Q_{n+1}=1$ identically in
$\alpha$, and with $\varsigma_i=\sin((ij-1)\vartheta+\beta)$ the products of $w_i=\varsigma_i^2/(\varsigma_{i-1}\varsigma_{i+1})$ telescope identically in $\beta$. The
degenerate coordinates are exactly the zeros of the unshifted chains --- $\mathcal Q_0=0$ at the small $L$, $\varsigma_p=0$ at the middle of the triple, $p=j^{-1}$
--- and taking the phases of the two chains of order $\epsilon$, $\alpha=a\epsilon$ and $\beta=b\epsilon$ with $a,b$ the unknown \emph{shift
coefficients}, resolves them, giving the degenerate coefficients explicitly in terms of $a$ and $b$:
\[
C_s=\frac{a^2}{\sin^2\vartheta},\qquad c_\pm=\pm2a\cot\vartheta,\qquad s=-\frac{b^2}{\sin^2(j\vartheta)},\qquad
A_{\mathrm l,\mathrm r}=\mp\frac{\sin^2(j\vartheta)}{\sin(2j\vartheta)\,b}.
\]
Hence, for $j\ne1$, the leading data carry two free parameters $(a,b)$: by Proposition~\ref{prop:reduction} (Appendix~\ref{app:existence}), for $\alpha,\beta=O(\epsilon)$ one has $\Lambda,\Xi=1+O(\epsilon)$, so the leading-order equations hold for all $(a,b)$.
\end{proposition}

\begin{proof}
For the $L$-chain, with $\psi=n\vartheta+\alpha$, $\sin^2\psi-\sin(\psi-\vartheta)\sin(\psi+\vartheta)=\sin^2\vartheta$ gives $\mathcal Q_n^2-\mathcal Q_{n-1}\mathcal Q_{n+1}=1$ for every $\alpha$. For the $R$-chain, the telescoping of the products of the $w_i$ is the content of the identities (T1)--(T5) of Lemma~\ref{lem:telescope}, which hold for arbitrary $\beta$. For the coefficients, put $\alpha=a\epsilon$, $\beta=b\epsilon$ and expand to leading order. At the small $L$, $C_s\epsilon^2=\mathcal Q_0^2=\sin^2(a\epsilon)/\sin^2\vartheta$, and at its neighbours $\sin^2(\pm\vartheta+a\epsilon)/\sin^2\vartheta=1\pm2a\epsilon\cot\vartheta+O(\epsilon^2)$. At the centre of the triple, $\varsigma_p=\sin\beta\approx b\epsilon$ and $\varsigma_{p-1}\varsigma_{p+1}=\sin(\beta-j\vartheta)\sin(\beta+j\vartheta)\approx-\sin^2(j\vartheta)$, so $w_p\approx-b^2\epsilon^2/\sin^2(j\vartheta)$. At its neighbours $w_{p\pm1}=\varsigma_{p\pm1}^2/(\varsigma_p\varsigma_{p\pm2})\approx\pm\sin^2(j\vartheta)/(b\epsilon\sin(2j\vartheta))$, so $w_{p-1}=A_{\mathrm l}/\epsilon$ and $w_{p+1}=A_{\mathrm r}/\epsilon$ with the values in the display.
\end{proof}

Two cases are special. For $j=1$ the gap sits at the first $L$ and the triple wraps across the fan ($R_m,R_1,R_2$), so both chains meet at the same run
boundary and their phases are tied: the family is $|\alpha|=|\beta|$, with the next order selecting $\beta=-\alpha$ (Theorem~\ref{thm:exist}(ii)). For the type $0$ there is no triple, and the family is the
$L$-shift $\alpha$ alone, with every $R$-coefficient equal to $C_s$.

Proposition~\ref{prop:reduction} in Appendix~\ref{app:existence} shows that on the shifted chains the full system, not only its leading order, reduces to two scalar conditions.

\section{Existence}
\label{sec:shifts}

The leading-order system leaves the two chain shifts free. This section states that the whole deformed Bethe system fixes them, so that the
leading data of Section~\ref{sec:leading} are the limits of exact solutions at every placement $j\not\equiv0$. The proof is in
Appendix~\ref{app:existence}. On the shifted chains of Section~\ref{sec:leading}, with phases $\alpha,\beta$, put
\[
\Lambda=C_1\Bigl(\frac{\varsigma_0}{\varsigma_1}\Bigr)^2,\qquad \Xi=\frac{\sin^2(j\vartheta)}{\varsigma_0^2\,C_m}.
\]
The whole system then reduces exactly to the two scalar equations
\[
\textup{(III)}\ \ \Lambda=1+\epsilon,\qquad\qquad \textup{(IV)}\ \ \Xi=1
\]
(Proposition~\ref{prop:reduction}): the equation at the $l$-th $R$-site, written as its left side $(R_l)$, equals $\Xi\Lambda^l$. This
$\Lambda$ is not the Floquet multiplier of Section~\ref{sec:spectrum}. The Jacobian of (III), (IV) in the two phases is $-4$ at every
placement $j\ne1$ (Proposition~\ref{prop:jac4}); at $j=1$ the determinant $-4$ is that of $(\Lambda,G)$, $G=(\Xi-1)/\alpha$, in the variables $(\alpha,t)$ with $\beta=t\alpha$ (proof of
Theorem~\ref{thm:exist}(ii)). The implicit function theorem then applies.

\begin{theorem}[Existence]
\label{thm:exist}
Let $m\ge3$ be odd, $k,j\in(\Z/m)^\times$ and $\vartheta=k\pi/m$. There are $\epsilon_0>0$ and real-analytic $\alpha(\epsilon),\beta(\epsilon)$, real for real $\epsilon$, on
$|\epsilon|<\epsilon_0$, vanishing at $\epsilon=0$ and nonzero for $\epsilon\ne0$, such that for $0<|\epsilon|<\epsilon_0$ the shifted chains of
Section~\ref{sec:leading} with phases $\alpha(\epsilon),\beta(\epsilon)$ solve the deformed Bethe system of $L^mR^m$ exactly, with the leading
data of type $j$ and Galois index $k$. Moreover:
\begin{enumerate}
\item[\textup{(i)}] for $j\ne1$, $\alpha=-\tfrac\epsilon2\cot\vartheta+O(\epsilon^2)$ and $\beta=\tfrac\epsilon2\cot(j\vartheta)+O(\epsilon^2)$;
\item[\textup{(ii)}] for $j=1$, $\alpha=-\tfrac\epsilon2\cot\vartheta+O(\epsilon^2)$ and $\alpha+\beta=-\epsilon^2\cos\vartheta/(4\sin^3\vartheta)+O(\epsilon^3)$.
\end{enumerate}
\end{theorem}

So the shifts are fixed by a Jacobian equal to $-4$ at every placement $j\not\equiv0$, the case $j=1$ after the change of variable $\beta=t\alpha$;
part (ii) gives the drift $\varepsilon:=-(\alpha+\beta)$ away from the $\beta=-\alpha$ of the case $j=1$ described after Proposition~\ref{prop:shifts}. Only the type $0$
($j\equiv0$) lies outside Theorem~\ref{thm:exist}: there the $R$-chain is constant, (T5) of Lemma~\ref{lem:telescope} gives $1-1/w_l=0$, and no choice of phases satisfies
$(R_l)$. That type, like the fixed type, needs a different family.

\section{The corank of the leading-order Jacobian}
\label{sec:lemma1}

Theorem~\ref{thm:exist} gives, at every placement $j\not\equiv0$, an exact solution with the leading data of that
placement. It is the only one on the shifted-chain family. This section asks whether it is isolated in the full Bethe system. The first
requirement is that the linearised system be as non-degenerate as the two chain shifts suggest. So the leading-order Jacobian $\hat J_0$,
constructed below, should have corank exactly $2$. Theorem~\ref{thm:corank} states what is proved, by two routes.
\begin{itemize}
\item At $j=1$ the matrix splits into an $L$-block and an $R$-block. Their determinants are two continuants (determinants of tridiagonal matrices), evaluated in closed form in
Appendix~\ref{sub:continuants}. The corank is then exactly $2$ (Section~\ref{sub:corankproof}).
\item At $j\ne1$ the blocks are coupled. The determinant of the finite block is computed by reduction modulo the ramified prime
$\mathfrak p=(1-\zeta_m)$ of $\Q(\zeta_m)$ above $m$, where, after halving the $L$-rows, it is the unit $-j^4/2$ (Appendix~\ref{sub:modp}). This is proved except at the roots of a sextic $\mathsf P$, given in \eqref{eq:Ppoly} below, and checked at those
roots for $m<260$. The step from the finite block to the corank is verified numerically for every prime $m\le43$.
\end{itemize}
Section~\ref{sub:uniqueness} explains why corank two does not by itself give isolatedness, which remains open. The results of
Part~\ref{part:aminus} are collected in Section~\ref{sec:aminussummary}.

\emph{The leading-order Jacobian.} We write $\hat J_0$ for it. Its construction is as follows. Take the $2m$ deformed
Bethe equations \eqref{eq:bethe} in logarithmic form, one for each site $b$ of the word, $\sum_aK_{ab}Z_a+Q_b\log(1-x_b)=2\pi in_b+c_b\log(1+\epsilon)$ (Section~\ref{sub:relation}). Differentiate each equation with respect to the $2m$ logarithmic coordinates $Z_a$, one for each site $a$, and evaluate the partial derivatives on the leading data of the class. The derivatives are taken in the logarithmic coordinates, so the columns of the degenerate coordinates (exponent $p_a\ne0$ in the sense of Section~\ref{sub:degenerate}) are finite and are evaluated directly. The only columns that diverge are those of the two neighbours $L_{j\pm1}$ of the small $L$, where $1-x_a\to0$ and the diagonal term grows like $\epsilon^{-1}$. These two columns are multiplied by $\epsilon$ before $\epsilon\to0$. The result is a $2m\times2m$ matrix, with rows indexed by the equations and
columns by the sites. Its entries are those of the level matrix $K$ plus the diagonal $D$ of Section~\ref{sec:torsion}. For the $L$-sites they are $4/(1-x_a)$ on the diagonal and $-2$ between neighbours; for the $R$-sites, $2\min(i,l)$ between $R_i$ and $R_l$ plus a diagonal term. The $L$-diagonal at $L_1$ carries in addition the entry $2m$ of $K$, and the $L$--$R$ entries are the ramps of $K$. For $j\ne1$ the row of the small $L_j$ vanishes identically, as the proof of Theorem~\ref{thm:corank} shows, so
$\operatorname{corank}\hat J_0\ge1$. For $j\ne1$ the two chain shifts are flat directions of the leading-order system
(Proposition~\ref{prop:shifts}), which suggests corank $2$.

\subsection{The corank theorem}
\label{sub:corankproof}

The proof at $j\ne1$ (the \hyperref[thm:modp]{finite-block lemma}) inverts one block of the reduced matrix of the \hyperref[prop:modp]{reduction lemma}, and that block is singular exactly at the roots of the polynomial below. It comes from the determinant of that block, $\det M=-\mathsf P/(24p^2(p-1))$, computed in Appendix~\ref{sub:modp}:
\begin{equation}
\mathsf P(p)=2p^6-3p^5-p^4+15p^3-13p^2+12\pmod m,\qquad p=j^{-1} .
\label{eq:Ppoly}
\end{equation}
Its vanishing is a limitation of that method, not a property of the classes: at every root with $m<260$ the determinant was computed directly and
is the same value $-j^4/2$ as at the other placements (see after the \hyperref[thm:modp]{finite-block lemma}). Since $\mathsf P$ is irreducible over $\Q$, it has on average one
root modulo a prime (Chebotarev); among the primes $5\le m<1000$, $68$ have no root in $\{2,\dots,m-1\}$, $65$ have one, $21$ two, $9$ three and
$3$ four.

\begin{theorem}[The corank of $\hat J_0$]
\label{thm:corank}
Let $m\ge5$ be prime, $k\in(\Z/m)^\times$, and let $j\in(\Z/m)^\times$ be a placement. At $j=1$, $\operatorname{corank}\hat J_0=2$. At $j\ne1$,
$\operatorname{corank}\hat J_0=1+\operatorname{corank}M''$, where $M''$ is the block of $\hat J_0$ on all sites other than the small $L$ and its two
neighbours; the finite block of $M''$ is nonsingular when $\mathsf P(j^{-1})\not\equiv0\pmod m$, so $\operatorname{corank}\hat J_0\le4$; and
by a floating-point computation, $\operatorname{corank}\hat J_0=2$ at every placement of every prime $m\le43$.
\end{theorem}

The proof below uses four lemmas, proved in Appendix~\ref{app:corank}: the arc-continuant and $L$-block lemmas at $j=1$, and the
reduction and finite-block lemmas at $j\ne1$. A continuant, the determinant of a tridiagonal matrix, satisfies a three-term recurrence in its size.

\begin{proof}[Proof of Theorem~\ref{thm:corank}]
\emph{The case $j=1$.} We show $\operatorname{corank}\hat J_0=2$. Recall the construction of $\hat J_0$ at the start of this section. Read $\psi^{\mathsf T}\hat J_0=0$ column by column. The two $1+c\epsilon$ columns are multiples of $e_{L_2}$ and $e_{L_m}$, so $\psi_{L_2}=\psi_{L_m}=0$. The
finite-$L$ columns are $L_3,\dots,L_{m-1}$; each is tridiagonal in $L$ alone, and since $L_1,L_m$ are not among them the system closes on
$\psi_{L_3},\dots,\psi_{L_{m-1}}$ with the boundary values already $0$. Its matrix is the finite-$L$ block, nonsingular by
the \hyperref[prop:lblock]{$L$-block lemma}, so $\psi_{L_n}=0$ for $2\le n\le m$.

Every $R$-column gives $2i\psi_{L_1}-2\psi_{L_m}+\sum_l2\min(i,l)\psi_{R_l}+D_i\psi_{R_i}=0$, with $D_i$ the diagonal. Second differences in $i$
annihilate $2i\psi_{L_1}$ and $-2\psi_{L_m}$ and reduce $\sum_l2\min(i,l)\psi_{R_l}$ to $-2\psi_{R_i}$, giving
$d_{i-1}\psi_{R_{i-1}}-(2+2d_i)\psi_{R_i}+d_{i+1}\psi_{R_{i+1}}=0$ for $2\le i\le m-1$, with $d_i$ as in the
\hyperref[prop:arc]{arc-continuant lemma}. Since $D_2=D_m=0$ the chain breaks at both ends of the triple,
and $i=3,\dots,m-1$ is a closed system on the finite $R$-sites: the arc matrix, nonsingular by the \hyperref[prop:arc]{arc-continuant lemma}. Hence $\psi_{R_i}=0$ for
$3\le i\le m-1$, and the broken relation at $i=2$ reads $-\psi_{R_1}-2\psi_{R_2}=0$.

Three unknowns remain, $\psi_{L_1},\psi_{R_2},\psi_{R_m}$, and the columns $L_1$, $R_m$ and $R_2$ all give $\psi_{L_1}+\psi_{R_m}=0$, while the column
$R_1$ is satisfied identically on substituting $\psi_{R_1}=-2\psi_{R_2}$ and $\psi_{R_m}=-\psi_{L_1}$. So the solution space is free in $\psi_{L_1}$ and
$\psi_{R_2}$, of dimension exactly $2$.

\emph{The case $j\ne1$.} In $\hat J_0$ the columns of the two neighbours $L_{j\pm1}$ are nonzero multiples of $e_{L_{j\pm1}}$ (nonzero because $c_\pm=\mp\cot^2\vartheta\ne0$ on the solutions of Theorem~\ref{thm:exist}),
so a left kernel vector $\psi$ has $\psi_{L_{j\pm1}}=0$ and the corank of $\hat J_0$ is that of the block $M'''$ on the other $2m-2$ sites. In
$M'''$ the row and column of the small $L_j$ vanish: its diagonal is $4+D_{L_j}=4-4=0$ and its only couplings are to $L_{j\pm1}$. Hence
$\operatorname{corank}\hat J_0=1+\operatorname{corank}M''$ with $M''$ the block on the finite $L$-sites and all $m$ $R$-sites, and the theorem asks
for $\operatorname{corank}M''=1$. The finite block $B$ of the \hyperref[thm:modp]{finite-block lemma} is $M''$ with the three sites of the triple removed, so
$\det B\ne0$ gives $\operatorname{corank}M''\le3$, with equality to $1$ if and only if the $3\times3$ Schur complement of $B$ in $M''$ has rank $2$. The
reduction modulo $\mathfrak p$ cannot decide this last step: modulo $\mathfrak p$ the entries of $M''$ are those of the \hyperref[prop:modp]{reduction lemma}
extended to the triple ($D=0$ at its ends and $-1$ at its centre are the values of $(i-p)^2-1$ there), and the corank of the reduced $M''$ is
not $1$ (by exact computation over $\F_m$ it is $3$ at every placement $j\ne1$ for primes $m\le61$); the exact corank is decided at higher $\mathfrak p$-adic order, as the arc continuant of
the \hyperref[prop:arc]{arc-continuant lemma} is. It was computed in floating point for every placement $j\ne1$ and every $k$ of every prime $5\le m\le43$ ($3720$
placements): $\hat J_0$ has two singular values below $10^{-15}$ times the largest and the third at least $4\cdot10^{-8}$ of the largest, so
$\operatorname{corank}\hat J_0=2$ there. A proof uniform in $m$ would need the closed form of the Schur complement, as in the \hyperref[prop:arc]{arc-continuant lemma}.
\end{proof}

\begin{remark}[The dichotomy is exact]
At $j=1$ the block splits as $L\oplus R$ and both factors have closed forms, but their determinants have positive $\mathfrak p$-adic valuation (Appendix~\ref{sub:continuants}, the \hyperref[prop:lblock]{$L$-block} and \hyperref[prop:arc]{arc-continuant} lemmas), so the reduction modulo $\mathfrak p$ is singular and gives no information. At $j\ne1$ the block is
coupled, the splitting is unavailable, and the $\mathfrak p$-adic reduction of the finite block is a unit. Neither route covers a single placement the
other does; together they cover the finite block at all of them except the placements with $\mathsf P(p)\equiv0\pmod m$ discussed above.
\end{remark}

\subsection{From the corank to uniqueness}
\label{sub:uniqueness}

Corank two is what a Lyapunov--Schmidt reduction on the two chain shifts needs; it does not by itself make the solution of Theorem~\ref{thm:exist} isolated. In a
Lyapunov--Schmidt reduction, with the two chain shifts as coordinates on the kernel, the range equations are solved by the implicit function
theorem and the solutions are the zeros of two bifurcation equations in $(a,b)$. The first-order residual of the shifted chains is, in the
logarithmic equations, exactly $(\log\Lambda-\log(1+\epsilon))\,c+(\log\Xi)\,\mathbf1_R$ (Proposition~\ref{prop:reduction}), and both vectors $c$ and
$\mathbf1_R$ lie in the range of $\hat J_0$, numerically at every placement computed (at $j=1$ only $c$ does, which is the tying of the phases noted after Proposition~\ref{prop:shifts}). So the
bifurcation equations vanish at first order in $\epsilon$ and begin at second order; the Jacobian $-4$ of Proposition~\ref{prop:jac4} is that of
(III), (IV) restricted to the shifted-chain family, not that of the bifurcation equations. Uniqueness within the family is therefore proved
(Theorem~\ref{thm:exist}); isolatedness in the full system is not, and is left open.

\section{Summary of Part~\ref{part:aminus}}
\label{sec:aminussummary}

The following theorem collects what Sections~\ref{sec:lines}--\ref{sec:lemma1} prove about the $A_-$ sector of $L^mR^m$ at prime $m$.
Remark~\ref{rem:aminusgaps} lists what it leaves open.

\begin{theorem}[The $A_-$ sector of $L^mR^m$ at prime $m$]
\label{thm:aminus}
Let $m\ge5$ be prime.
\begin{enumerate}
\item[\textup{(i)}] \emph{Labelling and spins.} The $\tfrac12(m^2-1)$ leading data of Section~\ref{sec:leading} correspond bijectively to
$(A_-\smallsetminus0)/\!\pm$ through their kernel vectors, and the $\tfrac1m\Z/\Z$-component of the spin of every vacuum of a real-analytic real solution with the leading data at $\pm u$ is $q_-(u)$.
\item[\textup{(ii)}] \emph{Existence.} For every placement $j\in\{1,\dots,m-1\}$ and every $k$, the leading data $(j,k)$ are the limit as
$\epsilon\to0$ of an exact solution of the deformed Bethe system. These are $\tfrac12(m-1)^2$ distinct $A_-$ classes: all but the $m-1$ of
the types $0$ and $\infty$.
\item[\textup{(iii)}] \emph{Rank of the linearisation.} That solution is the only one with these leading data on the shifted-chain family. The leading-order Jacobian has
corank exactly two at $j=1$; at $j\ne1$ its finite block is nonsingular when $\mathsf P(j^{-1})\not\equiv0\pmod m$, and for every $j$ when $m<260$,
and, numerically, the corank is two at every placement of every prime $m\le43$.
\end{enumerate}
\end{theorem}

\begin{proof}
\textup{(i)} is Theorem~\ref{thm:lines}, and \textup{(ii)} is Theorem~\ref{thm:exist}; the classes are distinct because their leading data are.
\textup{(iii)} Uniqueness on the family is the implicit function theorem of Theorem~\ref{thm:exist}. The corank at $j=1$ is Theorem~\ref{thm:corank};
the finite block is the \hyperref[thm:modp]{finite-block lemma} and the evaluation recorded after it; the corank for $m\le43$ is the computation in the proof of Theorem~\ref{thm:corank}
(Section~\ref{sub:corankproof}).
\end{proof}

\begin{remark}[What Theorem~\ref{thm:aminus} does not cover]
\label{rem:aminusgaps}
\textup{(i)} The types $0$ and $\infty$, together $m-1$ classes. Their leading data and kernel vectors are those of Theorem~\ref{thm:aminus}(i). Existence is not proved; at the type $0$ the shifted chains cannot satisfy the $R$-equations (see after Theorem~\ref{thm:exist}). If such solutions exist and are real-analytic with real data, Theorem~\ref{thm:lines}\textup{(iii)} gives their spin components.
\textup{(ii)} Isolatedness of the solutions of Theorem~\ref{thm:exist} in the full system, at every placement: the bifurcation equations of the
Lyapunov--Schmidt reduction begin at second order in $\epsilon$ (Section~\ref{sub:uniqueness}). The corank statement itself is proved only at $j=1$,
and verified for $m\le43$ otherwise.
\textup{(iii)} That these classes, with the $m-1$ of \textup{(i)}, are the whole sector rests on (F1) and on a classification of the possible degenerate leading data, which is not proved (Remark~\ref{rem:exhaustive}); at $m=5,7,11$ the
numerically harvested sectors are complete.
\textup{(iv)} At $m=3$ Theorem~\ref{thm:exist} gives the classes with $j=1,2$, but the division by $6$ in Theorem~\ref{thm:lines} fails and the
spins rest on Observation~\ref{obs:l3r3minus}.
\textup{(v)} The handle-gluing operator at these classes. They are $A_-$ classes in the sense of Section~\ref{sub:degenerate}, by their kernel
vectors; the value $\mathcal H=2^mm^2$ is observed at $m=5,7,11$ and is not computed here, Theorem~\ref{thm:Hall} not reaching degenerate
classes. Corollary~\ref{cor:diagonal} obtains the corresponding moduli from the fusion rules instead.
\end{remark}

\part{\texorpdfstring{The modular data and the example $L^3R^3$}{The modular data and the example L3R3}}
\label{part:modular}

Here $\mathrm{TFT}_0$ denotes the factor of the $S$-fold TQFT that remains when the pointed factor $\mathcal P=U(1)_{\pm2}\boxtimes\Adec$ is split off (Assumption~\ref{ass:F}). The orthogonal splitting of the quadratic form on the sign group is proved for $L^mR^m$ (Proposition~\ref{prop:signsum}); the factorisation of the TQFT is part of Assumption~\ref{ass:F}. That $q$ attains the semion value $\pm N/4$ is also checked on $427$ cyclic words (Observation~\ref{obs:e1even}). The subscript $0$ does not refer to the unit. This part assembles the modular data of $\mathrm{TFT}_0$ from the Bethe results of Parts~\ref{part:results} and~\ref{part:aminus}, and works
the example $L^3R^3$. Section~\ref{sec:modular} assumes the dictionary \eqref{eq:dict} and the grafting fusion rules of \cite{EG20}
(Assumption~\ref{ass:F}; the notion of grafting is \cite[\S3.3]{EG11}); apart from Observation~\ref{obs:wcrow} and the computations of Remark~\ref{rem:Ssign}, which are numerical, and the enumerations of Remark~\ref{rem:scope}, everything else in it is proved from the stated hypotheses, Corollary~\ref{cor:diagonal} assuming the existence of the classes of types $0$ and $\infty$. Under Assumption~\ref{ass:F},
the spins at the non-degenerate classes give $T$ (Corollary~\ref{cor:T}). $T$ and the moduli of the vacuum column then determine $S$, which is
the grafting formula (Proposition~\ref{prop:S}). For every word of odd trace whose nonzero kernel vectors are non-degenerate, under the non-degeneracy condition of Theorem~\ref{thm:Hall}, the unit is $\sigma V^*$,
which is input~(b), here derived from Assumption~\ref{ass:F}. Then $S$ is the grafting of $(A_-,q_-)$ and $(A_+,q_+)$ in the frame $s_o^2=\lambda/N_-N_+$, in which the unit is the larger of the two identity rows (Theorem~\ref{thm:main}). For
$L^mR^m$, these data are the image under $\sigma_\ell$, $\ell\equiv2^{-1}$, of the unitary grafting of the linking forms; this is an identity of explicit matrices and
needs no assumption on the TQFT (Corollary~\ref{cor:galois}). Whenever $m^2+4$ is prime, the unitary grafting is the modular data of the double of every pseudo-unitary untwisted Haagerup--Izumi category of $\Z_m$ \cite[Thm.~A(2)]{GSY26}. Corollary~\ref{cor:diagonal} reaches $L^mR^m$ itself for every prime $m\ge5$ with $m^2+4$ prime, assuming the existence of the $A_-$ classes of types $0$ and $\infty$. Section~\ref{sec:l3r3} collects the data of the paper at $L^3R^3$, which gives the untwisted
Haagerup double at $m=3$. The spin formula of Theorem~\ref{thm:lines}(iii) needs $m\ge5$, so there the $A_-$ spins are numerical (Observation~\ref{obs:l3r3minus}).

\section{The modular data from the Bethe data and the fusion rules}
\label{sec:modular}

This section derives the modular data of $\mathrm{TFT}_0$ from the Bethe data, under one assumption: that the fusion rules are those of the
Evans--Gannon grafting of the two metric groups \cite[\S3.3]{EG11}, \cite{EG20} (Assumption~\ref{ass:F}). The Bethe side supplies two pieces. The first is the
moduli of the vacuum column of $S$, through $S_{0\alpha}^2=\mathcal H_\alpha^{-1}$ and Theorems~\ref{thm:torsion} and~\ref{thm:Hall}. The
second is the spins, through the fibering operator and Theorem~\ref{thm:rogers}. Corollary~\ref{cor:T} turns the spins into the $T$-matrix: by
the Cauchy theorem every twist has odd order, so the twists are the forms $q_\mp$. Proposition~\ref{prop:S} shows that $T$, the moduli and the
fusion rules determine $S$, which is the grafting formula, and that the forms decide which identity class is the unit.
Theorem~\ref{thm:main} combines these for every word of odd trace whose nonzero kernel vectors are all non-degenerate. The unit is $\sigma V^*$, which is
input~(b), here derived, and $S$ is the grafting in the frame $s_o^2=\lambda/N_-N_+$, in which the unit is the larger of the two identity rows. For $L^mR^m$, Corollary~\ref{cor:galois}
identifies the data of Proposition~\ref{prop:S}, as an algebraic identity, with the Galois conjugate, under $\sigma_\ell$ with $\ell\equiv2^{-1}$, of the unitary grafting, which is the modular data of the double of every pseudo-unitary untwisted Haagerup--Izumi category of $\Z_m$ whenever $m^2+4$ is prime \cite[Thm.~A(2)]{GSY26}. That these are the data of $\mathrm{TFT}_0$ requires Assumption~\ref{ass:F} and, for primes $m\ge5$ with $m^2+4$ prime, the existence of the classes of types $0$ and $\infty$ as real-analytic real solutions (Corollary~\ref{cor:diagonal}).
Corollary~\ref{cor:diagonal} extends Theorem~\ref{thm:main} to $L^mR^m$ for every prime $m\ge5$ with $m^2+4$ prime, using Part~\ref{part:aminus} on the degenerate $A_-$
classes; the classes of types $0$ and $\infty$ remain open. Observation~\ref{obs:wcrow} records that the row of $b$ is seen in the Bethe
data without Assumption~\ref{ass:F}.

\subsection{Setting}

Let $\varphi$ be a hyperbolic word in $L,R$ of odd trace $k\ge5$, with $A_\mp=\coker(\varphi\mp1)$, $N_\mp=|A_\mp|$, $\Phi$ and $q_\mp$ as in
Section~\ref{sub:spinform}. Write $\lambda$ for the stretch factor of $\varphi$, $r=\#L+\#R$ for the
length of the word, and $\mathcal K_\mp=\{v\in(\Z/N_\mp)^2:(\varphi\mp1)v\equiv0\}$ for the kernel group of Section~\ref{sec:spectrum}.

The Bethe classes are the two identity classes, of $V^*$ and of $\sigma V^*$, and, for each $v\in\mathcal K_\mp\smallsetminus0$ with
$f(\mathsf m_t)\neq0$ for every middle slope $\mathsf m_t$, the class of Theorem~\ref{thm:allwords} with coordinates
\[
x_t=s_t\ (L\text{-step}),\qquad x_t=1-s_t\ (R\text{-step}),\qquad s_t=\frac{f(\mathsf k_t)^2}{f(\mathsf m_t)^2},\qquad
f(e)=\sin\frac{\pi\langle v,e\rangle}{N_\mp};
\]
we call these the \emph{non-degenerate} classes. For them $\mathcal H_\alpha=2^{\#L}N_\mp$ (Theorem~\ref{thm:Hall}, under the hypothesis
$\det(F+D_R)\ne0$ stated there), and $\mathcal H(V^*)=2^{\#L}\lambda N_-N_+$, $\mathcal H(\sigma V^*)=2^{\#L}\lambda^{-1}N_-N_+$
(Theorem~\ref{thm:torsion}).

\begin{assumption}[The fusion rules]
\label{ass:F}
The TQFT factorises as $\mathrm{TFT}_0\boxtimes\mathcal P$, with $\mathcal P$ pointed on the sign group $(\Z/2)^{\#L}$ and $\mathrm{TFT}_0$ a
modular category whose simple objects are $o$ (the unit, denoted $0$ in \eqref{eq:dict}), $b$, $a_g$ for $g\in(A_-\smallsetminus0)/\!\pm$ and $c_l$ for
$l\in(A_+\smallsetminus0)/\!\pm$, the two identity classes corresponding to $o,b$ in some order and the kernel-vector classes of
$A_\mp$ to the $a_g$, $c_l$, the class of the kernel vector $\operatorname{adj}(\varphi-1)g$ to $a_g$ and that of
$\operatorname{adj}(\varphi+1)l$ to $c_l$. With $X=b+\sum_ga_g+\sum_lc_l$, the fusion rules are those of the grafting. The fusion rules of the Haagerup double were computed in \cite[\S4.3.1]{EP06} and \cite[\S3.2]{EG11}; the grafted form is \cite[\S5.6]{EG20}:
\[
b\otimes b=o+X,\quad b\otimes a_g=X+a_g,\quad b\otimes c_l=X-c_l,\quad a_g\otimes c_l=X,
\]
\[
a_g\otimes a_{g'}=X+a_{g+g'}+a_{g-g'},\qquad c_l\otimes c_{l'}=X-c_{l+l'}-c_{l-l'},
\]
with the conventions $a_0=o+b$ and $c_0=b-o$.\footnote{In \cite{EG20} the convention printed as $d_0=0-b$ should read $d_0=b-0$; this is the convention $c_0=b-o$ used here.} The dictionary \eqref{eq:dict} holds: $S_{0\alpha}^2=\mathcal H_\alpha^{-1}$, and the spins are read off the fibering operator, $T_{\alpha\alpha}/T_{00}=\mathcal F_\alpha/\mathcal F_0$.
\end{assumption}

The two metric groups enter with opposite signs, $+a_{g\pm g'}$ against $-c_{l\pm l'}$; this is the unequal footing of the two groups in the
grafting \cite[\S5.6]{EG20}, as in Lemma~\ref{lem:disc}.

Since every simple object of $\mathrm{TFT}_0$ is self-dual, $S^2=1$, and $S$, being unitary, is real. By the factorisation,
$S_{0\alpha}^2=2^{\#L}/\mathcal H_\alpha$ in $\mathrm{TFT}_0$, so the moduli of the vacuum column $s_\alpha:=S_{\alpha0}$ are
\begin{equation}
s_{a_g}^2=\frac1{N_-},\qquad s_{c_l}^2=\frac1{N_+},\qquad\{s_o^2,s_b^2\}=\Bigl\{\frac\lambda{N_-N_+},\frac1{\lambda N_-N_+}\Bigr\},
\label{eq:moduli}
\end{equation}
the last pair according as the unit is $\sigma V^*$ or $V^*$.

\subsection{\texorpdfstring{The $T$-matrix and the $S$-matrix}{The T-matrix and the S-matrix}}

\begin{corollary}[The $T$-matrix]
\label{cor:T}
Under Assumption~\ref{ass:F}, the spins of $\mathrm{TFT}_0$, relative to the unit, are $\theta_o=\theta_b=0$, $\theta_{a_g}=q_-(g)$ at every
non-degenerate class of $A_-$, and $\theta_{c_l}=q_+(l)$ at every non-degenerate class of $A_+$; the $T$-eigenvalues are $e(\theta)$.
\end{corollary}

\begin{proof}
The global dimension of $\mathrm{TFT}_0$ is $1/S_{00}^2=\lambda^{\mp1}N_-N_+$, of norm $(N_-N_+)^2$. By the Cauchy theorem for spherical fusion
categories \cite[Thm~3.9]{BNRW16}, every rational prime dividing the order of $T$ divides $N_-N_+$; so every twist has odd order. The pointed
factor has spins in $\tfrac14\Z$, since a quadratic form on an elementary abelian $2$-group takes values in $\tfrac14\Z$. Hence the twist of
a class is determined by the spins of its vacua modulo $\tfrac14$, which an element of odd order determines uniquely. At the identity classes
the Rogers sum is $c_\alpha=\#R$ exactly (Theorem~\ref{thm:ceff}). By \eqref{eq:Fphase} this fixes their spins modulo $\tfrac18$, which also determines an element of odd order, so $\theta_o=\theta_b$. At a non-degenerate class, Theorem~\ref{thm:rogers} gives the spin modulo
$\tfrac14$ relative to the identity classes as $-\Phi(u)/4N_\mp$, and the unique element of odd order in that class is $-[4^{-1}]_{N_\mp}\Phi(u)/N_\mp=q_\mp(u)$,
since $4[4^{-1}]_N-1\in N\Z$.
\end{proof}

\begin{proposition}[The $S$-matrix]
\label{prop:S}
Let $\mathrm{TFT}_0$ be as in Assumption~\ref{ass:F}, with spins $\theta_o=\theta_b=0$, $\theta_{a_g}=q_-(g)$, $\theta_{c_l}=q_+(l)$ for
nondegenerate quadratic forms $q_\mp$ on $A_\mp$, and vacuum column with moduli \eqref{eq:moduli}. The overall sign of $S$ is not fixed by these data (Remark~\ref{rem:Ssign}); normalise $s_o>0$. Then
\[
s_b>0,\qquad s_{a_g}=\frac1{\sqrt{N_-}}\ \text{for all }g,\qquad s_{c_l}=\frac{\operatorname{sgn}(s_b-s_o)}{\sqrt{N_+}}\ \text{for all }l,
\]
and the unit is the identity class with $s_o^2=\lambda/N_-N_+$ if and only if $\gamma_-=\gamma_+$, and the one with $s_o^2=1/\lambda N_-N_+$ if
and only if $\gamma_-=-\gamma_+$. In either case $S$ is the Evans--Gannon grafting \cite[\S3.3]{EG11} of $(A_-,q_-)$ and $(A_+,q_+)$, whose fusion rules \cite[\S5.6]{EG20} are those of Assumption~\ref{ass:F}, in the frame fixed by the unit
\textup{(}$s_o^2=\lambda/N_-N_+$ or $1/\lambda N_-N_+$\textup{)}.
\end{proposition}

\begin{proof}
The balancing equation in a modular category reads $S_{ij}=\bar t_i\bar t_j\sum_kN_{ij}^kt_ks_k$, with $t_i=e(\theta_i)$ the $T$-eigenvalues relative to the unit (the objects are
self-dual). Put
$P=s_b+\sum_gt_{a_g}s_{a_g}+\sum_lt_{c_l}s_{c_l}$, $p_x=t_{a_x}s_{a_x}$ on $A_-$ with $p_0=s_o+s_b$, and $p_y=t_{c_y}s_{c_y}$ on
$A_+$ with $p_0=s_b-s_o$. The fusion rules give
\[
S_{bb}=s_o+P,\quad S_{ba}=\bar t_aP+s_a,\quad S_{bc}=\bar t_cP-s_c,\quad S_{ac}=\bar t_a\bar t_cP,
\]
\[
S_{a_ga_{g'}}=\bar t_g\bar t_{g'}\bigl(P+p_{g+g'}+p_{g-g'}\bigr),\qquad S_{c_lc_{l'}}=\bar t_l\bar t_{l'}\bigl(P-p_{l+l'}-p_{l-l'}\bigr).
\]
(i) $P=0$. Otherwise reality of $S_{ac}$ for all $c$ forces $e(q_+(l)-q_+(l'))=\pm1$, so $q_+$ is constant, say $\equiv\varrho$, on
$A_+\smallsetminus0$, $N_+$ being odd, and $q_+(2l)=4q_+(l)$ gives $3\varrho\equiv0$. Then $|G(q_+)|^2=|1+(N_+-1)e(\varrho)|^2$ is $N_+^2$
if $\varrho=0$ and $N_+^2-3N_++3$ otherwise, whereas a nondegenerate form on a group of odd order has $|G(q_+)|^2=N_+$; neither is
possible for $N_+=k+2\ge7$.

(ii) With $P=0$, orthogonality of the rows $o$ and $b$ reads $2s_os_b+\tfrac12(N_+^{-1}-N_-^{-1})=0$, so $s_os_b=1/N_-N_+>0$ and $s_b>0$. By
\eqref{eq:moduli}, in both cases $s_o+s_b=N_-^{-1/2}$ and $|s_o-s_b|=N_+^{-1/2}$.

(iii) Let $B_-$ be the bilinear form of $q_-$. As $t_{a_{g\pm g'}}=t_gt_{g'}e(\pm B_-(g,g'))$,
\[
S_{a_ga_{g'}}=e(B_-(g,g'))\,s_{a_{g+g'}}+e(-B_-(g,g'))\,s_{a_{g-g'}}\qquad(s_{a_0}:=s_o+s_b),
\]
and reality forces $s_{a_{g+g'}}=s_{a_{g-g'}}$ whenever $B_-(g,g')\neq0$, since $B_-\notin\tfrac12\Z$ for odd $N_-$. Taking $g=g'$: $s_{a_{2g}}=s_o+s_b=N_-^{-1/2}$ whenever $q_-(g)\ne0$, which covers every non-isotropic element.
For isotropic $x\ne0$ choose a non-isotropic $y$ and $g=(x+y)/2$, $g'=(x-y)/2$; then $B_-(g,g')=-q_-(y)/2\ne0$, so $s_{a_x}=s_{a_y}$. Thus
$s_{a_g}=N_-^{-1/2}$ for all $g$. The same argument on $A_+$, with the sign of the $c$-block, gives $s_{c_l}=s_b-s_o$ for all $l$.

(iv) Now $\sum_gt_{a_g}=(G(q_-)-1)/2$, so $P=0$ reads
$s_b+\tfrac12(\gamma_--N_-^{-1/2})+\tfrac12(s_b-s_o)(\sqrt{N_+}\gamma_+-1)=0$. If $s_o^2=\lambda/N_-N_+$ then $s_b=\tfrac12(N_-^{-1/2}-N_+^{-1/2})$
and $s_b-s_o=-N_+^{-1/2}$, and the equation is $\gamma_-=\gamma_+$; if $s_o^2=1/\lambda N_-N_+$ it is $\gamma_-=-\gamma_+$. All entries of $S$ are now
determined by the balancing equation, and they are those of the grafting.
\end{proof}

\begin{remark}[The overall sign of $S$]
\label{rem:Ssign}
The data of Assumption~\ref{ass:F} determine $S$ only up to $S\mapsto-S$. Both signs give $S^2=1$ and the same Verlinde coefficients. With $T_{oo}=1$, the relation $(ST)^3=e(c/8)S^2$ changes $c$ by $4$ modulo $8$ under this sign. So the sign is fixed by the central charge, and the representative with $s_o>0$ need not be the right one.

This can be tested on the full set of vacua, without Assumption~\ref{ass:F}. Put $G=\sum_\alpha(\mathcal H_{\mathrm{unit}}/\mathcal H_\alpha)\,e(\theta_\alpha)$, where $\theta_\alpha$ is the spin relative to the unit. Then $|G|^2=\sum_\alpha\mathcal H_{\mathrm{unit}}/\mathcal H_\alpha$, and $G/|G|=\operatorname{sgn}(S_{oo})\,e(c/8)$.

We computed $G$ numerically for the six words of Remark~\ref{rem:scope}, from the $\mathcal H_\alpha$ and spins of all vacua, with $c=\#L-6\#R$. In each case $\sum_\alpha\mathcal H_\alpha^{-1}=1$ and $G/|G|=\pm e(c/8)$. The sign is $-$ for $L^2RLR^2$, $L^2R^2LR$ and the two words of length $10$, and $+$ for $L^2RLRLR^2$ and $L^2R^2LRLR$. This confirms the congruence of Theorem~\ref{thm:fareygeneral} modulo $4$ independently of the fusion rules, given input~(b). Modulo $8$ it shows that $S_{oo}<0$ for the first four words. For $L^mR^m$ with $m$ odd the comparison can be made for $\mathrm{TFT}_0$ alone. The pointed factor $\mathcal P$ has spins, measured from the unit, with Gauss sum $\sqrt{2^m}\,e(-m/8)$ (Propositions~\ref{prop:signsum} and~\ref{prop:spinq}(iii)), so its central charge is $c_{\mathcal P}\equiv-m\pmod8$. Hence $c(\mathrm{TFT}_0)\equiv c-c_{\mathcal P}\equiv-5m+m\equiv4\pmod8$. The grafting with $s_o>0$ has central charge $0$ modulo $8$ (computed for $m=3,5,7,9$), so the offset is again $4$, and $S_{oo}<0$.
\end{remark}

\begin{theorem}[The modular data of $\mathrm{TFT}_0$]
\label{thm:main}
Let $\varphi$ be a hyperbolic word of odd trace $k\ge5$. Assume Assumption~\ref{ass:F}, the hypothesis of Theorem~\ref{thm:Hall}, and that every nonzero kernel vector of $\varphi$ is non-degenerate.
Then the modular data of $\mathrm{TFT}_0$ are determined: $T$ is given by the forms $q_\mp$ of \eqref{eq:qform}, the unit is $\sigma V^*$, and
$S$ is the Evans--Gannon grafting of $(A_-,q_-)$ and $(A_+,q_+)$ in the frame $s_o^2=\lambda/N_-N_+$, up to the overall sign of Remark~\ref{rem:Ssign}. Since $\gamma_-=\gamma_+$, no unitary
grafting has these forms \textup{(}Proposition~\ref{prop:S}\textup{)}; for $L^mR^m$, Corollary~\ref{cor:galois} identifies the unitary data of
which these are a Galois conjugate.
\end{theorem}

\begin{proof}
Corollary~\ref{cor:T} gives $T$; Lemma~\ref{lem:disc} gives $\gamma_-=\gamma_+$; Proposition~\ref{prop:S} then gives $S$ and places the unit at
the identity class with $\mathcal H=2^{\#L}\lambda^{-1}N_-N_+$, which is $\sigma V^*$ (Theorem~\ref{thm:torsion}).
\end{proof}

\begin{corollary}[The unitary partner of the diagonal words]
\label{cor:galois}
Let $\varphi=L^mR^m$ with $m$ odd, and let $\sigma_\ell$ be the automorphism $\zeta\mapsto\zeta^\ell$ of $\Q(\zeta_{N_-N_+})$ with
$\ell\equiv2^{-1}\pmod{N_-N_+}$. The modular data of Proposition~\ref{prop:S} with the spins $q_\mp$ and $s_o>0$ --- the form that Theorem~\ref{thm:main}
gives to the modular data of $\mathrm{TFT}_0$, up to the sign of $S$ --- are the image under $\sigma_\ell$ of the unitary Evans--Gannon grafting of $(A_-,q^{\mathrm{lk}})$
and $(A_+,q^{\mathrm{lk}})$, the forms of Proposition~\ref{prop:linking}.
\end{corollary}

\noindent That such a Galois conjugate is non-unitary modular data with the same fusion rules, whose positive column is that of $b$, is observed in \cite[\S3.3]{EG11}.

\begin{proof}
By Proposition~\ref{prop:linking}, the forms $q^{\mathrm{lk}}$ have $\gamma_-=+1$ and $\gamma_+=-1$, so by Proposition~\ref{prop:S} their grafting has
its unit at $s_o^2=1/\lambda N_-N_+$, which is the unitary frame. Its entries, listed in the proof of Proposition~\ref{prop:S}, lie in
$\Q(\zeta_{N_-N_+})$, since $N_-=m^2$ and $N_+=\mu\equiv1\pmod4$. The automorphism $\sigma_\ell$ sends $e(q^{\mathrm{lk}})$ to $e(2^{-1}q^{\mathrm{lk}})=e(q_\mp)$
(Section~\ref{sec:metric}) and $\cos2\pi B^{\mathrm{lk}}$ to $\cos2\pi B_\mp$, where $B$ denotes the bilinear form of each quadratic form. It
fixes $m$, and it sends $\sqrt\mu$ to $(\ell/\mu)\sqrt\mu=(2/\mu)\sqrt\mu=-\sqrt\mu$, as $\mu\equiv5\pmod8$. Hence it exchanges the values
$\tfrac12(m^{-1}\mp\mu^{-1/2})$ of $s_o$ and $s_b$, which gives $s_o^2=\lambda/N_-N_+$, and it changes the sign of $s_{c_l}=s_b-s_o$. These are
the data of Proposition~\ref{prop:S} with the spins $q_\mp$.
\end{proof}

At $m=3$ the unitary data are those of \cite[(2.5)]{EG11}, \cite[Thm.~3(ii)]{EG17}, with forms $x_1x_2/3$ on $A_-\cong\Z_3^2$ and $6l^2/13$ (isometric to the $5l^2/13$ of Proposition~\ref{prop:linking}) on
$A_+\cong\Z_{13}$, so that $\gamma_-=1=-\gamma_+$. Now let $\mu=m^2+4$ be prime, as for $m=3,5,7,13,15,17$. By \cite[Thm.~A(2), Cor.~5.6]{GSY26} the form on the block of order $\mu$ is then determined by its Jacobi symbol, which is $-1$ for $q^{\mathrm{lk}}$ (Proposition~\ref{prop:linking}). So whenever $\mu$ is prime the unitary grafting is the modular data $\D^0\Hg_m$. The first odd $m$ not covered are $m=9,11$, where $\mu=85,125$; there the identification is not proved here.

\begin{corollary}[The diagonal words]
\label{cor:diagonal}
Let $\varphi=L^mR^m$ with $m\ge5$ prime and $\mu=m^2+4$ prime, for instance $m=5,7,13,17,37,47$, and assume Assumption~\ref{ass:F}. Suppose that the $m-1$ classes of $A_-$ of types $0$ and
$\infty$ exist as solutions real-analytic in $\epsilon$ with real coordinates, with the leading data of Theorem~\ref{thm:aminus}\textup{(i)} \textup{(}Remark~\ref{rem:aminusgaps}\textup{(i)}\textup{)}. Then
the unit is $\sigma V^*$, and $S$ is, up to the sign of Remark~\ref{rem:Ssign}, the grafting of $(A_-,q_-)$ and $(A_+,q_+)$ in the frame $s_o^2=\lambda/N_-N_+$: by
Corollary~\ref{cor:galois}, the image under $\sigma_\ell$, $\ell\equiv2^{-1}$, of the modular data of the double of every pseudo-unitary untwisted Haagerup--Izumi category of $\Z_m$ \cite[Thm.~A(2)]{GSY26}.
\end{corollary}

\begin{proof}
First, every nonzero $A_+$ kernel vector is non-degenerate. Since $\mu$ is prime, $\mathcal K_+$ is a line over $\F_\mu$, so $\langle v,e\rangle\equiv0\pmod\mu$ with $v\in\mathcal K_+\smallsetminus0$ forces $e\in\mathcal K_+$, i.e.\ $(\varphi+1)e\equiv0\pmod\mu$.
The middle slopes are $e=(i,1)$ and $e=(1+im,i)$, $1\le i\le m$ (proof of Theorem~\ref{thm:lines}). The second coordinate of $(\varphi+1)e$ is $mi+2$, resp.\ $m-2i$ modulo $\mu$. It vanishes only at $i\equiv-2m^{-1}$, resp.\ $i\equiv m/2$, and both residues equal $\tfrac12(m^2+m+4)$, which lies strictly between $m$ and $\mu$. So no middle slope is degenerate, and Corollary~\ref{cor:T} gives $\theta_{c_l}=q_+(l)$. On $A_-$,
Theorem~\ref{thm:aminus}\textup{(i)} gives the $\tfrac1m\Z/\Z$-component of the spin as $q_-$, and under Assumption~\ref{ass:F} that component is
the whole twist: by the Cauchy argument of Corollary~\ref{cor:T} every twist has order prime to $6$ ($3\nmid m^2(m^2+4)$), while
$24\theta\equiv c_\alpha\in\tfrac1m\Z$. The moduli $s_{a_g}^2=1/N_-$ of \eqref{eq:moduli}, which Theorem~\ref{thm:Hall} does not give at these
degenerate classes, are forced by Assumption~\ref{ass:F}: the dimensions $d_\alpha=s_\alpha/s_o$ form a character of the fusion ring, and
$a_g\otimes c_l=X=b\otimes a_g-a_g=b\otimes c_l+c_l$ gives $d_{a_g}=d_b+1$ and $d_{c_l}=d_b-1$ for all $g,l$, dividing by $d_{c_l}$ and $d_{a_g}$, which are nonzero since $S_{0\alpha}\ne0$, and $d_X=d_{a_g}d_{c_l}=d_b^2-1$.
Since $X$ has $\tfrac12(N_-{-}1)$ summands $a_g$ and $\tfrac12(N_+{-}1)$ summands $c_l$, with $N_\pm=k\pm2$, this reads $d_b^2-kd_b+1=0$, so $d_b=\lambda^{\pm1}$, $\sum_\alpha d_\alpha^2=\lambda^{\pm1}N_-N_+$, and $s_{a_g}^2=d_{a_g}^2/\sum_\alpha d_\alpha^2=1/N_-$. The same argument gives $s_{c_l}^2=1/N_+$, so the hypothesis $\det(F+D_R)\ne0$ of Theorem~\ref{thm:Hall} is not used. The
proof of Theorem~\ref{thm:main} then applies.
\end{proof}

\begin{observation}[The Wilson line of charge $c$ is the row of the second identity object]
\label{obs:wcrow}
Let $\tr\varphi$ be odd. At the base vacuum $V_{\varepsilon_0}$ \textup{(}Proposition~\ref{prop:spinq}\textup{)} of every class of
Theorem~\ref{thm:allwords},
\[
W_c=+1\ \text{on the }A_-\text{ sector},\qquad W_c=-1\ \text{on the }A_+\text{ sector},
\]
and at the identity classes $W_c(V^*)=\lambda$ \textup{(}Proposition~\ref{obs:stretch}\textup{)}, and $W_c=\lambda^{-1}$ at the Galois conjugate of the lift $V^*$ (with $z_L=h'(\mathsf m)/h'(\mathsf k)$ as in Theorem~\ref{prop:propagated}). That vacuum lies in the class of $\sigma V^*$ but is not $\sigma V^*$ itself.
\end{observation}

The sign of $W_c$ depends on the lift within the class, since the vacua of a class differ by signs of the $z_L$. At the all-plus vacuum $\sigma V^*$ of the conjugate class the value is $W_c=(-1)^{\#R}\lambda^{-1}$ rather than $\lambda^{-1}$. Indeed $\sigma V^*$ differs from the Galois conjugate of the lift in the sign of $z_L$ at every $L$-site (Section~\ref{sec:bethe}), and $W_c=\prod_az_a^{c_a}$ with $\sum_Lc_L=\sum_\rho b_\rho=\#R$ (Theorem~\ref{thm:Kformula}); this does not affect $W_c^2$, which is all that enters $\mathcal H$ (Theorem~\ref{thm:torsion}).

This was checked in closed form at every class of Theorem~\ref{thm:allwords} on $25$ words of odd trace and length at most $10$, in
both sectors, with no exception. That $W_c^2=1$ at these classes is Theorem~\ref{thm:Hall}; the sign is not proved. In the dictionary of
\cite[(2.23)--(2.25)]{GKL23}, a Wilson line acts on the vacuum $\alpha$ by $S_{\beta\alpha}/S_{0\alpha}$, and the line of charge $c$ is the
object $\beta$ at the other identity class \cite[(2.22)]{GKL23}, \cite[(3.12), (3.17), footnote~8]{JS26}, for $LR^n$ and $L^2R^n$. Up to the sign $(-1)^{\#R}$ of the lift, the values above are
therefore the ratios $S_{b\alpha}/S_{o\alpha}$ with the unit at $\sigma V^*$:
\[
\frac{s_b}{s_o}=\lambda^{-1},\qquad\frac{S_{bb}}{S_{ob}}=\lambda,\qquad\frac{S_{ba_g}}{S_{oa_g}}=+1,\qquad\frac{S_{bc_l}}{S_{oc_l}}=-1,
\]
which is the row of $b$ in the grafted $S$ of Proposition~\ref{prop:S}, and the second row of \cite[(2.25)]{GK22}. So the opposite signs with
which the two metric groups enter, $+a_g$ against $-c_l$ in $b\otimes a_g$ and $b\otimes c_l$, are observed in the Bethe data, with the dictionary of \cite{GKL23} and input~(b) but without
Assumption~\ref{ass:F}.

\begin{remark}[Scope]
\label{rem:scope}
The hypothesis of Theorem~\ref{thm:main} that every nonzero kernel vector is non-degenerate holds, among the $102$ cyclic words of odd trace at least $5$
and length at most $11$, for $6$ words, $L^2RLR^2$ among them. For $L^mR^m$ with $m$ prime the $A_-$ classes are degenerate, and
Corollary~\ref{cor:diagonal} replaces Theorem~\ref{thm:rogers} there by Theorem~\ref{thm:aminus}; its other
hypothesis, that the $A_+$ kernel vectors are non-degenerate, holds whenever $m^2+4$ is prime (proof of Corollary~\ref{cor:diagonal}). Among the primes $5\le m\le59$ these are $m=5,7,13,17,37,47$; for the others some $A_+$ kernel vector is degenerate (a computation), and the corollary does not apply. At $m=3$ Theorem~\ref{thm:lines} is not
available, and Observation~\ref{obs:l3r3minus} takes its place. For other words the spins of the degenerate classes, which have a coordinate at
infinity, are not covered by Theorem~\ref{thm:rogers}.
\end{remark}

\section{\texorpdfstring{The worked example: $L^3R^3$}{The worked example: L3R3}}
\label{sec:l3r3}

This section collects the data of the paper at one word, $\varphi=L^3R^3$. It is the shortest word of odd trace with non-cyclic $A_-$, and
by Proposition~\ref{obs:untwisted} it gives the metric groups of the untwisted Haagerup double, the case the paper is aimed at. Here $r=6$, $\#L=\#R=3$,
$\tr\varphi=11$, $A_-\cong\Z_3\times\Z_3$, $A_+\cong\Z_{13}$, $d=\tfrac12(3+\sqrt{13})$ and $\lambda=d^2=\tfrac12(11+3\sqrt{13})$. The section
writes out the following.
\begin{itemize}
\item The level matrix, the charges and the matrix $\mathbf N$ (Theorem~\ref{thm:Kformula}).
\item The positive vacuum in closed form, with $c_{\eff}=3$ and $W_c(V^*)=\lambda$ (Theorems~\ref{thm:nahmclosed} and~\ref{thm:ceff},
Proposition~\ref{obs:stretch}).
\item The twelve classes and their handle-gluing operators, with the sum rule and the torsion-side check (Proposition~\ref{prop:Hspectrum},
Theorems~\ref{thm:torsion} and~\ref{thm:complete}, Proposition~\ref{prop:exclusion}).
\item The half-index central charge $c\equiv9\pmod{24}$ (Theorems~\ref{thm:farey} and~\ref{thm:fareygeneral}).
\item The spins, with the splitting of the sign group into the semion and the three-fermion theory (Proposition~\ref{prop:signsum},
Corollary~\ref{cor:arf}).
\end{itemize}
The spin formula of Theorem~\ref{thm:lines}(iii) needs $m\ge5$, so at $m=3$ the $A_-$ spins rest on numerics. Observation~\ref{obs:l3r3minus} finds four of
them, with spins given by the hyperbolic form $x_1x_2/3$. With these, the argument of Theorem~\ref{thm:main} gives the modular data: under
Assumption~\ref{ass:F}, up to the sign of $S$, the Galois conjugate of the Haagerup double by $\sigma_\ell$, $\ell\equiv2^{-1}$ (Corollary~\ref{cor:galois}).

By Proposition~\ref{obs:untwisted} this is the untwisted
member $\D^{\omega=0}\Hg_3$, that is the modular data of the Drinfeld double of the even part of the Haagerup subfactor; the
two metric groups $\Z_3\times\Z_3$ and $\Z_{13}$ are exactly the pair underlying the $\Z_2$-equivariantised Tambara--Yamagami
categories from which Evans--Gannon graft that double \cite[\S5.6]{EG20}, and the $12$ classes counted by Theorem~\ref{thm:count},
$\#(A_-/\!\pm)+\#(A_+/\!\pm)=5+7$, are its rank. Where $d$ appears below as $\tfrac12(3+\sqrt{13})$ it is
simultaneously the quantum dimension of $\rho$ in the Haagerup fusion ring $\rho^2=1+\sum_{g\in\Z_3}g\rho$
and the square root of the stretch factor of the word. So the arithmetic in this section is the double's
arithmetic, read off a word in $L$ and $R$.

\paragraph{The level matrix.} Theorem~\ref{thm:Kformula} gives $\Kraw$ directly, as the member $t=0$ of the family
$\Kraw(t)=K_0+t\,uu^{\mathsf T}$ of \eqref{eq:BKA}, $u=Q^{-1}c=(\tfrac32,0,0,1,2,3)^{\mathsf T}$. The other two criteria of Section~\ref{sub:charges}
select the same member: $u^{\mathsf T}K_0^{-1}u=3$, so $\det K(t)=64(3t+1)$ and $K(t)/2$, which is integral only for even integers $t$, is
unimodular only at $t=0$; and $\mathbf N(t)$ has $(L_2,L_2)$ entry $-2(5t+3)/(3t+1)$, for integer $t$ an integer only at $t\in\{0,1,-1,-3\}$, of which the
other entries exclude $t=1,-1,-3$. The matrix is
\[
\Kraw=\begin{pmatrix}
\tfrac52&-\tfrac12&-\tfrac12&1&2&3\\
-\tfrac12&1&-\tfrac12&0&0&0\\
-\tfrac12&-\tfrac12&1&-1&-1&-1\\
1&0&-1&2&2&2\\
2&0&-1&2&4&4\\
3&0&-1&2&4&6
\end{pmatrix},
\qquad
\mathbf N=\begin{pmatrix}
-2&-4&-4&-2&0&2\\
-4&-6&-8&-4&0&2\\
-4&-8&-10&-6&0&2\\
-2&-4&-6&-3&1&1\\
0&0&0&1&0&1\\
2&2&2&1&1&0
\end{pmatrix},
\]
with $\mathbf N$ symmetric, integral and of signature $(3,3)$, hence indefinite. The $R$-block of $\Kraw$ is the fan
$2\min(a,b)$ of Theorem~\ref{thm:Kformula}, the half-integers sit in the $L$-block only, so the minimal rescaling is $Q=(2,2,2,1,1,1)$ --- charge $2$ at
each $L$-tetrahedron --- and
\[
K=Q\Kraw Q=\begin{pmatrix}
10&-2&-2&2&4&6\\ -2&4&-2&0&0&0\\ -2&-2&4&-2&-2&-2\\
2&0&-2&2&2&2\\ 4&0&-2&2&4&4\\ 6&0&-2&2&4&6
\end{pmatrix},
\]
all entries even, $\det K=64=2^r$, and $K/2$ unimodular. The $L$-block is $K_{LL}=2E_{LL}$ with $E_{LL}$ the cyclic tridiagonal $(-1,2,-1)$ plus
$m=3$ at the $(1,1)$ entry, as in Proposition~\ref{prop:signsum}. The charge vector of Theorem~\ref{thm:Kformula} is $c=(3,0,0,1,2,3)$.

\paragraph{The closed form.} By Theorem~\ref{thm:nahmclosed}, with $\mathcal Q_j=1+jd$,
\[
\delta_{L_j}=\frac{\mathcal Q_{j-1}\mathcal Q_{j+1}}{\mathcal Q_j^2},\qquad \delta_{R_j}=\frac{d^2}{\mathcal Q_j^2},\qquad j=1,2,3,
\]
all in $\Q(\sqrt{13})$ and all in $(0,1)$, with $\delta_{L_3}=1-1/d^2=3/d$ the largest. The Rogers sum gives $c_{\eff}=3=\#L$
(Theorem~\ref{thm:ceff}), and $W_c(V^*)=\lambda$ (Proposition~\ref{obs:stretch}), the fan product telescoping as
$\mathcal Q_3(\mathcal Q_3/\mathcal Q_4)^3\cdot((d+1)/d)^3=d^2$.

\paragraph{The spectrum.} There are $m^2+3=12$ classes, each carrying $2^m=8$ vacua, so $96$ raw vacua. The two identity classes and the six
$A_+$ classes are proved (Proposition~\ref{prop:Hspectrum}(i)--(ii) and Theorem~\ref{thm:complete}, every run having length $3$; the $A_+$ values under the non-degeneracy condition, which holds there by computation); the four $A_-$
classes, the hypothesis of Proposition~\ref{prop:Hspectrum}(iii), are found numerically (Observation~\ref{obs:l3r3minus}). Then
\[
\mathcal H\in\bigl\{\,2^3\cdot9=72\ (\times4),\qquad 2^3\cdot13=104\ (\times6),\qquad 2^3\lambda^{\pm1}\cdot117\ (\times1\text{ each})\,\bigr\},
\]
the multiplicities being $\tfrac{m^2-1}2=4$ and $\tfrac{m^2+3}2=6$. The sum rule
$\sum_\alpha2^3/\mathcal H_\alpha=\tfrac49+\tfrac6{13}+\tfrac{\lambda+\lambda^{-1}}{117}=\tfrac{52+54+11}{117}=1$ closes exactly. On the torsion side, the
characteristic polynomial of multiplication by $\det(I-d\varphi|_T)$ on $\C[x,y,z]/I$ splits over $\Q$ on $\tr\partial=+2$ as
\[
(t+117)^2(t+9)^4(t-13)^6,
\]
against $(k-2,k^2-4,k+2)=(9,117,13)$ with multiplicities $(\tfrac{k-3}2,2,\tfrac{k+1}2)=(4,2,6)$ --- an exact identity over $\Q$, computed with a Gr\"obner basis, and the
$|A_-||A_+|=9\cdot13=117$ of Theorem~\ref{thm:torsion}. On the Markov locus $\tr\partial=-2$ the same polynomial is
$(t-7)^2(t-3)^2(t^4+106t^3+3604t^2+43190t+90699)^2$ with the quartic irreducible, which is Proposition~\ref{prop:exclusion}.

\paragraph{The central charge.} The single $R$-run end is $R_3$, $\mathbf1^{\mathsf T}K_{\mathcal E\mathcal E}\mathbf1=6=2\#R$, so $\theta_{V^*}\equiv-\tfrac34$ and
\[
c\equiv\#L-6\#R=3-18\equiv9\pmod{24},
\]
which is Theorem~\ref{thm:farey} at $m=3$ and Theorem~\ref{thm:fareygeneral} at this word.

\paragraph{The spins.} By Proposition~\ref{prop:signsum} the $8$ vacua of a class carry the discriminant form of $K_{LL}$, with
$q(e_1)=(m+2)/4\equiv\tfrac14$ and $q(\mathbf1)=m/4=\tfrac34$, the semion; the splitting is $U(1)_{\pm2}\boxtimes \Adec$ with $|\Adec|=4$ and $q|_{\Adec}$
half-integral, so under Assumption~\ref{ass:F} $\mathrm{TFT}[L^3R^3]=\mathrm{TFT}_0\boxtimes U(1)_{\pm2}\boxtimes \Adec$ with $\Adec$ the three-fermion theory --- the pointed category on
$\Z_2\times\Z_2$ whose three non-trivial objects all have spin $\tfrac12$, equivalently $SO(8)_1$, Arf-odd by Corollary~\ref{cor:arf}. Concretely $\Adec$
is generated by the flips of $z_2$ and $z_3$, the semion of $U(1)_{\pm2}$ is the flip of all three signs, and the flip of $z_1$ alone is semion times
fermion; it is the Arf-odd alternative to the toric code, which has the same order but spins $0,0,\tfrac12$. The factor
$\mathrm{TFT}_0$ that remains, carrying the two metric groups, is, under Assumption~\ref{ass:F} and with the numerically found $A_-$ spins of Observation~\ref{obs:l3r3minus}, up to the sign of $S$ a Galois conjugate of the Haagerup double, by $\sigma_\ell$ with $\ell\equiv2^{-1}$, which sends $\lambda\mapsto\lambda^{-1}$ (Corollary~\ref{cor:galois}): this splitting is
how the double appears inside the $S$-fold TQFT, with the pointed factors
$U(1)_{\pm2}\boxtimes \Adec$ accounting for the $2^{\#L}=8$ vacua of a class and $\mathrm{TFT}_0$ for the $12$
classes. The Gauss sum of $q$ on the sign group is
$\operatorname{Re}((1+i)^3)(1+i^3)=\sqrt8\,e^{2\pi i\cdot3/8}$. This is the base-centred form. By Proposition~\ref{prop:spinq}(iii) the base vacuum is the all-flip of $\sigma V^*$, and the spins measured from the unit have Gauss sum $\sqrt8\,e^{-2\pi i\cdot3/8}$. So the pointed factor has $c_{\mathcal P}\equiv-3\equiv5\pmod8$, and $c(\mathrm{TFT}_0)\equiv9-5=4\pmod8$, against $0$ for the grafting with $s_o>0$: here $S_{oo}<0$ (Remark~\ref{rem:Ssign}).

\begin{observation}[The $A_-$ spins at $L^3R^3$]
\label{obs:l3r3minus}
At $\varphi=L^3R^3$ the deformed Bethe equations at $\epsilon=10^{-2}$ and $10^{-3}$ have four $A_-$ classes, each continued to $\epsilon=10^{-5}$ with
$\mathcal H_\alpha/2^3\to9$; all four are real. In the notation $z_b=C_b\epsilon^{p_b}(1+o(1))$ of Observation~\ref{obs:strata} their exponent vectors are
$(0,0,2\,|\,0,0,0)$, $(2,0,0\,|\,2,-1,-1)$, $(0,2,0\,|\,-1,2,-1)$ and $(-2,-2,-2\,|\,-1,-1,2)$. Along each continuation $\arg\mathcal F/2\pi$ is
constant modulo $\tfrac14$ to $10^{-8}$, with values $\tfrac1{24},\tfrac18,\tfrac18,\tfrac5{24}$, against $\tfrac18$ at both identity classes.
Relative to the identity classes these are $\tfrac1{6},0,0,\tfrac1{12}\pmod{\tfrac14}$, which lift uniquely to $\tfrac13\Z/\Z$ as
$\tfrac23,0,0,\tfrac13$: the values of the hyperbolic form $x_1x_2/3$ on $(A_-\setminus0)/\!\pm$. The anisotropic form, which takes only the
values $\tfrac13,\tfrac23$ there, is excluded. The same reading on the six $A_+$ classes gives the values $l^2/13$, $l=1,\dots,6$, the form $q_+$ up to isometry, as Theorem~\ref{thm:rogers}
requires. Since the $A_+$ kernel vectors of $L^3R^3$ are non-degenerate, the proof of Theorem~\ref{thm:main} then applies at $m=3$, with these
values on $A_-$ in place of Theorem~\ref{thm:rogers}: the unit is $\sigma V^*$ and $S$ is the grafting of $(\Z_3^2,x_1x_2/3)$ and $(\Z_{13},q_+)$.
These values are numerical: Theorem~\ref{thm:rogers}, which gives the spins modulo $\tfrac14$ at the non-degenerate classes, does not reach the
degenerate ones, and at $m=3$ Theorem~\ref{thm:lines} is not available.
\end{observation}

\subsection*{Use of Artificial Intelligence}
I have been studying the approach of \cite{GKLSY21,GK22,  GKL23}  since the appearances of their manuscripts in 2022-23 and more recently \cite{JL25,JS26} in 2025-26.
I had discussions with Pinhas Grossman during a visit to UNSW in 2023, but could make no progress in solving the Bethe equations.
The breakthrough for me came through a talk in  September 2025, at the Newton Institute on the occasion of Peter Goddard’s 80th birthday, 
given by Werner Nahm on Nahm sums.  This talk triggered me into going back and looking at the paper by Nahm, Recknagel, Ternhoven \cite{NRT92}.
I was  aware of this paper because Andreas Recknagel was my postdoc at Swansea shortly afterwards.
I turned the Bethe equations into Nahm form and then
 quickly found in closed form  the positive solutions and with the help of Mathematica could compute the integral effective central charge
from the corresponding dilogarithms. Finding other solutions proved difficult with Mathematica. Understanding and solving the deformed equations in order to obtain all the Bethe vacua and primary labels only came with the heavy lifting of  Claude. Corresponding and prompting  since June 2026 enabled Claude as a Research Assistant to assist in solving many of my questions and help in writing  this manuscript. I take responsibility for the final manuscript.

\subsection*{Acknowledgements}

I am grateful to Pinhas Grossman, Sergei Gukov and Carlos Nunez for helpful discussions. I also appreciate the support of the Leverhulme Foundation for an  Emeritus Fellowship, and  the  Quantum Symmetries Programme Reunion at MSRI/SLMath, Australian Research Council Discovery International Award DP170103265 at UNSW, Sydney and the Dublin Institute for Advanced Studies
where some of this work was carried out.

\appendix

\section{Proofs for Section~\ref{sec:runK}}
\label{app:Kformula}

This appendix proves Theorem~\ref{thm:Kformula}, the closed form of the level matrix $K$, on which all later sections rest, and
Proposition~\ref{prop:closed}, the amphichirality of $\Kraw$ for $L^mR^m$, used in Section~\ref{sec:farey} and in Theorem~\ref{thm:nahmclosed}.

\subsection{Proof of Theorem~\ref{thm:Kformula}}

The proof has six parts, proved in this order: that the matrix of the theorem satisfies the determinant law; that it satisfies
$B\Kraw=A$; that $\operatorname{rank}B=r-1$; that $Q^{-1}c$ spans $\ker B$; that it is positive definite; and that it is the only symmetric
solution with $\det K=2^r$. The
second and fourth use the edge equations of Section~\ref{sub:gluing}; the first and fifth use the Schur complement of the matrix itself. The only external input is the Neumann--Zagier rank
theorem \cite[Proposition~2.3]{NZ85}, used in the third, which needs the layered triangulation to be geometric \cite{Gueritaud06}.
The rows of $\Kraw$ listed after Theorem~\ref{thm:Kformula} are used throughout. As there, for an $R$-run $\rho$ we write $L^+(\rho)$ for the first $L$ after $\rho$ and $L^-(\rho)$ for the last $L$ before it.

\paragraph{The life of an edge.} Each flip creates one slope, and the edges of the layered triangulation are in bijection with the flips
(Section~\ref{sub:gluing}). Fix the flip $i$, with letter $\ell_i$ and Farey basis $(c_1,c_2)$ before it. If $\ell_i=L$ the basis after the flip is
$(c_1,c_1+c_2)$ and the newborn slope is $e=2c_1+c_2$; if $\ell_i=R$ the basis is $(c_1+c_2,c_2)$ and $e=c_1+2c_2$. In either case $e$ is the
sum of the two columns of the basis after flip $i$, i.e.\ the Farey slope of flip $i+1$. After flip $i+1$ the basis contains $e$ as one
column: as the second column if $\ell_{i+1}=L$, as the first if $\ell_{i+1}=R$. Flip $i+2$ keeps the first column if $\ell_{i+2}=L$ and the second if
$\ell_{i+2}=R$. Hence:

\begin{lemma}[Life of an edge]
\label{lem:life}
The edge $e$ born at flip $i$ is the Farey slope of flip $i+1$. If $\ell_{i+2}=\ell_{i+1}$ it is the dying slope of flip $i+2$. If
$\ell_{i+2}\ne \ell_{i+1}$ it is the kept slope of flip $i+2$ and of every flip of the run of $\ell_{i+2}$'s starting there, and it is the dying slope
of the first flip after that run.
\end{lemma}

In the second case the run of $\ell_{i+2}$'s is a whole run of the word, since $\ell_{i+1}\ne \ell_{i+2}$. The multiplicities check: an edge kept at
$a$ flips appears with multiplicity $1+2+2a+1$, and $\sum_e(4+2a_e)=4r+2r=6r$, six edges per tetrahedron.

\paragraph{The edge rows.} Each edge $e$ gives one edge equation. In each tetrahedron we eliminate $Z''=i\pi-Z-Z'$, so that the equation is linear in the variables $(Z,Z')$ of the $r$ tetrahedra.
Its coefficients of $Z$ form the row of $e$ in $A$, and its coefficients of $Z'$ the row of $e$ in $B$; these are the matrices of \eqref{eq:BKA}.
Each tetrahedron containing $e$ contributes to the column of that tetrahedron. The contribution depends on the position of $e$ in the tetrahedron (newborn or dying, kept, or Farey) and on its letter, through the assignment of Section~\ref{sub:gluing}.
At an $L$-tetrahedron the dying and newborn edges carry $Z$, the kept edges $Z'$ and the Farey edges $Z''$; at an $R$-tetrahedron $Z$ and $Z'$ are exchanged.
An edge contributes its multiplicity times its parameter, and a Farey edge carries $Z''=i\pi-Z-Z'$, so it contributes $-2$ to both $A$ and $B$. The contributions are therefore
\[
\begin{array}{lcc}
 & \text{$L$-tetrahedron} & \text{$R$-tetrahedron}\\
\text{newborn or dying (multiplicity 1)} & A\mathrel{+}=1 & B\mathrel{+}=1\\
\text{kept (multiplicity 2)} & B\mathrel{+}=2 & A\mathrel{+}=2\\
\text{Farey (multiplicity 2)} & A\mathrel{-}=2,\ B\mathrel{-}=2 & A\mathrel{-}=2,\ B\mathrel{-}=2
\end{array}
\]
By Lemma~\ref{lem:life} the row of $e$ depends only on $(\ell_i,\ell_{i+1},\ell_{i+2})$ and, when $\ell_{i+1}\ne \ell_{i+2}$, on the run $\varsigma$ starting
at $i+2$; there are eight patterns. Write $R_j$ for the $j$-th $R$ of a run and $L_s$ for the $s$-th $L$ of a run; when $\ell_i=\ell_{i+1}$ the
sites $i$ and $i+1$ lie in the same run.
\begin{itemize}
\item[$LLL$] $A=(1,-2,1)$, $B=(0,-2,0)$ at $(i,i+1,i+2)$; $i+1$ is a middle $L$.
\item[$RRR$] $A=(0,-2,0)$, $B=(1,-2,1)$ at $(i,i+1,i+2)=(R_{j-1},R_j,R_{j+1})$ of one run.
\item[$RLL$] $A=(0,-2,1)$, $B=(1,-2,0)$ at $(R_b,L_1,L_2)$: $R_b$ the last $R$ of a run $\rho$ of length $b$, $L_1=L^+(\rho)$, $L_2$ the
next $L$.
\item[$LRR$] $A=(1,-2,0)$, $B=(0,-2,1)$ at $(L^-(\rho),R_1,R_2)$ for a run $\rho$ of length $b\ge2$.
\item[$LLR$] $A=1$ at $i$, $-2$ at $i+1=L^-(\rho)$, $2$ at every $R$ of the run $\rho$ starting at $i+2$, $1$ at $L^+(\rho)$; $B=-2$ at
$L^-(\rho)$.
\item[$RRL$] $A=-2$ at $i+1=R_b$; $B=1$ at $R_{b-1}$, $-2$ at $R_b$, $2$ at every $L$ of the $L$-run $\varsigma$ starting at $i+2$, $1$ at the
first $R$ after $\varsigma$. Here the run $\rho\ni R_b$ has length $b\ge2$.
\item[$LRL$] as $RRL$ with $\rho$ of length $1$: $A=1$ at $i=L^-(\rho)$, $-2$ at $R_1$; $B=-2$ at $R_1$, $2$ at every $L$ of $\varsigma$, $1$
at the first $R$ after $\varsigma$.
\item[$RLR$] as $LLR$ with the $L$-run of length $1$: $A=-2$ at $i+1=L$, $2$ at every $R$ of the run $\rho'$ starting at $i+2$, $1$ at
$L^+(\rho')$; $B=1$ at $i=R_b$, the last $R$ of the run $\rho$ before $L$, and $-2$ at $L$. Here $L=L^+(\rho)=L^-(\rho')$.
\end{itemize}

\begin{proof}[Proof of the determinant law $\det K=2^r$]
Eliminate each $R$-run's fan block $F=2M$, $M_{ij}=\min(i,j)$. Since $M^{-1}$ is tridiagonal $(2,-1)$ with last diagonal entry $1$, one has
$M^{-1}(1,\dots,b)^{\mathsf T}=e_b$ and $M^{-1}\mathbf1=e_1$; so the ramp $\mathsf r=2(1,\dots,b)$ and the row $s=-2\cdot\mathbf1$ contribute
$\mathsf r^{\mathsf T}F^{-1}\mathsf r=2b$, $s^{\mathsf T}F^{-1}s=2$ and $\mathsf r^{\mathsf T}F^{-1}s=-2$ to the $L$-Schur complement. The $+2b$ removes the $2b$ on the
diagonal of the first $L$ after the run, and the $-2$ cross term cancels the cyclic $L$--$L$ coupling across the run. What remains decomposes by
$L$-runs, $S=\bigoplus_q2M_{a_q}^{-1}$ --- after the $R$-runs are eliminated the $L$-runs decouple completely --- with $\det=2^{a_q}$ per block,
including $\#L=1$, where the single $L$ is adjacent to itself. Hence $\det K=\prod_q\det(2M_{b_q})\prod_q2^{a_q}=2^{\#R}2^{\#L}=2^r$.
\end{proof}

\begin{proof}[Proof that the rows of $B\Kraw=A$ hold]
For each of the eight patterns $LLL,\dots,RLR$ listed above, $\sum_tB[e,t]\Kraw[t,\cdot]$ is computed from the rows of $\Kraw$ listed in Section~\ref{sec:runK} and compared with $A[e,\cdot]$ column by
column. In each case only the columns named in the computation can be nonzero; every other column is zero on both sides.

\emph{$LLL$.} $-2\Kraw[i+1,\cdot]$ for a middle $L$ is $-2$ at $i+1$ and $+1$ at its two adjacent $L$'s, $i$ and $i+2$.

\emph{$RRR$.} The second difference along the run of $2\min(j,l)$ in $j$ is $-2\delta_{jl}$; the ramp $j$ at $L^+(\rho)$ is linear in $j$ and
the entry $-1$ at $L^-(\rho)$ is constant, so both difference to zero.

\emph{$RLL$.} $\Kraw[R_b,\cdot]-2\Kraw[L_1,\cdot]$: at $R_l\in\rho$, $2l-2l=0$; at $L_1$, $b-2(1+b/2)=-2$; at $L_2$, $-2\cdot(-\tfrac12)=1$;
at $L^-(\rho)$, $-1-2(-\tfrac12)=0$, the $-\tfrac12$ being the adjacency of $L^-(\rho)$ and $L_1$ across $\rho$. ($L_1$ is not an $L^-$,
since $\ell_{i+2}=L$.)

\emph{$LRR$.} $-2\Kraw[R_1,\cdot]+\Kraw[R_2,\cdot]$: at $R_1$, $-4+2=-2$; at $R_l$, $l\ge2$, $-4+4=0$; at $L^+(\rho)$, $-2+2=0$; at
$L^-(\rho)$, $2-1=1$.

\emph{$LLR$.} $-2\Kraw[L^-(\rho),\cdot]$: this $L$ is not an $L^+$, its predecessor $i$ being an $L$, so its row is $1$ on the diagonal,
$-\tfrac12$ at $i$ and at $L^+(\rho)$, its two adjacent $L$'s, and $-1$ at every $R$ of $\rho$. Times $-2$: $-2$, $1$, $1$, $2$, as required.

\emph{$RRL$.} $\Kraw[R_{b-1},\cdot]-2\Kraw[R_b,\cdot]+2\sum_{s=1}^a\Kraw[L_s,\cdot]+\Kraw[R'_1,\cdot]$, where $\varsigma=L_1\cdots L_a$,
$L_1=L^+(\rho)$ and $L_a=L^-(\rho')$ for the run $\rho'$ starting at $R'_1$. At $R_l\in\rho$: $2\min(b-1,l)-4l$ from the first two rows
and $2l$ from the ramp of $L_1$, in total $2\min(b-1,l)-2l=-2\delta_{lb}$. At $L^-(\rho)$: $-1-2(-1)+2(-\tfrac12)=0$, from $R_{b-1}$, from $R_b$, and from the adjacency of $L^-(\rho)$ and $L_1$ across $\rho$. At $L_1$: $(b-1)-2b+2(1+b/2)+2(-\tfrac12)=0$, the last term from
$L_2$ if $a\ge2$ and from $\Kraw[R'_1,L_1]=-1$ if $a=1$. At a middle $L_s$, $1<s<a$: $2(-\tfrac12+1-\tfrac12)=0$. At $L_a$, $a\ge2$:
$2(-\tfrac12+1)-1=0$, the $-1$ from $R'_1$. At $R'_l\in\rho'$: $2(-1)+2\min(1,l)=0$, the $-1$ from $L_a=L^-(\rho')$. At $L^+(\rho')$:
$2(-\tfrac12)+1=0$, the adjacency of $L_a$ and $L^+(\rho')$ against the ramp of $R'_1$. When runs coincide --- $\rho'=\rho$ for a two-run
word, $L^+(\rho')=L_1$, or $L^-(\rho)\in\varsigma$ --- the extra entries are the ones already listed, landing on the same columns, and the same
sums give the same totals.

\emph{$LRL$.} The same combination without the row $R_{b-1}$ and with $b=1$: $-2\Kraw[R_1,\cdot]+2\sum_s\Kraw[L_s,\cdot]+\Kraw[R'_1,\cdot]$.
At $R_1$: $-4+2=-2$. At $L^-(\rho)=i$: $2+2(-\tfrac12)=1$, from $\Kraw[R_1,L^-(\rho)]=-1$ and the adjacency of $L^-(\rho)$ and $L_1$. At
$L_1$: $-2+2(1+\tfrac12)-1$ or $-2+2(1+\tfrac12)+2(-\tfrac12)$, both $0$. The columns $L_s$, $L_a$, $R'_l$ and $L^+(\rho')$ are as in $RRL$.

\emph{$RLR$.} $\Kraw[R_b,\cdot]-2\Kraw[L,\cdot]$ with $L=L^+(\rho)=L^-(\rho')$: at $R_l\in\rho$, $2l-2l=0$; at $L$, $b-2(1+b/2)=-2$; at
$L^-(\rho)$, $-1-2(-\tfrac12)=0$; at $R'_l\in\rho'$, $-2(-1)=2$; at $L^+(\rho')$, $-2(-\tfrac12)=1$. For $\#L=1$, $\rho'=\rho$ and
$L^+(\rho')=L$: the entries add to $-1$ at $L$ and $2$ at each $R_l$, which is what the pattern gives when the dying site $L^+(\rho')$ is $L$
itself.
\end{proof}

\begin{proof}[Proof that $\operatorname{rank}B=r-1$ and that $Q^{-1}c$ spans $\ker B$]
The identity $A=B\Kraw$ gives
$(A\,|\,B)=B\,(\Kraw\,|\,I)$, where $(X\,|\,Y)$ denotes the $r\times2r$ matrix whose columns are those of $X$ followed by those of $Y$, with $(\Kraw\,|\,I)$ of full row rank, so $\operatorname{rank}B=\operatorname{rank}(A\,|\,B)$. The layered
triangulation of a hyperbolic once-punctured-torus bundle is geometric \cite{Gueritaud06}, so \cite[Proposition~2.3]{NZ85} applies. It states that for an
ideal triangulation of a hyperbolic manifold with $n$ tetrahedra and $h$ cusps, the matrix $R=(R'\,|\,R'')$ of the edge relations
$\prod_\nu z_\nu^{r'_{j\nu}}(1-z_\nu)^{r''_{j\nu}}=\pm1$ has rank $n-h$. Here $n=r$ and $h=1$. Their variables are $(\log z,\log(1-z))$ and
ours are $(Z,Z')=(\log z,\log(1-z)-\log z)$, so $(R'\,|\,R'')=(A-B\,|\,B)$, a column operation on $(A\,|\,B)$; hence
$\operatorname{rank}(A\,|\,B)=\operatorname{rank}R=r-1$, and $\operatorname{rank}B=r-1$.

Put $u=Q^{-1}c$: $u=b/2$ at $L^+(\rho)$ for each run $\rho$ of length $b$, $0$ at every other $L$, and $j$ at the $j$-th $R$ of every run.
Then $Bu=0$ row by row, from the list of patterns: $LLL$, $-2u_{i+1}=0$ at a middle $L$; $RRR$, $(j-1)-2j+(j+1)=0$; $RLL$, $b-2(b/2)=0$;
$LRR$, $-2+2=0$; $LLR$, $-2u_{L^-(\rho)}=0$, this $L$ not being an $L^+$; $RRL$, $(b-1)-2b+2(b/2)+1=0$, the $L$'s of $\varsigma$ other than
$L_1=L^+(\rho)$ contributing $0$; $LRL$, $-2+2\cdot\tfrac12+1=0$; $RLR$, $b-2(b/2)=0$. As $u\ne0$ and $\ker B$ is one-dimensional, $u$ spans
it.
\end{proof}

Positive definiteness follows from the Schur decomposition in the proof of the determinant law, every block being positive definite.

\emph{Uniqueness.} Write $K_0$ for the matrix $\Kraw$ of the theorem. Since $\operatorname{rank}B=r-1$ and $u:=Q^{-1}c$ spans $\ker B$, every
solution of $B\Kraw=A$ differs from $K_0$ by $uw^{\mathsf T}$ for some $w$, and symmetry forces $w=tu$; so the symmetric solutions are
$\Kraw(t)=K_0+t\,uu^{\mathsf T}$, the one-parameter family of Section~\ref{sub:charges}. As $K_0$ is invertible,
$\Kraw(t)=K_0\bigl(I+t\,K_0^{-1}uu^{\mathsf T}\bigr)$, where $K_0^{-1}uu^{\mathsf T}$ is a rank-one matrix. By the matrix determinant lemma $\det(I+ab^{\mathsf T})=1+b^{\mathsf T}a$, applied with $a=tK_0^{-1}u$ and $b=u$,
$\det\Kraw(t)=\det K_0\,\bigl(1+t\,u^{\mathsf T}K_0^{-1}u\bigr)$, with $u^{\mathsf T}K_0^{-1}u>0$ as $K_0$ is positive definite, so $\det K=2^r$ holds only at $t=0$.

\subsection{Proof of Proposition~\ref{prop:closed}}
\label{app:closed}

\begin{proof}[Proof of Proposition~\ref{prop:closed}]
$(\Kraw P)^2=I$ splits into (i) $V^2+CF=I$, (ii) $VC+CV^{\mathsf T}=0$, (iii) $FV+V^{\mathsf T}F=0$, with (i) equivalent to its transpose.

(iii): with $\mathbf1$ the all-ones vector, $Fe_0=2\mathbf1$ and $Fe_{m-1}=2v$, so $FV=2\mathbf1v^{\mathsf T}-2v\mathbf1^{\mathsf T}$ and $V^{\mathsf T}F=2v\mathbf1^{\mathsf
T}-2\mathbf1v^{\mathsf T}$, which sum to zero.

(ii): the interior rows of $C$ are $(-\tfrac12,1,-\tfrac12)$, annihilating both constants and linear ramps; the boundary rows give
$\tfrac{m+2}2-\tfrac12-\tfrac12=\tfrac m2$ and $-\tfrac12(m-1)+m-\tfrac12=\tfrac m2$. The remaining boundary entries vanish: row $0$ of $Cv$ is $\tfrac{m+2}2-1-\tfrac m2=0$, and row $m-1$ of $C\mathbf1$ is
$-\tfrac12+1-\tfrac12=0$. Hence $C\mathbf1=\tfrac m2e_0$ and $Cv=\tfrac m2e_{m-1}$, and
$VC+CV^{\mathsf T}=\tfrac m2(e_0e_{m-1}^{\mathsf T}-e_{m-1}e_0^{\mathsf T}+e_{m-1}e_0^{\mathsf T}-e_0e_{m-1}^{\mathsf T})=0$.

(i), interior rows: for $1\le i\le m-2$, $(CF)_{ij}=2\min(i+1,j+1)-\min(i,j+1)-\min(i+2,j+1)$, the second difference of $k\mapsto\min(k,j+1)$, which
vanishes except at the bend $i=j$ where it equals $1$; and $V^2$ has no support on those rows.

(i), boundary rows: expanding, $V^2=e_0v^{\mathsf T}-m\,e_0\mathbf1^{\mathsf T}-e_{m-1}v^{\mathsf T}+e_{m-1}\mathbf1^{\mathsf T}$, so
$(V^2)_{0j}=j+1-m$ and $(V^2)_{m-1,j}=-j$. Row $0$: $(CF)_{0j}=\tfrac{m+2}2F_{0j}-\tfrac12F_{1j}-\tfrac12F_{m-1,j}$ with $F_{0j}=2$, $F_{1j}=2$ for
$j=0$ and $4$ for $j\ge1$, $F_{m-1,j}=2(j+1)$; for $j=0$ this is $(m+2)-1-1=m$, for $j\ge1$ it is $(m+2)-2-(j+1)=m-j-1$. Row $m-1$:
$(CF)_{m-1,j}=-\tfrac12F_{m-2,j}+F_{m-1,j}-\tfrac12F_{0j}$ gives $j$ for $j\le m-2$ and $m$ for $j=m-1$. Adding $(V^2)_{0j}$ and $(V^2)_{m-1,j}$ gives $\delta_{0j}$ in row $0$ and $\delta_{m-1,j}$ in row $m-1$. At $m=2$ the two boundary rows are distinct and the computation above applies, the doubled edge entering as $C_{01}=-1$. At $m=1$,
where they coincide, $V=0$, $C=\tfrac12$ and $F=2$, so $\Kraw P\Kraw=P$ directly.
\end{proof}

\section{The periodic Jacobi matrix and the proof of the torsion formula}
\label{app:torsion}
\label{sub:jacobi}

This appendix proves the two halves (A) and (B) of Section~\ref{sec:torsion}, used in the proof of Theorem~\ref{thm:torsion}. It also proves
two lemmas on frieze rows, Lemmas~\ref{lem:cotsol} and~\ref{lem:floqvec}, which half (A) needs. These two lemmas and the identity for one $R$-run,
Lemma~\ref{lem:perrun}, are used again in the proof of Theorem~\ref{thm:Hall} (Appendix~\ref{app:Hall}).

Recall that $|A_-||A_+|=|\det(\varphi-1)\det(\varphi+1)|=(\tr\varphi)^2-4$. Here
$\mathcal H=W_c^2\det(K+D)\prod_a(1-x_a)$ by \eqref{eq:Habelian}, where $K$ is the run-structure level matrix of Theorem~\ref{thm:Kformula} and
$D=\operatorname{diag}\bigl(Q_a^2x_a/(1-x_a)\bigr)$ is the diagonal matrix by which the Hessian of the twisted superpotential differs from $K$, both indexed by the $r$ sites (letters) of the
word; $D_R$ and $D_L$ denote the diagonal blocks of $D$ on the $R$-sites and on the $L$-sites. Let $h=h_+$ be the expanding height function of
Theorem~\ref{prop:propagated}, so that $x=V^*$. Along an $R$-run from columns $(p,q)$, with $u_j=h(p+jq)/h(q)=u_0+j$, one has $D_j=u_j^2-1$
and $1-x_j=u_j^{-2}$; along an $L$-run from $(p',q')$, with $v_k=h(q'+kp')/h(p')=v_0+k$, one has $D_k=4/(v_k^2-1)$ and $1-x_k=(v_k^2-1)/v_k^2$.

We use the notation of Section~\ref{sec:torsion}: $D$, the fans $F_b$, the $L$-Schur complement $S$, its couplings $\beta$, its transfer matrix
$T$, and the two halves (A) and (B). We write $F$ for $F_b$ when the run is understood; $\beta_q$ is the coupling across the $q$-th
$R$-run ($q=1,\dots,n_R$). The \emph{Casoratian} of two solutions $f,g$ of a three-term recurrence $y_{i-1}+y_{i+1}=e_iy_i$ is $f_ig_{i+1}-f_{i+1}g_i$, the discrete Wronskian, constant in $i$. A
\emph{Floquet solution} is a solution with $y_{i+n}=\Lambda y_i$ for all $i$.

\subsection*{Two lemmas on frieze rows}
The mechanism of half (A), and of Theorem~\ref{thm:Hall}, is a single fact about a trigonometric frieze row, valid for any row and any word,
together with an explicit Floquet solution.

\begin{lemma}[The recurrence of a frieze row]
\label{lem:cotsol}
Let $\hat u_i=\sin\psi_i/\sin\vartheta$, $\psi_i=i\vartheta+\phi$, be a frieze row with $1-x_i=\hat u_{i-1}\hat u_{i+1}/\hat u_i^2$. The three-term
recurrence
\[
y_{i-1}+y_{i+1}=\frac{2}{1-x_i}\,y_i
\]
has the fundamental system $\mathfrak f_i=\cot\psi_i$, $\mathfrak g_i=i\cot\psi_i-\cot\vartheta$, with Casoratian $\mathfrak f_i\mathfrak g_{i+1}-\mathfrak f_{i+1}\mathfrak g_i=-1$. It is the recurrence of
the $L$-Schur complement $S_\alpha$ (Schur complement of $K+D$ at the class $\alpha$ whose coordinates give the frieze row) inside an $L$-run \textup{(}where $S_{ii}=4/(1-x_i)$ and the couplings are $-2$\textup{)}, and also, up to the factor $\tfrac12$, that of
$J=F^{-1}+D_R^{-1}$ inside an $R$-run \textup{(}off-diagonal $-\tfrac12$, diagonal $1+1/(\hat u_{i-1}\hat u_{i+1})=1/(1-x_i)$ at the rows $i<b$, and $\tfrac12+1/(\hat u_{b-1}\hat u_{b+1})$ at the last row\textup{)}, where
$F=(2\min(i,k))$ is the fan and $F^{-1}$ the path matrix $\tfrac12\operatorname{tridiag}(-1,2,-1)$ with last diagonal entry $\tfrac12$. On an $R$-run, $x_i$ in the lemma is the frieze variable and not the Bethe shape $x_{R_i}$; for the latter, $1/(1-x_{R_i})=\hat u_i^2$.
\end{lemma}

\begin{proof}
With $u_i=\sin\psi_i$ put $y_i=z_i/u_i$; the recurrence becomes
$u_{i+1}z_{i-1}+u_{i-1}z_{i+1}=2u_iz_i$, and $z_i=\cos\psi_i$ solves it by $\sin(\psi+\vartheta)\cos(\psi-\vartheta)+\sin(\psi-\vartheta)\cos(\psi+\vartheta)=\sin2\psi$,
while $z_i=i\cos\psi_i-\cot\vartheta\sin\psi_i$ solves it because the extra terms are
$\sin(\psi-\vartheta)\cos(\psi+\vartheta)-\sin(\psi+\vartheta)\cos(\psi-\vartheta)=-\sin2\vartheta$ and $-2\cot\vartheta(\sin^2\psi-\sin^2\vartheta)+2\cot\vartheta\sin^2\psi=\sin2\vartheta$.
The Casoratian is $\cot\psi_0\cot\psi_1+\cot\vartheta(\cot\psi_1-\cot\psi_0)=(\cos\psi_0\cos\psi_1-\cos\vartheta)/(\sin\psi_0\sin\psi_1)=-1$. That the recurrence is the one of $S_\alpha$ inside an $L$-run follows from Theorem~\ref{thm:Kformula}: there $K_{ii}=4$,
$K_{i,i\pm1}=-2$, so $S_{ii}=4+D_i=4+4x_i/(1-x_i)=4/(1-x_i)$ and the row of $S_\alpha$ is $-2y_{i-1}+\tfrac4{1-x_i}y_i-2y_{i+1}=0$. That it is
also the recurrence of $J=F^{-1}+D_R^{-1}$ inside an $R$-run follows from the inverse of the fan matrix, $F^{-1}=\tfrac12\operatorname{tridiag}(-1,2,-1)$ with
last diagonal entry $\tfrac12$, together with $D_j^{-1}=(1-x_{R_j})/x_{R_j}=1/(\hat u_j^2-1)=1/(\hat u_{j-1}\hat u_{j+1})$ by the frieze relation.
\end{proof}

Fix a word and the slope data of Theorem~\ref{thm:allwords}: at the $t$-th tetrahedron the basis $(p_t,q_t)$, the kept slope
$\mathsf k_t$, the middle slope $\mathsf m_t=p_t+q_t$. Fix also a class $\alpha$ of that theorem, with coordinates $x$, and write $S$ for the
$L$-Schur complement of $K+D$ with $D=\operatorname{diag}(Q_a^2x_a/(1-x_a))$ evaluated at those coordinates. It is a periodic Jacobi matrix
(Section~\ref{sec:torsion}) on the $L$-sites, with diagonal $4/(1-x)$ and couplings $-2$ inside the $L$-runs, and across each $R$-run of
length $b$ the two corner corrections and the coupling
\begin{equation}
S_{L_aL_a}=\frac4{1-x_{L_a}}+J^{11}-2,\qquad S_{L'_1L'_1}=\frac4{1-x_{L'_1}}+J^{bb},\qquad \beta_q=-J^{1b},
\label{eq:corners}
\end{equation}
for the $q$-th $R$-run, with $J^{ik}:=(J^{-1})_{ik}$. Here $\beta_q$ is the coupling of Section~\ref{sec:torsion} across the run, written in a second form.
In \eqref{eq:corners} $L_a$ is the last $L$ before the run, $L'_1$ the first after it, and $J=F^{-1}+D_R^{-1}$ on the run. To see this, the Schur
complement subtracts $\mathsf r^{\mathsf T}G\mathsf r$, $s^{\mathsf T}Gs$ and $\mathsf r^{\mathsf T}Gs$ from the entries $(L'_1,L'_1)$, $(L_a,L_a)$ and
$(L_a,L'_1)$ of $K+D$, where $G=(F+D_R)^{-1}$ and $\mathsf r=Fe_b$, $s=-Fe_1$ are the columns of $K$ coupling $L'_1$ and $L_a$ to the run
(Theorem~\ref{thm:Kformula}); and $FGF=F-J^{-1}$, an instance of the Woodbury matrix identity, here simply
$F-FGF=F(F+D_R)^{-1}D_R=\bigl(D_R^{-1}(F+D_R)F^{-1}\bigr)^{-1}=(F^{-1}+D_R^{-1})^{-1}$. So the three entries change by $-(2b-J^{bb})$, $-(2-J^{11})$ and $+(2-J^{1b})$: the $-2b$ cancels the $2b$ in $K_{L'_1L'_1}=4+2b$, the
$+2$ cancels the cyclic coupling $-2$, and the $-2$ at $L_a$ remains, which is \eqref{eq:corners}.

\begin{lemma}[The universal Floquet vector]
\label{lem:floqvec}
Let $\theta:\Z^2\to\C$ be additive, put $f(e)=\sin\theta(e)$, and let $x_t=f(\mathsf k_t)^2/f(\mathsf m_t)^2$ at $L$-flips and
$x_t=1-f(\mathsf k_t)^2/f(\mathsf m_t)^2$ at $R$-flips along the developed word, all finite and different from $0$ and $1$, and suppose $\det(F+D_R)\ne0$ on every $R$-run, so that the corner entries \eqref{eq:corners} are defined; the case
$\theta(e)=\pi\langle v,e\rangle/N$, $v\in\Z^2$, $N\ge1$, is the closed form of Theorem~\ref{thm:allwords}. Let $h$ be any linear functional on
$\Z^2$, and put
\[
H(e)=h(e)\cot\theta(e),\qquad Y_t=H(\mathsf m_t)-H(\mathsf k_t)\quad(t\text{ an }L\text{-site}).
\]
Index the $L$-sites along the developed word, i.e.\ along the infinite periodic sequence of letters obtained by repeating the word (the \emph{unrolled} recurrence).
The lemma is local: it does not assume the closing-up condition $Y_{t+n}=\Lambda Y_t$, $n=\#L$, which holds only for particular $\theta$ and $h$ and is verified in each application (the proof of (A) below, and Theorem~\ref{thm:Hall}). Then $SY=0$, identically in $\theta$ and $h$, where $S$ here denotes the infinite three-term operator along the developed word whose rows are those of $S$, given locally. This holds for the interior rows $-2Y_{i-1}+\tfrac4{1-x_i}Y_i-2Y_{i+1}$ and for the junction rows, which carry the corner
entries \eqref{eq:corners}.
\end{lemma}

\begin{proof}
\emph{Interior rows.} Along an $L$-run with kept slope $p$, entered from the basis $(p,q)$, the middle slopes are $c_i=q+ip$, so with
$\vartheta=\theta(p)$, $\phi=\theta(q)$ and $\psi_i=i\vartheta+\phi$,
\[
Y_i=\bigl(h(q)+ih(p)\bigr)\cot\psi_i-h(p)\cot\vartheta=h(q)\mathfrak f_i+h(p)\mathfrak g_i
\]
in the fundamental system of Lemma~\ref{lem:cotsol}, whose row is $\hat u_i=f(c_i)/f(p)=\sin\psi_i/\sin\vartheta$ with
$1-x_{L_i}=\hat u_{i-1}\hat u_{i+1}/\hat u_i^2$ (Ptolemy). So the interior rows hold, and the formal extensions $Y_0$, $Y_{a+1}$ of the
run's solution are the same expression at $c_0=q$ and $c_{a+1}=q+(a+1)p$.

\emph{Junction rows.} Let the $R$-run following the run have length $b$, kept slope $q'=c_a$ and middle slopes $w_j=p+jq'$
($w_0=p$, $w_1=c_{a+1}$), so $1-x_{R_j}=\hat v_j^{-2}$ with $\hat v_j=f(w_j)/f(q')=\sin\chi_j/\sin\vartheta'$, $\chi_j=j\vartheta'+\vartheta$,
$\vartheta'=\theta(q')$; the next $L$-run has basis $(w_b,q')$ and kept slope $w_b$. Eliminating $Y_{a-1}$ and $Y'_2$ by the
interior recurrence, the two rows of $S$ at the junction read
\[
E_a=(J^{11}-2)Y_a+2Y_{a+1}-J^{1b}Y'_1,\qquad E'_1=J^{bb}Y'_1+2Y'_0-J^{1b}Y_a ,
\]
with $Y_a=H(q')-H(p)$, $Y_{a+1}=H(w_1)-H(p)$, $Y'_1=H(w_{b+1})-H(w_b)$, $Y'_0=H(q')-H(w_b)$. Here $J^{ik}$ denotes the entries of $J^{-1}$. Put $\eta_k=H(w_k)-H(q')=h(p)\mathfrak f_k+h(q')\mathfrak g_k$, $k=0,\dots,b+1$, with $\mathfrak f_k=\cot\chi_k$ and $\mathfrak g_k=k\cot\chi_k-\cot\vartheta'$. By Lemma~\ref{lem:cotsol} the sequence $\eta$ solves the recurrence of the $R$-row. So the interior rows of $J\eta$ vanish, and the first and last rows give
\[
J\eta=\tfrac12\eta_0\,e_1+\tfrac12(\eta_{b+1}-\eta_b)\,e_b .
\]
Here $\eta_0=-Y_a$, $\eta_1=Y_{a+1}-Y_a$, $\eta_b=-Y'_0$ and $\eta_{b+1}-\eta_b=Y'_1$. The first component of $\eta=J^{-1}(J\eta)$ reads $Y_{a+1}-Y_a=-\tfrac12J^{11}Y_a+\tfrac12J^{1b}Y'_1$, which is $E_a=0$. The $b$-th component reads $-Y'_0=-\tfrac12J^{1b}Y_a+\tfrac12J^{bb}Y'_1$, which is $E'_1=0$. A singleton $L$-run carries both corrections at
one site, and its row is $E'_1+E_a$ of the two adjacent junctions, because the interior recurrence at $i=1$ is an identity for the
formal extensions. The identities hold wherever $J^{-1}$ exists.
\end{proof}

\emph{Proof of (A).} At the $L$-site $t$, with kept slope $\mathsf k_t$ and middle slope $\mathsf m_t$, set $Y^-_t=1/(h(\mathsf k_t)h(\mathsf m_t))$, and extend $Y^-$ to the $L$-sites of consecutive periods of the word by $t\mapsto\varphi t$. Then $Y^-$
solves the three-term recurrence of $S$ along the unrolled sequence: for every $L$-site $t$, with the indices running on through consecutive periods,
\[
\beta_{t-1}Y^-_{t-1}+d_tY^-_t+\beta_tY^-_{t+1}=0 ,
\]
including at the junction rows, where one $L$-run ends and the next begins. Since $h(\varphi e)=\lambda h(e)$, one has $Y^-_{t+n}=\lambda^{-2}Y^-_t$ with $n=\#L$: $Y^-$ is a Floquet solution of $S$ with multiplier $\lambda^{-2}$. This is the limit $\epsilon\to0$, taken row by row, of Lemma~\ref{lem:floqvec}, applied with the additive angle
$\theta(e)=\epsilon h(e)$, $\epsilon>0$ small. As $\epsilon\to0$ the coordinates $\sin^2(\epsilon h(\mathsf k_t))/\sin^2(\epsilon h(\mathsf m_t))$ tend to $h(\mathsf k_t)^2/h(\mathsf m_t)^2$,
which is $V^*$ (Theorem~\ref{prop:propagated}); and with the contracting height $h'$ in place of the lemma's $h$,
$\epsilon\bigl(H(\mathsf m_t)-H(\mathsf k_t)\bigr)\to h'(\mathsf m_t)/h(\mathsf m_t)-h'(\mathsf k_t)/h(\mathsf k_t)$, a constant multiple of $Y^-_t$, because
$h'(\mathsf m_t)h(\mathsf k_t)-h'(\mathsf k_t)h(\mathsf m_t)$ is a constant multiple of $\det(\mathsf k_t,\mathsf m_t)=\det(\mathsf k_t,\mathsf m_t-\mathsf k_t)=1$.
Since $h(\varphi v)=\lambda h(v)$, $Y^-|_L$ is a Floquet solution with multiplier $\lambda^{-2}$. $S$ is symmetric, so the transfer matrix has
determinant $1$, the other multiplier is $\lambda^2$, and $\tr T=\lambda^2+\lambda^{-2}$. Hill's formula $\det S=(-1)^n\prod\beta_i(\tr T-2)$
then gives (A), since
\[
\lambda^2+\lambda^{-2}-2=(\tr\varphi)^2-4=|A_-||A_+| .
\]
Inside $L$-runs the recurrence has fundamental system $1/v_k$, $v_k^2-1$ (Casoratian $3$), and $Y^-\propto1/v_k$ there.

\emph{Reduction of (B).} The $L$-runs give $\prod_L(1-x)=\prod_qv_0v_{a+1}/(v_1v_a)$, since $1-x_k=v_{k-1}v_{k+1}/v_k^2$. Over one period all
intermediate heights cancel between consecutive $R$- and $L$-run factors, leaving
$\bigl[h(p)h(q)/h(p+q)\bigr]\bigl[h(\varphi p+\varphi q)/(h(\varphi p)h(\varphi q))\bigr]=1/\lambda$; the $(-2)$'s of the $\#L-n_R$ couplings inside
$L$-runs and the $n_R$ factors $-\beta_q/2$, one for each $R$-run, combine with $(-1)^{\#L}$ to $2^{\#L}$. So (B) follows from a single identity for each $R$-run.

\begin{lemma}[The identity for one $R$-run, for any frieze row]
\label{lem:perrun}
Let $(u_j)_{j\in\Z}$ be any solution of the frieze relation $u_j^2-u_{j-1}u_{j+1}=1$ \textup{(}Remark~\ref{rem:friezedef}\textup{)}, and
for $b\ge0$ let $D_R=\operatorname{diag}(u_1^2-1,\dots,u_b^2-1)$, $1-x_j=u_j^{-2}$ $(j=1,\dots,b)$, and $\beta_q=-2-\mathsf r^{\mathsf T}(F_b+D_R)^{-1}s$,
with $\mathsf r=2(1,\dots,b)^{\mathsf T}$ and $s=-2\cdot\mathbf1$ the ramp and the constant row of Theorem~\ref{thm:Kformula}, the coupling of
the $L$-Schur complement across the run. Then
\begin{equation}
G_b(u):=\det(F_b+D_R)\prod_{j=1}^b(1-x_j)\cdot\bigl(-\tfrac{\beta_q}2\bigr)=\frac{u_0\,u_{b+1}}{u_1\,u_b}.
\label{eq:perrun}
\end{equation}
For the affine row $u_j=u+j$ of the expanding
eigenvector the right side is $u(u+b+1)/\bigl((u+1)(u+b)\bigr)$. The identity \eqref{eq:perrun} is also valid for the trigonometric row
$u_j=\sin(j\vartheta+\phi)/\sin\vartheta$ of a general Bethe class \textup{(}Theorem~\ref{thm:friezeall},
Remark~\ref{rem:friezedef}\textup{)}, with the same formula in terms of these $u_j$.
\end{lemma}

\begin{proof}
The argument below uses the chain only through $u_1^2-1=u_0u_2$ --- the frieze relation --- and through the shift of the row by one index, so it is
insensitive to which frieze row $(u_j)$ is. Let $B_b(u)=\begin{psmallmatrix}F_b+D&s\\\mathsf r^{\mathsf T}&-2\end{psmallmatrix}$, so $\det B_b=\det(F_b+D)\,\beta_q$ and
$G_b=-\tfrac12\det B_b\prod_ju_j^{-2}$, with $G_0=1$. Eliminate the first $R$-site: its pivot is $\pi=2+u_1^2-1=u_1^2+1$, assumed nonzero (both sides of the resulting identity are polynomial in the $u_j$, so the case $\pi=0$ follows by continuity), and it couples by $2$ to every
other $R$-site, by $-2$ to the border column and by $2$ to the border row. Writing $2\min(i,j)=2+2\min(i-1,j-1)$ and putting $c:=2-4/\pi$ (a scalar
used only in this proof), the Schur complement is
$\begin{psmallmatrix}F_{b-1}+D'+c\mathbf1\mathbf1^{\mathsf T}&-c\mathbf1\\\mathsf r'^{\mathsf T}+c\mathbf1^{\mathsf T}&-c\end{psmallmatrix}$, with $D'$ the
same diagonal in the shifted variable $u+1$ and $\mathsf r'=2(1,\dots,b-1)$. Adding the border column to every $R$-column removes $c\mathbf1\mathbf1^{\mathsf
T}$ and the $c\mathbf1^{\mathsf T}$ in the border row, leaving $B_{b-1}(u+1)$ with its last column scaled by $c/2$. Hence
$\det B_b=\pi\tfrac c2\det B_{b-1}'=(u_1^2-1)\det B_{b-1}'=u_0u_2\det B_{b-1}'$, where $B'$ is the same expression for the row shifted by one index, so
$G_b=\tfrac{u_0u_2}{u_1^2}G_{b-1}'$, which telescopes to $\prod_{j=1}^bu_{j-1}u_{j+1}/u_j^2=u_0u_{b+1}/(u_1u_b)$. The step $\pi c/2=\pi-2=u_1^2-1=u_0u_2$
is the frieze relation and is the only place the chain enters.
\end{proof}

\section{Proof of Hill's formula at a general class}
\label{app:Hall}

This appendix proves Theorem~\ref{thm:Hall}, which gives the handle-gluing values $\mathcal H_\alpha=2^{\#L}|A_\mp|$ used in
Sections~\ref{sec:spectrum}, \ref{sec:count} and~\ref{sec:modular}.

\begin{proof}[Proof of Theorem~\ref{thm:Hall}]
\emph{Non-degeneracy.} Every kept, dying and newborn slope of the word is the middle slope of another tetrahedron: the newborn slope
of step $t$ is the middle of step $t+1$, and a basis vector at step $t$ is the middle slope of the last flip that replaced that column
(the word contains both letters, so every column is replaced; when that flip lies in the previous period of the developed word, $f\ne0$ there follows from the periodicity (b) in the proof of Theorem~\ref{thm:allwords}). So $f(\mathsf m_t)\ne0$ for all $t$ gives $f\ne0$ on all four slopes of every
tetrahedron, and by Ptolemy $x_t\notin\{0,1,\infty\}$: the class has all coordinates finite and nonzero, the hypothesis of
Lemma~\ref{lem:floqvec}, and $D_R$ is invertible. The second hypothesis is what defines $S$ (it is automatic for $b=1$, where
$J=(2-x_R)/2x_R$ and $x_R<1$, since $1-x_R$ is a ratio of squares of real numbers).

\emph{Periodicity.} Let $h$ be the height functional of Theorem~\ref{prop:propagated}, $h(\varphi e)=\lambda h(e)$, and $h'$ the
contracting one. Developing the word once more sends every slope $e$ to $\varphi e$, and
$\langle v,\varphi e\rangle=\langle\varphi^{-1}v,e\rangle\equiv\pm\langle v,e\rangle\pmod N$ for $v\in\mathcal K_\mp$, so, $\cot$ being
$\pi$-periodic, $H(\varphi e)=\pm\lambda H(e)$. Hence the vectors $Y^\pm$ of Lemma~\ref{lem:floqvec} built from $h$ and $h'$ satisfy
$Y^\pm_{t+r}=\pm\lambda^{\pm1}Y^\pm_t$: they are Floquet solutions of $S$ with multipliers $\pm\lambda$ and $\pm\lambda^{-1}$. They are
not both zero: $Y^+\equiv Y^-\equiv0$ would force, $h,h'$ being a basis of the dual, $\cot(\pi\langle v,\mathsf m_t\rangle/N)\,\mathsf m_t=
\cot(\pi\langle v,\mathsf k_t\rangle/N)\,\mathsf k_t$ for every $L$-site, i.e.\ both cotangents zero, i.e.\ $\langle v,\mathsf m_t\rangle\equiv
\langle v,\mathsf k_t\rangle\equiv N/2$, and then $\langle v,\mathsf m_t-\mathsf k_t\rangle\equiv0$ for the dying slope $\mathsf m_t-\mathsf k_t$, which is
excluded. A nonzero Floquet solution with multiplier $\Lambda$ gives $\tr T_\alpha=\Lambda+\Lambda^{-1}$ ($\det T_\alpha=1$), so
\[
\Delta_\alpha=\tr T_\alpha-2=\pm(\lambda+\lambda^{-1})-2=\pm\tr\varphi-2=\pm|\det(\varphi\mp1)|=\pm|A_\mp| .
\]
Hill's formula $\det S=(-1)^{\#L}\prod\beta_i\,\Delta_\alpha$ holds with the couplings of \eqref{eq:corners}, including the folded cases
$\#L=1,2$. For $\#L=1$, $\det S=d+2\beta=-\beta(\tr T-2)$, and for $\#L=2$, $\det S=d_1d_2-(\beta_1+\beta_2)^2=\beta_1\beta_2(\tr T-2)$; the couplings are nonzero, since $G_b\ne0$ and $\det(F+D_R)\ne0$ by hypothesis, so the transfer
matrix $T_\alpha$ is defined.

\emph{$P_\alpha$.} By Lemma~\ref{lem:perrun} on the $R$-row $\hat v_j$ and the $L$-run product $\prod_L(1-x)=\hat u_0\hat u_{a+1}/(\hat
u_1\hat u_a)$, the $(-2)$'s and $\beta_q$'s combining to $2^{\#L}$ exactly as in the reduction of (B),
\[
P_\alpha=2^{\#L}\prod_{L\text{-runs}}\frac{f(q)\,f(c_{a+1})}{f(c_1)\,f(c_a)}\prod_{R\text{-runs}}\frac{f(p)\,f(g_{b+1})}{f(g_1)\,f(g_b)} .
\]
Since $c_{a+1}=g_1$, $c_a=q'$ is the $q$ of the next $L$-run (an $R$-run keeps it), $g_b$ is the $p$ of the next $L$-run and $g_{b+1}$
its $c_1$, the product telescopes around the word to
\[
\frac{f(p_0)}{f(\varphi p_0)}\cdot\frac{f(q_0)}{f(\varphi q_0)}\cdot\frac{f(\varphi c_1)}{f(c_1)},\qquad c_1=p_0+q_0 ,
\]
for the basis $(p_0,q_0)$ at the start of the first $L$-run. By periodicity (b) of Theorem~\ref{thm:allwords},
$f(\varphi e)=\pm(-1)^{\langle w,e\rangle}f(e)$, and $\langle w,p_0\rangle+\langle w,q_0\rangle+\langle w,p_0+q_0\rangle$ is even, so
the product is $\pm1$: $P_\alpha=\pm2^{\#L}$ with the sector sign.

\emph{$W_c^2$.} By \eqref{eq:wcmono} at $\epsilon=0$ and $1-x_{L_i}=\hat u_{i-1}\hat u_{i+1}/\hat u_i^2$, the $L$-run product
$\prod_{i=1}^a(1-x_{L_i})^{-2i}$ telescopes to $\hat u_a^{2(a+1)}/(\hat u_0^2\hat u_{a+1}^{2a})=f(c_a)^{2a+2}/\bigl(f(q)^2f(c_{a+1})^{2a}\bigr)$,
and the following $R$-run contributes $(1-x_{R_1})^{-a}=f(g_1)^{2a}/f(q')^{2a}$ with $g_1=c_{a+1}$, $q'=c_a$; the pair leaves
$f(c_a)^2/f(q)^2=f(q')^2/f(q)^2$, the ratio of the squares of the second basis vector at the end and at the start of the $L$-run. Around
the word this telescopes to $f(\varphi q_0)^2/f(q_0)^2=1$.

\emph{$\mathcal H_\alpha$.} $\mathcal H_\alpha=W_c^2\,\Delta_\alpha P_\alpha=1\cdot(\pm|A_\mp|)(\pm2^{\#L})=2^{\#L}|A_\mp|$.
\end{proof}

\section{Proof of completeness}
\label{app:complete}

This appendix proves Theorem~\ref{thm:complete}, the completeness of the kernel-vector classes for words with no singleton run. The theorem
is used in item (A6) of the Introduction, in Corollary~\ref{cor:interpretation}, in the count of classes in Section~\ref{sec:count}, and in
Section~\ref{sec:spins} through the labelling of the non-degenerate classes by kernel vectors. It uses
Lemma~\ref{lem:sine}.

\begin{proof}[Proof of Theorem~\ref{thm:complete}]
\emph{Rows.} Fix an $L$-run of length $a\ge2$, entered from the basis $(p,q)$, with middle slopes $c_i=q+ip$, and put
$u_i=\pm x_{L_i}^{-1/2}$. By Theorem~\ref{thm:friezeall} and Remark~\ref{rem:friezesign} the signs can be chosen so that
$u_i^2-u_{i-1}u_{i+1}=1$ at every interior site (a flip of $u_c$ changes the sign of the relation at $c\pm1$ and at no other site: go through the interior sites $i=2,3,\dots$ in turn and, if the relation at $i$ has the wrong
sign, flip $u_{i+1}$, which corrects the relation at $i$ and affects only the relation at $i+2$, still to come). The row then satisfies $u_{i+1}=\tau u_i-u_{i-1}$ with $\tau$ constant, extended to $u_0$ and
$u_{a+1}$ by the recurrence. To see that $\tau$ is constant, let $i$ and $i+1$ be two consecutive interior sites. Subtracting the relations
$u_i^2-u_{i-1}u_{i+1}=1$ and $u_{i+1}^2-u_iu_{i+2}=1$ gives
\[
u_i^2+u_iu_{i+2}=u_{i+1}^2+u_{i-1}u_{i+1},\qquad\text{i.e.}\qquad u_i(u_i+u_{i+2})=u_{i+1}(u_{i+1}+u_{i-1}),
\]
and dividing by $u_iu_{i+1}\ne0$ (non-degeneracy) gives $(u_i+u_{i+2})/u_{i+1}=(u_{i-1}+u_{i+1})/u_i$. So the quantity $\tau=(u_{i-1}+u_{i+1})/u_i$ takes the same value at
every interior site. For $a=3$ there is a single interior site and $\tau$ is its value; for $a=2$, where there is no interior site, $\tau$ is defined by $\tau u_1u_2=u_1^2+u_2^2-1$ so that
the extended row satisfies the relation. If $\tau\ne\pm2$, the solutions $c\mu^{\pm i}$ of the recurrence, $\mu+\mu^{-1}=\tau$, have $u_i^2-u_{i-1}u_{i+1}=0$ and are excluded by the frieze relation, and Remark~\ref{rem:friezedef}(ii) gives $\vartheta,\phi\in\C$, $\sin\vartheta\ne0$, and a constant $A$ (the \emph{amplitude}, the overall multiplier of the sine row) with
$u_i=A\sin(i\vartheta+\phi)/\sin\vartheta$ for $i=0,\dots,a+1$. Substituting in $u_i^2-u_{i-1}u_{i+1}=1$ and using
$\sin^2\psi-\sin(\psi-\vartheta)\sin(\psi+\vartheta)=\sin^2\vartheta$ gives $A^2=1$, and replacing $\phi$ by $\phi+\pi$ if necessary, $A=1$. We set $F(p)=\sin\vartheta$, $F(c_i)=\sin(i\vartheta+\phi)$: a function on
the run's slopes with $x_{L_i}=F(p)^2/F(c_i)^2$. If $\tau=\pm2$ the row is affine, $u_i=\pm(u_0+i)$ at $\tau=2$ and $u_i=\pm(-1)^i(u_0+i)$ at $\tau=-2$; only $u_i^2$ enters below, so the second case reduces to the first, and we set $F(p)=1$,
$F(c_i)=u_0+i$ instead, a linear function. The same with $q'=c_a$ and $g_j=p+jq'$ for the following $R$-run, $\hat v_j=\pm(1-x_{R_j})^{-1/2}$,
$F(q')=\sin\vartheta'$, $F(g_j)=\sin(j\vartheta'+\phi')$. Non-degeneracy gives $F\ne0$ on every slope of every run, including the
extensions $c_0,c_{a+1}$ and $g_0,g_{b+1}$, since $1-x_{L_1}=u_0u_2/u_1^2$, $1-x_{L_a}=u_{a-1}u_{a+1}/u_a^2$ and
$x_{R_1},x_{R_b}$ likewise. Each run's $F$ is defined up to a scalar, which is immaterial below.

\emph{The equations at the run ends.} With the rows in place, the interior $L$-site equations are identities (Proposition~\ref{prop:runreduce}); the $R$-site equations are the $(R_k)$ below, and
summation by parts along the $R$-row gives $\prod_ix_{R_i}^{2\min(i,k)}=(\hat v_0/\hat v_k)^2(\hat v_{b+1}/\hat v_b)^{2k}$,
$\prod_jx_{R_j}^{2j}=(\hat v_{b+1}/\hat v_b)^{2b}(\hat v_0/\hat v_b)^2$ and $\prod_jx_{R_j}=\hat v_0\hat v_{b+1}/\hat v_1\hat v_b$.
Reading the exponents off Theorem~\ref{thm:Kformula}, the equations at the sites $R_1,\dots,R_b$ of the run, at the last $L$ before
it and at the first $L$ after it (whose run has basis $(g_b,q')$ and row $u'_i$) are
\begin{align*}
(R_k)\quad&1=u_a^2\hat v_0^2\bigl[u'^{-2}_1(\hat v_{b+1}/\hat v_b)^2\bigr]^k,\qquad k=1,\dots,b,\\
(L_a)\quad&(u_{a-1}u_{a+1})^2=u_{a-1}^2\,u'^2_1\,(\hat v_1\hat v_b/\hat v_0\hat v_{b+1})^2,\\
(L'_1)\quad&(u'_0u'_2)^2=u'^{-2b}_1u'^2_2\,u_a^2\,(\hat v_{b+1}/\hat v_b)^{2b}(\hat v_0/\hat v_b)^2 .
\end{align*}
Since $b\ge2$, $(R_1)$ and $(R_2)$ give $u_a^2\hat v_0^2=1$ and $u'^2_1=(\hat v_{b+1}/\hat v_b)^2$, and then
$(L_a)$ and $(L'_1)$ reduce to $u_{a+1}^2=u'^2_1(\hat v_1\hat v_b/\hat v_0\hat v_{b+1})^2$ and $u'^2_0=u_a^2\hat v_0^2/\hat v_b^2$.
Now the two runs share the three slopes $p=g_0$, $q'=c_a$ and $c_{a+1}=g_1$, and the next $L$-run shares $q'$, $g_b$ and
$g_{b+1}$ with the $R$-run; writing $\rho_e=F_L(e)/F_R(e)$ and $\varsigma_e=F_R(e)/F_{L'}(e)$ for the ratios on the shared slopes,
$u_a\hat v_0=\rho_{q'}/\rho_p$, $u_{a+1}=(\rho_{c_{a+1}}/\rho_p)\hat v_1/\hat v_0$, $u'_1=(\varsigma_{g_b}/\varsigma_{g_{b+1}})\hat v_{b+1}/\hat v_b$
and $u'_0=(\varsigma_{g_b}/\varsigma_{q'})/\hat v_b$, so the four equations say exactly
\[
\rho_p^2=\rho_{q'}^2=\rho_{c_{a+1}}^2,\qquad \varsigma_{q'}^2=\varsigma_{g_b}^2=\varsigma_{g_{b+1}}^2 :
\]
\emph{the rows of adjacent runs agree, up to a common factor and up to the sign of each entry, on the three slopes they share.}
Conversely these conditions imply $(R_k)$ for all $k$, $(L_a)$ and $(L'_1)$. So, for a word with no singleton run, the Bethe equations
\eqref{eq:bethex} are equivalent to the junction conditions at every junction.

\emph{Gluing.} The three shared slopes form a Farey triangle $(e_1,e_2,e_1+e_2)$ --- $(p,q',p+q')$ at an $L\to R$ junction, $(g_b,q',g_b+q')$ at an $R\to L$ junction --- and on it a sine row has values
$(\sin A,\sin B,\sin(A+B))$ with $(A,B)$ the angles of $(e_1,e_2)$: $(\vartheta,a\vartheta+\phi)$ for the $L$-run, $(\phi',\vartheta')$ for the
$R$-run. By Lemma~\ref{lem:sine}, the junction condition forces $(\phi',\vartheta')\equiv\pm(\vartheta,a\vartheta+\phi)\pmod\pi$ with a common sign, which is
absorbed into the $R$-run's parametrisation. Inductively there is a single additive $\Theta:\Z^2\to\C/\pi\Z$ with
$F(e)^2=\sin^2\Theta(e)$ on every slope of the developed word, hence $x_t=\sin^2\Theta(\mathsf k_t)/\sin^2\Theta(\mathsf m_t)$ at $L$-steps and
$1-$that at $R$-steps. Lemma~\ref{lem:sine} also excludes a sine row adjacent to an affine one, so either all runs are affine or none is. In the
affine case $F$ is a linear function $h$ on every run, consistency on the shared basis makes it one linear function on $\Z^2$ (the signs of the three shared entries agree because $F\ne0$ on them), and
the wrap-around junction gives $h\circ\varphi=ch$ for a scalar $c$; so $h$ is an eigenfunctional of $\varphi$, and $x$ is $V^*$ or
$\sigma V^*$ (Theorem~\ref{prop:propagated}).

\emph{Closure.} In the sine case the wrap-around junction compares the last run's row on the slopes $\varphi p_0,\varphi q_0$ with the
first run's row on $p_0,q_0$: Lemma~\ref{lem:sine}, applied to the first run's Farey triangle and to its image under $\varphi$, on which $x$ takes the same values by periodicity, gives $\Theta\circ\varphi\equiv\pm\Theta\pmod\pi$ on a basis, hence on $\Z^2$. Lift $\Theta$ to an
additive $\tilde\Theta:\Z^2\to\C$: choose lifts of $\Theta(e_1)$ and $\Theta(e_2)$ to $\C$ and extend linearly, which agrees with
$\Theta$ modulo $\pi$ because both are additive on the free group $\Z^2$. Then $\tilde\Theta\circ(\varphi\mp1)$ takes values in $\pi\Z$, so $\tilde\Theta\circ(\varphi\mp1)=\pi m$
for some $m\in\operatorname{Hom}(\Z^2,\Z)$ and $\tilde\Theta=\pi\,m\circ(\varphi\mp1)^{-1}=\pi\langle v,\cdot\rangle/N$ with
$v\in\Z^2$, $N=|\det(\varphi\mp1)|=|A_\mp|$: the angles are rational multiples of $\pi$, with denominator $N$. Periodicity now reads
$\langle v,(\varphi\mp1)e\rangle\equiv0\pmod N$ for all $e$, i.e.\ $v\in\mathcal K_\mp$ (proof of Theorem~\ref{thm:allwords}), and
$v\not\equiv0$ since $F\ne0$. So $x$ is the class of Theorem~\ref{thm:allwords} at $v$. Finally $x$ determines each run's $(\vartheta_\rho,\phi_\rho)$
modulo $\pi$ up to a common sign (a row of length $\ge2$ determines $\cos\vartheta$ up to sign, the sign being the alternating flip $(\vartheta,\phi)\mapsto(-\vartheta,-\phi)$ modulo $\pi$; and the sign choices compatible with the frieze
relation are the shifts by $\pi$), the junctions correlate the signs, so $\Theta$ is determined up to sign and $v$ up to sign modulo $N$. No $x$ arises from both sectors, since $\Theta\circ\varphi\equiv\Theta$ and $\Theta\circ\varphi\equiv-\Theta$ together give $2\Theta\equiv0$, and then $f$ vanishes somewhere. No kernel-vector class is $V^*$ or $\sigma V^*$, since a sine row is never affine (Lemma~\ref{lem:sine}).

\emph{Where the hypothesis is used.} That no run is a singleton enters at three points. In \emph{Rows}, a run of length $\ge2$ carries at
least two coordinates, which fix $\tau$ and the amplitude; a singleton run carries one coordinate for a row with two parameters. At
\emph{the run ends}, $b\ge2$ is needed to use both $(R_1)$ and $(R_2)$. On an $L$-run of length $\ge2$ the site $L_a$ serves only the
junction after the run and $L'_1$ only the junction before the next one; at a singleton $L$-run one site would serve both. In
\emph{Closure}, a row of length $\ge2$ determines $\cos\vartheta_\rho$ up to the sign just described. Remark~\ref{rem:singleton}
describes what changes at a singleton run.
\end{proof}

\section{The step identities}
\label{app:steps}

This appendix proves the step identities \eqref{eq:steps} used in Step~3 of the proof of Theorem~\ref{thm:rogers}, with $R$, $z$ and $\mathcal G$
as defined there.

\begin{lemma}[The step identities]
\label{lem:steps}
The two identities \eqref{eq:steps} hold for all real $\alpha,\beta$ at which their arguments are defined.
\end{lemma}

\begin{proof} (Algebraically, the first is Abel's five-term identity \cite{Zagier07} at $(x,y)=(z,\bar z)$: $z\bar z=s$ and
$z(1-\bar z)/(1-z\bar z)=z(\alpha,\alpha+\beta)$; the second is the same identity applied to $1-z$.) Write $\psi=\alpha+\beta$,
$\Theta_x=\log|\sin x|$, $\omega(X,Y)=X\,dY-Y\,dX$ (a $1$-form on the $(\alpha,\beta)$-plane, with $dX$, $dY$ the differentials of the functions $X$, $Y$), and
\[
\eta(\alpha,\beta)=\omega(\Theta_\alpha,\Theta_\beta)+\omega(\Theta_\beta,\Theta_\psi)+\omega(\Theta_\psi,\Theta_\alpha).
\]
Consider the first identity on a connected component $C$ of the complement of the lines $\alpha,\beta,\psi,2\alpha+\beta\in\pi\Z$, on which all
terms are real-analytic. For real $s\notin\{0,1\}$, $d\operatorname{Re}R(s)=-\tfrac12\omega(\log|1-s|,\log|s|)$; with $\log|s|=2(\Theta_\alpha-\Theta_\psi)$ and
$\log|1-s|=\Theta_\beta+\Theta_{2\alpha+\beta}-2\Theta_\psi$ (Ptolemy) this expands to $\eta(\alpha,\beta)-\eta(\alpha,\psi)$. On $C$, $z$ and $1-z$ are not
real, so $\arg z=-\beta+\pi k_0$ and $\arg(1-z)=\alpha+\pi k_1$ with $k_0,k_1\in\Z$ constant, and
$d\operatorname{Re}R(z)=\tfrac12\omega(\log|z|,\log|1-z|)-\tfrac12\omega(\arg z,\arg(1-z))$; since $\log|z|=\Theta_\alpha-\Theta_\psi$ and $\log|1-z|=\Theta_\beta-\Theta_\psi$,
this gives $d\mathcal G=\eta(\alpha,\beta)-\omega(\alpha,\beta)$. As $\omega(\alpha,\psi)=\omega(\alpha,\beta)$, the left side of the first identity has
zero differential on $C$, so it is constant there. Each component is a bounded polygon with edges in at least three of the four directions,
so it has an edge on a line $\alpha=m\pi$ or $\beta=m\pi$. On approaching an interior point of an edge $\alpha=m\pi$, $s\to0$, $z\to0$,
$z(\alpha,\psi)\to0$, $\arg(1-z)\to0$, and the constant becomes $-m\pi\bigl(\arg z-\arg z(\alpha,\psi)\bigr)=-m\pi^2(k_0+m-k_0')\in\pi^2\Z$.
On approaching an edge $\beta=m\pi$, $s\to1$, $z\to1$, and $z':=z(\alpha,\alpha+m\pi)=e^{-i\alpha}/2\cos\alpha$ satisfies $1-z'=\overline{z'}$, so
$2\operatorname{Re}R(z')=\operatorname{Re}\bigl(R(z')+R(1-z')\bigr)=\pi^2/6$ and $\mathcal G(\alpha,\alpha+m\pi)=\pi^2/6-m\pi\arg z'$; the constant becomes
$\pi^2/6-\pi^2/3-m\pi(\alpha+\pi k_1)+\pi^2/6+m\pi(\alpha-\pi\epsilon)\in\pi^2\Z$, where $\arg z'=-\alpha+\pi\epsilon$ at the edge point. The second identity is proved
in the same way on the complement of $\alpha,\beta,\psi,\alpha+2\beta\in\pi\Z$: at $\alpha=m\pi$ one has $s\to1$, $z\to0$ and
$z(\psi,\beta)\to e^{-i\beta}/2\cos\beta$, and at $\beta=m\pi$ one has $s\to0$ and $z,z(\psi,\beta)\to1$. Writing $\arg z=-\beta+\pi k_0$ and $\arg(1-z)=\alpha+\pi k_1$ on the component,
$\arg\bigl(e^{-i\beta}/2\cos\beta\bigr)=-\beta+\pi\epsilon'$ at the edge point and $\arg\bigl(1-z(\psi,\beta)\bigr)=\psi+\pi k_1''$ on the component, the
constant is $m\pi^2(\epsilon'-k_0)$ at an edge $\alpha=m\pi$ and $m\pi^2(m+k_1''-k_1)$ at an edge $\beta=m\pi$; both lie in $\pi^2\Z$.
\end{proof}

\section{The exact reduction and the proof of existence}
\label{app:existence}

This appendix proves the existence theorem, Theorem~\ref{thm:exist}. Its results are used as follows. The telescoping identities of
Lemma~\ref{lem:telescope} enter the spin computation in the proof of Theorem~\ref{thm:lines}\textup{(iii)}, the proof of
Proposition~\ref{prop:shifts}, and the exclusion of the type $0$ after Theorem~\ref{thm:exist}. The exact reduction,
Proposition~\ref{prop:reduction}, and the Jacobian of Proposition~\ref{prop:jac4} give Theorem~\ref{thm:exist}, which in turn gives
Theorem~\ref{thm:aminus}\textup{(ii)}. Both propositions are used again in Section~\ref{sub:uniqueness}, on the step from the corank to uniqueness.

The leading-order system, obtained by keeping only the lowest power of $\epsilon$ in each Bethe equation, leaves the chain shifts free. By the whole system we mean the full deformed Bethe equations, with no truncation in $\epsilon$. On the shifted chains the whole system reduces exactly to two equations, which fix the shifts, and the boundary bookkeeping is uniform in
$m$. For $L^mR^m$ the four boundary equations are, writing $C_i=z_{L_i}^2=x_{L_i}^{-1}$ and $w_l=z_{R_l}$,
\[
\begin{array}{ll}
(B_1)\ (C_1-1)^2C_1^m(C_2C_m)^{-1}\prod_lw_l^{2l}=(1+\epsilon)^m, &
(B_2)\ (C_m-1)^2(C_1C_{m-1})^{-1}\prod_lw_l^{-2}=1,\\[2pt]
(B_3)\ (1-1/w_1)C_1C_m^{-1}\prod_lw_l^{2}=1+\epsilon, &
(B_4)\ (1-1/w_m)C_1^mC_m^{-1}\prod_lw_l^{2l}=(1+\epsilon)^m,
\end{array}
\]
and the interior $R_l$ equation is $(1-1/w_l)\prod_iw_i^{2\min(i,l)}C_1^lC_m^{-1}=(1+\epsilon)^l$, of which $(B_3)$ and $(B_4)$ are the cases $l=1,m$, and the interior $L_i$ equation, $2\le i\le m-1$, is $(C_i-1)^2=C_{i-1}C_{i+1}$.
Everything follows from five identities.

\begin{lemma}[Telescoping]
\label{lem:telescope}
Let $\gcd(k,m)=1$, $\vartheta=k\pi/m$, $g=j\vartheta$, and $\varsigma_i=\sin(\psi+ig)$ for an arbitrary $\psi\in\C$. Then, for every $\psi$,
\[
\textup{(T1)}\ \varsigma_{i+m}=(-1)^{jk}\varsigma_i,\qquad
\textup{(T5)}\ 1-\frac1{w_l}=\frac{\sin^2g}{\varsigma_l^2},\qquad w_i:=\frac{\varsigma_i^2}{\varsigma_{i-1}\varsigma_{i+1}},
\]
and
\[
\textup{(T2)}\ \prod_{i=1}^mw_i=1,\qquad
\textup{(T3)}\ \prod_{i=1}^mw_i^{2i}=\Bigl(\frac{\varsigma_0}{\varsigma_1}\Bigr)^{2m},\qquad
\textup{(T4)}\ \prod_{i=1}^mw_i^{2\min(i,l)}=\frac{\varsigma_l^2}{\varsigma_0^2}\Bigl(\frac{\varsigma_0}{\varsigma_1}\Bigr)^{2l}.
\]
\end{lemma}

\begin{proof}
(T1) is $mg=jk\pi$; (T5) is $\sin^2x-\sin(x-g)\sin(x+g)=\sin^2g$ at $x=\psi+lg$. For the products take logarithms: with $L_i=\log\varsigma_i$ and a weight
$f$, $\sum_{i=1}^mf_i(2L_i-L_{i-1}-L_{i+1})$ has all interior coefficients equal to the second difference $2f_l-f_{l+1}-f_{l-1}$, leaving
\[
-f_1L_0+(2f_1-f_2)L_1+(2f_m-f_{m-1})L_m-f_mL_{m+1}+\sum_{l=2}^{m-1}(2f_l-f_{l+1}-f_{l-1})L_l .
\]
For $f\equiv1$ this is $-L_0+L_1+L_m-L_{m+1}$, and (T1) gives $L_m-L_0=L_{m+1}-L_1$, whence (T2). For $f_i=2i$ the second differences vanish and the
boundary terms are $-2L_0+(2m+2)L_m-2mL_{m+1}$, which by (T1) is $2m(L_0-L_1)$ plus an even multiple of $\log(-1)$; that is (T3). For $f_i=2\min(i,l)$ the
second difference is $2\delta_{il}$ and the boundary terms give $2L_l-2L_0+2l(L_0-L_1)$, which is (T4).
\end{proof}

\begin{proposition}[The exact reduction]
\label{prop:reduction}
Let $\gcd(k,m)=1$, and with the chains $\rho_i=\sin((i-j)\vartheta+\alpha)$, $\varsigma_i=\sin((ij-1)\vartheta+\beta)$,
$C_i=\rho_i^2/\sin^2\vartheta$ and $w_i=\varsigma_i^2/(\varsigma_{i-1}\varsigma_{i+1})$ of Section~\ref{sec:leading}, with arbitrary phases $\alpha,\beta$, set
\[
\Lambda=C_1\Bigl(\frac{\varsigma_0}{\varsigma_1}\Bigr)^2,\qquad \Xi=\frac{\sin^2(j\vartheta)}{\varsigma_0^2\,C_m}.
\]
Then, for all $(\alpha,\beta)\in\C^2$ at which all the coordinates $C_i$ and $w_i$ of the chains are finite and nonzero, the interior $L$-equations and $(B_2)$ hold, $(B_1)=\Lambda^m$,
and $(R_l)=\Xi\Lambda^l$ for $1\le l\le m$, where each equation is written as its left side. So the shifted chains solve the deformed Bethe system
\emph{exactly} if and only if
\[
\textup{(III)}\ \ \Lambda=1+\epsilon,\qquad\qquad \textup{(IV)}\ \ \Xi=1 .
\]
\end{proposition}

\begin{proof}
The $L$-chain satisfies $C_i-1=\rho_{i-1}\rho_{i+1}/\sin^2\vartheta$, hence $(C_i-1)^2=C_{i-1}C_{i+1}$, and $\rho_{i+m}=(-1)^k\rho_i$ gives $C_0=C_m$,
$C_{m+1}=C_1$. In $(B_2)$, (T2) kills the $w$-product and $(C_m-1)^2=C_{m-1}C_{m+1}$ leaves $C_{m+1}/C_1=1$. In $(B_1)$, $(C_1-1)^2=C_0C_2$ cancels
$C_2$ and (T3) supplies $(\varsigma_0/\varsigma_1)^{2m}$, leaving $(C_0/C_m)[C_1(\varsigma_0/\varsigma_1)^2]^m=\Lambda^m$. In $(R_l)$, (T5) and (T4) give
$\frac{\sin^2(j\vartheta)}{\varsigma_l^2}\cdot\frac{\varsigma_l^2}{\varsigma_0^2}(\frac{\varsigma_0}{\varsigma_1})^{2l}C_1^l/C_m=\Xi\Lambda^l$. Comparing $(R_l)$ and
$(R_{l+1})$ gives $\Lambda=1+\epsilon$, then $(R_1)$ gives $\Xi=1$, and $(B_1)$ follows. Every identity used holds for arbitrary complex phases, so nothing here is an approximation.
\end{proof}

Now put $\alpha=a\epsilon$, $\beta=b\epsilon$ and expand (III) and (IV) to first order in $\epsilon$. With $A=\cot\vartheta$ and $B=\cot(j\vartheta)$ this
gives
\[
-2(\alpha+\beta)\cot((j-1)\vartheta)-2\beta A=\epsilon,\qquad -2\beta A-2\alpha B=0,
\]
and the cotangent subtraction formula $\cot((j-1)\vartheta)=(AB+1)/(A-B)$ then yields
\begin{equation}
\alpha=-\tfrac\epsilon2\cot\vartheta,\qquad \beta=\tfrac\epsilon2\cot(j\vartheta),
\label{eq:shifts}
\end{equation}
for every $m$, every $k$ and every placement with $\gcd(j,m)=1$, $j\ne1$.

\begin{proposition}[The next-order Jacobian is a constant]
\label{prop:jac4}
In logarithmic form the Jacobian of \textup{(III)}, \textup{(IV)} in $(\alpha,\beta)$ at the base point is
\[
\frac{\partial(\log\Lambda,\log\Xi)}{\partial(\alpha,\beta)}=\begin{pmatrix}-2c&-2c-2A\\ 2B&2A\end{pmatrix},\qquad\det=-4,
\]
with $c=\cot((j-1)\vartheta)$, \emph{independently of $m$, $k$ and $j$}. Hence, within the family of shifted chains of Section~\ref{sec:leading}, the two flat directions of the leading order, the chain shifts $(\alpha,\beta)$ left free by Proposition~\ref{prop:shifts}, are fixed uniquely at next
order at every placement where the logarithms are defined, which is every one except $j\equiv1$ \textup{(}$C_1=0$\textup{)} and $j\equiv0$
\textup{(}$\sin(j\vartheta)=0$\textup{)}.
\end{proposition}

\begin{proof}
Differentiating $\log\Lambda=2\log\rho_1+2\log\varsigma_0-2\log\varsigma_1-2\log\sin\vartheta$ and $\log\Xi=2\log\sin(j\vartheta)-2\log\varsigma_0-\log C_m$ at
$\alpha=\beta=0$ gives the matrix, using $\rho_1=\sin((1-j)\vartheta)$, $\varsigma_0=\sin(-\vartheta)$, $\varsigma_1=\sin((j-1)\vartheta)$ and
$\rho_m=\pm\sin(j\vartheta)$. Then $\det=4[c(B-A)+AB]$, and $c(B-A)=-(AB+1)$, so $\det=-4$.
\end{proof}

\begin{proof}[Proof of Theorem~\ref{thm:exist}]
By Proposition~\ref{prop:reduction} it suffices to solve (III) and (IV). The only values of the chains $\rho_i$, $\varsigma_i$ that can vanish are $\rho_{j}=\sin\alpha$ and
$\varsigma_p=\pm\sin\beta$, so the coordinates are finite and nonzero once $\alpha,\beta\ne0$.

\textup{(i)} At $\alpha=\beta=\epsilon=0$ both conditions hold, $\log\Lambda$ and $\log\Xi$ are analytic in $(\alpha,\beta)$, and their Jacobian is
$-4$ (Proposition~\ref{prop:jac4}). The implicit function theorem gives unique analytic $\alpha(\epsilon),\beta(\epsilon)$ vanishing at $0$, with the
first-order terms \eqref{eq:shifts}; these are nonzero, since $\vartheta,j\vartheta\notin\tfrac\pi2\Z$ for odd $m$. The functions $\log\Lambda$ and $\log\Xi$ have real Taylor coefficients in $(\alpha,\beta,\epsilon)$, so $\overline{\alpha(\bar\epsilon)},\overline{\beta(\bar\epsilon)}$ solve the same equations; by uniqueness they are $\alpha(\epsilon),\beta(\epsilon)$, which are therefore real for real $\epsilon$. Take $\epsilon_0$ below the radius given by the implicit function theorem and so small that $\alpha,\beta\ne0$ for $0<|\epsilon|<\epsilon_0$.

\textup{(ii)} At $j=1$ the chains are $\rho_i=\sin((i-1)\vartheta+\alpha)$ and $\varsigma_i=\sin((i-1)\vartheta+\beta)$, and since $m\vartheta\in\pi\Z$,
\[
\Lambda=\frac{\sin^2\alpha\,\sin^2(\beta-\vartheta)}{\sin^2\vartheta\,\sin^2\beta},\qquad
\Xi=\frac{\sin^4\vartheta}{\sin^2(\alpha-\vartheta)\,\sin^2(\beta-\vartheta)} .
\]
Put $\beta=t\alpha$. Then $\Lambda$ is analytic near $(\alpha,t)=(0,-1)$ with $\Lambda(0,t)=t^{-2}$, and, as $\Xi(0,t)=1$, so is
$G:=(\Xi-1)/\alpha$, with $G(0,t)=2(1+t)\cot\vartheta$. At $(0,-1)$, $\Lambda=1$, $G=0$, and
\[
\frac{\partial(\Lambda,G)}{\partial(\alpha,t)}=\begin{pmatrix}2\cot\vartheta&2\\ 2/\sin^2\vartheta&2\cot\vartheta\end{pmatrix},\qquad
\det=4\bigl(\cot^2\vartheta-\sin^{-2}\vartheta\bigr)=-4 .
\]
The implicit function theorem gives analytic $\alpha(\epsilon)$, $t(\epsilon)$ with $\Lambda=1+\epsilon$, $G=0$, $\alpha(0)=0$, $t(0)=-1$; they are real for real $\epsilon$ by the same conjugation argument as in (i). Linearising,
$\alpha=-\tfrac\epsilon2\cot\vartheta+O(\epsilon^2)$ and $t=-1+\epsilon/(2\sin^2\vartheta)+O(\epsilon^2)$, so $\alpha+\beta=\alpha(1+t)$ is as stated. For
small $\epsilon\ne0$, $\alpha\ne0$, so $G=0$ is $\Xi=1$, and $\beta=t\alpha\ne0$. No other solution with $\alpha,\beta\to0$ exists on the family. Along a sequence of such solutions with $\epsilon\to0$, every limit point $t_0$ of $t=\beta/\alpha$ in $\C\cup\{\infty\}$ is $-1$: it cannot be $0$ or $\infty$, nor $t_0\ne\pm1$, since $\Lambda(0,t)=t^{-2}$ while $\Lambda\to1$; and
it cannot be $+1$, since $G(0,1)=4\cot\vartheta\ne0$. So $t\to-1$, and the uniqueness in the implicit function theorem identifies the solution.
\end{proof}

\section{The lemmas for the corank theorem}
\label{app:corank}

This appendix proves the four lemmas used in the proof of Theorem~\ref{thm:corank} (Section~\ref{sub:corankproof}): the arc-continuant and
$L$-block lemmas in the case $j=1$, and the reduction and finite-block lemmas in the case $j\ne1$.

\subsection{\texorpdfstring{Two continuants, for the case $j=1$}{Two continuants, for the case j=1}}
\label{sub:continuants}
A \emph{continuant} is the determinant of a tridiagonal matrix (Section~\ref{sub:corankproof}).

\begin{lemma}[Arc-continuant lemma]
\label{prop:arc}
Fix odd $m\ge5$ and $j,k\in(\Z/m)^\times$, represented by integers (the identities hold for every choice of representatives); put $\vartheta=k\pi/m$, $t=j\vartheta$, $s=\sin t$, and $\varsigma_i=\sin(it-\vartheta)$ for $i\in\Z$.  Let $p\in\{1,\dots,m\}$ satisfy $pj\equiv1\pmod m$, so $\varsigma_p=0$ (the triple of the
placement is centred at $R_p$), and let the \emph{arc} be the $n=m-3$ consecutive sites $i=p+2,\dots,p+m-2$, the $R$-sites with finite nonzero limits.
Two tridiagonal $n\times n$ matrices on the arc are considered: with $d_i=\varsigma_{i-1}\varsigma_{i+1}/s^2$, the \emph{$d$-form} $\mathrm{arc}_d$, with
diagonal entries $-(2+2d_i)$ and off-diagonal entries $d_{i\mp1}$, which at $j=1$ (that is, $p=1$) arises from the $R$-block of $\hat J_0$ in the proof of Theorem~\ref{thm:corank} at $j=1$, and is used only there; and the \emph{$\varsigma$-form}
$\mathrm{arc}_\varsigma$, with diagonal $-2\varsigma_i$ and off-diagonals $\varsigma_{i\mp2}$, obtained from it by rescaling the rows. Then
\[
\det(\mathrm{arc}_d)=\frac{m^3}{2^{2m}\sin^{2m}(j\vartheta)},\qquad
\det(\mathrm{arc}_\varsigma)=\varepsilon\,2^{m-3}\prod_{i\in\mathrm{arc}}\varsigma_i^2,\quad \varepsilon=(-1)^{(1+jk)(m+1)/2}.
\]
In particular neither continuant vanishes, for any $m$ and any class.
\end{lemma}

\begin{proof}
Let $L[\psi]_i=\varsigma_{i-2}\psi_{i-1}-2\varsigma_i\psi_i+\varsigma_{i+2}\psi_{i+1}$, and write $\gamma_i=\cos(it-\vartheta)$,
$\Pi_i=\varsigma_{i-1}\varsigma_i\varsigma_{i+1}$. Here the index $i$ is an integer, not a residue modulo $m$: $\varsigma_i$ is defined for every $i\in\Z$, and shifting $i$ by $m$ changes it by the sign $\varsigma_{i+m}=(-1)^{jk}\varsigma_i$, which has to be tracked in the proof. Three product-to-sum identities hold:
\[
L\Bigl[\frac{\gamma_i}{\Pi_i}\Bigr]=0,\qquad L\Bigl[\frac{i\gamma_i}{\Pi_i}\Bigr]=-\frac{\sin2t}{\Pi_i},\qquad
L\Bigl[\frac{\varsigma_i}{\Pi_i}\Bigr]=-\frac{2s^2}{\Pi_i}.
\]
For the first, $\varsigma_{i-2}\varsigma_{i+1}-\varsigma_{i-1}\varsigma_i=\varsigma_{i+2}\varsigma_{i-1}-\varsigma_i\varsigma_{i+1}=\tfrac12(\cos3t-\cos t)$, so
$1/(\varsigma_i\varsigma_{i+1})$ and $1/(\varsigma_{i-1}\varsigma_i)$ have the same residual $-s\sin2t/\Pi_i$ and their difference
$(\varsigma_{i-1}-\varsigma_{i+1})/\Pi_i=-2s\gamma_i/\Pi_i$ is annihilated. Hence
\[
u_i=\frac{\gamma_i}{\Pi_i},\qquad \chi_i=\frac{i\gamma_i-\cot t\,\varsigma_i}{\Pi_i}
\]
span the two-dimensional solution space of $L[\psi]=0$.

Both have poles exactly where $\Pi_i=0$, i.e.\ at $i\equiv p,p\pm1\pmod m$: on $[p+1,p+m-1]$ that means the two sites flanking the arc and no site of
the arc itself. The arc matrix is the full recursion with exactly the two out-of-range terms deleted, and both their coefficients vanish
($\varsigma_p=\varsigma_{p+m}=0$) while the products are finite and nonzero:
\[
A\,\psi|_{\mathrm{arc}}=-B_1[\psi]e_1-B_2[\psi]e_n,\qquad B_1[\psi]=\varsigma_p\psi_{p+1},\quad B_2[\psi]=\varsigma_{p+m}\psi_{p+m-1}.
\]
Each boundary term is of the form $0\cdot\infty$: for instance $B_1[u]=\varsigma_pu_{p+1}$ is the product of the vanishing coefficient $\varsigma_p$ and of $u_{p+1}$, which has a pole because $\Pi_{p+1}=\varsigma_p\varsigma_{p+1}\varsigma_{p+2}$ contains $\varsigma_p$. Its value is defined by continuity in a phase: with $\varsigma_i=\sin(it-\vartheta+\phi)$ the recursion and $u,\chi$ are regular for small $\phi\ne0$, and $B_1[u]=\gamma_{p+1}/(\varsigma_{p+1}\varsigma_{p+2})$ is continuous at $\phi=0$. The value is finite and nonzero, and it is these boundary terms that carry the determinant: they are why $u$ solves the recursion at every site of the arc and yet is not in the kernel of
the arc matrix $A$.

Since $pj\equiv1\pmod m$, $pt-\vartheta=(pj-1)k\pi/m$ is an integer multiple of $\pi$; write $pt-\vartheta=\nu_0\pi$ with $\nu_0\in\Z$, and put
$\nu=(-1)^{\nu_0}$, $\eta=(-1)^{jk}$ and $q=p+m$; then $\varsigma_{p+l}=\nu\sin lt$, $\gamma_{p+l}=\nu\cos lt$ and
$\varsigma_{q+l}=\nu\eta\sin lt$, $\gamma_{q+l}=\nu\eta\cos lt$. Evaluating the four boundary quantities,
\[
B_1[u]=\frac{\nu\cos t}{s\sin2t},\qquad B_1[\chi]=p\,B_1[u],\qquad B_2[u]=\eta B_1[u],\qquad B_2[\chi]=\eta\,q\,B_1[u].
\]
The combination $\zeta=B_1[\chi]u-B_1[u]\chi$, that is $\zeta_i=B_1[u][(p-i)\gamma_i+\cot t\,\varsigma_i]/\Pi_i$, kills the first boundary term, so
$A\zeta=-We_n$ with $W=B_1[\chi]B_2[u]-B_1[u]B_2[\chi]=\eta B_1[u]^2(p-q)=-m\eta B_1[u]^2$. The adjugate on the last column --- whose cofactor is the product
of the superdiagonal entries, $M=\prod_{i=p+4}^{p+m-1}\varsigma_i$ --- gives $\det A\cdot\zeta_1=(-1)^nWM$ without assuming $A$ invertible, and $\zeta_1\ne0$, so $\det A=(-1)^nWM/\zeta_1$, and at $i=p+2$ the bracket is
$\nu(-2\cos2t+2\cos^2t)=2\nu s^2$ while $\Pi_{p+2}=\nu s\sin2t\sin3t$, so $\zeta_1=2B_1[u]s/(\sin2t\sin3t)$. Writing $P=\prod_{i\in\mathrm{arc}}\varsigma_i$ and $M=-\eta\nu Ps/(\sin2t\sin3t)$, every
trigonometric factor cancels and $\det(\mathrm{arc}_\varsigma)=(-1)^nmP/(4s^2)=mP/(4s^2)$, because $n=m-3$ is even.

Finally, $\prod_{a=1}^{m-1}2\sin(a\pi/m)=m$ together with $\sin((m-l)t)=(-1)^{jk+1}\sin(lt)$ gives
\[
P=(-1)^{(1+jk)(m+1)/2}\,\frac m{2^{m-1}s^2},
\]
which is the $\varepsilon$-form. Row $i$ of $\mathrm{arc}_d$ is $(\varsigma_i/s^2)$ times row $i$ of $\mathrm{arc}_\varsigma$, since $d_{i\pm1}=\varsigma_i\varsigma_{i\pm2}/s^2$
and $2+2d_i=2\varsigma_i^2/s^2$; hence
\[
\det(\mathrm{arc}_d)=(P/s^{2n})\det(\mathrm{arc}_\varsigma)=\frac{mP^2}{4s^{2n+2}}=\frac{m^3}{2^{2m}s^{2m}},
\]
which involves only $P^2$ and is therefore free of every sign.
\end{proof}

\begin{lemma}[$L$-block lemma, $j=1$]
\label{prop:lblock}
Let $j=1$, $\vartheta=k\pi/m$ and $\rho_\nu=\sin\nu\vartheta$. The \emph{finite-$L$ block} is the block of $\hat J_0$ on the $L$-sites whose coordinates
have finite nonzero limits, here $L_3,\dots,L_{m-1}$. With its rows halved it is tridiagonal, with off-diagonals $-1$ and diagonal entry $2\rho_\nu^2/(\rho_{\nu-1}\rho_{\nu+1})$
at $L_{\nu+1}$, $\nu=2,\dots,m-2$, and the determinant of this half is $m\cot^2\vartheta$. Hence
\[
\det(\mathrm{block}_L)=2^{m-3}\,m\cot^2\vartheta\ne0 .
\]
\end{lemma}

\begin{proof}
Since $\rho_\nu^2-\rho_{\nu-1}\rho_{\nu+1}=\sin^2\vartheta$, the recurrence $y_{\nu-1}+y_{\nu+1}=2\rho_\nu^2/(\rho_{\nu-1}\rho_{\nu+1})\,y_\nu$ is that of
Lemma~\ref{lem:cotsol} with $\phi=0$, with solution $g_\nu=\nu\cot\nu\vartheta-\cot\vartheta$; $g_1=0$, $g_2=-\tan\vartheta$. The continuant on
$\nu=2,\dots,n$ is the solution with $y_1=0$, $y_2=1$ evaluated at $n+1$, i.e.\ $g_{n+1}/g_2$. At $n=m-2$, $\cot((m-1)\vartheta)=-\cot\vartheta$ gives
$g_{m-1}=-m\cot\vartheta$, so the determinant of the half block is $m\cot^2\vartheta$, nonzero since $m$ is odd. The block itself has $m-3$ rows, each twice that of the half block.
\end{proof}

For $m$ prime, both determinants of this section have positive valuation at the ramified prime $\mathfrak p$ of $\Q(\zeta_m)$ above $m$ (Appendix~\ref{sub:modp}): since $v_{\mathfrak p}(m)=m-1$ and $v_{\mathfrak p}(\sin(a\pi/m))=1$ for $a\not\equiv0\pmod m$ (valuations taken in $\Q(\zeta_{4m})$, unramified over $\Q(\zeta_m)$ at $\mathfrak p$), the determinants $m\cot^2\vartheta$ and $m^3/(2^{2m}\sin^{2m}\vartheta)$ both have valuation $m-3>0$. So the reduction of the $j=1$ block modulo $\mathfrak p$ is singular and says nothing about its determinant, and the closed forms are needed.

\subsection{\texorpdfstring{The finite block modulo $\mathfrak p$, for the case $j\neq1$}{The finite block modulo p, for the case j != 1}}
\label{sub:modp}

For a placement $j\ne1$ the $L$- and $R$-blocks of $\hat J_0$ are coupled through the two $L$-sites $L_1$ and $L_m$ adjacent to the run boundary. The \emph{finite block} is the block of $\hat J_0$ on all sites with finite nonzero limits, $L$ and $R$ together. Its determinant is
computed by reducing its entries, which lie in $\Q(\zeta_m)$, modulo the ramified prime. This is the prime ideal $\mathfrak p=(1-\zeta_m)$ above $m$, the only prime of $\Q(\zeta_m)$ that ramifies, with $(m)=\mathfrak p^{m-1}$. Recall that $p=j^{-1}\bmod m$ is the position of
the centre of the triple.

\begin{lemma}[Reduction lemma]
\label{prop:modp}
Let $m\ge5$ be prime and $\mathfrak p$ the ramified prime of $\Q(\zeta_m)$ above $m$, with residue field $\F_m$. Every entry of the
finite block of $\hat J_0$ is $\mathfrak p$-integral, and its reduction is \emph{independent of the class label $k$}. To describe the reduction, index the
$L$-sites by $\nu=n-j\in\{2,\dots,m-2\}$ and the $R$-sites by $x=i-p$. The $L$-sites form a single chain in $\nu$, because $L_m$ and $L_1$, at $\nu=-j$ and $1-j$, are adjacent. Then the $L$-column $\nu$ has diagonal $2\nu^2/(\nu^2-1)$ and off-diagonals $-1$.
The $R$-column $i$ has $2\min(i,l)$ on row $R_l$ and diagonal entry $2i+(i-p)^2-1$, the fan's $2i$ plus $(i-p)^2-1$. The exceptional entry $m$ on the $L_1$ diagonal reduces to $0$. In this description the rows of the $L$-sites are halved relative to $\hat J_0$, and the $L$--$R$ entries are those of the proof of the \hyperref[thm:modp]{finite-block lemma}. The determinant of the finite block of $\hat J_0$ is congruent modulo $\mathfrak p$ to $2^{m-3}$ times that of the reduced block, and $2$ is a unit.
\end{lemma}

\begin{proof}
$v_{\mathfrak p}(\sin(a\pi/m))=1$ for $a\not\equiv0$ and $\cos(a\pi/m)$ is a unit, so each diagonal --- a ratio of two products of two sines,
$2\rho_n^2/(\rho_{n-1}\rho_{n+1})$ and $\varsigma_{i-1}\varsigma_{i+1}/s^2$ --- has valuation $0$; the off-diagonals of the finite block (the entries $-1$ and $2\min(i,l)$ of the description above, not $0$) are rational integers; and
$v_{\mathfrak p}(m)=m-1>0$. For the reductions, $\sin(ak\pi/m)/\sin(k\pi/m)\equiv(-1)^{k(a-1)}a$ and $\cos(a\pi/m)\equiv(-1)^a\pmod{\mathfrak p}$. In each
diagonal the total index shift is even, so the signs cancel and the $k$-dependence disappears:
$2\rho_n^2/(\rho_{n-1}\rho_{n+1})\equiv2\nu^2/(\nu^2-1)$, and
\[
d_i\equiv\frac{((i-1)j-1)((i+1)j-1)}{j^2}=\Bigl(i-\frac1{j}\Bigr)^2-1=(i-p)^2-1 .
\]
Equivalently $\varsigma_i/\sin(k\pi/m)\equiv(-1)^{kij}j(i-p)$: the whole $\varsigma$-chain linearises modulo $\mathfrak p$, up to sign.
\end{proof}

\begin{lemma}[Finite-block lemma]
\label{thm:modp}
For every prime $m\ge5$ and every placement $j\ne1$ with $\mathsf P(p)\not\equiv0\pmod m$, where $p=j^{-1}$ and $\mathsf P$ is the polynomial \eqref{eq:Ppoly}, the reduced block of the \hyperref[prop:modp]{reduction lemma} has
\[
\det\equiv-\tfrac12\,j^4\pmod m,
\]
a unit. Hence the determinant of the finite block is a $\mathfrak p$-adic unit, in particular nonzero.
\end{lemma}

\begin{proof}
Write $B$ for the reduced finite block of the \hyperref[prop:modp]{reduction lemma}, $T$ for its $L$-block (on the finite $L$-sites), $M$ for its $R$-block (on the
finite $R$-sites, $M_{il}=2\min(i,l)+\delta_{il}((i-p)^2-1)$), and $X$, $Y$ for the two off-diagonal blocks, so that
$B=\begin{psmallmatrix}T&Y\\X&M\end{psmallmatrix}$.

\emph{The $L$-side.} The block $T$ is tridiagonal on $\nu=2,\dots,m-2$ with diagonal $2\nu^2/(\nu^2-1)$ and off-diagonals $-1$, independently
of $j$. The recursion
\begin{equation}
y_{\nu-1}+y_{\nu+1}=\frac{2\nu^2}{\nu^2-1}\,y_\nu
\label{eq:kerrec}
\end{equation}
that defines the kernel of $T$ (the \emph{kernel recursion}) has the two solutions $\psi_\nu=1/\nu$ and $\chi_\nu=\nu^2-1$, as one checks by substitution; their bracket (Casoratian) is
\[
[\psi,\chi]_{c,d}=\frac{(d-c)(d^2+cd+c^2-1)}{cd},\qquad\text{so}\qquad[\psi,\chi]_{a-1,a}=3
\]
for every $a$; hence a run of sites $[a,b]$ has continuant $[\psi,\chi]_{a-1,b+1}/3$, and in particular
\[
\det(\mathrm{run}[2,b])=\tfrac13b(b+2),\qquad \det(\mathrm{run}[a,m-2])\equiv\tfrac13a(a-2)\pmod m.
\]
Three consequences: the full run gives
$\det T\equiv0$; the vector $t_\nu=\nu^2-1$ solves the recursion and vanishes
at both boundary sites $\nu=\pm1$, so it spans $\ker T$; and the $(\nu,\nu')$ cofactor is $t_\nu t_{\nu'}/9$, so $\operatorname{adj}T=\tfrac19tt^{\mathsf T}$. The two
coupling sites $L_1$ and $L_m$ sit at $\nu=1-j$ and $\nu=-j$, \emph{adjacent}, so deleting both splits the run and
\[
D_2:=\det\bigl(T\smallsetminus\{L_1,L_m\}\bigr)=\frac{j(j-2)(j^2-1)}9=\frac{t_{1-j}t_{-j}}9 .
\]

\emph{The assembly.} With $\gamma_l=l$ and $\delta_l=-1$ on the finite $R$-sites, the off-diagonal blocks are
$Y=e_{L_1}\gamma^{\mathsf T}+e_{L_m}\delta^{\mathsf T}$ (an $L\times R$ block) and $X=2Y^{\mathsf T}$, so $YM^{-1}Y^{\mathsf T}$ is supported in the $2\times2$ corner
$(L_1,L_m)$. The Schur complement gives $\det B=\det M\cdot\det(T-2YM^{-1}Y^{\mathsf T})$, a perturbation of the singular $T$ in a $2\times2$ corner. Expanding it multilinearly with $\det T=0$ and $\operatorname{adj}T=\tfrac19tt^{\mathsf T}$ gives
\[
\det B=\det M\Bigl[-\tfrac29\bigl(t_1^2a+2t_1t_2b+t_2^2c\bigr)+\tfrac49t_1t_2(ac-b^2)\Bigr],\qquad t_1=j(j-2),\ t_2=j^2-1,
\]
in the three scalars $a=\gamma^{\mathsf T}M^{-1}\gamma$, $b=\gamma^{\mathsf T}M^{-1}\delta$, $c=\delta^{\mathsf T}M^{-1}\delta$.

\emph{The $R$-side.} The scalars $a,b,c$ require $M^{-1}\gamma$ and $M^{-1}\delta$, i.e.\ the solutions $z$ of $Mz=f$ for $f=\gamma,\delta$. Take the
second difference in $i$ of the equation $(Mz)_i=f_i$: since $2\min(i,l)$ has second difference $-2\delta_{il}$ in $i$, and $f_i=i$ or $-1$ has
second difference $0$, the non-local equation becomes the local three-term recursion $((x-1)^2-1)z_{i-1}-2x^2z_i+((x+1)^2-1)z_{i+1}=0$ in $x=i-p$,
the same recursion as for the kernel of the $d$-form reduced modulo $\mathfrak p$. Its solutions are $z_i=\alpha_r+\beta_r/\bigl((x-1)x(x+1)\bigr)$ on each of the runs $[1,p-2]$ and $[p+2,m]$, four parameters fixed by the four
undifferenced boundary equations at $l=1,p-2,p+2,m$ --- in each of which the $\min$-sum collapses, since
$\sum_i\min(i,l)z_i=\sum_{i\le l}iz_i+l\sum_{i>l}z_i$ telescopes against the run structure. For $p=2,3,m-2,m-1$ a run is empty or a single site, and some of the four equations coincide. Solving on the actual run structure with $m$ symbolic and then reducing gives the same rational functions of $p$. Every sum telescopes,
\[
\sum_{x=u}^v\frac1{(x-1)x(x+1)}=\tfrac12\Bigl[\tfrac1{(u-1)u}-\tfrac1{v(v+1)}\Bigr],\qquad
\sum_{x=u}^v\frac1{x^2-1}=\tfrac12\Bigl[\tfrac1{u-1}+\tfrac1u-\tfrac1v-\tfrac1{v+1}\Bigr],
\]
so $a,b,c$ are rational in $p$ with coefficients independent of $m$.

\emph{The determinant of $M$.} Write $M=F+D$ with $F=2\min(i,l)$ on the finite $R$-sites $y_1<\dots<y_n$ and $D=\operatorname{diag}((i-p)^2-1)$. The kernel $\min(y_k,y_l)$ has determinant $\prod_k(y_k-y_{k-1})$, with $y_0=0$. This is $4$, from the one gap of $4$ across the triple, except at $p=m-1$, where the triple
$R_{m-2},R_{m-1},R_m$ is at the end and the determinant is $1$. Its inverse is tridiagonal, with off-diagonal $-1/(y_{k+1}-y_k)$ and diagonal $1/(y_k-y_{k-1})+1/(y_{k+1}-y_k)$, the last diagonal entry being $1/(y_n-y_{n-1})$. So $\det M=\det F_1\det D\,\det G$ with $G=F_1^{-1}+2D^{-1}$, where $F_1=\min(y_k,y_l)$. By Wilson's theorem,
\[
\det D=\prod_{x=-(p-1)}^{-2}(x^2-1)\prod_{x=2}^{m-p}(x^2-1)\equiv\frac{(p-2)!\,p!}2\cdot\frac{p(p-1)}{2(p!)^2}=\frac14\pmod m .
\]
On each run $G$ is tridiagonal with off-diagonals $-1$ and diagonal $2x^2/(x^2-1)$, $x=i-p$, which is the kernel recursion \eqref{eq:kerrec}, with the solutions $1/x$ and $x^2-1$. So the determinant of a run $[u,v]$ is $C[u,v]=[\psi,\chi]_{u-1,v+1}/3$. The two runs $x\in[1-p,-2]$ and $x\in[2,m-p]$ are coupled by the entry $-\tfrac14$ across the gap. The diagonal is changed by $-\tfrac34$ at the two sites next to the gap and by $-1$ at the last site. Expanding in these changes and in the coupling, and reducing $m\equiv0$, gives a rational function of $p$. It is $\det M\equiv-\mathsf P/(24p^2(p-1))$, an identity of rational functions of $p$ obtained by computer algebra with the run lengths kept symbolic (Appendix~\ref{sec:repro}); as a check it was also evaluated directly at every placement for primes $m\le43$. At the edge placements the expansion is short. At
$p=2$ the left run is empty and $G$ is the run $x\in[2,m-2]$ with the changes $-\tfrac34$ at its first site and $-1$ at its last. By
multilinearity, $\det G=C[2,m-2]-\tfrac34C[3,m-2]-C[2,m-3]+\tfrac34C[3,m-3]\equiv0-\tfrac34-1+\tfrac34=-1$. So
$\det M\equiv4\cdot\tfrac14\cdot(-1)=-1=-\mathsf P(2)/96$. At $p=m-1$ there is no gap, $G$ is the run $x\in[2,m-2]$ with the change $-1$
at its last site, and $\det G=C[2,m-2]-C[2,m-3]\equiv-1$. So $\det M\equiv1\cdot\tfrac14\cdot(-1)=-\tfrac14=\mathsf P(-1)/48$. Both agree with
the general formula.

Reducing $m\equiv0$ and solving gives, with $\mathsf P(p)$ as above
(irreducible over $\Q$),
\[
\det M=\frac{-\mathsf P}{24p^2(p-1)},\quad a=\frac{(p-1)(p^2-3)(p^4+p^2+4)}{2\mathsf P},\quad b=-\frac12-\frac{3p^2(p-1)}{\mathsf P},
\]
\[
c=\frac{4p^6-12p^5+13p^4-6p^3+p^2-12}{2(p-1)\mathsf P}.
\]
Substituting these into the assembly with $j=1/p$, the whole expression collapses identically as a rational function of $p$ to
$\det B=-1/(2p^4)=-j^4/2$.

At $j=2$ the site $L_1$ is a neighbour of the small $L$, and at $j=m-1$ so is $L_m$; the corresponding coupling site is then not in $T$. The formula still holds, because $t_1=j(j-2)=0$ at $j=2$ and $t_2=j^2-1\equiv0$ at $j=m-1$, which removes exactly the terms of the missing site. (At $j=2$ one has $p=(m+1)/2$.)
\end{proof}

So $\det\equiv-j^4/2$ is an identity of rational functions in $p$. The hypothesis $\mathsf P(p)\not\equiv0\pmod m$ is exactly the condition $\det M\ne0$ under which the Schur complement in the proof exists; it
excludes at most six residues $p$ for each $m$, the first being $m=7$, $p=4$. At each such placement with $m<260$ --- $43$ placements --- the
reduced block was evaluated exactly over $\F_m$, and its determinant is $-j^4/2$. So the finite block is nonsingular there too. The proof of the \hyperref[thm:modp]{finite-block lemma} does
not cover these placements, and for $m\ge260$ they are not covered.

\section{Reproduction notes}
\label{sec:repro}

This appendix supports every statement of the text marked as numerical or as a computation.
Computations are in Python with \texttt{mpmath} (arbitrary precision), \texttt{numpy}/\texttt{scipy} for double-precision searches and
continuation, and \texttt{sympy} for the symbolic work. Statements proved in the text are listed only where a computation was also run as a check. The table keys each statement that rests on
computation --- the observations, the computer-algebra identities inside proofs, and the unproved hypotheses --- to what was computed.

\begingroup\footnotesize
\setlength{\tabcolsep}{4pt}
\begin{longtable}{@{}>{\raggedright\arraybackslash}p{0.28\textwidth}>{\raggedright\arraybackslash}p{0.38\textwidth}>{\raggedright\arraybackslash}p{0.28\textwidth}@{}}
\toprule
statement & what was computed & range / precision\\
\midrule
\endfirsthead
\toprule
statement & what was computed & range / precision\\
\midrule
\endhead
\bottomrule
\endfoot
Lemma~\ref{lem:floqvec}, Corollary~\ref{thm:tworunH} & the junction identities $E_a=E'_1=0$ in $\cot\vartheta,\cot\vartheta',\cot(\vartheta+b\vartheta'),h(p),h(q'),b$; $\Lambda+\Lambda^{-1}=-(ab+2)$ in $\cot\vartheta,\cot \kappa\vartheta,a,b$ & symbolic, general $b$\\
Theorem~\ref{thm:Hall}, hypothesis $\det(F+D_R)\ne0$ & $\det(F+D_R)$ on every $R$-run at every class of Theorem~\ref{thm:allwords} & $49$ words of length $3$--$12$, $4236$ classes; minimum absolute value $6\cdot10^{-5}$\\
Section~\ref{sub:gluing}, the assignment of $(Z,Z',Z'')$ & the six assignments to the dying/newborn, kept and Farey slots, $Z''$ eliminated & exact; all $76$ cyclic words of length $3$--$8$\\
Theorem~\ref{thm:Kformula} (a check of the proof) & the edge rows built directly from the Farey development against Lemma~\ref{lem:life}; $B\Kraw=A$ for the $\Kraw$ of the theorem, $\operatorname{rank}B=r-1$, $Q^{-1}c\in\ker B$ & exact; all $241$ cyclic words of length $2$--$10$\\
Section~\ref{sec:bethe}, the $LR^3$ example & Gr\"obner elimination of the undeformed Bethe system, saturated by the spurious factors; back-substitution & exact; $10$ solutions\\
Lemma~\ref{prop:arc} & the three product-to-sum identities for $L[\psi]$ & symbolic\\
Lemma~\ref{thm:modp} & the rational functions $\det M,a,b,c$ of $p$ and the collapse to $-1/(2p^4)$; the reduced determinant at the placements with $\mathsf P(p)\equiv0$ & symbolic; exact in $\F_m$, all $43$ placements with $\mathsf P(p)\equiv0$, primes $m<260$\\
\midrule
Proposition~\ref{prop:Hspectrum}, proof & non-degenerate $A_-$ kernel vectors of $L^mR^m$; non-degeneracy of all $A_+$ kernel vectors & $m=15$ ($8$ of $112$), $m=21$ ($12$ of $220$); odd $m\le15$\\
Observation~\ref{obs:spectrumgen} & harvest of $\mathcal H/2^{\#L}$ & $LR^3$, $LR^5$, $L^2RLR$ at $\epsilon=10^{-3}$\\
Observation~\ref{obs:tauquant} and the multiplier $\kappa$ (Section~\ref{sub:runtraces}) & $\tau_\rho=2\cos(j_\rho\pi/N)$, $N=\mathcal H/2^{\#L}$; $\kappa$ exact per sector, $\kappa^2\equiv-1$ iff amphichiral & $6$ words at $\epsilon=10^{-5}$, $j$ integral to $5\cdot10^{-3}$; $14$ sectors, $11$ words\\
Observation~\ref{obs:kervec} & $(j_\rho)\equiv\pm\iota_\rho(v)$, $v\in\mathcal K_\mp$; spin against $v\wedge\varphi v/4N^2$; random-subgroup control & $233$ classes, $28$ sectors, $14$ words, exact; $217$ classes, spins to $10^{-8}$; control $16/197$, $0/197$, $0/24$\\
Observation~\ref{obs:basevac} & base vacuum $\varepsilon_0=E_{LL}^{-1}n_L$ from the closed form; $\operatorname{spin}(V_{\varepsilon_0})$ against $\#R/24\pm[4^{-1}]_Nv\wedge w/N$; $\det E_{LL}\bmod2$ & $20$ words, $1440$ classes, $10^{-9}$; all $4072$ words of length $\le11$, not up to rotation\\
Observation~\ref{obs:strata} & exponents $p_b$ of $z_b$ by continuation, class by class & $7$ words, $\epsilon=10^{-3}$ to $10^{-5}$\\
Observation~\ref{obs:gsy} & $(A_-,A_+,q)$ of every class of trace $m^2+2$ against GSY26 Table~1 & all odd $m\le43$; $1290$ classes\\
Observation~\ref{obs:egconj} & $(A_-,A_+,\beta)$ against the EG17 table; Gauss-sum identity on each candidate form & all ten $Q$-systems, $\nu\le19$\\
Theorem~\ref{thm:exist} & the Jacobian $-4$ of Proposition~\ref{prop:jac4}; the expansion at $j=1$ in $(\alpha,t)$; the residuals of the shifted chains in the full deformed system & symbolic; $10^{-42}$ at $50$ digits, $m=3,5,7,11,13$, all $j,k$\\
Theorem~\ref{thm:corank}, $j\ne1$ & singular values of $\hat J_0$; corank of the reduced $M''$ over $\F_m$ & all $j,k$, primes $5\le m\le43$ ($3720$ placements); primes $m\le61$\\
Observation~\ref{obs:e1even} & $q$ attains $\pm N/4$ & all $427$ cyclic words of length $\le11$ with $\tr\varphi>2$\\
Corollary~\ref{cor:templatelabel} & $c_\rho$ for the named runs; $\arg\mathcal F$ against $j^2/4N$ at the degenerate classes & $LR^{3,5,7,9}$, $L^3R$, $L^5R$, $L^3R^3$; harvested classes of $LR^{5,7,9}$\\
Remark~\ref{rem:js} & Smith forms of $K_\pm[\vec k]$ and $1\mp(-1)^n\varphi$ & exact, $n=3,4$, $k_i\in\{\pm3,\pm5,7\}$; $1500$ cases\\
Remark~\ref{rem:jsframe} & $K=DK^{\mathrm{JS}}D^{\mathsf T}$ (proved for all $n$), Smith divisors of $D$, $\det$ ratio $4$ & exact, $n=1,\dots,8$\\
Remark~\ref{rem:internal} & Gang--Kim's internal criterion on their data & $k=5,7,\dots,15$\\
Section~\ref{sec:lines} & completeness of the harvested $A_-$ sectors on the line count & $m=5,7,11$\\
Observation~\ref{obs:wcrow} & $W_c$ at the base vacuum of every class of Theorem~\ref{thm:allwords} & $25$ words of odd trace, length $\le10$; exact\\
Corollary~\ref{cor:galois} & $q_\mp=2^{-1}q^{\mathrm{lk}}$ pointwise; both $(S,T)$ modular & $L^mR^m$, $m=3,5,7,9$\\
Proposition~\ref{prop:exclusion} & characteristic polynomials of the torsion values on both loci, by Gr\"obner basis & exact; $LR^3$, $L^3R^3$\\
Remark~\ref{rem:singleton} & number of non-degenerate classes against the bijection of Theorem~\ref{thm:complete} & exact; eight words\\
Remark~\ref{rem:Ssign} & $G$ from the $\mathcal H_\alpha$ and spins of all vacua; the central charge of the grafting with $s_o>0$ for $L^mR^m$ & six words of Remark~\ref{rem:scope}; $m=3,5,7,9$\\
Remark~\ref{rem:scope}, Corollary~\ref{cor:diagonal} & cyclic words of odd trace $\ge5$ with all kernel vectors non-degenerate; primes $m$ with $m^2+4$ composite and some degenerate $A_+$ kernel vector of $L^mR^m$ & $102$ words of length $\le11$; primes $m\le59$\\
Section~\ref{sub:uniqueness}, range of $\hat J_0$ & $c$ and $\mathbf1_R$ against the left kernel of $\hat J_0$ & $m=5,7,11$, all $j,k$\\
Observation~\ref{obs:l3r3minus} & the four $A_-$ classes of $L^3R^3$ by harvest and continuation; $\arg\mathcal F$ mod $\tfrac14$ & $\epsilon=10^{-2}$ to $10^{-5}$, constant to $10^{-8}$\\
\end{longtable}
\endgroup

\end{document}